\documentclass[aps,pre,reprint,amsmath,amssymb,longbibliography,floatfix]{revtex4-2}
\usepackage[T1]{fontenc}
\usepackage{graphicx,bm,booktabs,needspace}
\usepackage{tikz}
\usetikzlibrary{arrows.meta,positioning,fit}
\usepackage{xcolor}
\definecolor{figureblue}{HTML}{24518A}
\definecolor{figureorange}{HTML}{B85A32}
\usepackage{amsthm}
\usepackage{mathrsfs}
\usepackage[hidelinks]{hyperref}

\newtheorem{theorem}{Proposition}
\newtheorem{proposition}[theorem]{Proposition}
\newtheorem{conjecture}{Conjecture}

\newcommand{\E}{\mathbb E}
\newcommand{\one}{\bm 1}
\newcommand{\T}{\mathsf T}

\newcommand{\dB}{\mathrm{dB}}

\graphicspath{{figures/}}
\makeatletter
\newcommand{\labelalias}[2]{\@ifundefined{r@#2}{}{%
\protected@write\@auxout{}{\string\newlabel{#1}{\csname r@#2\endcsname}}}}
\makeatother
\begin{document}
\title{Protected domains and cooperative invasion on weighted networks}
\author{Abhishek Chowdhury}
\email{achowdhury@iitbbs.ac.in}
\affiliation{School of Basic Sciences, Indian Institute of Technology Bhubaneswar, India}
\begin{abstract}
Reciprocal contacts can favor cooperation by protecting a lineage until its descendants compete as a group. An individual is chosen to reproduce from the whole population, with a probability increasing with its averaged cooperation payoff. Its offspring replaces a neighbor through contacts weighted equally in both directions. We compare a single cooperator taking over defectors with the reverse invasion, using the same distribution of introduction sites. When payoffs weakly affect reproduction, every network with at least six sites and a path visiting them all admits weights favoring cooperation at any benefit--cost ratio $b/c>1$. This holds for mutants placed uniformly or arising in newborn offspring. Unequal contacts sustain aligned domains around two protected sites, converting individual reproduction into collective competition. Cooperation cannot have an advantage when its benefit does not exceed its cost, at any positive selection strength with positive reproduction rates. Where mutants arise changes which contact patterns can use this mechanism. On a five-site path, the weak-selection threshold has the sharp lower bound seven under uniform placement; newborn introduction never favors cooperation at this order. We classify the contact patterns with at most five sites that can favor cooperation under either rule; newborn introduction requires a four- or five-site cycle. Tracing copies backward to their ancestors separates domain competition from the probability of reaching it. Exact polynomial inequalities prove sharp threshold bounds for paths and the four- and five-site cycles. The same ancestry reconstructs evolution without selection on weighted paths and cycles and fixes the slowest relaxation rate on every connected reciprocal network. Cooperation thus depends on both the collective competitors created by contacts and the route by which a new lineage enters them.
\end{abstract}
\maketitle
\section{Introduction}\label{sec:intro}
An individual that pays a cost to benefit its neighbors has an immediate disadvantage. Its descendants may nevertheless create a neighborhood in which the same trait succeeds. Understanding when this change of scale is possible is a central problem in the evolution of cooperation~\cite{NowakMay1992,Nowak2006,AllenNowak2014}. It is also a problem of stochastic dynamics. A rare lineage must survive its initial surroundings, form correlations and reach configurations that persist long enough to compete. We compare the probability that one cooperator replaces a population of defectors with the probability of the reverse invasion, using the same distribution of introduction sites. This eventual replacement is called fixation. The difference between the two probabilities measures the cooperative invasion advantage; its response to selection depends on correlations that a mean density does not determine.

The order of reproduction and replacement makes this question particularly sharp. In death--birth updating an individual dies first and neighboring offspring compete for the vacancy. On large sparse unweighted regular graphs this rule favors cooperation at weak selection, to first order in selection strength, when $b/c$ exceeds the degree, with $b$ the benefit conferred by cooperation and $c$ its cost~\cite{OhtsukiEtAl2006}. In source-first Birth--death updating, denoted Bd, selection chooses a reproducing individual from the whole population according to its fecundity, the relative propensity to reproduce, and its offspring then replaces a neighbor. The capital letter marks the selected step~\cite{AllenLippnerNowak2019}. The neighbors who receive a cooperator's benefit also compete with it for reproduction~\cite{Taylor1992b,HauertMiekisz2018}. Homogeneous-population analyses found that this competition can prevent cooperation under averaged donation payoffs~\cite{OhtsukiEtAl2006,TaylorDayWild2007,GrafenArchetti2008,DebarreHauertDoebeli2014}. At weak selection the obstruction holds on every unweighted regular graph and, more generally, on isothermal weighted graphs, where all sites have equal overwrite rates without selection~\cite{AllenLippnerNowak2019}.

Allen, Lippner and Nowak asked whether nonisothermal weighted graphs could overcome this obstruction~\cite{AllenLippnerNowak2019}. Intermediate-selection examples exist under other social-good payoff conventions~\cite{McAvoyRaoHauert2021}. Weighted one-dimensional imitation has analytical and numerical antecedents~\cite{IwataAkiyama2016}, and joined communities can help under death-first updating~\cite{Fotouhi2018Conjoining}. For each fixed $b/c>1$ and positive selection strength, growing decorated graphs can favor cooperation under death-first updating with accumulated payoffs and exponential dependence of fecundity on payoff~\cite{SvobodaChatterjee2025}. Temporal contacts also alter cooperation under imitation updating with accumulated payoffs~\cite{LiMengZhouMasudaWang2026}. Here we fix reciprocal contacts, averaged payoffs and source-first updating. The response of a specified graph can be calculated by tracing neutral copying backward to common ancestors~\cite{AllenEtAl2017,AllenEtAl2021,McAvoyAllen2021}. Numerical evaluation of the weak-selection criterion finds no source-first donation advantage against neutrality among the connected unweighted graphs with three to eight sites studied in Ref.~\cite[Fig.~4A]{WangSu2026}. The structural question is what positive contact weights can accomplish, including hierarchies with widely separated scales.

Reciprocal contact heterogeneity can overcome the source-first obstruction by protecting a founder. Each site holds one individual. Reciprocal contact weights are equal in both directions, but normalizing them by each parent's total contact weight can make the copying probabilities asymmetric. A site can send offspring efficiently to a neighbor that rarely sends offspring back. We call a site with a small total overwrite rate a \emph{pin}. Its \emph{followers} adopt its trait through fast copying paths, possibly through intermediate sites; the aligned pin and its followers form a persistent \emph{domain}. Two such domains on a six-site path compete with leading payoffs $b-c$ and zero. Their threshold, the benefit--cost ratio above which the first-order fixation response favors cooperation, approaches one as the contact scales separate. Every fixed longer path admits the same mechanism. More generally, any fixed network with at least six sites and a spanning path, one visiting every site, can be weighted to approach threshold one. One first chooses a finite promoting path weighting and then adds sufficiently weak positive contacts. Six sites are necessary among paths to reach this infimum, the lowest threshold approachable over positive weightings. The copying routes explain the size of this construction. A pin sends copies to a follower that must pass most of its own copies onward to avoid overwriting the pin. Each separately aligned domain therefore uses a pin and at least two followers, giving two triples on the six-site path.

 At finite positive weights, a strict bound excludes a cooperative fixation advantage whenever $b\le c$. This bound holds at every positive selection strength for which fecundity is positive and strictly increasing with payoff. Its origin is local. Across a contact between opposite traits, the cooperator encounters the defector and vice versa. The cooperator's possible advantage in the fraction of cooperative encounters is strictly less than one, so it has the lower payoff whenever $b\le c$.

Smaller contact patterns reveal what prevents a founder from realizing that collective advantage. On a five-cycle, two pins share a fluctuating \emph{fringe} of non-pin sites. The position of the intervening domain wall, the boundary between traits, raises the threshold infimum to $3/2$; the same wall calculation gives three on a four-cycle. On a five-site path, a founder in an unprotected dimer, a two-site group aligned by mutual copying, must first win an establishment contest, raising the infimum to seven when mutants are placed uniformly. These are distinct physical penalties. At each fixed admissible positive selection strength, meaning one for which all fecundities remain positive, separated domains retain the limiting threshold one. Selection within a shared fringe or during dimer establishment changes the other thresholds and can eliminate promotion.

The \emph{introduction rule}, the distribution of the initial mutant site, determines access to these competitors. Uniform introduction samples sites equally; birth-associated introduction follows a mutation in newborn offspring and samples incoming copying flux. A \emph{support} is the graph of retained positive contacts. Both rules can approach threshold one on these spanning-path supports, with stronger contact separation needed for the birth-associated construction. A weighting is a \emph{weak promoter} if cooperation is favored for some benefit--cost ratio at sufficiently weak selection. We determine the existence of such weightings on every support through five sites. Birth-associated introduction permits them precisely when a four- or five-cycle is present. Uniform introduction also permits a spanning five-site path, which visits all five sites, without either cycle. The five-site path therefore loses every weak promoter under birth-associated introduction. Its frequently overwritten dimer receives new traits often, whereas its protected pin receives them rarely.

We must distinguish the competition between persistent domains from the probability that a new lineage reaches them. Neutral ancestry, the backward copying process without selection, reduces first-order response to two lineages followed until they meet. Its Green function, the time-integrated propagator before coalescence, maps the introduction source to the measured response. Changing the introduction rule or the observable leaves this neutral operator fixed. Exact elimination of fast configurations transforms the source and observable together with the generator; retaining excursion times also keeps the physical clock. On a shared fringe, neutral ancestry on reciprocal contacts obeys the same boundary-value equation as voltage in a resistor network. Holding the two pins at potentials one and zero makes the voltage at each fringe site equal to the probability that its ancestry reaches the first pin before the second, and these voltages set how strongly each pin's encounter partners agree with it. When promotion is possible, a convex identity in the introduction distribution identifies a singleton source attaining the lowest weak-selection threshold at fixed weights.

Turning a limiting mechanism into a sharp bound requires control of every positive weighting. An edgewise lower bound explains why the threshold can approach one only when weighted disagreement histories concentrate on contacts with weak copying in both directions. At fixed size this requires diverging ratios between contact weights. Rooted-forest expansions express the ancestral inverse through positive sums over acyclic transition graphs. On the five-site path, grouping polynomial terms by their degree in the outer weights compares the limiting contest with all finite contacts and gives a sharp envelope at fixed inner-contact ratio. These methods determine the threshold infima on all supports through four sites and seventeen of the twenty-one connected five-site supports. The four remaining infima reduce to a single conjectured inequality on the complete five-site graph. Exact finite-contrast bounds prove this inequality on part of the positive-weight domain, while controlled singular limits constrain its boundary. Exact finite tests sample the remaining high-contrast region. The promoter-existence classification is complete independently of this bound.

The same ancestry also determines neutral evolution, with reproduction independent of payoff. On positive reciprocal paths and cycles, exterior powers, whose entries are minors of a one-lineage propagator, reconstruct the full configuration dynamics. This representation proves diagonalizability and reduces the pair solve to a Sylvester equation, a linear matrix equation for the response. Ancestral coalescence fixes the slowest relaxation rate on every connected reciprocal support; the exterior representation evaluates it on paths and cycles. Separating stationary and relaxing one-lineage modes further bounds the response for every introduction rule without solving each pair problem anew. The operator construction therefore connects an invasion criterion to the dynamics of approach to consensus.

The response extends to general pairwise two-strategy games, distinct encounter and replacement kernels, and larger supports that preserve protected domains. Comparisons with directed copying and death-first updating separate the benefit--cost barrier from the origin of individual or collective protection. Collective competition has antecedents in demographic and multilevel-selection models~\cite{TraulsenNowak2006,SimonFletcherDoebeli2013}; here contact hierarchies generate the competitors and their payoff-dependent transfer rates at fixed population size. For quantitative biology, encounter averaging, offspring routing and mutation introduction are separate experimental inputs. Exact finite examples quantify how alternative conventions change the invasion criterion, including a selection-dependent mutation source under accumulated payoffs. For statistical physics, the construction connects an absorbing spin process, coalescing ancestry and collective slow variables~\cite{HolleyLiggett1975,SoodAntalRedner2008,CastellanoFortunatoLoreto2009}. All limits here keep the number of sites fixed.

The paper follows the passage from a founder to a structural bound. Section~\ref{sec:model} specifies the game, introduction experiment and clock. Section~\ref{sec:response-theory} develops the ancestral response and fast-state elimination, then compares separated domains, a shared fringe and dimer establishment, retaining entrance probabilities and physical times. Section~\ref{sec:classification} turns those mechanisms into all-weight bounds and the support classification. Section~\ref{sec:v4-observable-algebra} extracts full neutral evolution and relaxation from the same ancestry. Section~\ref{sec:extensions} tests the results under finite selection and changes of payoffs, contacts and copying rule; Sec.~\ref{sec:discussion} develops their implications and remaining questions. Readers primarily interested in the collective mechanism can follow Secs.~\ref{sec:model}--\ref{sec:classification}. The biological conventions in Sec.~\ref{sec:model} and their comparisons in Sec.~\ref{sec:v7-robustness} explain which features of an application determine its invasion experiment. Readers interested in neutral stochastic dynamics can proceed from the ancestral pair operator in Sec.~\ref{sec:response-theory} to its full dynamical consequences in Sec.~\ref{sec:v4-observable-algebra}.

The appendices complete the arguments along the same route. Appendix~\ref{app:v10-response} derives the response identities and strict barrier, including its directed and accumulated-payoff extensions. Appendix~\ref{app:v10-operators} treats stochastic operators, neutral propagation and relaxation. Appendix~\ref{app:v10-domains} proves the collective reductions, source estimates, fringe bounds and physical-time limits. Appendix~\ref{app:path-classification} develops the sign and grading bounds, support classification and dense comparisons. Appendix~\ref{app:v7-robustness} treats general games, changed encounter kernels, added contacts, birth-associated introduction and death-first copying.

\section{The game and the invasion experiment}
\label{sec:model}
We consider an inherited binary trait on a finite connected, loopless,
undirected graph. Each of its $N\ge2$ sites contains one individual,
with $n_i=1$ for cooperation and $n_i=0$ for defection. This section fixes the payoff convention, the copying rule, the two
introduction experiments and their observables;
Table~\ref{tab:notation} collects the principal notation. The configuration is $n=(n_1,\ldots,n_N)$. Paths and cycles on $N$ sites are denoted $P_N$ and $C_N$; $K_N$ denotes the complete graph, with an edge between every pair of sites.

\subsection{Payoffs and reciprocal contacts}
\label{sec:game}
A cooperative encounter costs the producer $c>0$ and benefits its
partner by $b\ge0$. With $C$ denoting cooperation and $D$ defection, the row player's payoff is
\begin{equation}
 \begin{array}{c|cc}
       & C & D\\ \hline
 C&b-c&-c\\
 D&b&0
 \end{array}.
 \label{eq:donation-matrix}
\end{equation}
Microbial public-benefit production motivates this choice. A quantitative application would relate encounters to delivered benefit through a specified transport and uptake model~\cite{AllenGoreNowak2013,GoreYoukVanOudenaarden2009,NadellDrescherFoster2016}.
When a benefit retained by the producer can be represented by fixed
pairwise payoffs, it changes the matrix and falls under the general
two-strategy criterion of Sec.~\ref{sec:v7-robustness}. Nonlinear production or concentration-dependent uptake requires a specified payoff model. If it creates interactions involving several neighboring traits, the response can require ancestral correlations beyond the pairs used below. Reciprocal weights satisfy
$w_{ij}=w_{ji}>0$ on the support, the set of allowed edges, and vanish elsewhere, including $w_{ii}=0$.
They define the site strengths and contact kernel,
\begin{equation}
 s_i=\sum_jw_{ij},\qquad P_{ij}=\frac{w_{ij}}{s_i},
 \qquad P\one=\one .
 \label{eq:kernel}
\end{equation}
Here $\one$ is the all-ones vector, so each row of $P$ sums to one. We use the same kernel for averaged interactions and replacement,
\begin{equation}
 f_i(n)=-cn_i+b(Pn)_i.
 \label{eq:payoff}
\end{equation}
Thus $(Pn)_i$ is the cooperative fraction among the encounters sampled by individual $i$. A weight $w_{ij}$ is the relative frequency of contact between $i$ and
$j$, used both for interaction and for the dispersal of offspring. Only
ratios of weights matter, and the strength $s_i$ is the sum of the incident contact weights. The total reproduction opportunity is fixed, and fecundity sets each parent's share. Its contact weights partition the destination probabilities without multiplying its reproduction rate. Accumulated encounters and strength-proportional reproduction are separate model choices, compared in Sec.~\ref{sec:v7-robustness}.
At fixed neighboring traits, defection gains exactly $c$, since there
are no self-contacts. It is a strict best reply in every environment,
so all defection is the unique Nash equilibrium, the state in which
no individual gains by changing its own strategy. Strict dominance also excludes mixed-strategy equilibria. Yet all cooperation gives everyone $b-c>0$ when $b>c$.
In the dynamics below, traits are inherited through copying; individuals
do not switch to their best reply. Their realized payoffs can differ
because their neighborhoods differ. Individual payoff dominance and population-level fixation advantage
therefore answer distinct questions about the same game.

The contact average is evaluated before the nonlinear fecundity map, without separate sampling of encounter rewards. The averaging convention matters on heterogeneous graphs. Accumulating
encounters gives $s_i f_i$ and can change the selection
criterion~\cite{MaciejewskiFuHauert2014,SantosPacheco2005,Masuda2007,SzolnokiPercDanku2008,SzaboFath2007}.
Distributing a donor's fixed benefit budget uses the donor's
normalization, with different consequences for cooperation~\cite{McAvoyAllenNowak2020,McAvoyRaoHauert2021}.
 With distinct interaction and replacement kernels, neutral ancestry
follows replacement and the benefit contraction changes;
Sec.~\ref{sec:v7-robustness} extends the common-kernel bounds to a
specified mixed-encounter model.

\subsection{Copying dynamics and physical time}
\label{sec:update}
Fecundity, the relative propensity to reproduce, is $F_i=\phi(\delta f_i)>0$, with selection strength
$\delta\ge0$ and normalization $\phi(0)=\phi'(0)=1$ near
neutrality. We use $\phi(x)=1+x$ and $\phi(x)=e^x$.
A different positive slope-to-value ratio
rescales weak response without changing its threshold.
We write $r=b/c$ and $\kappa=\delta c$ for the dimensionless
benefit--cost ratio and selection intensity. Since $-c\le f_i\le b$,
the linear map is positive on every configuration precisely when
 $0\le\kappa<1$; the minimum $f_i=-c$ is attained by a lone
cooperator. The exponential map is positive for every finite $\kappa\ge0$.

At each global Poisson attempt, a source site, the reproducing parent, is chosen in proportion to
fecundity and its recipient through $P$. Writing $n^{j\leftarrow i}$
for the configuration after site $j$ copies site $i$, the rates and backward generator acting on an observable $g$ are
\begin{align}
 r_{ij}(n)&=\frac{F_i(n)P_{ij}}{\sum_kF_k(n)},\label{eq:rates}\\
 (\mathcal A g)(n)&=\sum_{i,j}r_{ij}(n)
       [g(n^{j\leftarrow i})-g(n)] .
 \label{eq:backward}
\end{align}
Time is measured in global attempts, including copies between identical
traits; $t$ is the global attempt time, and $\widehat t=t/N$ measures attempts per site. We also call $\widehat t$ the ancestral clock because a neutral lineage at $j$ then jumps to $i$ at rate $P_{ij}$.
Keeping the total attempt rate fixed compares networks at equal numbers of reproduction or copying opportunities. A hierarchy is a limit in which contact ratios diverge at fixed $N$. In the protected-domain hierarchies almost every attempt copies between identical traits, so physical waiting times retain those unsuccessful attempts. A rule in which each site reproduces at a rate proportional to its strength weights pins and followers differently and can change fixation probabilities. Multiplying every transition rate at a given configuration by the same positive factor preserves fixation probabilities and changes holding times. Section~\ref{sec:v7-robustness} compares strength-dependent reproduction with the present rule. The rule in Eq.~\eqref{eq:rates} is source-first
Birth--death (Bd) updating; the capital letter marks the step at which
selection acts. Choosing the recipient first and letting its
neighbors compete produces a different response~\cite{OhtsukiEtAl2006,AllenLippnerNowak2019}. Neutrality gives a weighted
invasion process~\cite{SoodAntalRedner2008,BaronchelliCastellanoPastorSatorras2011}. A fixed fitness for each type would describe constant
selection; here fecundity depends on the configuration through the game.

The configuration representation below is the starting point for fast-state elimination and neutral spectral analysis. Spin notation retains its stochastic interpretation. In the configuration basis, represent $n_i$ by its diagonal occupation projector. The operators $\sigma_i^+$ and $\sigma_i^-$ change $0$ to $1$ and $1$ to $0$, respectively, and annihilate the other local state; $\widehat r_{ij}$ is the diagonal rate operator formed from Eq.~\eqref{eq:rates}. The forward generator is
\begin{equation}
 \mathcal L=\sum_{i,j}\left[n_i(\sigma_j^+-(1-n_j))
 +(1-n_i)(\sigma_j^--n_j)\right]\widehat r_{ij}.
 \label{eq:spin-generator}
\end{equation}
The diagonal rate operator evaluates the configuration before the occupations change. The probability vector obeys $\partial_t|p\rangle=\mathcal L|p\rangle$ and $\langle\mathbf1|p\rangle=1$, where $\langle\mathbf1|$ sums all configuration probabilities. Write $\mathsf A$ for the configuration matrix of the backward operator $\mathcal A$. The forward operator $\mathcal L$ has matrix $\mathsf Q=\mathsf A^{\T}$. Appendix~\ref{sec:physics} constructs a two-species Doi--Peliti representation of the same generator using cooperative and defective bosons constrained to one occupant per site~\cite{Doi1976SecondQuantization,Peliti1985PathIntegral}. There is no quantum-state normalization or assumption of detailed balance in configuration space.

\subsection{Fixation, introduction and absorption time}
\label{sec:observables}

At positive weights, every mixed configuration eventually reaches
$\mathbf0$ or $\mathbf1$. Indeed, connectedness permits a sequence of
copies spreading either surviving type, so the finite mixed sector is
transient. Both consensus states are absorbing, although only defection
is a Nash equilibrium. The stationary measures are their convex
mixtures; copying cannot recreate an absent type.

Let $h_C(n)$ be the probability that the population eventually becomes
all cooperative when started from $n$. This fixation probability, also
called a committor in first-passage theory, satisfies
\begin{equation}
 \begin{gathered}
 \mathcal A h_C=0\quad\hbox{on mixed states},\\
 h_C(\mathbf0)=0,\qquad h_C(\mathbf1)=1.
 \end{gathered}
 \label{eq:committor}
\end{equation}
We compare two introduction experiments. One places a single cooperator
uniformly in an all-defector population; the other places a single
defector uniformly in an all-cooperator population. Comparing these two experiments under the same introduction rule is
what we call paired introduction. If $e_i$ has its
sole cooperator at $i$, their respective fixation probabilities are
\begin{equation}
 \rho_C=\frac1N\sum_i h_C(e_i),\quad
 \rho_D=\frac1N\sum_i[1-h_C(\one-e_i)].
 \label{eq:rho}
\end{equation}

Uniform introduction chooses the site independently of its replacement
rate, as in external seeding at an equally sampled site. Mutation arising during reproduction samples the recipient of a
copying event~\cite{McAvoyRaoHauert2021}. We compare these introduction rules on the same network and determine which supports admit promotion under each.

For a common, selection-independent distribution $\mu_i$ of introduction
sites, the averages in Eq.~\eqref{eq:rho} use $\mu_i$ in place of $1/N$.
We also call $\mu$ the introduction source or preparation; it specifies the initial mutant location.
The two introduction rules considered here are
\begin{equation}
 \mu_i^U=\frac1N,\qquad
 \mu_i^{\rm b}=\frac{\vartheta_i}{N},\qquad \vartheta_i=\sum_kP_{ki}.
 \label{eq:v7-preparation}
\end{equation}
The superscripts $U$ and $\mathrm b$ denote uniform and birth-associated introduction. Evolutionary graph theory calls $\vartheta_i$ the temperature of site
$i$~\cite{LiebermanHauertNowak2005}. It is an incoming copying weight unrelated to thermodynamic temperature. The birth-associated rule is called temperature initialization~\cite{AdlamChatterjeeNowak2015,McAvoyRaoHauert2021};
a mutation that arises in a newborn offspring appears at the site that
the offspring replaces. A weighting is isothermal when all $\vartheta_i$ are
equal. For reciprocal weights on a connected support this holds exactly
when all strengths $s_i$ are equal, because $P$ is reversible with
stationary distribution proportional to $s$; protected pins are the
coldest sites of a strongly nonisothermal weighting.
In either homogeneous background every source has the same fecundity,
so the recipient of a birth has distribution $\mu^{\rm b}$ even at
finite selection. Superscripts identify the introduction rule when
the comparison requires it; unmarked quantities refer to uniform
introduction.
At neutrality both types have the same fixation probability
$\rho_0^\mu=\sum_i\mu_i\pi_i$. The reproductive value $\pi_i$ is the neutral fixation probability of a singleton introduced at $i$~\cite{Maciejewski2014}; Eq.~\eqref{eq:reproductive} computes it. The averaged probability equals $1/N$ under uniform introduction and on isothermal graphs, and generally differs from $1/N$ otherwise.

The invasion advantage is $\Delta=\rho_C-\rho_D$, and its weak susceptibility, the derivative at neutral selection, is
\begin{equation}
 \mathcal S=\left.\partial_\delta\Delta\right|_{\delta=0}.
 \label{eq:response}
\end{equation}
Let $\tau$ be the absorption time, measured in global attempts. For a cooperator introduced at a uniformly chosen site, define
\begin{equation}
 T_{\rm abs}=\E[\tau],\qquad T_{\rm fix}=\E[\tau\mid C\text{ fixes}].
 \label{eq:times-definition}
\end{equation}
The unconditional mean includes extinction; the conditional mean counts
successful invasions. A superscript $\mu$ denotes the corresponding times for introduction distribution $\mu$. These definitions allow selection. The domain limits in Sec.~\ref{sec:domain-times} retain their physical holding times. We say that the
network promotes cooperation when $\Delta>0$. This compares the two
traits against each other, whereas $\rho_C^\mu>\rho_0^\mu$ compares cooperative invasion with its own neutral baseline. For this donation game and a common selection-independent singleton source, trait complementation gives $\partial_\delta\rho_C^\mu|_0=\mathcal S^\mu/2$ and $\partial_\delta\rho_D^\mu|_0=-\mathcal S^\mu/2$, as derived in Appendices~\ref{app:path-pair-representation} and \ref{app:v7-preparation-source}. The two criteria therefore agree at sufficiently small positive selection when the first-order coefficient is nonzero.

This ranking has a population interpretation under equal rare rates of
introducing the absent type at uniformly chosen sites. If each invasion
resolves before the next, the consensus populations communicate at
rates proportional to fixation probabilities~\cite{FudenbergImhof2006,McAvoyAllen2021},
and the stationary probability $X_C^{\rm rare}$ of the all-cooperator
population is
\begin{equation}
 X_C^{\rm rare}=\frac{\rho_C}{\rho_C+\rho_D},
 \qquad
 X_C^{\rm rare}>\frac12\ \Longleftrightarrow\ \Delta>0 .
 \label{eq:rare-abundance}
\end{equation}
The rare-introduction limit is taken at fixed weights. Along a contact hierarchy the introduction rate must decrease enough to keep invasions separated on their absorption timescale~\cite{WuGokhaleWangTraulsen2012}. For the protected six-site path this is ensured by $\nu q\to0$, where $\nu$ is the total introduction rate per global attempt (Sec.~\ref{sec:domain-times}).

\begin{table*}[t]
\caption{Principal notation. The equation or section gives the defining
relation. Unmarked quantities refer to uniform introduction; the
superscript $\mathrm b$ marks birth-associated introduction
(called temperature initialization in evolutionary graph theory).}
\label{tab:notation}
\begin{ruledtabular}
\begin{tabular}{p{0.12\textwidth}p{0.33\textwidth}p{0.15\textwidth}p{0.33\textwidth}}
$w_{ij}$, $s_i$, $P_{ij}$&{\raggedright weights, strengths, kernel, Eq.~\eqref{eq:kernel}\par}&
$\pi_i$, $Z$&{\raggedright reproductive value and its inverse-strength normalizer, Eq.~\eqref{eq:reproductive}\par}\\
$\vartheta_i$&{\raggedright incoming copying weight (temperature); $N^{-1}\vartheta_i$ is the neutral overwrite rate, Eq.~\eqref{eq:v7-preparation}\par}&
$\gamma_{ij}$, $L_\gamma$, $\mathsf B$&{\raggedright ancestral conductance, its Laplacian, benefit matrix, Eqs.~\eqref{eq:edge-balance}, \eqref{eq:JK}\par}\\
$\mu^U$, $\mu^{\mathrm b}$&{\raggedright uniform and birth-associated introduction, Eq.~\eqref{eq:v7-preparation}\par}&
$D_{ij}^\mu$&{\raggedright integrated disagreement, Eqs.~\eqref{eq:disagreement}, \eqref{eq:v7-pair-preparation}\par}\\
$\rho_C$, $\rho_D$, $\Delta$&{\raggedright fixation probabilities and their difference, Eqs.~\eqref{eq:rho}, \eqref{eq:response}\par}&
$K$, $J$, $R=K/J$, $R_*$&{\raggedright cost/benefit contractions, threshold and its infimum over positive weights, Eqs.~\eqref{eq:JK}, \eqref{eq:threshold}\par}\\
$\mathcal S$, $\chi_b$, $\chi_c$&{\raggedright weak, benefit and cost susceptibilities, Eqs.~\eqref{eq:response}, \eqref{eq:criterion}\par}&
$\mathfrak a$, $\mathfrak A_\mu$&{\raggedright fringe assortment and response amplitude, Eqs.~\eqref{eq:agreements}, \eqref{eq:contested-threshold}\par}\\
$r=b/c$, $\kappa=\delta c$&{\raggedright benefit--cost ratio, selection intensity\par}&
$q$, $Q$, $u=q^{-1}$, $v=q/Q$&{\raggedright $P_6$ pin--follower and follower--follower weights, and hierarchy ratios\par}\\
$T_{\rm abs}$, $T_{\rm fix}$&{\raggedright mean absorption and success-conditioned fixation times, Eq.~\eqref{eq:times-definition}\par}&
pins, followers&{\raggedright protected sites and the sites they align, Sec.~\ref{sec:protected-mechanism}\par}\\
$\varpi_k$&{\raggedright wall weights of a contested chain, Eq.~\eqref{eq:c5-wall}\par}&
$\widetilde t=t/q$&{\raggedright slow time of domain competition, Sec.~\ref{sec:finite-selection}\par}\\
$t$, $\widehat t=t/N$&{\raggedright global attempt and ancestral times, Sec.~\ref{sec:update}\par}&
$\mathsf M$, $g^\mu$&{\raggedright pair matrix and pair source, Eq.~\eqref{eq:v7-pair-preparation}\par}\\
$\delta$, $f_i$, $F_i$&{\raggedright selection strength, payoff and fecundity, Eq.~\eqref{eq:payoff}, Sec.~\ref{sec:update}\par}&
$h_C$, $\Phi$&{\raggedright fixation committor and its neutral value, Eqs.~\eqref{eq:committor}, \eqref{eq:reproductive}\par}\end{tabular}
\end{ruledtabular}
\end{table*}

At finite weights the support is connected and every fecundity is positive. When a further limit in selection strength or in a shape parameter is taken, its order relative to the hierarchy limit is specified where it is used. Common weight rescaling and vertex relabeling preserve the experiment, including reversal of a path. Exchanging cooperation and defection relates the paired introduction experiments, but generally changes their selected transition rates.
\section{Exact response and collective competition}
\label{sec:response-theory}\label{sec:domains}
The central question is how a single disadvantaged individual can become a persistent competitor. We first follow that event on a six-site path, then compute the exact response that tests it. Comparing a shared fringe and an unprotected dimer identifies which part of the mechanism controls each threshold.
\subsection{Protection from normalized reciprocal contacts}
\label{sec:protected-mechanism}
Number the path consecutively and choose
\begin{equation}
 \begin{gathered}
 (w_{12},w_{23},w_{34},w_{45},w_{56})=(Q,q,1,q,Q),\\
 u=q^{-1},\qquad v=q/Q,\qquad (u,v)\longrightarrow(0,0).
 \end{gathered}
 \label{eq:p6-weights}
\end{equation}
Sites 3 and 4 are the pins. Their followers are sites 1,2 and 5,6, respectively; copying first reaches the adjacent site and then propagates outward. The unit edge joins the pins and is their bridge. A central site sends fractions $1/(1+u)$ toward its followers and
$u/(1+u)$ across the bridge. Its adjacent follower sends only
$v/(1+v)$ back toward it because the outer edge $Q$ dominates that
follower's transmissions. Reciprocal contacts therefore produce directed influence after normalization. Directed social interactions can also alter cooperation~\cite{SuAllenPlotkin2022}, although that model uses accumulated payoffs and death-first updating.
At neutrality,
\begin{equation}
 \vartheta_3=\vartheta_4=\frac q{Q+q}+\frac1{q+1},\qquad
 r_{32}=\frac{q}{6(q+1)}.
 \label{eq:v10-pin-rates}
\end{equation}
The incoming overwrite rate tends to zero, while the pin continues to copy into its follower at a finite rate. In the hierarchy limit the two pins carry almost all neutral reproductive value, the site-dependent fixation weight in Eq.~\eqref{eq:reproductive}. Protection refers to this dynamical
asymmetry, not to an imposed difference between the two traits. In the terms of Sec.~\ref{sec:observables}, a pin is a cold site,
$\vartheta_3,\vartheta_4=O(u+v)$. In voter-model language it behaves as an emergent soft
zealot, rarely converted yet continually converting its
neighbors~\cite{Mobilia2003,MobiliaPetersenRedner2007,MasudaGibertRedner2010,YildizEtAl2013}. Unlike an imposed
zealot, its protection comes from the weights alone, and it changes at
finite positive weights.

At $u=v=0$, the pins retain their traits while their followers align. At finite positive weights the pins can change. Denote the four aligned states by their domain traits
$00,10,01,11$, with the left-domain trait written first. A pin introduction establishes its triple with probability
tending to one, while a follower introduction disappears. Uniform
placement leaves mass $1/6$ in each mixed aligned state and loses $2/3$
during formation, as shown in Fig.~\ref{fig:domains}.
\begin{figure*}[t]
\centering

\begin{tikzpicture}[x=1cm,y=1cm,font=\small,>=Stealth,
 box/.style={draw=black!45,rounded corners=2pt,align=center,inner sep=6pt},
 copy/.style={->,line width=.85pt,draw=figureblue},
 flow/.style={->,line width=.7pt,draw=black!60},
 rate/.style={->,line width=.85pt,draw=black!85}]
\node[anchor=west] at (0,6.1) {\textbf{(a)} Reciprocal contacts and directed copying};
\foreach \i/\x in {1/2,2/4,5/10,6/12}{\node[circle,draw,fill=white,minimum size=5mm,inner sep=0pt] (p\i) at (\x,5.3) {\i};}
\foreach \i/\x in {3/6,4/8}{\node[circle,draw,fill=figureblue,text=white,minimum size=5mm,inner sep=0pt] (p\i) at (\x,5.3) {\i};}
\foreach \i/\j/\w in {1/2/Q,2/3/q,3/4/1,4/5/q,5/6/Q}{\draw (p\i)--node[above] {$\w$}(p\j);}
\draw[copy] (p3.south) to[bend left=35] node[below=2pt] {$P_{32}\to1$} (p2.south);
\draw[copy] (p4.south) to[bend right=35] node[below=2pt] {$P_{45}\to1$} (p5.south);
\node[align=left,anchor=west,font=\footnotesize] at (12.8,5.1) {Dark sites are pins;\\arrows show outward copying.};
\node[font=\footnotesize] at (8,3.88)
 {Reverse copying $P_{23}=P_{54}\sim q/Q$; bridge $P_{34}=P_{43}\sim1/q$};
\draw[black!15] (0,3.5)--(16,3.5);
\node[anchor=west] at (0,3.0) {\textbf{(b)} Establishment from a uniform cooperator introduction};
\node[box] (start) at (2,1.25) {One cooperator\\among six sites};
\node[box] (pin) at (7,2) {At pin 3 or 4\\probability \(2/6\)};
\node[box] (fol) at (7,.4) {At a follower\\probability \(4/6\)};
\node[box,text=figureblue] (mixed) at (13.2,2) {One cooperative triple\\\(10\) or \(01\), total mass \(1/3\)};
\node[box] (zero) at (13.2,.4) {Cooperator disappears\\\(00\), mass \(2/3\)};
\draw[flow] (start)--(pin);
\draw[flow] (start)--(fol);
\draw[flow] (pin)--node[above] {fast alignment}(mixed);
\draw[flow] (fol)--(zero);
\draw[black!15] (0,-.25)--(16,-.25);
\node[anchor=west] at (0,-.8) {\textbf{(c)} Competition in slow time \(\widetilde t=t/q\)};
\node[box] (dd) at (2,-2) {\(00\)\\all defectors};
\node[box] (mid) at (8,-2) {\(10\) or \(01\)\\one domain of each type};
\node[box,text=figureblue] (cc) at (14,-2) {\(11\)\\all cooperators};
\draw[rate] (mid)--node[above=4pt] {\((1-h_{\rm dom})/3\)}(dd);
\draw[rate] (mid)--node[above=4pt] {\(h_{\rm dom}/3\)}(cc);
\node at (8,-3) {Each mixed state exits at total rate \(1/3\); no transition between \(10\) and \(01\)};
\path[use as bounding box] (0,-3.25) rectangle (16,6.45);
\end{tikzpicture}

\caption{The initial individual and the eventual domain competitor.
Panel (a) locates the slowly overwritten pins and their followers, which adopt a pin's trait through fast copying. A pin and its two aligned followers form a triple. Edge labels are relative contact weights. Blue arrows show dominant outward copying probabilities conditional on the parent; their neutral event rates are $P_{ij}/6$. Gray arrows in panel (b) distinguish pin introductions, which establish an
aligned triple as $q\to\infty$ and $q/Q\to0$, from follower introductions,
which disappear. Panel (c), with black rate arrows, is the resulting four-state process
in slow time $\widetilde t=t/q$; its arrows are rates in that clock.
Here $h_{\rm dom}$ is the probability that cooperation wins the domain contest. The entrance
weights belong to the original uniform-singleton experiment. At finite weights each pin is overwritten at neutral rate $\vartheta_3/6=\vartheta_4/6$, with $\vartheta_3,\vartheta_4$ given by Eq.~\eqref{eq:v10-pin-rates}; this rate vanishes along the hierarchy.}
\label{fig:domains}
\end{figure*}
A cooperative triple has leading payoff $b-c$ at each site and a
defective triple has payoff zero. A rare bridge copy changes a pin;
its followers usually align before the next rare copy. The collective advantage $b-c$ is therefore realized only after
establishment. The effective competitors are assembled by the
individual copying rules, with no separate rule for group reproduction.

\subsection{The conserved coordinate and covariance response}
\label{sec:neutral-coordinate}\label{sec:absorbing-response}
To connect domain formation with the microscopic fixation experiment, we first identify the conserved neutral coordinate. Source-first copying weights sites by inverse strength~\cite{AntalRednerSood2006,Maciejewski2014,SoodAntalRedner2008},
\begin{equation}
 Z=\sum_i s_i^{-1},\qquad
 \pi_i=\frac{s_i^{-1}}Z,\qquad
 \Phi(n)=\sum_i\pi_i n_i .
 \label{eq:reproductive}
\end{equation}
The identity
\begin{equation}
 \pi_jP_{ij}=\pi_iP_{ji}
 =\frac{\gamma_{ij}}Z,\qquad
 \gamma_{ij}=\frac{w_{ij}}{s_is_j}
 \label{eq:edge-balance}
\end{equation}
makes the two copying directions cancel in $\mathcal A_0\Phi$, where $\mathcal A_0=\mathcal A|_{\delta=0}$ is the neutral backward operator.
Stopping this bounded martingale at consensus gives $h_C^0=\Phi$ and
$\rho_C^{0,\mu}=\rho_D^{0,\mu}=\rho_0^\mu=\sum_i\mu_i\pi_i$; uniform introduction gives $1/N$. Here $\pi_i$ is the neutral fixation probability of a singleton at $i$,
called its reproductive value. It differs from the strength-proportional
stationary measure of $P$. Ancestry follows a
received copy backward from $j$ to $i$, at rate $P_{ij}/N$.

Let $L_\gamma$ be the symmetric Laplacian with off-diagonal entries
$-\gamma_{ij}$ and diagonal entries $\sum_j\gamma_{ij}$. The first drift correction is
\begin{equation}
 \left.\partial_\delta\mathcal A\Phi(n)\right|_0
 =\frac{-c\,n^\T L_\gamma n+
               b(Pn)^\T L_\gamma n}{NZ}.
 \label{eq:drift}
\end{equation}
The derivative of the global rate normalization cancels because it
multiplies $\mathcal A_0\Phi=0$. The two terms separate the cost contribution to fixation drift across
trait interfaces from the benefit due to different cooperative environments,
\begin{gather}
 n^\T L_\gamma n=\sum_{i<j}\gamma_{ij}(n_i-n_j)^2,
 \label{eq:dirichlet} \\
 (Pn)^\T L_\gamma n=
 \sum_{i<j}\gamma_{ij}(n_i-n_j)[(Pn)_i-(Pn)_j].
 \label{eq:benefit}
\end{gather}
The first is strictly positive on mixed configurations and penalizes
interfaces. The second compares cooperative environments across each
interface; its sign depends on the correlations generated by copying.

Writing $h_C=\Phi+\delta h_1+O(\delta^2)$,
$\mathcal A_1=\partial_\delta\mathcal A|_{\delta=0}$ and
$g_{\rm sel}=\mathcal A_1\Phi$, the absorbing Dirichlet problem, with $h_1=0$ at both consensus states, gives
\begin{equation}
 h_1=(-\mathcal A_{0,\mathrm{mix}})^{-1}g_{\rm sel}
 =\int_0^\infty e^{t\mathcal A_{0,\mathrm{mix}}}g_{\rm sel}\,dt .
 \label{eq:absorbing-susceptibility}
\end{equation}
The restriction to mixed states stops each history at consensus. Its
Green function in Eq.~\eqref{eq:absorbing-susceptibility} integrates
the selection drift over precisely that transient part of the evolution. Write $\E_\delta$ for expectation at
selection strength $\delta$ with a uniformly placed cooperative singleton.
At neutrality, $\E_0[n_i(t)]=1/N$ at every site and time. Nevertheless,
the covariance $\mathcal C_{ij}(t)=\E_0[n_i(t)n_j(t)]-\E_0[n_i(t)]\E_0[n_j(t)]$
evolves, and
\begin{equation}
 \left.\partial_\delta\frac{d}{dt}\E_\delta[\Phi(n(t))]\right|_0
 =\frac{\operatorname{tr}[(bP^\T L_\gamma-cL_\gamma)\mathcal C(t)]}{NZ}.
 \label{eq:covariance}
\end{equation}
The mean contribution cancels by $P\one=\one$ and $L_\gamma\one=0$. For another preparation the mean $m_i(t)=\E_0^\mu [n_i(t)]$ need not be spatially constant, and the numerator also contains $m^{\T}(bP^{\T}L_\gamma-cL_\gamma)m$.
Thus the neutral mean density is stationary while correlations carry
the response. This is the fluctuation mechanism missed by the
spatially uniform mean-field approximation, which neglects these correlations. Weak selection controls the perturbation
of rates; the neutral fluctuations themselves are treated exactly.

\subsection{Ancestral pairs and the invasion threshold}
\label{sec:pair-response}
The symmetric parts of the two drift matrices have zero row sums, so both terms can be expressed through the pair disagreements $(n_i-n_j)^2$. The integrated disagreement
\begin{equation}
 D_{ij}=\int_0^\infty\E_0[(n_i(t)-n_j(t))^2]\,dt,\qquad D_{ii}=0
 \label{eq:disagreement}
\end{equation}
has a direct ancestral interpretation. Trace each received copy back
to its source. Two such lineages move independently until they meet
and share one ancestor, the coalescing duality of voter-type
models~\cite{HolleyLiggett1975,Liggett1985}; Sec.~14.3 of Ref.~\cite{AldousFillReversible} treats the finite-graph construction. Denote their meeting time in global attempts by $\tau_{\rm meet}$. Before that meeting, uniform singleton
introduction makes their disagreement probability $2/N$, so
$D_{ij}=(2/N)\E_{ij}[\tau_{\rm meet}]$. Consequently
\begin{equation}
 (\vartheta_i+\vartheta_j)D_{ij}
 -\sum_k(P_{ki}D_{kj}+P_{kj}D_{ik})=2,\quad i\ne j,
 \label{eq:pair-system}
\end{equation}
where $\vartheta_i/N$ is the neutral overwrite rate of site $i$.
This replaces $2^N-2$ transient configuration unknowns by
$\binom N2$ ancestral pairs. The source $2=N\cdot(2/N)$ combines the
disagreement probability $2/N$ of a uniformly introduced singleton with
the factor $N$ that converts the ancestral clock to global attempts.
Equation~\eqref{eq:pair-system} and the criterion below are the
source-first, averaged-payoff case of the coalescence-time criterion of
Allen \textit{et al.} for arbitrary weighted graphs~\cite{AllenEtAl2017};
related forms appear in Refs.~\cite{AllenEtAl2021,McAvoyAllen2021}. The
new questions here concern its behavior over all positive weights.

The lineage walk behind Eq.~\eqref{eq:pair-system} is reversible with
respect to the reproductive values, because $\pi_jP_{ij}=\pi_iP_{ji}$,
and $P_{ki}=\gamma_{ki}/(Z\pi_i)$. Its harmonic functions are therefore
the potentials of an electrical network on the support with
conductances $\gamma_{ij}=w_{ij}/(s_is_j)$ and vertex weights $\pi_i$.
The ancestral walk and these inverse-strength conductances appear in the
supplementary information of Ref.~\cite{AllenEtAl2017}; Ref.~\cite{DoyleSnell1984} develops the connection between reversible walks, hitting probabilities and voltages. We call it the network of ancestry. Three
quantities of this paper are read off it. The cost contraction
$K=\int_0^\infty\E_0[n^{\T}L_\gamma n]\,dt$ is the time-integrated
Dirichlet energy of the trait configuration, a domain-wall energy with
couplings $\gamma_{ij}$. This quadratic form measures costly trait interfaces in the absorbing process; it does not define an equilibrium Gibbs weight. The edge bound behind the threshold barrier at one,
Eq.~\eqref{eq:edge-gap-bound} below, states that the environment
difference across an interface never exceeds the interface itself. And
the probability that a lineage started at $x$ reaches a site $p_1$
before a site $p_2$ is the voltage at $x$ when $p_1$ is held at
potential one and $p_2$ at zero; Sec.~\ref{sec:c5-contested} rests on
this reading. One identity follows at once. Contracting
Eq.~\eqref{eq:pair-system} with the weights $\pi_i\pi_j$ gives
\begin{equation}
 \sum_{i<j}\gamma_{ij}(\pi_i+\pi_j)D_{ij}=Z\Bigl(1-\sum_i\pi_i^2\Bigr),
 \label{eq:coalescence-sum-rule}
\end{equation}
the statement that every pair of lineages coalesces exactly once,
weighted by reproductive value; on an isothermal graph it fixes $K$
outright.

Contracting the solution $D$ of Eq.~\eqref{eq:pair-system} with the drift forms~\eqref{eq:dirichlet} and \eqref{eq:benefit} defines
\begin{equation}
 \begin{gathered}
 \mathsf B=\tfrac12(P^\T L_\gamma+L_\gamma P),\\
 K=\sum_{i<j}\gamma_{ij}D_{ij},\qquad
 J=-\sum_{i<j}\mathsf B_{ij}D_{ij}.
 \end{gathered}
 \label{eq:JK}
\end{equation}
The symmetric benefit matrix $\mathsf B$ has zero row sums. Complementary introductions
have the same first-order drift (Appendix~\ref{app:path-pair-representation}), hence
\begin{equation}
 \mathcal S=\frac{2}{NZ}(bJ-cK).
 \label{eq:criterion}
\end{equation}
Equivalently, $\mathcal S=b\chi_b-c\chi_c$, with
$\chi_b=2J/(NZ)$ and $\chi_c=2K/(NZ)>0$. For $J>0$,
\begin{equation}
 R(w)=K/J,\qquad \mathcal S>0\ \Longleftrightarrow\ b/c>R(w).
 \label{eq:threshold}
\end{equation}
If $J\le0$, nonnegative benefit cannot compensate positive cost at
first order. A positive susceptibility ensures positive advantage at
sufficiently small selection for fixed weights; this interval can
shrink as the contact ratios become singular. We call a weighting with $J>0$ a weak promoter because some positive benefit--cost ratio then favors cooperation at sufficiently weak selection. A support is promoting if it admits such a weighting. The threshold infimum $R_*$ of a support is the
infimum of $R$ over its weak promoters, and $R_*=+\infty$ when there
are none.

Changing the introduction rule leaves the ancestral-pair matrix
unchanged and changes its source, as in the treatment of general
mutant-appearance distributions in Refs.~\cite{AllenTarnita2014,McAvoyAllen2021}. In place of Eq.~\eqref{eq:pair-system},
\begin{equation}
 (\vartheta_i+\vartheta_j)D_{ij}^{\mu}
 -\sum_k(P_{ki}D_{kj}^{\mu}+P_{kj}D_{ik}^{\mu})
 =N(\mu_i+\mu_j).
 \label{eq:v7-pair-preparation}
\end{equation}
Indeed, if two ancestors occupy distinct sites $k,l$ at time zero,
a singleton drawn from $\mu$ makes them disagree with probability
$\mu_k+\mu_l$. Write $\mathsf M$ for the pair matrix on the left of Eqs.~\eqref{eq:pair-system} and \eqref{eq:v7-pair-preparation}, and $g^\mu_{ij}=N(\mu_i+\mu_j)$ for the right side. The pair Green function $\mathsf M^{-1}$, the inverse of the pair matrix with coalescence as an absorbing boundary, integrates this source.
For birth-associated introduction the right side is $\vartheta_i+\vartheta_j$;
the resulting $D_{ij}^{\rm b}$ is the expected number of ancestral
jumps before meeting, including the final coalescing jump. Uniform
introduction weights elapsed meeting time. The contractions
in Eq.~\eqref{eq:JK} and the criterion in Eq.~\eqref{eq:criterion}
retain their form with $D^{\mu},K^{\mu},J^{\mu}$. The introduction rule thus
changes which histories enter the response, even when the copying
generator is held fixed.

The source also has a useful conservation identity,
\begin{equation}
 \sum_{i<j}\gamma_{ij}(\pi_i+\pi_j)D_{ij}^{\mu}
 =NZ\sum_i\mu_i\pi_i(1-\pi_i).
 \label{eq:v10-source-sum}
\end{equation}
Weighting the pair equations by $\pi_i\pi_j$ makes the internal ancestral jumps cancel; only coalescence remains. Uniform introduction recovers Eq.~\eqref{eq:coalescence-sum-rule}. Thus the same inverse tests different preparations, while its conservation relation changes with their neutral fixation weights. Appendix~\ref{app:v10-source-sum-rule} derives the identity with the source retained throughout.

Source linearity also compares all singleton preparations at fixed weights.
Let $K^{(i)},J^{(i)}$ be the contractions for introduction exclusively at site $i$.
The singleton has a positive integrated interface cost, so $K^{(i)}>0$.
Linearity of the pair equation gives
\begin{equation}
 \frac{J^\mu}{K^\mu}
 =\sum_i\frac{\mu_iK^{(i)}}{\sum_j\mu_jK^{(j)}}\,\frac{J^{(i)}}{K^{(i)}}.
 \label{eq:I-source-convexity}
\end{equation}
The benefit per unit cost is therefore a convex average of the site
responses, weighted by their cost contractions as well as their arrival
probabilities. For a preparation with $J^\mu>0$, it follows that
\begin{equation}
 R^\mu\ge\min_{i:J^{(i)}>0}\frac{K^{(i)}}{J^{(i)}}.
 \label{eq:I-best-source}
\end{equation}
Equality requires all positive introduction mass to lie on sites maximizing $J^{(i)}/K^{(i)}$. Thus mixing sites cannot improve on the best single-site preparation at fixed weights. If all $J^{(i)}\le0$, no singleton preparation permits weak promotion. Mixing favorable and adverse sites can reverse
the sign of the benefit response. The small-path exclusions and the pin--dimer
comparison below illustrate these two consequences.
\subsection{Elimination with entrance and readout}
\label{sec:effective-memory}
The pair calculation and the domain reduction both have the form $\ell M^{-1}g$, with a column source $g$ and a row readout $\ell$ that extracts the observable. For the pair response, $M=\mathsf M$ is the pair matrix in Eq.~\eqref{eq:v7-pair-preparation}, $g=g^\mu$ has entries $g^\mu_{ij}=N(\mu_i+\mu_j)$, and the readout has entries
\begin{equation}
 \ell_{ij}=\frac{2}{NZ}\left(-b\mathsf B_{ij}-c\gamma_{ij}\right),
 \qquad i<j.
 \label{eq:v10-pair-readout}
\end{equation}
Then $\ell M^{-1}g=\mathcal S^\mu$. Changing the introduction rule changes the source $g$, while a different measured contraction changes $\ell$. The same operation applies to other invertible operators $M$. Partition its coordinates into retained states $S$ and eliminated states $F$, and set
\begin{equation}
 \mathcal M_S=M_{SS}-M_{SF}M_{FF}^{-1}M_{FS}.
 \label{eq:v10-schur-kernel}
\end{equation}
If $M_{FF}$ and $\mathcal M_S$ are invertible, solving the lower block equation and substituting gives the exact identity
\begin{align}
 \ell M^{-1}g={}&\ell_FM_{FF}^{-1}g_F\notag
 +(\ell_S-\ell_FM_{FF}^{-1}M_{FS})\mathcal M_S^{-1}\notag\\
 &\hspace{1em}\times(g_S-M_{SF}M_{FF}^{-1}g_F).
 \label{eq:v10-decorated-schur}
\end{align}
The reduced kernel, entrance, readout and direct fast contribution all matter. This is the source-dependent form of the Schur or Feshbach reduction~\cite{DussonSigalStamm2021}; for Markov chains the corresponding elimination is the stochastic complement~\cite{Meyer1989}. It is exact before a separation of timescales is invoked. Eliminating several blocks successively gives the same response as eliminating their union, whenever the intermediate inverses exist. The kernel identity is the Schur quotient formula~\cite{CrabtreeHaynsworth1969}; the full response identity follows by successive substitution in the same linear system, with the source, readout and direct term transformed at each step. Successive Schur elimination is therefore exact for nested contact scales. A leading-order truncation at each stage requires an error bound uniform in the ratios still to be reduced.

Retaining the readout can also shorten a singular calculation. If its leading row annihilates the slow nullvector, an undetermined component of the asymptotic pair solution cannot affect the response coefficient at that order. The dense five-site comparison in Sec.~\ref{sec:five-site-supports} uses this cancellation.

For the path in Eq.~\eqref{eq:p6-weights}, partition the forward generator $\mathsf Q=\mathsf A^{\T}$ into the four aligned states $S$ and the remaining 60 configurations $F$.
For $\operatorname{Re}z>0$, let
$\widetilde p_S(z)=\int_0^\infty e^{-zt}p_S(t)\,dt$.
Their exact Laplace-space elimination gives
\begin{align}
 \Sigma(z)&=\mathsf Q_{SF}(zI-\mathsf Q_{FF})^{-1}
                    \mathsf Q_{FS},\notag\\
 \bigl[zI-\mathsf Q_{SS}-\Sigma(z)\bigr]\widetilde p_S(z)
 & =p_S(0) \label{eq:main-self-energy}\\
 &\quad+\mathsf Q_{SF}(zI-\mathsf Q_{FF})^{-1}p_F(0). \notag
\end{align}
The self-energy $\Sigma(z)$ sums excursions through fast configurations,
including their duration. The source on the right is equally necessary.
It determines how the original singleton enters the domain process. At zero frequency the corresponding backward operator is
\begin{equation}
 L_{\rm harm}
 =\mathsf A_{SS}-\mathsf A_{SF}\mathsf A_{FF}^{-1}\mathsf A_{FS}
 =u[L_{\rm dom}+O(u+v)] ,
 \label{eq:schur}
\end{equation}
where $\mathsf A_{XY}$ are blocks of the backward generator and $L_{\rm dom}$ is the slow domain generator of Eq.~\eqref{eq:domain-generator}. It computes hitting
probabilities, while physical times also require a differentiated
resolvent or the equivalent corrected Poisson source, which includes time spent on eliminated excursions. The time-domain
memory equation is derived in Appendix~\ref{app:p6-protected}.

Physically, every exit from the aligned states needs a rare bridge
copy, whereas followers realign at rates of order one. Formally, a fixed
sequence of follower-aligning copies bounds the fast inverse, every row
leaving $S$ carries an exact factor $u$, and at $u=0$ the aligned
domains are closed for all nearby $v$. For a generator of this form the
Schur complement divided by $u$ converges with an error linear in the
perturbation, uniformly on compact positive-fecundity
domains, as stated in Proposition~\ref{prop:I-fixed-selection} and proved in Appendix~\ref{app:p6-protected}.
The reduction connects the aggregation of nearly decomposable
chains~\cite{SimonAndo1961} with singular-perturbation and metastability
methods~\cite{SlowAvrachenkovHaviv2004,AvrachenkovFilarHowlett2013,BovierDenHollander2015}.
The operator $L_{\rm harm}$ generates the trace of the process on the aligned states,
whose clock omits fast excursions~\cite[Sec.~6.1]{BeltranLandim2010}.
The introduction source must also be compared with the scale of the response, as in Sec.~\ref{sec:v10-preparation-comparison}.
\subsection{Fixation and first-passage distributions}
\label{sec:finite-selection}\label{sec:domain-times}
The elimination in Eq.~\eqref{eq:schur} turns the four aligned states of Fig.~\ref{fig:domains}(c) into a collective competition with its own holding times. Let $F_{\rm dom}=\phi(\delta(b-c))$ and $d_0=\phi(0)=1$ be the cooperative and defective domain fecundities, and set $h_{\rm dom}=F_{\rm dom}/(F_{\rm dom}+d_0)$.

\begin{proposition}[Protected-domain competition at fixed selection]
\label{prop:I-fixed-selection}
For source-first copying with common-kernel averaged donation payoffs, take the six-site path with weights $(Q,q,1,q,Q)$ with uniform singleton introduction. Let $u=q^{-1}\to0$ and $v=q/Q\to0$. Suppose $\phi$ is locally Lipschitz, with $\phi(0)=1$, and restrict $(\kappa,\kappa r)=(\delta c,\delta b)$ to a compact domain on which $\phi>0$ on $[-\kappa,\kappa r]$, so that all fecundities are bounded away from zero uniformly in the weights. In slow time $\widetilde t=t/q$, where $t$ counts global attempts, the harmonic Schur complement of Eq.~\eqref{eq:schur} satisfies $L_{\rm harm}/u=L_{\rm dom}+O(u+v)$, so the aligned domains have limiting generator
\begin{equation}
 L_{\rm dom}=\frac13
 \begin{pmatrix}
 0&0&0&0\\
 1-h_{\rm dom}&-1&0&h_{\rm dom}\\
 1-h_{\rm dom}&0&-1&h_{\rm dom}\\
 0&0&0&0
 \end{pmatrix}
 \quad (00,10,01,11).
 \label{eq:domain-generator}
\end{equation}
For the original uniform-singleton experiment, the fixation probabilities and paired advantage satisfy
\begin{align}
 \rho_C&=\frac{h_{\rm dom}}3+O(u+v),\notag\\
 \rho_D&=\frac{1-h_{\rm dom}}{3}+O(u+v), \label{eq:p6-selected}\\
 \Delta&=\frac{F_{\rm dom}-d_0}{3(F_{\rm dom}+d_0)}+O(u+v). \notag
\end{align}
Starting from a uniformly introduced cooperator, the mean absorption time and the mean time conditioned on cooperative fixation satisfy
\begin{equation}
 T_{\rm abs}=q[1+O(u+v)],\qquad
 T_{\rm fix}=3q[1+O(u+v)] .
 \label{eq:p6-times}
\end{equation}
All errors are uniform on the stated parameter domain and in the ratio $v/u$.
\end{proposition}
\begin{proof}[Proof]
Appendix~\ref{app:p6-protected} derives the reduced generator and controls both the entrance probabilities and the time spent on eliminated excursions.
\end{proof}
Each mixed aligned state has total leading fecundity $3(F_{\rm dom}+d_0)$. The two pins together carry fecundity $F_{\rm dom}+d_0$, so selection redistributes reproduction between them while keeping their joint share equal to $1/3$. Each sends a bridge copy with probability asymptotic to $1/q$. The contest therefore lasts $3q$ global attempts on average, with cooperative exit probability $h_{\rm dom}$. Uniform introduction reaches it with probability tending to $1/3$; conditioning on successful fixation selects this slow contest. These two factors account for $T_{\rm abs}\sim q$ and $T_{\rm fix}/T_{\rm abs}\to3$.

A fixed $b>c$ therefore gives positive
limiting advantage at fixed positive selection for strictly increasing
fecundity, so at every fixed finite admissible selection strength the limiting threshold remains $b/c=1$. For the two fecundity maps of Sec.~\ref{sec:update}, the limit $\Delta_\infty=\lim_{u,v\to0}\Delta$ in Eq.~\eqref{eq:p6-selected} is
\begin{equation}
 \Delta_\infty=
 \begin{cases}
 \displaystyle\frac{\delta(b-c)}{3[2+\delta(b-c)]}\,,
       &\phi(x)=1+x,\\[2mm]
 \displaystyle\frac13\tanh\!\frac{\delta(b-c)}2\,,
       &\phi(x)=e^x.
 \end{cases}
 \label{eq:selected-maps}
\end{equation}
For $\phi(x)=1+x$ the formula holds on the admissible range $0\le\kappa<1$ of Sec.~\ref{sec:update}, which keeps fecundity positive on establishment configurations, where a lone cooperator has payoff $-c$, as well as on the aligned states. At weak selection, $\mathcal S\to(b-c)/6$.

In a joint limit with growing $\kappa r=\delta b$, the bridge and establishment corrections enter the fecundity arguments as $\delta bu$ and $\delta bv$. Extending the reduction requires controlling these fecundities and the fast relaxation time. The inverse of the killed fast generator in Appendix~\ref{app:p6-protected} is uniformly bounded on the stated compact domains; with growing selection its size must be compared with the slow time $q$.

The entrance probabilities and holding rate also determine the full limiting time distributions (Appendix~\ref{app:p6-protected}). For every $\widetilde t>0$,
\begin{align}
 \Pr(\tau/q>\widetilde t)&\longrightarrow\tfrac13e^{-\widetilde t/3},\notag\\
 \Pr(\tau/q>\widetilde t\mid C\text{ fixes})&\longrightarrow e^{-\widetilde t/3}.
 \label{eq:p6-tails}
\end{align}
The limiting unconditional distribution assigns probability $2/3$
to zero, because the formation transient is compressed to zero on the slow clock.

Conditioning on extinction retains both fast failed introductions
and unsuccessful domain competitors. Their survival probability and
mean are
\begin{equation}
 \begin{gathered}
 \Pr(\tau/q>\widetilde t\mid C\text{ becomes extinct})\\
 \longrightarrow\frac{1-h_{\rm dom}}{3-h_{\rm dom}}e^{-\widetilde t/3},\quad \widetilde t>0,\\
 T_{\rm ext}^C=\frac{3(1-h_{\rm dom})}{3-h_{\rm dom}}q+O(1+q^2/Q).
 \end{gathered}
 \label{eq:v7-extinction-time}
\end{equation}
The atom at zero has weight $2/(3-h_{\rm dom})$. At neutrality the mean is
$3q/5$ to leading order, compared with $3q$ for successful fixation.
Subtracting the successful subdistribution from the unconditional
one proves the limit; the time resolvent controls the mean error.

The single slow exit from a mixed domain contest makes successful-fixation times exponential in this limit. Spectral and distributional approaches also distinguish the timescales of fixation in other evolutionary models~\cite{AshcroftTraulsenGalla2015,HathcockStrogatz2019}. The additive finite-weight correction to the mean is $O(1+q^2/Q)$, bounded when $Q$ is at least of order
$q^2$. Appendix~\ref{app:p6-protected} also gives the full limiting
propagator and multivariate occupation generating function.

The rare-introduction abundance tends to $h_{\rm dom}$ when introductions of
total rate $\nu$ per global attempt remain separated, $\nu q\ll1$. Figure~\ref{fig:finite} tests selected fixation and neutral waiting times against the finite microscopic equations. Its choice $Q=q^3$ follows Eq.~\eqref{eq:p6-weights}, with both error parameters $u=q^{-1}$ and $v=q^{-2}$ vanishing. When $b/c-1$ is comparable to $u+v$, formation corrections are of the same order as the limiting advantage. For the plotted weighting $Q=q^3$, Eq.~\eqref{eq:I-long-path-leak-errors} gives the weak-selection threshold $R^U=1+6/q+O(q^{-2})$. Resolving the finite-weight sign in this regime at fixed $\delta>0$ requires the corresponding corrections to both the reduced rates and the entrance probabilities.
\begin{figure*}[t]
\centering
\includegraphics[width=.9\textwidth]{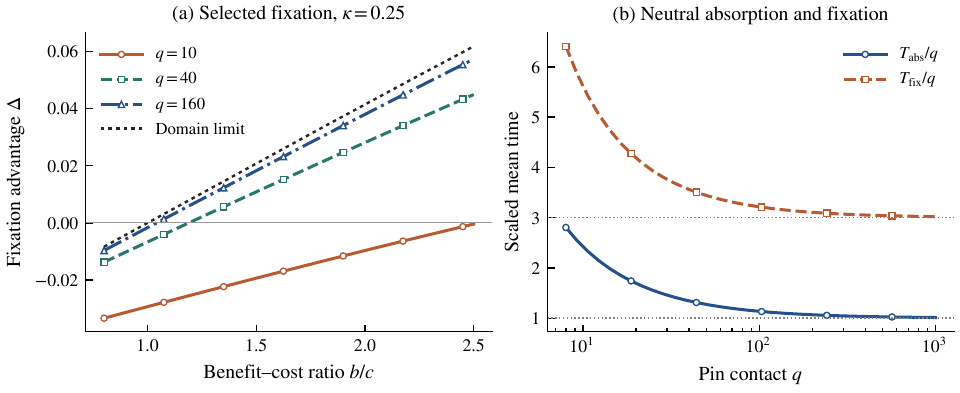}
\caption{Finite microscopic checks on $(q^3,q,1,q,q^3)$ under uniform singleton introduction. The paired fixation advantage is $\Delta=\rho_C-\rho_D$.
(a) The complete fixation-probability (committor) calculation for exponential fecundity at
$\kappa=\delta c=0.25$, compared with the limiting advantage
$\tfrac13\tanh[\kappa(b/c-1)/2]$.
(b) Neutral mean absorption time and mean fixation time conditioned on success approach $T_{\rm abs}/q=1$ and $T_{\rm fix}/q=3$, respectively.
The two panels test selected fixation and neutral time separately.
}
\label{fig:finite}
\end{figure*}

\subsection{A contested fringe and its ancestral voltage}
\label{sec:c5-contested}
The non-pin sites of a two-pin contest form its fringe. In separated domains they align with their own pin; a shared fringe fluctuates under both pins' influences. Freeze the pins at opposite types and let $h_{\rm hit}(x)$ be the probability that neutral ancestry from $x$ reaches pin $p_1$ before $p_2$. A first-step decomposition gives $\sum_iP_{ix}[h_{\rm hit}(i)-h_{\rm hit}(x)]=0$ at every fringe site $x$, with boundary values one at $p_1$ and zero at $p_2$. Multiplying by $\pi_x$ and using Eq.~\eqref{eq:edge-balance} turns this into $\sum_i\gamma_{ix}[h_{\rm hit}(i)-h_{\rm hit}(x)]=0$, Kirchhoff's current-balance equation for a resistor network with conductances $\gamma_{ix}$. Its voltage is therefore the ancestral hitting probability~\cite{DoyleSnell1984}. Define
\begin{gather}
 a_1=\sum_jP_{p_1j}h_{\rm hit}(j),\qquad
 a_2=\sum_jP_{p_2j}[1-h_{\rm hit}(j)],\notag\\
 \mathfrak a=a_1+a_2-1.
 \label{eq:agreements}
\end{gather}
The fringe assortment $\mathfrak a$ is the difference between the cooperative-neighbor fractions sampled by the opposed pins. Sample either opposed pin with nonzero probability, call its type $X$, and draw a partner using its row of $P$, with type $Y$. Conditional on this contest, $\mathfrak a=\operatorname{Cov}(X,Y)/\operatorname{Var}(X)$. The condition $b\mathfrak a>c$ below takes Hamilton's regression form, in which the benefit weighted by association exceeds the cost~\cite{Hamilton1964,Hamilton1970,Queller1992}. This conditional coefficient concerns present traits; it does not measure population-wide relatedness through shared descent.

At stationarity a chain between opposed frozen pins has one domain wall. Its probability on bond $ij$ is proportional to the ancestral resistance $1/\gamma_{ij}$. Detailed balance for this wall walk is equivalent to voltage division along the chain; it does not assert detailed balance for the full absorbing copying process. This makes the imperfect agreement with each pin calculable before the slow contest is solved.

The pin--pin edge is the bridge, with conductance $\gamma_{\rm br}$; unlike a bridge in the graph-theoretic sense, it may lie on a cycle, as in $C_4$ and $C_5$. The reduction requires the fringe to relax before a pin changes, the bridge to dominate ancestral conductance, and source remainders to stay small relative to the pin-pair response. In the ancestral clock $\widehat t=t/N$, a lineage jumps from $j$ to $i$ at rate $P_{ij}$. All hitting-time means in this voltage reduction use that clock. Let $\tau_F$ be the largest mean hitting time of either pin from a fringe site, and $\tau_2$ the largest mean time for a pair to reach the pin-pair state $S=\{p_1,p_2\}$ or coalesce. For introduction rule $\mu$, write $D_S^\mu$ for the integrated disagreement at the pins, put $G_\mu=N\max_{i<j}(\mu_i+\mu_j)$ and define
\begin{equation}
 \Lambda_\mu=\frac{G_\mu\tau_2}{D_S^\mu},\qquad
 \varepsilon_\mu=\frac{N|\mu_{p_1}-\mu_{p_2}|\tau_F}{D_S^\mu}.
 \label{eq:I-contest-source-bounds}
\end{equation}
These ratios measure fast-source accumulation and pin-source imbalance relative to that response. They remain meaningful when arrival amplitudes vanish.

\begin{proposition}[Ancestral voltage of a protected contest]
\label{prop:v10-contest}
Consider a sequence of positive weightings on a fixed support with two adjacent pins $p_1,p_2$. Let $\varepsilon_T=(\vartheta_{p_1}+\vartheta_{p_2})\tau_F$, and let $\varepsilon_\gamma$ be the sum of nonbridge conductances divided by the bridge conductance. Suppose $\varepsilon_T,\varepsilon_\gamma\to0$ and $\mathfrak a\to\mathfrak a_\infty>0$. For a common singleton source with bounded normalized fast-source accumulation $\Lambda_\mu=O(1)$ and vanishing normalized pin-source imbalance $\varepsilon_\mu\to0$, set the positive amplitude $\mathfrak A_\mu=2\gamma_{\rm br}D_S^\mu/(NZ)$, where $\gamma_{\rm br}$ is the bridge conductance. Then
\begin{equation}
 \mathcal S^\mu=\mathfrak A_\mu
 [b\mathfrak a-c+o(b+c)],\qquad
 R^\mu\longrightarrow\mathfrak a_\infty^{-1}
.
 \label{eq:contested-threshold}
\end{equation}
Uniform introduction satisfies these source bounds; birth-associated introduction does so when $\tau_F$ stays bounded. For the following fringe bounds, suppose additionally that along a subsequence the copying matrices converge to $\bar P$ and ancestry from every fringe site reaches a pin in this limiting matrix. Then $\mathfrak a_\infty=\mathfrak a(\bar P)$. Two fringe sites give $\mathfrak a_\infty\le1/3$, and three give $\mathfrak a_\infty\le2/3$, even for nonreversible limiting fringe copying. Equality in the three-site bound has the contested five-cycle geometry. Perfect limiting assortment requires at least four fringe sites under these hypotheses.
\end{proposition}
\begin{proof}[Proof]
The stopped ancestral-pair decomposition in Appendix~\ref{app:circuit} controls the response coefficients and source remainders. Continuity of the invertible killed-fringe matrix identifies the limiting harmonic function. Appendix~\ref{app:fringe-floors} proves the fringe inequalities and equality conditions using exact computer-assisted polynomial sign comparisons.
\end{proof}

The errors are coefficientwise. Under uniform introduction and concentration of reproductive value on the pins, $\mathfrak A_U=4\xi_{\rm pin}/[N(1+\xi_{\rm pin})^2][1+o(1)]$, where $\xi_{\rm pin}=s_{p_2}/s_{p_1}$. Birth-associated introduction has
\begin{equation}
 \frac{\mathfrak A_{\rm b}}{\mathfrak A_U}
 =\frac{\vartheta_{p_1}+\vartheta_{p_2}}2[1+o(1)].
 \label{eq:v10-circuit-access}
\end{equation}
Protection suppresses this amplitude through the small overwrite rates of the pins. The leading ratio agrees with the ratio of direct pin-introduction masses. Appendix~\ref{app:circuit} justifies this agreement by bounding the additional fringe-source terms relative to $D_S^\mu$ under the source hypotheses of Proposition~\ref{prop:v10-contest}.

The regression interpretation of $\mathfrak a$ connects the contest to kin selection. Limited dispersal can increase both association among relatives and competition among their descendants, with exact cancellation in some patch models~\cite{Taylor1992a}. The voltage calculation separates these effects in the present dynamics. It measures association with the partner whose cooperation benefits the pin, while the reduced exit rates determine whose offspring replace the opposed pin. Persistent domains change this competition through their collective payoffs. Competition is therefore included in the controlled response $\mathfrak A_\mu(b\mathfrak a-c)$; the conditional assortment by itself is not a general population-level selection criterion.

For the four-cycle $(1,q,Q,q)$ with $q\to\infty$ and $q^2/Q\to0$, the
chain consists of two strongly coupled sites, and the three bonds
between the pins have asymptotically equal resistance $Q[1+o(1)]$. The limiting wall distribution is therefore uniform on its three positions, and each pin agrees with its chain neighbor with probability $2/3$. Thus $\mathfrak a=1/3$, and Eq.~\eqref{eq:contested-threshold}
gives $R\to3$. The all-weight
inequality of Appendix~\ref{app:c5-supports} shows that no weighting of
$C_4$ does better.

On the five-cycle the chain has three sites. Number the cycle so that
pins $1$ and $2$ share the unit bridge, and write
\begin{equation}
 \begin{gathered}
 (w_{12},w_{23},w_{34},w_{45},w_{51})=(1,q,Q_1,Q_2,q),\\
 u=q^{-1},\qquad y=\frac{q^2}{Q_1+Q_2},\qquad
 \beta=\frac{Q_1}{Q_1+Q_2}.
 \end{gathered}
 \label{eq:c5-weights}
\end{equation}
Here $y$ measures residual follower-to-pin copying relative to bridge copying, which we call leakage. For $u,y\to0$ at fixed $\beta\in(0,1)$, each pin sends its copies almost
entirely to its chain end, whereas the bridge and the chain ends
overwrite it only at rates of order $u$ and $uy$. The chain ends copy
inward, and the middle site divides its copies between them in the
proportions $\beta$ and $1-\beta$. With pin $1$ cooperative and pin $2$ defective, let $k$ count cooperators along the ordered chain $5,4,3$. Its neutral stationary wall weights are
\begin{equation}
 (\varpi_0,\varpi_1,\varpi_2,\varpi_3)=\tfrac13\,(1-\beta,\,1,\,1,\,\beta),
 \label{eq:c5-wall}
\end{equation}
shown in Fig.~\ref{fig:c5}(a), so $a_1=1-\varpi_0$, $a_2=1-\varpi_3$ and $\mathfrak a=1-\varpi_0-\varpi_3=2/3$ for every
$\beta$, because the two boundary bonds carry one third of the
resistance between the pins whatever the split of the chain
(Appendix~\ref{app:wall-cycles}). With $\xi_{\rm pin}=1$ and uniform introduction mass $2/5$ on the pins,
Eq.~\eqref{eq:contested-threshold} gives
\begin{equation}
 \mathcal S\longrightarrow\frac{2b-3c}{15},\qquad R\longrightarrow\frac32 .
 \label{eq:c5-response}
\end{equation}
The uniform response has an order-one amplitude along the hierarchy. Figure~\ref{fig:c5}(b) shows how finite selection changes this contest; Sec.~\ref{sec:contested-selection} derives that dependence.

\labelalias{prop:c5}{eq:c5-expansion}
For the weights~\eqref{eq:c5-weights} with $u,y\to0$ at fixed
$\beta\in(0,1)$, uniform introduction gives
\begin{equation}
 R=\frac{3+4y}{2-y}+\bigl[6+\beta(1-\beta)\bigr]u+O\bigl(u(u+y)\bigr),
 \label{eq:c5-expansion}
\end{equation}
and birth-associated introduction gives
$\mathcal S^{\mathrm b}=u(2b-3c)/15+O(u(u+y))$ and $R^{\mathrm b}\to3/2$.
The neutral mean times are $T_{\rm abs}=q\,(1+y/3)^{-1}[1+O(u)]$ and
$T_{\rm fix}=\tfrac52T_{\rm abs}[1+O(u)]$.

Appendix~\ref{app:c5-reduction} controls the wall-class reduction.
Appendix~\ref{app:c5-slow-competition} derives the exit rates, their closed
$y$ dependence and the physical times. The pair solution in
Appendix~\ref{app:c5-pair} determines the $u$ correction, while
Appendix~\ref{app:circuit} controls the birth-associated amplitude.

The weak limit can also tolerate arbitrarily unequal core contacts. Define the two cycle leakage ratios $\lambda_{\mathrm c1}=q^2/Q_1$ and $\lambda_{\mathrm c2}=q^2/Q_2$. If both vanish as $q\to\infty$, then $R^U,R^{\mathrm b}\to3/2$ uniformly over $\beta$, including ratios approaching either endpoint. The birth response still carries the extra factor $u$. A nonvanishing leakage distinguishes the introduction rules. At fixed $\lambda_{\mathrm c1}>0$ with $\lambda_{\mathrm c2}\to0$,
\begin{equation}
 R^U\longrightarrow\frac32,\qquad
 R^{\mathrm b}\longrightarrow\frac{24(2\lambda_{\mathrm c1}+1)}{\lambda_{\mathrm c1}+16}.
 \label{eq:I-c5-boundary-threshold}
\end{equation}
The second limit exceeds $3/2$ because follower introductions contribute at the same order as the rare pin introductions. A regular pair inverse controls these joint weak-selection limits (Appendix~\ref{app:I-c5-boundary}). The selected fixation and time expansions are proved for $\beta$ in compact subsets of $(0,1)$. Extending them to $\beta\to0,1$ requires a wall-class reduction uniform in both leakage ratios, $\lambda_{\mathrm c1}=y/\beta$ and $\lambda_{\mathrm c2}=y/(1-\beta)$, including entrance probabilities and physical times.

For a contested cycle $C_{m+2}$ with $m\ge2$ fringe sites, the contacts $(1,q,Q,\ldots,Q,q)$ with $q\to\infty$ and $q^2/Q\to0$ give
\begin{equation}
 R\longrightarrow\frac{2m-3}{2m-4}\quad(m\ge3),\qquad
 R\longrightarrow3\quad(m=2).
 \label{eq:cN-contested}
\end{equation}
For $m\ge3$, the limiting series has two end resistances $Q$, two adjacent resistances $2Q$, and $m-3$ remaining resistances $4Q$. For $m=2$, all three resistances equal $Q$. Summing these resistances gives Eq.~\eqref{eq:cN-contested}. Appendices~\ref{app:circuit} and \ref{app:wall-cycles} derive the voltage, both preparation amplitudes and the cycle formula.

The values are $3/2$, $5/4$ and $7/6$ for $N=5,6,7$. Beyond the
five-cycle they are limits of the contested mechanism alone; the
threshold infimum of $C_N$ is one for $N\ge6$, because $C_N$ contains
the path $P_N$ (Appendix~\ref{app:v7-support-extension}). A contested domain
approaches the barrier $b=c$ only as its chain grows, whereas two
separated domains approach it on six sites.

Proposition~\ref{prop:v10-contest} explains why the small fringe cannot achieve the perfect assortment of the separated triples. Outside the hypotheses of Proposition~\ref{prop:v10-contest}, Section~\ref{sec:five-site-supports} treats other regimes through all-weight polynomial bounds and controlled hierarchy limits; the four dense lower bounds remain Conjecture~\ref{conj:v10-dense}.

\begin{figure*}[t]
\centering
\begin{minipage}[t]{0.52\textwidth}\vspace{0pt}
\centering
\begin{tikzpicture}[x=1cm,y=1cm,font=\small,
 site/.style={circle,draw=black!65,fill=white,minimum size=4.8mm,inner sep=0pt},
 pin/.style={site,fill=figureblue,text=white},
 coop/.style={circle,draw=black!65,fill=figureblue!25,minimum size=3.6mm,inner sep=0pt},
 defe/.style={circle,draw=black!65,fill=white,minimum size=3.6mm,inner sep=0pt},
 edge/.style={draw=black!65,line width=.75pt}]
\node[anchor=west] at (-1.9,1.95) {\textbf{(a)} Five-cycle $C_5$};
\node[pin] (n1) at (-0.75,1.15) {\scriptsize1};
\node[pin] (n2) at (0.75,1.15) {\scriptsize2};
\node[site] (n3) at (1.2,-0.25) {\scriptsize3};
\node[site] (n4) at (0,-1.1) {\scriptsize4};
\node[site] (n5) at (-1.2,-0.25) {\scriptsize5};
\draw[edge] (n1)--node[above=2pt] {$1$}(n2);
\draw[edge] (n2)--node[right=2pt] {$q$}(n3);
\draw[edge] (n3)--node[below right=0pt] {$Q_1$}(n4);
\draw[edge] (n4)--node[below left=0pt] {$Q_2$}(n5);
\draw[edge] (n5)--node[left=2pt] {$q$}(n1);
\begin{scope}[shift={(2.35,0.95)}]
\node[anchor=west,font=\footnotesize] at (-.10,.75) {Pinned ends; fluctuating chain};
\foreach \x/\lab in {0/1,.55/5,1.05/4,1.55/3,2.1/2}
 {\node[font=\scriptsize] at (\x,.39) {$\lab$};}
\foreach \k/\a/\b/\c/\lab/\wall in {0/defe/defe/defe/{(1-\beta)/3}/.275,1/coop/defe/defe/{1/3}/.80,2/coop/coop/defe/{1/3}/1.30,3/coop/coop/coop/{\beta/3}/1.825}{
 \draw[black!40,line width=.5pt] (0,-.62*\k)--(2.1,-.62*\k);
 \draw[black!85,line width=.85pt] (\wall,-.62*\k-.17)--(\wall,-.62*\k+.17);
 \node[coop] at (0,-.62*\k) {};
 \node[\a] at (.55,-.62*\k) {};
 \node[\b] at (1.05,-.62*\k) {};
 \node[\c] at (1.55,-.62*\k) {};
 \node[defe] at (2.1,-.62*\k) {};
 \node[anchor=west,font=\footnotesize] at (2.38,-.62*\k) {$k=\k$:\ $\lab$};
}
\node[font=\footnotesize,anchor=west] at (-.1,-2.45) {Shaded $C$; white $D$; wall tick};
\end{scope}
\end{tikzpicture}
\end{minipage}\hfill
\begin{minipage}[t]{0.44\textwidth}\vspace{0pt}
\centering
\includegraphics[width=\linewidth]{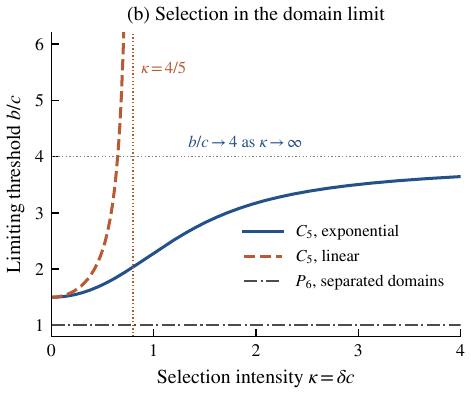}
\end{minipage}
\caption{The contested domain of the five-cycle under uniform singleton introduction. (a) Weights of
Eq.~\eqref{eq:c5-weights}; dark graph nodes are the dynamically protected
pins. In the chain rows, pale shading denotes cooperation and white denotes defection. With pin $1$ cooperative and pin $2$ defective, the chain holds a
single wall, and $k$ counts its cooperative sites (shaded). The right-hand weights are neutral stationary probabilities, with $\beta=Q_1/(Q_1+Q_2)$; the end
pins in each row keep their traits. (b) Limiting threshold $b/c$ in the
hierarchy $q\to\infty$, $q^2/(Q_1+Q_2)\to0$ at $\beta=1/2$ against the selection intensity
$\kappa=\delta c$, for linear and exponential fecundity
(Sec.~\ref{sec:contested-selection}). The linear curve diverges as
$\kappa\uparrow4/5$; the exponential curve approaches four. The dash-dotted
line is the separated-domain limit of the six-site path. The five-cycle curves evaluate the exact limiting threshold equations of
Appendix~\ref{app:c5-selection}. Each curve takes the hierarchy limit at fixed $(\kappa,r)$, with $r=b/c$. The error bounds hold on compact positive-fecundity domains and do not give uniform convergence as $\kappa r$ grows. Approximating the curves at large fixed selection can require stronger contact separation (Appendix~\ref{app:c5-selection}).}
\label{fig:c5}
\end{figure*}

\subsection{Establishment and the introduction rule}\label{sec:five-site-mechanism}\label{sec:small-supports}
The shared-fringe contest can be compared with a support having a single protected founder. The five-site path $(A,1,q,Q)$ places the dimer $\{1,2\}$ beside the protected triple $\{3,4,5\}$ with pin 3 (Fig.~\ref{fig:paths}(c)). Each dimer site is frequently overwritten by its partner, so a cooperator arriving there must survive a one-against-one contest. A copy entering the triple's protected pin aligns its followers with probability tending to one in the hierarchy below. These distinct entrance events account for the establishment factors below.
Take $u=1/q\to0$, $v=q/Q\to0$ at fixed $\eta=q/A>0$. Here $\eta\simeq P_{23}/P_{32}$ is the
rate at which the dimer copies into the pin relative to the reverse.
A cooperator introduced in the dimer first contests payoffs $(-c)$ and $b$,
so its establishment probability is
\begin{equation}
 p_C=\frac{\phi(-\delta c)}
 {\phi(-\delta c)+\phi(\delta b)},\qquad p_D=1-p_C .
 \label{eq:dimer-establishment}
\end{equation}
At neutrality $p_C=1/2$ and $p_C'(0)=-(b+c)/4$. This adverse
establishment response competes with the $b-c$ advantage of an aligned
cooperative domain. For $F_{\rm dom}=\phi(\delta(b-c))$, $d_0=\phi(0)$, the
left- and right-cooperative domain committors are
\begin{equation}
 h_L=\frac{\eta F_{\rm dom}}{\eta F_{\rm dom}+d_0 p_D},\qquad
 h_R=\frac{F_{\rm dom} p_C}{F_{\rm dom} p_C+\eta d_0}.
 \label{eq:dimer-slow-h}
\end{equation}
The factors $p_C,p_D$ enter again when a rare copy converts an existing
dimer. Including uniform introduction,
\begin{align}
\rho_C&\longrightarrow\tfrac25p_Ch_L+\tfrac15h_R,\notag\\
 \rho_D&\longrightarrow\tfrac25p_D(1-h_R)+\tfrac15(1-h_L),
 \label{eq:dimer-rho}
\end{align}
and therefore
\begin{equation}
 \mathcal S\longrightarrow
 \frac{2\eta}{5(1+2\eta)^2}
 [(1-2\eta)b-(7+2\eta)c].
 \label{eq:dimer-response}
\end{equation}
The limiting threshold is $(7+2\eta)/(1-2\eta)$ for $0<\eta<1/2$;
the benefit coefficient vanishes at $\eta=1/2$ and becomes adverse
above it. Letting $\eta\to0$ approaches seven but also suppresses
the amplitude. The controlled expansion is
\begin{equation}
 \mathcal S=\frac{2\eta}{5}(b-7c)
           +O(\eta u+\eta^2+v),
 \label{eq:dimer-uniform}
\end{equation}
coefficientwise in $b,c$. The ordering $1\ll q\ll A\ll Q$ makes
the remainder relatively small. Appendix~\ref{app:path-p5-sharpness} derives the establishment probabilities, domain committors and controlled expansion, including the general introduction masses considered below.

A pin introduction alone has leading paired threshold three. Indeed,
differentiating the two singleton committors in
Eq.~\eqref{eq:dimer-rho} gives
\begin{equation}
 \begin{aligned}
 \left.\partial_\delta(p_Ch_L)\right|_0
 &=-\frac{\eta[\eta(b+c)+2c]}{(1+2\eta)^2},\\
 \left.\partial_\delta h_R\right|_0
 &=\frac{\eta(b-3c)}{(1+2\eta)^2}.
 \end{aligned}
 \label{eq:p5-site-cancellation}
\end{equation}
As $\eta\to0$, the pin contributes $\eta(b-3c)$ to the derivative of cooperative fixation, while each of the two dimer sites contributes $-2\eta c$, to leading order. Their sum is $\eta(b-7c)$; uniform averaging and trait complementation give the factor $2/5$ in Eq.~\eqref{eq:dimer-uniform}. Equation~\eqref{eq:dimer-establishment} makes this cost explicit. Relative to its neutral value, the establishment probability changes by the factor $1-\delta(b+c)/2+O(\delta^2)$ for the two fecundity maps used here. An aligned domain gains from $b-c$, but copies entering a defective dimer must also survive this adverse local contest.

More generally, let $m_p>0$ and $m_{\rm dm}$ be the limiting probabilities of introduction at the pin and in the dimer, and put $\zeta_{\rm dm}=m_{\rm dm}/m_p$. Write $\overline R^\mu$ for the threshold limit along this hierarchy at fixed $\eta$. Differentiating the same two committors gives
\begin{equation}
 \overline R^{\mu}=
 \frac{3+\zeta_{\rm dm}(\eta+2)}{1-\zeta_{\rm dm}\eta},\qquad
 \zeta_{\rm dm}\eta<1.
 \label{eq:v10-p5-source}
\end{equation}
The limit is taken at fixed $\eta>0$. For $\zeta_{\rm dm}>0$ it increases with $\eta$ on $\zeta_{\rm dm}\eta<1$, approaching its infimum $3+2\zeta_{\rm dm}$ as $\eta\to0$ after the hierarchy limit. Pin-only introduction, $\zeta_{\rm dm}=0$, gives three for every $\eta>0$. Uniform introduction approaches seven, illustrating the source mixture in Eq.~\eqref{eq:I-source-convexity}. Equation~\eqref{eq:v10-p5-source} requires positive limiting pin mass and positive leading benefit coefficient, proportional to $m_p-m_{\rm dm}\eta$ in Eq.~\eqref{eq:v10-p5-source-proof}. If the leading benefit coefficient is nonpositive while the leading cost remains positive, the response is adverse at every fixed finite $b/c$. Higher orders are needed when both coefficients vanish. Birth-associated introduction has vanishing pin mass but nonzero dimer mass along this hierarchy, and hence negative leading benefit; Appendix~\ref{app:v7-preparation} proves the stronger exclusion at every positive weighting of the path.

Closing the path into a cycle lowers the uniform threshold infimum from seven to $3/2$ and retains an order-one response amplitude, as Eqs.~\eqref{eq:dimer-uniform} and \eqref{eq:c5-response} show. The cycle also retains promotion under birth-associated introduction, which removes every weak promoter of the path (Sec.~\ref{sec:access}). The contrast identifies dimer establishment as the additional cost of the path geometry.

\paragraph*{Strong internal contact without founder protection.}
Partition five sites into exchangeable groups $\mathcal G_2,\mathcal G_3$ of sizes two and three. Give their internal edges weights $w_{\mathcal G_2},w_{\mathcal G_3}>0$ and every cross edge weight $w_\times>0$. The pair problem has three permutation orbits, corresponding to pairs within either group and pairs across them. If $J_{\mathcal G_2},J_{\mathcal G_3}$ are the benefit coefficients for uniform introduction inside either group, symmetry and source linearity imply
\begin{equation}
 J^\mu=m_{\mathcal G_2}J_{\mathcal G_2}
 +(1-m_{\mathcal G_2})J_{\mathcal G_3}<0,
 \quad m_{\mathcal G_2}=\sum_{i\in\mathcal G_2}\mu_i.
 \label{eq:v10-dense-block}
\end{equation}
Each of $J_{\mathcal G_2}$ and $J_{\mathcal G_3}$ is minus a ratio of positive polynomials, as proved in Appendix~\ref{app:v10-dense-block}. Thus every singleton preparation fails to promote at weak selection on this family. Even when both internal-to-cross ratios diverge at fixed positive ratio, the physical benefit susceptibility tends to $-1/5$ for uniform and birth-associated introduction. The resident groups have overwhelmingly internal encounters, but a new type must establish within one of them. This comparison isolates the role of establishment in collective invasion~\cite{HauertImhof2012,MoawadAbbaraBitbol2024Demes}. The exclusion is proved for this exchangeable family under source-first copying with averaged payoffs; other weightings of the same complete support promote, with thresholds approaching $3/2$ (Table~\ref{tab:five-site}).
\section{Which supports permit an advantage}
\label{sec:classification}
The preceding constructions explain how collective invasion occurs. A bound over all positive weights has a different task. It must exclude competing arrangements, including singular hierarchies with several scales. Figure~\ref{fig:paths} compares the path geometries that mark the successive obstructions. We begin with a universal obstruction and then use the exact response to compare finite supports.
\begin{figure*}[t]
\centering
\begin{tikzpicture}[x=1cm,y=1cm,font=\small,
 site/.style={circle,draw=black!65,fill=white,minimum size=4.8mm,inner sep=0pt},
 pin/.style={site,fill=figureblue,text=white},
 edge/.style={draw=black!65,line width=.75pt}]
\node[anchor=west] at (0,2.8) {\textbf{(a)} Four-site path \(P_4\)};
\foreach \i/\x in {1/.5,2/2.2,3/3.9,4/5.6}
 {\node[site] (a\i) at (\x,1.6) {\scriptsize\i};}
\draw[edge] (a1)--node[above=3pt] {\(w_{12}\)}(a2);
\draw[edge] (a2)--node[above=3pt] {\(w_{23}\)}(a3);
\draw[edge] (a3)--node[above=3pt] {\(1\)}(a4);
\node at (3.05,.45) {No weak-selection promoters};
\node[anchor=west] at (8.1,2.8) {\textbf{(b)} Four-cycle \(C_4\)};
\node[pin] (b1) at (9,1.95) {\scriptsize1};
\node[pin] (b2) at (10.8,1.95) {\scriptsize2};
\node[site] (b3) at (10.8,.55) {\scriptsize3};
\node[site] (b4) at (9,.55) {\scriptsize4};
\draw[edge] (b1)--node[above=2pt] {\(1\)}(b2);
\draw[edge] (b2)--node[right=2pt] {\(q\)}(b3);
\draw[edge] (b3)--node[below=2pt] {\(Q\)}(b4);
\draw[edge] (b4)--node[left=2pt] {\(q\)}(b1);
\node[align=left,anchor=west] at (12,1.2) {Minimal promoter\\[3pt]\(\inf R=3\)};
\draw[black!15] (0,-.25)--(15.8,-.25);
\node[anchor=west] at (0,-.85) {\textbf{(c)} Five-site path \(P_5\)};
\foreach \i/\x in {1/.4,2/1.7,4/4.3,5/5.6}
 {\node[site] (c\i) at (\x,-1.85) {\scriptsize\i};}
\node[pin] (c3) at (3,-1.85) {\scriptsize3};
\draw[edge] (c1)--node[above=3pt] {\(A\)}(c2);
\draw[edge] (c2)--node[above=3pt] {\(1\)}(c3);
\draw[edge] (c3)--node[above=3pt] {\(q\)}(c4);
\draw[edge] (c4)--node[above=3pt] {\(Q\)}(c5);
\node[text=black!65] at (1.05,-2.45) {dimer};
\node[text=figureblue] at (4.3,-2.45) {protected triple};
\node at (3,-3.12) {Uniform introduction, \(\inf R=7\)};
\node at (3,-3.57) {Birth-associated introduction, no weak promoters};
\node[anchor=west] at (8.1,-.85) {\textbf{(d)} Six-site path \(P_6\)};
\foreach \i/\x in {1/8.4,2/9.7,5/13.6,6/14.9}
 {\node[site] (d\i) at (\x,-1.85) {\scriptsize\i};}
\node[pin] (d3) at (11,-1.85) {\scriptsize3};
\node[pin] (d4) at (12.3,-1.85) {\scriptsize4};
\draw[edge] (d1)--node[above=3pt] {\(Q\)}(d2);
\draw[edge] (d2)--node[above=3pt] {\(q\)}(d3);
\draw[edge] (d3)--node[above=3pt] {\(1\)}(d4);
\draw[edge] (d4)--node[above=3pt] {\(q\)}(d5);
\draw[edge] (d5)--node[above=3pt] {\(Q\)}(d6);
\node[text=figureblue] at (11.65,-2.45) {two protected triples};
\node at (11.65,-3.12) {Smallest path with threshold infimum \(1\)};
\node at (11.65,-3.57) {Both introduction rules};
\path[use as bounding box] (0,-3.90) rectangle (16,3.05);
\end{tikzpicture}

\caption{Contact patterns distinguish collective mechanisms. A weak promoter is a positive weighting that favors cooperation for some benefit--cost ratio at sufficiently weak selection. The threshold infimum is the lowest limiting ratio across such weights.
The weak-selection thresholds under uniform introduction distinguish these patterns; panels (c) and (d) also compare birth-associated introduction, which samples the recipient of a copying event.
$P_4$ represents the paths with at most four sites that do not promote
at weak selection. In $C_4$, the promoting pattern with the fewest sites and edges, two pins contest a
shared pair (Sec.~\ref{sec:c5-contested}); $P_5$ is the smallest
promoting path, and $P_6$ the smallest path approaching threshold one,
with two separated domains. Figure~\ref{fig:c5} develops the five-cycle contest.
Edge labels are relative weights; line thickness does not encode their
magnitude. Filled sites are pins, protected by the weights alone and overwritten at rates that vanish along each hierarchy. The limits are $q\to\infty$, $q^2/Q\to0$ on $C_4$;
$1\ll q\ll A\ll Q$ on $P_5$; and $q\to\infty$, $q/Q\to0$ on
$P_6$ under uniform introduction. Birth-associated introduction approaches the same six-site threshold when $q^2/Q\to0$; its five-site exclusion holds for every positive weighting. The six-site construction extends to every fixed longer path.}
\label{fig:paths}
\end{figure*}
\subsection{The benefit--cost barrier and path classification}
\label{sec:universal-bound}\label{sec:small-tree-obstructions}\label{sec:long-domains}
At finite selection the neutral fixation coordinate has drift
\begin{equation}
 \mathcal A\Phi=
 \frac{\sum_{i<j}(\gamma_{ij}/Z)
        [F_i-F_j](n_i-n_j)}{\sum_kF_k}.
 \label{eq:supermartingale}
\end{equation}
For $0\le b\le c$ on a discordant edge with $n_i=1,n_j=0$,
\begin{equation}
 f_i-f_j\le -(c-b)-b(P_{ij}+P_{ji})<0.
 \label{eq:edge-obstruction}
\end{equation}
These edge inequalities constrain both finite-selection fixation and the neutral response coefficients.
\begin{proposition}[Strict benefit--cost barrier]
\label{prop:I-barrier}
Consider source-first copying, Eq.~\eqref{eq:rates}, with common-kernel averaged donation payoffs on a connected loopless reciprocal support with $N\ge2$, positive retained weights and $c>0$. Let $\mu$ be a common selection-independent singleton introduction distribution. For strictly increasing fecundity and every $\delta>0$ at which all fecundities are positive,
\begin{equation}
 \rho_C^\mu<\rho_0^\mu<\rho_D^\mu,\qquad 0\le b\le c,\quad\delta>0.
 \label{eq:finite-no-go}
\end{equation}
Here $\rho_0^\mu=\sum_i\mu_i\pi_i$. The neutral contractions also obey the quantitative bound
\begin{equation}
 \begin{gathered}
 X_{\rm e}=\sum_{i<j}\gamma_{ij}(P_{ij}+P_{ji})D_{ij},\\
 K-J\ge X_{\rm e}>0,\qquad
 R-1\ge\frac{X_{\rm e}}{K-X_{\rm e}}\quad(J>0),
 \end{gathered}
 \label{eq:edge-gap-bound}
\end{equation}
with the source superscript suppressed on $D,K,J,X_{\rm e},R$. Thus $R>1$ on every weak promoter. Interpreting this coefficient inequality as a strict first-order selection response requires $\phi'(0)>0$, which the normalization of Sec.~\ref{sec:update} fixes to one.
\end{proposition}
\begin{proof}[Proof]
Appendix~\ref{app:path-universal-barrier} proves the fixation comparison by optional stopping and derives the coefficient bound by integrating the same edge inequality over neutral histories.
\end{proof}
In Eq.~\eqref{eq:edge-gap-bound}, $X_{\rm e}/K$ averages the sum of the two directional copying
probabilities with weights proportional to $\gamma_{ij}D_{ij}$.
Approaching $R=1$ therefore requires these weighted disagreement
histories to concentrate on edges with weak copying in both directions.
A strictly increasing map with $\phi'(0)=0$ retains the finite-selection comparison but has vanishing first-order response. Appendix~\ref{app:I-directed-barrier} extends the finite-selection barrier to strongly connected directed copying and to accumulated payoffs; its cut argument replaces the reciprocal edge balance used here.

Two protected domains approach $R=1$ on every fixed path with
$N\ge6$ sites. For $m_L,m_R\ge2$, $N=m_L+m_R+2$, choose
\begin{equation}
 (\underbrace{q^3,\ldots,q^3}_{m_L-1},q,1,q,
        \underbrace{q^3,\ldots,q^3}_{m_R-1}).
 \label{eq:long-path}
\end{equation}
The domain sizes are $m_L+1$ and $m_R+1$. The outer weights $q^3$ are illustrative; for threshold convergence under uniform introduction, any outer weight $Q$ with $q/Q\to0$ suffices on each fixed longer path, by Eq.~\eqref{eq:I-long-path-leak-errors}. The estimates that follow use $Q=q^3$. The two domains can have different holding times, but each cooperative exit probability is $h_{\rm dom}=F_{\rm dom}/(F_{\rm dom}+d_0)$ and
each pin is introduced with probability $1/N$. Thus
 \begin{gather}
 \rho_C=\frac{2h_{\rm dom}}{N}+O(q^{-1}),\quad
 \rho_D=\frac{2(1-h_{\rm dom})}N+O(q^{-1}), \notag \\
 \mathcal S\longrightarrow\frac{b-c}{N},  \label{eq:long-path-law}
 \end{gather}
with $T_{\rm abs}=q+O(1)$ and $T_{\rm fix}=Nq/2+O(1)$, Eq.~\eqref{app:path-general-mean-times}. Appendix~\ref{app:path-long} derives this construction and its controlled reduction.
Appendix~\ref{app:path-general-times} derives the unequal-domain
distributions and variances. The two-domain limit changes the effective competitors. Their relative fecundity is $\phi(\delta(b-c))/\phi(0)$, while their physical holding times retain the domain sizes. Here $N$ is fixed before the weight limit, because the error constants of the controlled elimination depend on the domain sizes (Appendix~\ref{app:path-long}). Domains growing with $q$ require alignment and first-passage bounds uniform in size; Appendix~\ref{app:I-long-path-correction} identifies the first-passage quantities as ancestral hitting-time moments of the follower cores.
The connecting weights stay positive throughout the construction;
setting them to zero first would leave domains unable to exchange traits.

The rate of approach has a one-lineage interpretation. For birth-associated introduction, follower copying contributes at relative order $q^2/Q$, as quantified in Eq.~\eqref{eq:I-long-path-leak-errors}. Strengthen every outer contact in Eq.~\eqref{eq:long-path} to $Q=q^4$ so that this contribution lies below the first correction. For both introduction rules,
\begin{equation}
 R^\mu=1+\frac{N/2+3}{q}+O_N(q^{-2}).
 \label{eq:I-long-path-correction}
\end{equation}
The coefficient is the sum of ancestral mean hitting times from the two adjacent followers to their pins, with a bridge correction. For path cores it depends on total size; the next correction involves second hitting-time moments and distinguishes the core sizes. The same first-passage calculation treats arbitrary fixed connected reciprocal cores attached to their pins at one vertex each. The coefficient becomes $N$ for internally regular cores and depends on the attachment vertex for irregular ones. Appendix~\ref{app:I-long-path-correction} derives the common inverse-strength formula and controls finite leakage. The expansion holds at fixed $N$. As $q\to\infty$, convergence to one already follows from $q/Q\to0$ for uniform introduction and $q^2/Q\to0$ for birth-associated introduction; the choice $Q=q^4$ makes the remainder $O_N(q^{-2})$ under both rules.

\paragraph*{Small-tree obstructions.}

For at most three sites every mixed state is a singleton or its
complement, with adverse benefit drift
\begin{equation}
 (Pe_k)^\T L_\gamma e_k
 =-\sum_{i\ne k}P_{ik}\gamma_{ik}<0.
 \label{eq:small-benefit}
\end{equation}
Solving the six pair equations on $P_4$ gives a negative site-resolved
benefit coefficient $J^{(\ell)}$ for every singleton introduction $\ell$
(Appendix~\ref{app:path-small}). The other four-site tree, the star, has pointwise nonpositive
benefit drift (Appendix~\ref{app:v10-minimal-cycle}). The source-convexity identity~\eqref{eq:I-source-convexity} then excludes promotion under every singleton mixture on these trees.

The barrier, the protected-domain construction, these small-tree exclusions and the five-site envelope of Sec.~\ref{sec:v4-p5-grading} together classify the paths.
\begin{proposition}[Sharp thresholds on paths and supports with a spanning path]
\label{thm:path-classification}
For source-first copying with common-kernel averaged donation payoffs, the infimum over positive path weights with $J>0$ under uniform singleton introduction is
\begin{equation}
 R_*(P_N)=
 \begin{cases}
 +\infty,&2\le N\le4,\\
 7,&N=5,\\
 1,&N\ge6.
 \end{cases}
 \label{eq:path-classification}
\end{equation}
Under birth-associated introduction, $R_*^{\mathrm b}(P_N)=+\infty$ for $2\le N\le5$ and $R_*^{\mathrm b}(P_N)=1$ for $N\ge6$. The infinite value denotes an empty promoter set. Every fixed connected support with $N\ge6$ and a spanning path has threshold infimum one under both rules. All finite infima are strict and unattained at positive weights, with $N$ fixed before the weight limit.
\end{proposition}
\begin{proof}[Proof]
Proposition~\ref{prop:I-barrier} gives $R^\mu>1$ under both rules on every support, so the infima one are unattained. The domain construction of Appendix~\ref{app:path-long} settles the uniform case $N\ge6$, and Eqs.~\eqref{eq:I-long-path-leak-errors} and \eqref{eq:v7-broader-birth-response} settle the birth-associated case. Appendix~\ref{app:v7-support-extension} controls birth-associated arrival and transfers the construction to supports with added contacts. The small-path signs are proved in Appendix~\ref{app:path-small}; the five-site results follow from Proposition~\ref{prop:I-p5-envelope} and the birth-associated sign in Appendix~\ref{app:v7-preparation}.
\end{proof}

\subsection{Sign structure of the threshold}
\label{sec:sign-structure}
The threshold is a ratio of two linear functionals of one linear solve.
This structure turns questions about all positive weightings into sign
questions for explicit polynomials. The same construction organizes the support classification and its sharp threshold bounds.

The pair matrix $\mathsf M$ on the left of Eq.~\eqref{eq:pair-system}
has diagonal entries $\vartheta_i+\vartheta_j$, nonpositive off-diagonal entries
$-P_{ki}$ and $-P_{kj}$, and row sums $P_{ij}+P_{ji}\ge0$, positive on
every edge. Thus $-\mathsf M$ generates the ancestral pair walk in ancestral time $\widehat t=t/N$, killed when the two lineages meet. On a connected support every pair can reach
an adjacent pair, so the walk is killed with probability one, and its
Green function $\mathsf M^{-1}=\int_0^\infty e^{-\widehat t\mathsf M}\,d\widehat t$ is finite
and entrywise nonnegative. Hence $\mathsf M$ is a nonsingular M-matrix,
with $\det\mathsf M>0$, for every positive
weighting~\cite[Chap.~6]{BermanPlemmons1994}. The introduction rule enters only
through the pair source vector $g^\mu$, with entries $g^\mu_{ij}=N(\mu_i+\mu_j)$,
equal to two for uniform introduction, so that $D^\mu=\mathsf M^{-1}g^\mu$.
For a contraction vector $\mathsf c$ on pairs, Cramer's rule gives
\begin{equation}
 \mathsf c^{\T}\mathsf M^{-1}g^\mu
 =-\frac{1}{\det\mathsf M}\det\begin{pmatrix}\mathsf M&g^\mu\\
 \mathsf c^{\T}&0\end{pmatrix}.
 \label{eq:bordered}
\end{equation}
Multiply each row of $(\mathsf M\ g^\mu)$ by the product of the
strengths occurring in its denominators, and clear the denominators of
$\mathsf c$ by a further positive factor; primes mark the cleared
quantities. Take $\mathsf c$ to be the coefficients of $K-r_{\rm lev} J$
in Eq.~\eqref{eq:JK}, with a rational trial threshold $r_{\rm lev}$. Then
\begin{equation}
 K-r_{\rm lev} J=\frac{\mathcal N_{\mathsf c}(w)}{\mathcal D_{\mathsf c}(w)},\quad
 \mathcal N_{\mathsf c}=-\det\begin{pmatrix}\mathsf M'&g^{\mu\prime}\\
 \mathsf c'^{\T}&0\end{pmatrix},
\end{equation}
where, for either named introduction rule, whose source is rational in the weights, $\mathcal N_{\mathsf c}$ is, up to a positive integer factor, an integer polynomial homogeneous in the weights, and $\mathcal D_{\mathsf c}$
is a positive multiple of $\det\mathsf M'>0$. Hence $K-r_{\rm lev} J$ has
the sign of $\mathcal N_{\mathsf c}$ at every positive weighting; $J$
likewise has the sign of its own numerator.

Two consequences organize the grading and support bounds below. First, a polynomial
whose coefficients are all positive proves an inequality such as
$R>r_{\rm lev}$ for every finite positive weighting. Limits of promoting weightings retain the non-strict bound, without a separate classification of their hierarchies. A benefit
numerator with only negative coefficients proves that a support has no
weak promoter. Second, along a hierarchy $w_e\propto\Lambda^{a_e}$ with
$\Lambda\to\infty$, the leading part of a numerator is its restriction to
the face of its Newton polytope, the convex hull of its exponent
vectors $\alpha$, on which $a\cdot\alpha$ is largest. For a numerator
of one sign this leading part cannot cancel. Dominance of a specified monomial requires linear inequalities in the scaling exponents $a_e$. These inequalities describe the normal cones of the Newton polytope, whose relative interiors group scaling directions selecting the same leading face~\cite{MaclaganSturmfels2015}. They identify candidate contact hierarchies. The coefficient comparison below then tests whether their limiting ratios attain a global bound.

The all-minors matrix-tree theorem explains why such positivity is
natural. Cofactors of $\mathsf M$ are sums over rooted spanning forests
of the ancestral pair graph, each weighted by a product of positive
rates~\cite{Chaiken1982,AGNamGunawardena2025}. The cost contraction $K$
inherits this positivity; the benefit contraction mixes signs, and its
numerator weighs the ancestral histories that favor cooperation against
those that do not. The polynomial statement is literal for the two named rational introduction rules; for arbitrary formal $\mu$, the source coordinates remain additional coefficients. Exact rational elimination or interpolation with proved degree and coefficient bounds gives identities on the whole domain. A sign conclusion then follows from an explicit positive denominator and nonnegative coefficients, sometimes after a finite exact change of variables. No extrapolation from sampled weights enters such an argument. Appendix~\ref{app:v10-support-identities} gives the reconstruction and coverage conditions. The fringe, leaky and three-pin inequalities, the small-support signs and the signed five-path grades contain computer-assisted steps using exact coefficient calculations; their appendix proofs specify the relevant polynomials. Separately, Fig.~\ref{fig:finite} compares numerical solutions of the full finite-state chain on $(q^3,q,1,q,q^3)$ with its two-domain limits. The selected advantage for exponential fecundity at $\kappa=0.25$ is compared with $\tfrac13\tanh[\kappa(b/c-1)/2]$, and the neutral mean times with $T_{\rm abs}/q\to1$ and $T_{\rm fix}/q\to3$. Symbolic calculations used SageMath, SymPy and Singular~\cite{SageMathSoftware,SymPySoftware,SingularSoftware}.

A coefficient test often determines the threshold itself. Write a
reduced fraction as $R=U/V$, with
$U=\sum_\alpha U_\alpha w^\alpha$ and
$V=\sum_\alpha V_\alpha w^\alpha$, taking absent coefficients as zero.
If the cost coefficients are positive on their support and some
$V_\alpha>0$, set
\begin{equation}
\theta_{\rm pos}=\min_{V_\alpha>0}\frac{U_\alpha}{V_\alpha}.
\label{eq:v10-coefficient-threshold}
\end{equation}
Every coefficient of $U-\theta_{\rm pos}V$ is nonnegative, so
$R\ge\theta_{\rm pos}$ whenever $V>0$. If $\theta_{\rm pos}>0$ and a minimizing exponent is a
vertex of the convex hull of the combined exponent sets of $U$ and $V$,
an integer direction making that vertex the unique maximizer defines a
hierarchy with $R\to\theta_{\rm pos}$. Both polynomials then have the same
unique leading monomial, with positive benefit coefficient. In all
23 promoting cases among the 32 response fractions reconstructed for sixteen supports under the two introduction rules (Table~\ref{tab:I-reduced-response-cases}), this coefficient bound equals
the proved sharp infimum. The matching domain constructions prove
sharpness even when a chosen minimizing exponent is not a vertex.

A face can contain several monomials. If their coefficients cancel, the comparison passes to subleading terms; nested weight scales require the corresponding successive balances. The coefficient inequalities used here hold before any limit is taken. Newton-polytope faces help locate those constructions without exhausting all possible hierarchies.
\subsection{Outer grading and the exact five-site envelope}
\label{sec:v4-p5-grading}
The dimer calculation of Eq.~\eqref{eq:dimer-response} explains how seven is approached. To prove that
no other weighting does better, we organize the exact ten-pair solution for uniform introduction on $(A,1,q,Q)$
as $R(A,1,q,Q)=U(A,q,Q)/V(A,q,Q)$ after clearing a positive denominator. The polynomials use the three free weights after the second contact is fixed to one. Then $U>0$
and $\operatorname{sgn}V=\operatorname{sgn}J$. Grading these
polynomials by total degree in the outer weights $A,Q$ gives
\begin{equation}
 \begin{gathered}
 U=\sum_{m=0}^{15}U_m,\quad V=\sum_{m=0}^{15}V_m,
 \\
 U_m\ge0,\quad V_m<0\ (m<15).
 \end{gathered}
 \label{eq:v4-p5-grades}
\end{equation}
The forest argument in Appendix~\ref{v4:alg:p5-forest} proves
the nonnegative cost grades. A forest selects at most one outgoing
transition from each pair state, so denominator powers can be bounded
and cleared term by term before any large polynomial is expanded. Exact coefficient comparison fixes the negative benefit grades; the exceptional grade $-V_{14}$ is a square term plus a positive-coefficient polynomial. 

The lower grades are the finite-contact penalties relative to the limiting outer contest. The following statement makes their effect on every positive weighting explicit.
\begin{proposition}[Outer-degree monotonicity and the five-site envelope]
\label{prop:I-p5-envelope}
For source-first copying with common-kernel averaged donation payoffs and uniform singleton introduction on $P_5=(A,1,q,Q)$, the following hold. (i) Fix $A,q,Q>0$. Along the outer ray $(A,Q)\mapsto(\ell A,\ell Q)$, a promoting interval exists exactly when its top benefit grade $V_{15}(A,q,Q)$ in Eq.~\eqref{eq:v4-p5-grades} is positive. In that case $J$ has a unique zero as a function of $\ell$, the ray promotes above this scale, and throughout that interval
\begin{equation}
 \frac{d}{d\ell}R(\ell A,1,q,\ell Q)<0.
 \label{eq:v4-p5-ray}
\end{equation}
(ii) Orient the path by reversal and common rescaling so that $q\ge1$, and define
\begin{equation}
 \Pi_5(q)=11q^4-16q^3-124q^2-128q-33.
 \label{eq:v4-p5-inner-polynomial}
\end{equation}
Promoting outer weights exist exactly when $q>q_c$, the unique positive root $q_c=4.536774892089\ldots$ of $\Pi_5$. For these $q$ the sharp envelope is
\begin{equation}
 \begin{aligned}
 R_5^{\rm env}(q)&=\inf_{A,Q>0,J>0}R(A,1,q,Q)\\
 &=\frac{2(q+1)^2(11q^2+28q+11)(7q+10)}{(2q+3)\Pi_5(q)}.
 \end{aligned}
 \label{eq:v4-p5-envelope}
\end{equation}
It is approached by $(A,Q)=(L,L^2)$ as $L\to\infty$, is not attained at positive $(A,Q)$, and decreases strictly from infinity to seven as $q$ increases from $q_c$, with $R_5^{\rm env}(q)=7+83/(2q)+O(q^{-2})$.
\end{proposition}
\begin{proof}[Proof]
The forest normalization in Appendix~\ref{v4:alg:p5-forest} fixes the cost signs structurally. Appendix~\ref{v4:p5} obtains the benefit-grade signs by exact symbolic coefficient extraction, then proves the ray monotonicity and the envelope.
\end{proof}
Strengthening the outer contacts suppresses the lower-grade penalties relative to the leading contest. The positive top-grade factorization in Eq.~\eqref{eq:v10-p5-strict-grade-bound} and the lower-grade signs give $U-7V>0$ before any limit is taken. The top-degree quotient is fractional-linear and increasing in $A/Q$ on the promoting domain, so its value as $A/Q\to0$ is the envelope. Explicitly, the ray condition is $Q\,\Pi_5(q)>-A\,q^4\Pi_5(1/q)$. Thus $q>q_c$ permits a suitable outer weighting; an individual ray still has to satisfy its own condition $V_{15}>0$.

\subsection{All supports with at most five sites}
\label{sec:five-site-supports}
We now compare all connected supports with at most five sites, using the contact constructions and sign arguments of Sec.~\ref{sec:sign-structure}. Adding an edge cannot raise a threshold infimum, because a vanishing
weight on the new edge recovers the old support without disconnecting
it (Appendix~\ref{app:v7-support-extension}). Upper bounds therefore
pass to supergraphs and lower bounds to spanning subgraphs. The lower
bounds on the classified supports follow from the sign arguments of
Sec.~\ref{sec:sign-structure}. On the four dense supports, Proposition~\ref{prop:I-barrier} gives the proved lower endpoint $R_*\ge1$ in Table~\ref{tab:five-site}; the conjectured bound $3/2$ would match the constructed upper value.

\begin{table}[!tbp]
\caption{Threshold infima of the connected five-vertex supports under
uniform ($R_*$) and birth-associated ($R_*^{\mathrm b}$) introduction;
$+\infty$ denotes an empty promoter set, so no positive weighting favors cooperation for any benefit--cost ratio at first order in selection. Every proved finite infimum is approached by a positive-weight hierarchy and is not attained. The intervals are the proved ranges for the four dense supports marked $\dagger$
(the wheel $W_4$, a four-cycle with a central hub; $K_5$ less a two-edge path; $K_5$ less an edge; and $K_5$). Appendix~\ref{app:v10-support-identities} names the individual sparse supports and fixes the contact order used in the polynomial calculations. For these four the value $3/2$ is approached by explicit hierarchies, and the matching lower bound is Conjecture~\ref{conj:v10-dense}.}
\label{tab:five-site}
\begin{ruledtabular}
\begin{tabular}{lccc}
Support class&Supports&$R_*$&$R_*^{\mathrm b}$\\
\hline
Contains $C_5$, proved floors&4&$3/2$&$3/2$\\
Contains $C_5$, dense$^\dagger$&4&$[1,3/2]$&$[1,3/2]$\\
Longest cycle $C_4$&6&$2$&$2$\\
Spanning path, no $C_4$ or $C_5$&4&$7$&$+\infty$\\
No spanning path&3&$+\infty$&$+\infty$
\end{tabular}
\end{ruledtabular}
\end{table}

\begin{proposition}[Small supports]
\label{prop:small-supports}
For source-first copying with common-kernel averaged donation payoffs and either uniform or birth-associated introduction, every connected support with at most
four sites has no weak promoter except $C_4$, the diamond ($K_4$ with one edge removed) and $K_4$,
whose threshold infimum is three. Among the twenty-one connected
five-vertex supports, seventeen have the threshold infima listed in
Table~\ref{tab:five-site}. The remaining four, the wheel $W_4$ (a four-site cycle joined to a central hub), $K_5$
less a two-edge path, $K_5$ less an edge and $K_5$, satisfy
$1\le R_*\le3/2$ and $1\le R_*^{\mathrm b}\le3/2$.
Across connected supports with $N\le5$, birth-associated introduction admits a weak promoter precisely when the support contains a $C_4$ or $C_5$ as a subgraph. Uniform introduction admits one precisely on a four-site support containing $C_4$, or a five-site support containing a spanning $P_5$, a path that visits every vertex once using retained edges.
\end{proposition}
\begin{proof}[Proof]
Exact coefficient arithmetic shows that the numerators of $K-3J$ on $C_4$, the diamond and $K_4$, of $J$ on the
paw (a triangle with a pendant vertex), and the five-site numerators in Appendix~\ref{app:c5-supports} have coefficients of a single sign under both rules. The five-path uniform bound is Proposition~\ref{prop:I-p5-envelope}. Each
upper value is approached by an explicit hierarchy (Appendix~\ref{app:v10-support-identities}), transferred to supergraphs by monotonicity (Appendix~\ref{app:v7-support-extension}). Supports with at most three sites have
adverse benefit drift by Eq.~\eqref{eq:small-benefit}, and $P_4$ and the
stars $K_{1,3}$ and $K_{1,4}$ are treated in Appendices~\ref{app:path-small} and \ref{app:v10-minimal-cycle}, respectively; the pointwise star argument applies to any number of leaves. The existence criteria follow from these exclusions and the grouping in Table~\ref{tab:five-site}. On the four dense supports Proposition~\ref{prop:I-barrier} gives the lower endpoint one under both rules, and the five-cycle hierarchy transferred to these supergraphs gives the upper endpoint $3/2$. These hierarchies contain weak promoters, so the existence criteria do not depend on Conjecture~\ref{conj:v10-dense}.
\end{proof}

In particular no support with fewer than four sites or fewer than four
edges promotes, $C_4$ is the only promoting support with four sites and
four edges, and no support with at most four sites has a threshold below
three. The value two on five-vertex supports whose longest cycle has four
sites is approached with $(a_1,a_2)=(1/2,1)$ in Eqs.~\eqref{eq:agreements} and \eqref{eq:contested-threshold}; the matching lower bound is the single-sign numerator of $K-2J$ in Appendix~\ref{app:v10-support-identities}. On $C_4$ with a pendant a strongly
coupled pair is contested by pin $p_1$ and by an aligned follower of
pin $p_2$, so $\mathfrak a=1/2$ and $R\to2$. In the hierarchy of
Appendix~\ref{app:v10-support-identities}, $p_1=3$ and $p_2=2$.
The pin $p_1$, whose agreement probability is $1/2$, overwrites $p_2$
about $q$ times more often than the reverse, and
$\xi_{\rm pin}=s_{p_2}/s_{p_1}\asymp q$. The uniform and
birth-associated amplitudes scale as $q^{-1}$ and $q^{-3}$, respectively,
along this hierarchy.

\begin{conjecture}[Dense five-site supports]\label{conj:v10-dense}
For every connected loopless reciprocal five-site graph, positive weights on its retained edges, and either uniform or birth-associated singleton introduction,
\begin{equation}
 2K^\mu-3J^\mu\ge0.
 \label{eq:v10-dense-conjecture}
\end{equation}
On $J^\mu>0$ this means $R^\mu\ge3/2$.
\end{conjecture}
 The proposed bound agrees with the sharp five-cycle limit and with the analytical regions and controlled hierarchy limits described below. Exact evaluations at sampled rational weights and numerical searches give further support; their coverage is specified in Appendix~\ref{app:I-dense-finite-tests}.

It suffices to prove the inequality on $K_5$. Any connected five-site weighting is the limit of complete weightings obtained by assigning $\epsilon>0$ to its missing contacts. The killed-pair inverse and both named introduction distributions remain continuous at the connected limit, so the undivided inequality passes to it, including when $J^\mu=0$. The case $K_5$ is itself one of the four. The value $3/2$ is approached even with every contact positive. At each fixed point of the five-cycle hierarchy, sufficiently small weights on the missing contacts change $R$ arbitrarily little. Thus one complete-support inequality would settle all four unresolved infima.

Finite contact contrast admits a complementary analytical bound on the complete graph $K_5$. Let $H=\max_{i<j}w_{ij}/\min_{i<j}w_{ij}$, the ratio of its largest and smallest contacts. For every common selection-independent singleton introduction rule, $H\le7/4$ implies $J^\mu<0$, while $H\le11/3$ implies $2K^\mu-3J^\mu>0$. Proposition~\ref{prop:I-source-spectral} of Sec.~\ref{sec:v4-observable-algebra} separates an adverse stationary ancestral contribution from the relaxing modes. A bounded spectral filter, Eq.~\eqref{eq:I-main-mode-filter} of Sec.~\ref{sec:v4-observable-algebra}, controls the latter through strength imbalance, without dividing by a small relaxation rate (Appendix~\ref{app:I-dense-bulk}). These regions allow arbitrary unequal contacts within their stated ranges. Their endpoints combine separate bounds on ancestral relaxation, the stationary contribution and strength imbalance; sharper endpoints require a tighter comparison of those contributions. Approaching a missing contact sends $H$ to infinity, beyond this finite-contrast estimate.

The leading copying graphs distinguish which candidate competitors survive the fastest dynamics. Suppose $w_e=p_e\Lambda^{a_e}$ with fixed $p_e>0$ and a unique largest incident exponent at each site. Draw its leading copying arrow toward that neighbor. Reciprocity permits only mutual pairs as directed cycles. On five sites the closed classes of the backward lineage walk, mutually reachable sets it cannot leave, therefore consist of one cold singleton, a cold singleton together with a closed pair, two cold singletons, or three cold singletons. A cold singleton has no incoming leading copying arrow. Appendix~\ref{app:I-leading-hierarchies} proves this classification. With several closed classes, the threshold also depends on the rates of rare transitions between them and the introduction mass reaching each class. These quantities require retaining the slower connecting contacts after the fastest copying graph has been identified.

Several exact results delimit the question. Bridge-dominated two-pin limits obey the voltage bounds of Proposition~\ref{prop:v10-contest}. Fixed positive two-pin shapes with residual follower-to-pin copying retain the $3/2$ floor under arbitrary singleton preparation (Appendix~\ref{app:v10-leaky-pins}). Three protected sites with a two-site fringe have a strict limiting bound $5/2$ under either named rule, with infimum $5/2$ when the fixed shape parameters approach their boundary after the hierarchy limit (Appendix~\ref{app:v10-three-pins}). The pin--dimer family in Appendix~\ref{app:v10-tied} keeps several contacts tied, with finite nonzero ratios in the first hierarchy limit. A subsequent limit of those ratios gives infima $3/2$ under uniform and $5/2$ under birth-associated introduction. The readout cancellation of Sec.~\ref{sec:effective-memory} reduces this calculation to two two-dimensional harmonic problems. The leaky, three-pin and tied bounds hold with their shape coefficients fixed and positive during the first hierarchy limit; further ordered limits of the shapes give the three-pin and tied infima. If a shape ratio degenerates together with the hierarchy parameter, an omitted transition can become comparable with those retained in the slow generator. Covering these joint limits requires an elimination with errors uniform in the shapes. 

The exchangeable pair--triple family in Eq.~\eqref{eq:v10-dense-block} is nonpromoting at weak selection for every singleton source. Exact rational evaluation of the ten-pair system gave $2K^\mu-3J^\mu>0$ under both rules on 200 dense graphs from logarithmic contact sets and structured families motivated by the one-pin and five-cycle limits. A separate search from 3000 double-precision starting points, followed by bounded coordinate refinement, found no violation. Appendix~\ref{app:I-dense-finite-tests} specifies the samples, normalized objective, refinement rule and exact checks of the stored minimizing contacts. Completing the conjecture requires a bound on the remaining high-contrast region, including unrestricted nested hierarchies; the finite tests check particular weights in that region. 

Among paths, six sites are necessary and sufficient for approaching threshold one (Proposition~\ref{thm:path-classification}). The conjecture would extend that minimum to all reciprocal supports under the two named introduction rules.

\subsection{The same support under a different introduction rule}
\label{sec:v10-preparation-comparison}
\label{sec:access}
The support criterion separates two ways of promoting a lineage. By Proposition~\ref{prop:small-supports}, a four- or five-cycle is necessary and sufficient on at most five sites under birth-associated introduction, while uniform introduction also admits the four five-site supports with a spanning path and no such cycle (Table~\ref{tab:five-site}). Promotion sustained by the pin--dimer establishment mechanism is therefore sensitive to preparation; the contested-fringe constructions admit both sources. Figure~\ref{fig:paths} compares the representative path mechanisms.

\begin{table}[tbp]
\caption{Sharp weak-selection threshold infima under uniform ($U$) and birth-associated ($\mathrm b$) introduction. Each infimum is over all positive contact weights on the indicated support. The value $+\infty$ means that no weighting favors cooperation at weak selection under that introduction rule.}
\label{tab:I-introduction-thresholds}
\begin{ruledtabular}
\begin{tabular}{lcccc}
Support&$C_4$&$C_5$&$P_5$&$P_N\ (N\ge6)$\\
\hline
$R_*$&$3$&$3/2$&$7$&$1$\\
$R_*^{\rm b}$&$3$&$3/2$&$+\infty$&$1$
\end{tabular}
\end{ruledtabular}
\end{table}

Birth-associated introduction samples incoming copying flux. Under constant fitness, uniform and birth-associated introduction can already classify the same graph differently~\cite{AdlamChatterjeeNowak2015,PavlogiannisEtAl2018,AllenEtAl2021}. Heterogeneous game models also require the place of mutation to enter the success criterion~\cite{MaciejewskiFuHauert2014}. Here the common kernel is fixed while the paired introduction rule changes. Table~\ref{tab:I-introduction-thresholds} compares the sharp threshold infima on the representative paths and cycles.

At fixed weights the source-convexity bound~\eqref{eq:I-best-source} identifies the best possible singleton introduction. The two named rules constrain this freedom differently. The five-site path loses all weak promoters under birth-associated introduction because its frequently overwritten dimer receives introductions while the protected pin receives few. On $(q^2,1,q,q^3)$, $\chi_b^{\rm b}=-14/(5q^2)+O(q^{-3})$ and $\chi_c^{\rm b}=8/(5q)+O(q^{-2})$; Appendix~\ref{app:v7-preparation} proves the all-weight exclusion. The contested five-cycle keeps its threshold, with the suppressed amplitude in Eq.~\eqref{eq:v10-circuit-access}. 

On the six-site path $(Q,q,1,q,Q)$ of Eq.~\eqref{eq:p6-weights}, the effect of introduction follows from one joint expansion.
With $u=1/q$ and $v=q/Q$, regularizing the pair equations throughout the
protected corner $u,v\to0$ gives
\begin{equation}
 \begin{aligned}
 \chi_b^{\mathrm b}&=\frac{u-2v}{6}+O((u+v)^2),\\
 \chi_c^{\mathrm b}&=\frac u6+\frac{22v}{9}+O((u+v)^2).
 \end{aligned}
 \label{eq:I-birth-joint-corner}
\end{equation}
The pin-source contribution is of order $u$, while follower leakage
changes both coefficients at order $v$. Within this corner the
birth-associated threshold approaches one if and only if $v=o(u)$,
equivalently $q^2/Q\to0$. The choice $Q=q^3$ meets this condition;
uniform introduction already permits $Q=q^2$.

The competing scale $v\asymp u$ is a specialization of the same
expansion. Put $\lambda_{\rm leak}=v/u=q^2/Q$. At fixed
$0<\lambda_{\rm leak}<1/2$,
\begin{equation}
 R^{\mathrm b}\longrightarrow
 \frac{3+44\lambda_{\rm leak}}{3-6\lambda_{\rm leak}}.
 \label{eq:I-birth-leak-threshold}
\end{equation}
The benefit response is eventually adverse for fixed
$\lambda_{\rm leak}\ge1/2$, with the equality case determined at the
next order. The joint expansion also makes the stronger-leakage regime
$q\ll Q\ll q^2$ adverse. Appendix~\ref{app:I-birth-leak} proves the
expansion, the threshold-one criterion and the critical sign. This competition between follower and pin introductions is the path counterpart of Eq.~\eqref{eq:I-c5-boundary-threshold}; the two geometries have different threshold functions.

Table~\ref{tab:I-finite-contrasts} compares representative finite weightings. Their contact contrast $H$ is the largest positive weight divided by the smallest. On the second six-site example, a pin sends a fraction $1/101$ of its offspring across the bridge, while its adjacent follower sends a fraction $1/10001$ back to the pin. These fractions quantify the asymmetry in offspring routing. The tabulated thresholds follow from first-order response at zero selection for fixed weights; the finite-selection domain limits of Sec.~\ref{sec:finite-selection} take the contact hierarchy at fixed selection intensity.
\begin{table*}[t]
\caption{Finite contact examples under averaged payoffs and source-first copying with $P_{ij}=w_{ij}/s_i$. Weights follow consecutive path or cyclic order. The contrast is $H=\max_e w_e/\min_e w_e$. The weak-selection thresholds $R^U$ and $R^{\mathrm b}$ use uniform and birth-associated introduction, respectively. All entries are rounded from exact rational solutions, checked independently against the full configuration chain. No promotion means that the benefit coefficient is negative for that rule.}
\label{tab:I-finite-contrasts}
\begin{ruledtabular}
\begin{tabular}{ccrrr}
Support & Consecutive weights & $H$ & $R^U$ & $R^{\mathrm b}$\\
\hline
$P_6$ & $(10^3,10,1,10,10^3)$ & $10^3$ & $2.317109$ & $11.494431$\\
$P_6$ & $(10^6,100,1,100,10^6)$ & $10^6$ & $1.063957$ & $1.249057$\\
$C_4$ & $(1,40,64000,40)$ & $64000$ & $3.677662$ & $6.276691$\\
$C_5$ & $(1,64,8388608,8388608,64)$ & $8388608$ & $1.602545$ & $1.612632$\\
$P_5$ & $(6400,1,80,40960000)$ & $40960000$ & $7.803543$ & No promotion
\end{tabular}
\end{ruledtabular}
\end{table*}
Along the six-site hierarchy approaching threshold one, the benefit susceptibility tends to $1/6$ under uniform introduction and is asymptotic to $u/6$ under birth-associated introduction, Eq.~\eqref{eq:I-birth-joint-corner}. The small pin source suppresses the amplitude, while $v=o(u)$ keeps follower-originated contributions small relative to that source.
Appendix~\ref{app:v7-support-extension} controls the source estimates
on larger protected supports.

The neutral baseline is needed to interpret this suppression. On the
symmetric six-site path with $v=o(u)$, Eq.~\eqref{eq:v7-birth-normalization}
gives $\rho_0^{\mathrm b}\sim u/6$, while $\rho_0^U=1/6$.
Together with Eq.~\eqref{eq:I-birth-joint-corner} and
$\left.\partial_\delta\rho_C^\mu\right|_0=\mathcal S^\mu/2$,
these relations give
$\left.\partial_\delta\rho_C^\mu\right|_0/\rho_0^\mu\to(b-c)/2$
under either rule. The derivative is taken before the hierarchy limit,
at fixed $b,c$. The absolute response per new mutant is suppressed under
birth-associated introduction, while its leading response relative to
neutral fixation is the same in this controlled limit.

\section{Neutral evolution and relaxation}
\label{sec:v4-observable-algebra}
The neutral ancestry behind the support bounds also determines finite-time evolution. On paths and cycles, site ordering reduces the full propagator and the response denominator to one-lineage matrices. This reduction calculates relaxation and pair response without diagonalizing the $2^N$-configuration generator. A second use of one-lineage modes separates the adverse stationary contribution to cooperation from the effect of unequal site strengths.
\subsection{An exterior representation of the full propagator}
A binary configuration supports a finite algebra of observables. With $\sigma_i=2n_i-1$, functions agreeing on every configuration are identified in
\begin{equation}
 \mathscr A_N=\mathbb R[n_1,\ldots,n_N]/(n_i^2-n_i)
 =\mathbb R[\sigma_1,\ldots,\sigma_N]/(\sigma_i^2-1).
 \label{eq:v4-boolean}
\end{equation}
Thus squarefree monomials $\sigma_S=\prod_{i\in S}\sigma_i$, products with distinct site indices, form a commuting basis, with $\sigma_S\sigma_T=\sigma_{S\triangle T}$. Here the symmetric difference $S\triangle T$ contains the indices belonging to exactly one of the two sets. Neutral copying preserves the nested spaces of observables of degree at most $k$. Each diagonal block in this degree ordering describes ancestral walkers stopped at their first collision. Reciprocity makes these blocks reversible. Block triangularity then makes the full spectrum the union of their real spectra, even when collisions couple different degrees. On a path or cycle, site ordering also permits these collision couplings to be removed by a triangular transformation.

In global-attempt time the one-lineage coefficient matrix has $\mathsf G_{ij}=P_{ij}/N$ for $i\ne j$ and $\mathsf G_{jj}=-\vartheta_j/N$. Thus $\mathcal A_0\sigma_j=\sum_i\mathsf G_{ij}\sigma_i$. Let $\mathsf G_{{\rm obs},0}$ denote the coefficient matrix of $\mathcal A_0$ on all ordered spin monomials. Simultaneous reversal of all traits separates even and odd spin degrees, denoted $p=e,o$; $\mathcal K_p$ contains the corresponding degrees from zero to $N$. On a path set $\mathsf G_e=\mathsf G_o=\mathsf G$. On a cycle set $\mathsf G_o=\mathsf G$ and reverse the two closing-edge off-diagonal signs to obtain $\mathsf G_e$, leaving its diagonal unchanged. The exterior power $\Lambda^k$ propagates $k$ antisymmetrized lineage indices; its matrix elements are minors of the one-lineage propagator. This auxiliary coefficient space represents the commuting physical observables.

\begin{proposition}[Neutral spectrum, propagation and relaxation]
\label{prop:I-neutral}
For neutral copying on every connected loopless reciprocal support with $N\ge2$ and fixed positive retained weights, the configuration spectrum is real and contained in $[-1,0]$ in global-attempt time. The eigenvalue $-1$ occurs, and zero has algebraic and geometric multiplicity two. The slowest nonzero decay rate is
\begin{equation}
 \lambda_{\rm rel}=\frac1N\lambda_{\min}(\mathsf M).
 \label{eq:I-relaxation-pair}
\end{equation}
Every initial distribution supported on mixed configurations has this absorption-tail exponent, also after conditioning on either consensus outcome.

On a path or cycle the full neutral generator is diagonalizable, and there are invertible triangular transformations $W_p$ such that
\begin{equation}
 \left.e^{t\mathsf G_{{\rm obs},0}}\right|_p
 =W_p\left[\bigoplus_{k\in\mathcal K_p}
       \Lambda^k(e^{t\mathsf G_p})\right]W_p^{-1}.
 \label{eq:v3-neutral-exterior}
\end{equation}
The same transformation serves both parity sectors on a path.
\end{proposition}
\begin{proof}[Proof]
Appendices~\ref{app:I-neutral-real-spectrum}, \ref{app:I-relaxation-gap} and \ref{app:v3-neutral-cas} prove the spectrum, relaxation and propagation statements, respectively.
\end{proof}
The identification of $\lambda_{\rm rel}$ with the ancestral-pair survival exponent has finite voter- and invasion-process antecedents~\cite{BenAriEtAl2021,AllardEtAl2024}; the proof retains the weighted physical clock and both conditioning outcomes.

Reality of the degree-block spectra alone does not exclude Jordan blocks coupling equal nonzero eigenvalues of different degrees. The path and cycle transformation $W_p$ removes those couplings by subtracting lower-degree collision terms. In the exterior coefficient space, copying moves an index or removes a colliding pair. Pair-removal operators commute, so their exponential cancels the collision term after one commutator when its coefficients solve the Sylvester equation in Appendix~\ref{app:v3-neutral-cas}. On a path this replaces $\sigma_i\sigma_j$ by $\sigma_i\sigma_j-1$, which vanishes when the two ancestral indices coincide. On a cycle the antiperiodic boundary condition in the even sector corrects the sign acquired when an ordered lineage crosses the closing edge.

The construction connects the full weighted dynamics to the free-fermion structure of one-dimensional copying models~\cite{Glauber1963,Felderhof1971,MayerSollich2004}. Determinantal propagation of noncolliding walks~\cite{KarlinMcGregor1959} and exact multipoint distributions for spatially inhomogeneous nearest-neighbor coalescing walks~\cite[Secs.~1.1--1.2]{GarrodEtAl2018} are related ancestral formulations. This similarity transformation applies despite the lack of configuration detailed balance. The pair solve becomes a Sylvester equation, a linear matrix equation $AX+XB=C$ whose unknown is antisymmetric here, with $O(N^3)$ dense work and $O(N^2)$ storage. Appendix~\ref{app:v3-neutral-cas} constructs $W_p$ by cancelling the collision term, and Appendix~\ref{app:v3-sylvester-time} derives the pair inverse in the physical clock. On a branching support, the same hop can acquire opposite ordering signs for different lineage positions within one parity sector. Extending this construction therefore requires accounting for the other occupied indices, which a single quadratic hopping coefficient cannot do (Appendix~\ref{app:v3-neutral-cas}).

\subsection{Coalescence, relaxation and the pair inverse}
A mixed population contains a disagreeing pair, and its two ancestral lineages must remain distinct. Conversely, at fixed positive weights the slowest pair-survival mode has a positive overlap with every mixed initial configuration. This gives the relaxation statement of Proposition~\ref{prop:I-neutral} on every connected reciprocal support; Appendix~\ref{app:I-relaxation-gap} makes the overlap and survival bounds explicit.
 Introduction and conditioning change the amplitudes and mean times. The tail-comparison factor $\min_i\pi_i$ can vanish along a contact hierarchy. Appendix~\ref{app:v10-domains} retains the entrance masses and physical clock to determine these limits for the protected hierarchies studied here. On a path the rate is the first
positive eigenvalue of $-\mathsf G$. On a cycle, if
$\nu_0^{(e)}\le\nu_1^{(e)}\le\cdots$ are the positive eigenvalues of
$-\mathsf G_e$, it is $\nu_0^{(e)}+\nu_1^{(e)}$.
For equal cycle weights this becomes
$\lambda_{\rm rel}=2[1-\cos(\pi/N)]/N$. This is the lowest decay rate of the separation between two ancestral lineages, which jumps in either direction at rate $1/N$ and is absorbed when the lineages meet.
The rate fixes the slowest exponential tail toward consensus; mean
fixation times also depend on the entrance amplitudes. On the six-site hierarchy $q\lambda_{\rm rel}\to1/3$, reproducing the domain tail $e^{-\widetilde t/3}$. On the five-cycle at fixed interior $\beta$ and vanishing leakage, $q\lambda_{\rm rel}\to2/5$. These spectral scales follow from the reduced generators, independently of the preparation-dependent entrance mass.

The spectral reduction also clarifies the five-site sign calculation. Let $\lambda_1^{\rm anc},\ldots,\lambda_4^{\rm anc}$ be the positive one-lineage rates on $P_5$, with the common factor $1/5$ removed, and $\mathsf e_k$ their elementary symmetric polynomials, defined by $\prod_{a=1}^4(z+\lambda_a^{\rm anc})=\sum_{k=0}^4\mathsf e_kz^{4-k}$. The pair spectrum contains those four rates and their six pair sums. Hence
\begin{align}
 \det\mathsf M&=\prod_a\lambda_a^{\rm anc}
       \prod_{a<b}(\lambda_a^{\rm anc}+\lambda_b^{\rm anc})\notag\\
 &=\mathsf e_4(\mathsf e_1\mathsf e_2\mathsf e_3-\mathsf e_3^2-\mathsf e_1^2\mathsf e_4)>0,\qquad \mathsf e_1=5.
 \label{eq:v10-p5-spectrum}
\end{align}
The factor in parentheses is the staircase Schur polynomial~\cite{Macdonald1995}. Its determinant formula evaluates the denominator from a one-lineage characteristic polynomial. Directed forests then count the denominator powers required for the cost numerator and prove coefficientwise nonnegativity of the cost numerator (Appendix~\ref{v4:alg:p5-forest}). The benefit readout combines forest terms of both signs; the graded coefficient comparisons in Sec.~\ref{sec:v4-p5-grading} resolve their balance. The spectral factorization and forest expansion thus simplify the denominator and cost numerator, while the benefit comparison completes the threshold bound.
On a reciprocal support with $N\ge3$, one-lineage modes also bound the response without solving the pair problem for every introduction rule. Let $\mathsf R=\operatorname{diag}(s_i^{-1})$, let $\mathsf A_{\rm s}=\mathsf R^{-1/2}L_\gamma\mathsf R^{-1/2}$ be the symmetric ancestral rate matrix with the clock factor removed, and put $\mathsf P_{\rm s}=\mathsf R^{1/2}P\mathsf R^{-1/2}$. In an orthonormal eigenbasis of $\mathsf A_{\rm s}$, with rates $0=\lambda_0<\lambda_1\le\cdots$, filter the skew-symmetric part of $\mathsf P_{\rm s}$ by
\begin{equation}
 (\mathsf T_{\rm filt})_{ab}
 =\frac{\lambda_a-\lambda_b}{\lambda_a+\lambda_b}
 \left(\frac{\mathsf P_{\rm s}-\mathsf P_{\rm s}^{\T}}2\right)_{ab},
 \quad (\mathsf T_{\rm filt})_{00}=0.
 \label{eq:I-main-mode-filter}
\end{equation}
The fraction is used away from $(a,b)=(0,0)$, including coincident positive rates, for which its numerator vanishes. After returning to the site basis, define the local response ratio $\mathfrak r_i=[(\mathsf T_{\rm filt})_{ii}-\pi_i]/(1-\pi_i)$.
\begin{proposition}[Source-independent spectral comparison]
\label{prop:I-source-spectral}
On a connected loopless reciprocal support with $N\ge3$ and positive retained weights, use averaged donation payoffs and the same encounter and replacement kernel. For every common selection-independent singleton introduction distribution $\mu$,
\begin{equation}
 \frac{J^\mu}{K^\mu}=\sum_i\zeta_i^\mu \mathfrak r_i\le\max_i \mathfrak r_i,
 \qquad \zeta_i^\mu=\frac{\omega_i^\mu(1-\pi_i)}{2K^\mu},
 \label{eq:I-main-local-response}
\end{equation}
where $\omega_i^\mu=s_i^{-1}[N\mu_i+\sum_jP_{ji}D_{ij}^\mu]\ge0$ and $\sum_i\zeta_i^\mu=1$. The local ratios $\mathfrak r_i$ depend on the weighted graph alone. Equal site strengths give $J^\mu/K^\mu=-1/(N-1)$ for every such source.
\end{proposition}
\begin{proof}[Proof]
Appendix~\ref{app:I-dense-bulk} derives the identity from the pair equation and proves nonnegativity and normalization of its averaging weights, basis independence of the filter and the equal-strength specialization.
\end{proof}
Introduction changes the weights in this convex average, so $\max_i\mathfrak r_i$, obtained from an $N$-dimensional spectral problem, bounds every common selection-independent singleton source. Computing the actual average still requires the pair response. Equal strengths make the skew part vanish, recovering the isothermal obstruction for every source. Unequal strengths can overcome this adverse response through the relaxing modes. The filter stays bounded as rates become small, which makes it useful for the finite-contrast comparison on $K_5$ in Sec.~\ref{sec:five-site-supports}.

\section{Selection, games and contact structure}
\label{sec:extensions}
The collective mechanisms can be tested by changing selection, payoff accounting and contacts. We first compare finite-selection effects on an aligned domain, a fluctuating fringe and dimer establishment, then vary the payoff matrix, encounter kernel and copying rule.
\subsection{Selection within a collective contest}
\label{sec:contested-selection}
On each fixed protected path with at least six sites the competing domains are internally homogeneous. Their
limiting advantage in Eq.~\eqref{eq:long-path-law} is positive for every
$b>c$ and fixed admissible positive selection strength. A contested chain is itself selected. Let $\widetilde\varpi$ be the invariant distribution of its fast generator before division by the total fecundity $F_\Sigma=\sum_iF_i$. Its physical stationary weights are proportional to $F_\Sigma\widetilde\varpi$, which cancels the local total-fecundity denominator in averaged rare exit rates. The common normalization still sets residence times. The contest therefore involves averages of fecundities; for exponential fecundity, replacing the fluctuating payoff by its mean changes the result. Appendix~\ref{app:c5-reduction} derives this clock weighting, and Appendix~\ref{app:c5-selection} retains it.

For linear fecundity on the
five-cycle at $\beta=1/2$, the hierarchy limit of $\Delta$ is an
explicit rational function of $\kappa$ and $r=b/c$
(Appendix~\ref{app:c5-selection}). The first-order changes of the two
pins' agreement probabilities cancel, and at second order selection
lowers their sum by $\kappa^2(r+1)(r+2)/9$. The limiting threshold
rises as
\begin{equation}
 r_{\rm crit}(\kappa)=\frac32+\frac{35}{16}\kappa^2-\frac{35}{32}\kappa^3
 +O(\kappa^4),
 \label{eq:c5-rcrit}
\end{equation}
with $r_{\rm crit}(1/2)=7/3$, and no benefit--cost ratio favors
cooperation in this limit for $4/5\le\kappa<1$. For exponential fecundity
every $\kappa>0$ has exactly one threshold,
$r_{\rm crit}=3/2+57\kappa^2/64+O(\kappa^4)$, which increases to
$4-2\ln2/\kappa+O(e^{-\kappa}/\kappa)$ at strong selection
(Fig.~\ref{fig:c5}(b)). Appendix~\ref{app:c5-selection} derives the threshold equations, and Appendix~\ref{app:v10-exponential-monotonicity} proves strict increase for the exponential map.

The wall dynamics explain the two selection limits through the auxiliary invariant distribution $\widetilde\varpi$ of the unnormalized generator. With linear fecundity at fixed $0<\kappa<1$, a larger benefit raises the cooperative pin's reward while strengthening the defective interior neighbor that removes its cooperative fringe. The auxiliary weight of that rewarding incursion decreases as $1/r$. The pin-fecundity averages entering Eq.~\eqref{eq:c5-exit} consequently approach $5(1-\kappa)$ and one as $r\to\infty$, identifying the boundary $\kappa=4/5$. The exact sign comparison establishes the exclusion at every finite benefit for $4/5\le\kappa<1$. With exponential fecundity, the strong-selection threshold follows a different balance. The wall state with one cooperative chain site has auxiliary weight $2e^{-\kappa-\kappa r/2}$ relative to the state with no cooperative chain sites and gives its cooperative pin fecundity $e^{\kappa(r-1)}$. Their product competes with a defective pin of fecundity one, giving $2e^{\kappa(r/2-2)}\sim1$ and hence the asymptote $4-2\ln2/\kappa$. Here $\kappa\to\infty$ along the threshold after the contact hierarchy has been taken.

 These thresholds belong to the limiting wall contest at each fixed $(\kappa,r)$, and the strong-selection asymptote is a property of that limit. The class-reduction errors are bounded on compact positive-fecundity domains, which exponential fecundity leaves as $\kappa r$ grows, so approaching the large-selection threshold with finite weights can require stronger contact separation (Appendix~\ref{app:c5-selection}). Separated domains approach the barrier on six sites at every fixed admissible selection strength after the hierarchy limit; the symmetric five-cycle approaches $3/2$ only
at weak selection.

Dimer establishment imposes a different finite-selection restriction.
For the five-site hierarchy of Sec.~\ref{sec:five-site-mechanism}, take
$u,v\to0$ at fixed $\eta>0$, $\kappa>0$ and $r=b/c$. Under uniform
introduction, the limiting comparison favors cooperation for some finite benefit precisely when
\begin{equation}
 \begin{cases}
  2\eta<1-\kappa,&\phi(x)=1+x,\quad 0<\kappa<1,\\
  2\eta<e^{-2\kappa},&\phi(x)=e^x.
 \end{cases}
 \label{eq:I-p5-selected-cutoffs}
\end{equation}
In either range there is exactly one benefit threshold, with promotion
above it; at or beyond the cutoff every finite benefit is adverse in
the limiting experiment. Appendix~\ref{app:I-p5-selected} reduces the
comparison to a cubic for linear fecundity and a quadratic for exponential
fecundity. Both recover $(7+2\eta)/(1-2\eta)$ as $\kappa\to0$ at fixed
$\eta<1/2$. The cutoff measures the competition between domain reproduction and local establishment. After the hierarchy limit, take $r\to\infty$ at fixed $\kappa$. Equation~\eqref{eq:dimer-establishment} gives $p_C\sim\phi(-\kappa)/\phi(\kappa r)$. The product $F_{\rm dom}p_C$ tends to $x_\infty=1-\kappa$ for linear fecundity and $x_\infty=e^{-2\kappa}$ for exponential fecundity. A cooperative pin founder then fixes with probability tending to $x_\infty/(x_\infty+\eta)$, while a defective founder at either dimer site fixes with probability tending to $\eta/(x_\infty+\eta)$. Their uniform comparison has limiting sign $x_\infty-2\eta$. The polynomial comparison in Appendix~\ref{app:I-p5-selected} determines when this endpoint sign corresponds to a finite-benefit threshold. In particular, the exponential domain reward and establishment penalty leave a residual factor $e^{-2\kappa}$ even as the benefit grows. These cutoffs are exact for the limiting committors~\eqref{eq:dimer-slow-h}, obtained at each fixed $(\eta,\kappa,r)$. At finite weights, a boundary at prescribed finite benefit depends on corrections to establishment and rare transfer. Locating the cutoff over all benefits additionally requires control as $r\to\infty$, beyond the compact positive-fecundity estimates of the present reduction.

\subsection{Games, encounters and copying rules}
\label{sec:v7-robustness}
The domain mechanism has consequences beyond the donation matrix.
For a symmetric two-strategy game, let $a_{XY}$ be the payoff to trait
$X$ against $Y$. Symmetry means that both player roles use the same
payoff rule; $a_{CD}$ and $a_{DC}$ may differ. Complementing the traits in the paired introduction
experiment cancels the cubic terms in the first-order drift. The same
neutral pair solution then gives
\begin{equation}
 \mathcal S_{\rm game}
 =\frac{(K+J)(a_{CC}-a_{DD})
       +(K-J)(a_{CD}-a_{DC})}{NZ}.
 \label{eq:v7-general-game}
\end{equation}
Since $K-J>0$, the structure coefficient
$\sigma_{\rm str}=(K+J)/(K-J)$ measures how strongly the paired
invasion comparison weights the diagonal payoff difference
\cite{TarnitaEtAl2009}. Here $\sigma_{\rm str}$ is defined by the paired fixation comparison and depends on its introduction rule. It also has an encounter interpretation. Across each discordant replacement contact, sum the two individuals' encounter fractions with their own type or with the opposite type. Averaging over neutral histories from the chosen introduction source and integrating in time with ancestral conductance weight $\gamma_{ij}$ gives $K+J$ for same-type encounters and $K-J$ for cross-type encounters. Thus $\sigma_{\rm str}$ compares these two weighted encounter totals. It acquires a stationary-abundance meaning when the two absent traits are introduced at equal rare total rates using that same source, as in Eq.~\eqref{eq:rare-abundance}. The threshold infima imply
$\sup\sigma_{\rm str}=2$ on $C_4$, $5$ on $C_5$, $4/3$ on $P_5$, and no finite
upper bound on $P_N$ for $N\ge6$, under uniform introduction.
In the protected-domain limit the competing aligned domains have
payoffs $a_{CC}$ and $a_{DD}$. Their competition can consequently
favor cooperation even when a unilateral change of strategy favors
defection. Appendix~\ref{app:v7-robustness} derives the criterion and
its fixed-selection domain limit.

Separating benefit encounters from replacement tests another model
choice~\cite{OhtsukiNowakPacheco2007}. For an arbitrary encounter kernel, Appendix~\ref{app:v7-encounter-kernel} bounds the change in benefit per unit cost by the change in encounter probabilities across replacement interfaces. An exactly solvable specialization lets a fraction $\alpha$ of encounters sample uniformly among
the other sites, while replacement continues through $P$. For $0\le\alpha\le1$, the
interaction kernel is
$Q_{\rm int}^{\alpha}=(1-\alpha)P+\alpha Q^{\rm unif}$, with
$Q^{\rm unif}_{ij}=(1-\delta_{ij})/(N-1)$. Its weak threshold is exactly
\begin{equation}
 R_\alpha=\frac{R}{1-\alpha-\alpha R/(N-1)},
 \label{eq:v7-encounter-map-main}
\end{equation}
where the denominator is positive. Thus the support infima become
$3/(1-2\alpha)$ on $C_4$, $12/(8-11\alpha)$ on $C_5$, $28/(4-11\alpha)$ on $P_5$, and
$(N-1)/(N-1-N\alpha)$ on $P_N$ for $N\ge6$.
There are no weak promoters when the corresponding denominator is
nonpositive. For exponential fecundity, the full selected generator
has an exact reduction to the original game with
$b_{\rm eff}=(1-\alpha)b$ and
$c_{\rm eff}=c+\alpha b/(N-1)$, including its physical clock. For linear fecundity the same factorization gives a configuration-dependent intensity in Eq.~\eqref{eq:I-linear-encounter-intensity}. Appendix~\ref{app:v7-encounter-kernel} derives the threshold map and both fecundity reductions.

Protected domains can survive added contacts. Adding edges
on a fixed vertex set cannot raise the infimum over positive weights,
so every support containing a spanning path with $N\ge6$ has
infimum one, including $C_N$ and $K_N$. More generally, the mechanism survives on any support whose vertices partition into two adjacent pins and two disjoint follower cores, each core connected, containing at least two sites and joined to its own pin. Give every core edge weight $Q$, split a total weight $q$ in fixed positive proportions over a nonempty set of edges from each pin to its own core, and give every other edge, including the bridge, unit weight. Taking $q\to\infty$ with $q^2/Q\to0$ retains the protected-domain limit under both introduction rules. Appendix~\ref{app:v7-support-extension} proves this extension and
controls the response for each introduction rule. Thus one-dimensional
support is a convenient setting for sharp bounds, while the
two-domain construction also applies to dense and branched networks.

\paragraph*{Update order, payoff accounting and directed contacts.}
Changing the update rule changes the relevant protection. Under uniform death followed by selected local reproduction, denoted dB, every individual has the same neutral overwrite rate. Collective persistence can nevertheless survive. With averaged payoffs and uniform introduction on the four-cycle $(1,q,Q,q)$, the ancestral dB response~\cite{AllenEtAl2017} gives
\begin{equation}
 R_{\dB}=\frac{(q+1)^2(Q+q)^2}{(Q-q^2)^2},\qquad Q\ne q^2.
 \label{eq:I-db-c4-threshold}
\end{equation}
At $Q=q^2$ the benefit coefficient vanishes while the cost coefficient remains positive. Every other finite positive weighting in this family is a weak promoter for sufficiently large $b/c$. For $Q=1$ and $q\to0$, two weakly connected dimers give $R_{\dB}\to1$, realizing the strong-pair limit of Ref.~\cite[Fig.~3]{AllenEtAl2017}. At fixed admissible fecundities, a vacancy in either dimer is filled by its partner with probability tending to one. Establishment within the pair is then independent of payoff, removing the source-first dimer penalty. Along the source-first hierarchy $Q=q^3$, $q\to\infty$, by contrast, $R_{\dB}\sim q^2$. The same support can therefore sustain collective competition through a different contact hierarchy. Appendix~\ref{app:I-db-c4} derives this comparison from four pair times.

Accumulated donation payoffs replace Eq.~\eqref{eq:payoff} by $s_i[-cn_i+b(Pn)_i]$. Neutral copying and its protected configurations remain unchanged, while site strength multiplies the local selection response. The cancellation $s_i\pi_i=1/Z$ removes its inverse-strength factor. Under uniform introduction, the benefit coefficient becomes negative on every finite example in Table~\ref{tab:I-finite-contrasts}, as an exact derivative of the full configuration chain verifies. These signs compare payoff conventions at the specified weights. Classifying accumulated-payoff promoters would require controlling the strength-weighted benefit contraction over all positive weights. The strict obstruction for $b\le c$ survives at every admissible positive selection strength for any common selection-independent introduction distribution (Appendix~\ref{app:I-directed-barrier}). A perturbative expansion along a hierarchy now requires $\delta(b+c)\max_i s_i\ll1$. At fixed selection, diverging strengths amplify the accumulated payoffs and must enter the selected rate reduction explicitly.

Mutations at births require a further change under accumulated payoffs. In an all-cooperator population, fecundity varies with site strength, so the introduction distribution of a new defector depends on selection. Appendix~\ref{app:I-accumulated-birth-source} derives this contribution to the paired fixation response. On the four-cycle in Table~\ref{tab:I-finite-contrasts}, it reverses the adverse benefit response obtained with the neutral introduction distribution held fixed, giving a threshold of approximately $1.404914$. Payoff accounting and the introduction experiment must therefore be specified together~\cite{MaciejewskiFuHauert2014,McAvoyAllenNowak2020}.

Strength-proportional reproduction changes a different ingredient. Copy rates proportional to $F_iw_{ij}$ replace the normalized rates proportional to $F_iP_{ij}$. Up to a common time factor, neutral ancestry then jumps at the symmetric rates $w_{ij}$ and every site's reproductive value is $1/N$. Cold pins lose their enhanced neutral fixation probability. For averaged payoffs the benefit coefficient is negative on all examples in Table~\ref{tab:I-finite-contrasts}, under uniform introduction and under the corresponding newborn distribution $s_i/\sum_j s_j$. These exact comparisons show how changing reproductive opportunities changes the sign at the same contacts. Excluding promoters at all positive weights requires proving that the benefit contraction of the strength-proportional pair system is nonpositive throughout that domain.

Reciprocity is also separable from the benefit--cost obstruction. On any strongly connected, loopless directed copying kernel, stationary ancestral flux balances across every cut. With the same kernel for averaged encounters, this balance still excludes cooperation for $b\le c$ at every admissible positive selection strength, and every weak promoter has threshold above one. Supports with at most three sites remain nonpromoting at weak selection. Appendix~\ref{app:I-directed-barrier} proves these consequences using reproductive values~\cite{Maciejewski2014}. On four or more sites, sharp directed thresholds require comparisons of the benefit and cost contractions beyond the strict bound $R>1$. The real-spectrum argument of Appendix~\ref{app:I-neutral-real-spectrum} uses edgewise reversibility of ancestry; this also holds on bidirected trees, as Appendix~\ref{app:I-directed-barrier} shows.

\section{Discussion}\label{sec:discussion}
A reciprocal network can favor a costly trait by changing what persists long enough to compete. A cooperator pays a cost in its present neighborhood, while a protected position lets its descendants reshape that neighborhood. The resulting domains compete through their collective payoffs, although reproduction still selects one individual at a time. This mechanism answers the weak-selection source-first question of Ref.~\cite{AllenLippnerNowak2019}. Every fixed support containing a spanning path with at least six sites admits weights approaching the sharp threshold one under either introduction rule. Smaller structures explain the obstacle to reaching it. A shared fluctuating fringe weakens collective agreement; an unprotected dimer imposes an establishment contest before the lineage can benefit from cooperation. In the dense pair--triple example strong internal contacts cluster the residents without protecting a founder, and weak selection disfavors cooperation there for every singleton source.

The controlled two-pin limit of Proposition~\ref{prop:v10-contest} gives this mechanism a kin-selection interpretation. Its criterion takes Hamilton's regression form $b\mathfrak a>c$~\cite{Hamilton1964,Hamilton1970,Queller1992}. The coefficient $\mathfrak a$ measures the association between a pin's present trait and its encounter partner's trait, conditional on the opposed-pin contest. Ancestral voltage computes that association, while the reduced exit rates retain the competition over replacement. Limited dispersal can increase both association and competition among relatives, with cancellation in some patch models~\cite{Taylor1992a}. The resulting fixation response is a positive, source-dependent multiple of $b\mathfrak a-c$. The conditional association becomes a weak-selection invasion criterion through the fast-fringe reduction and its dominant bridge. For contact structures outside this regime, the full ancestral-pair response determines the weak-selection sign.

Introduction determines which ancestral histories contribute to selection. The place at which mutations arise already affects success measures and even the neutral substitution rate on asymmetric structures~\cite{AllenTarnita2014,AllenEtAl2015}. Here it changes a complete support criterion through five sites. Uniform introduction measures ancestral residence time, whereas birth-associated introduction measures ancestral jumps. A long residence at a pin amplifies the former without necessarily increasing the latter, and newborn mutations reach the pins rarely along the protected hierarchy. Protection and access are consequently separate parts of the invasion mechanism. Even when both rules approach the same threshold, birth-associated introduction suppresses the leading response amplitude in proportion to its reduced access to the two pins, as quantified by Eq.~\eqref{eq:v10-circuit-access}.

The mechanisms also differ in their response to finite selection. On the protected paths, internal homogeneity leaves the limiting cooperative advantage positive for every $b>c$ at each fixed admissible positive selection strength. A shared fringe remains selected internally and requires averaging fecundity itself when its dependence on payoff is nonlinear. For uniform introduction on the symmetric five-cycle, the hierarchy-limit threshold rises with selection. Linear fecundity loses every promoting benefit in the admissible range $4/5\le\kappa<1$, where $\kappa=\delta c$; exponential fecundity retains a threshold that increases toward four. These conclusions hold in the hierarchy limit at fixed benefit and selection (Sec.~\ref{sec:contested-selection}), and the contact separation needed to approach them can grow with $\kappa r$. In the studied dimer hierarchy a larger benefit helps cooperation after establishment but also strengthens its defective rival during establishment. Its selection cutoffs persist even for exponential fecundity. The calculations distinguish two ways in which stronger selection can frustrate collective advantage. It can reduce the fringe agreement that enters collective transfer, or suppress the establishment required to reach an aligned domain. Separate aligned domains avoid both penalties.

An exact treatment matters because the invasion response depends on correlations and on rare transitions between long-lived configurations. Eliminating fast configurations must transform entrance probabilities and the fixation observable together with the generator. Physical holding times additionally require the frequency dependence of the return term or the equivalent corrected Poisson source, including relaxation within recurrent wall classes on a cycle. A limiting contest then explains a candidate threshold; comparison with all positive weightings determines whether it is sharp. On the five-site path, positive forest expansions and grading by contact degree control the full finite-weight inequality and its sharp inner-ratio envelope within one argument.

Neutral reciprocity constrains the dynamics as well as fixation. On every connected reciprocal support with fixed positive weights the full neutral spectrum is real and lies in $[-1,0]$ in global-attempt time (Proposition~\ref{prop:I-neutral}). On weighted paths and cycles, one-lineage matrices reconstruct the complete evolution through exterior powers, with parity-dependent boundary conditions on a cycle. Auxiliary fermionic variables organize the commuting stochastic observables, and the explicit transformation proves diagonalizability despite the lack of detailed balance in configuration space. The last surviving ancestral pair fixes the slowest decay rate from every mixed initial configuration, including after conditioning on fixation or extinction. Introduction changes the amplitudes and mean times. The decay-rate statement keeps weights fixed; its amplitudes can vanish along a contact hierarchy. On paths and cycles, the pair inverse requires cubic dense work in the number of sites. In this ordered representation, branching makes the sign of a hop depend on the positions of the other lineages within the same parity sector. Occupation-dependent hopping terms would account for this dependence (Appendix~\ref{app:v3-neutral-cas}). Neutral ancestry and the binary observable framework remain available on every graph.

The collective interpretation also clarifies the relation to demographic models. Group competition and local establishment followed by transfer have been imposed or derived in such models~\cite{TraulsenNowak2006,SimonFletcherDoebeli2013,Luo2014,Cooney2019,PiresBroom2024Communities}. Establishment followed by collective spread underlies the decorated-graph construction of Ref.~\cite{SvobodaChatterjee2025}, under death-first updating with accumulated payoffs. Here reciprocal contact hierarchies generate entrance, selected transfer and residence times at fixed size under source-first copying with averaged payoffs. In the deme model of Ref.~\cite{MoawadAbbaraBitbol2024Demes}, migration is independent of a deme's average fitness. An aligned domain here sends copies at a rate affected by its payoff, so the collective payoff enters the transfer competition directly.

The consequences extend beyond the donation game. In the paired weak-selection response, the structure coefficient $\sigma_{\rm str}$ weights the payoff difference between homogeneous cooperative and defective groups relative to the off-diagonal payoff difference~\cite{TarnitaEtAl2009}. On paths with at least six sites it has no finite upper bound under uniform introduction. Equation~\eqref{eq:v7-general-game} then has a direct interpretation. If a homogeneous cooperative group has a higher payoff than a homogeneous defective group, suitable contact weights can overcome any fixed finite disadvantage in the cross-type payoffs at sufficiently weak selection. This criterion concerns the paired fixation experiment and becomes a statement about stationary abundance in the rare-introduction regime of Eq.~\eqref{eq:rare-abundance}. At fixed selection, the separated-domain limit compares the diagonal fecundities of the aligned competitors directly.

The finite examples make the model's physical requirements concrete. The six-site weighting $(10^6,100,1,100,10^6)$ has thresholds approximately $1.064$ under uniform introduction and $1.249$ under birth-associated introduction (Table~\ref{tab:I-finite-contrasts}). It combines a roughly one-percent bridge fraction with a roughly one-in-ten-thousand return fraction from an adjacent follower. These offspring-routing probabilities connect the hierarchy to measurable biological scales. Neutral observations could test its dynamical mechanism without resolving a small selective advantage. Site-resolved neutral fixation probabilities are inversely proportional to the total contact weight at the introduction site~\cite{AntalRednerSood2006,Maciejewski2014}. In the protected six-site hierarchy with uniform introduction, the neutral mean time to successful fixation approaches three times the mean absorption time (Proposition~\ref{prop:I-fixed-selection}). Conditioning on success selects the long-lived domain contests from experiments that also include early extinction.

When transport and uptake change the encounter kernel while retaining the averaged donation payoff~\cite{AllenGoreNowak2013}, the separate-kernel bounds quantify the resulting weak-selection change; the mixed-encounter family admits an exact threshold map. Strength-proportional reproduction removes inverse-strength reproductive values. Accumulating payoffs changes both the local selection response and, for mutations at births, the introduction distribution~\cite{MaciejewskiFuHauert2014}. The latter correction restores promotion in a finite example whose fixed-source response is adverse. Contact weights alone therefore do not determine the invasion criterion. Encounter accounting, reproductive opportunities and mutation location must also be specified. A promoting weighting retains its invasion advantage after sufficiently weak contacts are added. Strengthening an arbitrary contact changes the normalized copying probabilities at both endpoints, so its effect on a fixed network requires a new response comparison.

The distinction between finding a promoter and determining its best possible threshold isolates the remaining small-graph problem. Promoter existence is settled through five sites. The sharp threshold infima are proved through four sites and for seventeen five-site supports; the other four reduce to the complete-graph inequality in Conjecture~\ref{conj:v10-dense}. If that inequality holds, six sites are the minimum needed to approach threshold one among reciprocal supports under either named introduction rule at weak selection. On the complete five-site graph, the bound holds for every singleton source when the largest contact is at most $11/3$ times the smallest, and controlled leaky, tied and three-pin limits constrain parts of the singular boundary. The spectral argument identifies an adverse stationary ancestral contribution that strength imbalance must overcome through relaxing modes. The remaining task is to bound that relaxing contribution relative to cost throughout the high-contrast region, including hierarchies with unrestricted nested scales.

For larger networks, the next structural question is which contact patterns admit protected collective competitors and which impose an establishment penalty. The first correction for a fixed follower core is determined by ancestral mean hitting times and its attachment geometry. Higher hitting-time moments enter subsequent corrections, so growing cores and arrays of domains require bounds on their growth relative to the contact scales. On the five-cycle, weak response remains controlled at extreme core ratios when both leakage ratios vanish (Appendix~\ref{app:I-c5-boundary}). The finite-selection and timing bounds use core proportions bounded away from zero and one. Extending them to the endpoints calls for a reduction controlled by both leakage ratios. A different set of questions concerns microscopic asymmetry. Nonreciprocal contacts can lose reversible ancestry, although every bidirected tree retains a reversible one-lineage measure. The directed cut argument preserves the strict benefit--cost barrier and excludes weak promoters on supports with at most three sites (Appendix~\ref{app:I-directed-barrier}). Determining sharp floors on larger supports calls for a stronger comparison of their benefit and cost responses; cut balance gives the bound $R>1$. Death-first updating permits a different form of protection, with weakly connected dimers approaching threshold one in the four-cycle comparison of Eq.~\eqref{eq:I-db-c4-threshold}. Changes of event clock, encounter normalization and mutation mechanism can likewise be compared by keeping the generator, introduction source and observable explicit.

The central conclusion is that reciprocal contacts can create collective competitors dynamically. Their payoffs explain what cooperation can gain; their establishment and accessibility determine whether a new lineage can realize that gain. Ancestral dynamics makes these stages calculable within the same finite stochastic process and connects the invasion criterion to its relaxation toward consensus.

\begin{acknowledgments}
A.C. acknowledges the support and hospitality of the International Centre for Theoretical Sciences (ICTS), Bengaluru, during participation in the programs \emph{Refinements in Enumerative Geometry and Physics 2026} (ICTS/REGP2026/06) and \emph{Monopole Moduli 2026} (ICTS/MM2026/09).

The author acknowledges the use of the free versions of ChatGPT (GPT-5 series) and Claude (Claude 4/5 series) for assistance with organizing the literature, secondary numerical cross-checking, and LaTeX formatting. All primary research, derivations, and conclusions remain the sole work of the author, who retains full responsibility for the manuscript.
\end{acknowledgments}

\section*{Data availability}
The programs and data supporting the findings of this study are available from the author upon reasonable request.
\appendix
\section{Ancestral response, sources and the strict barrier}\label{app:v10-response}
This appendix fixes the pair-clock normalization and derives the response identities used in the main text. Unless stated otherwise, graphs have positive reciprocal weights on a connected loopless support. A general singleton source is common to the cooperative and defective experiments and is independent of selection.
Write $\chi_\phi=\phi'(0)/\phi(0)$ when the fecundity normalization is not fixed to one. Throughout the appendices, generalized weak-response threshold statements assume $\chi_\phi>0$.
\label{app:path-selection-normalization}
\subsection{Pair representation and the meaning of a threshold}
\label{app:path-pair-representation}\label{app:response}

We use $\rho_C^U=\rho_C$, $\rho_D^U=\rho_D$,
$L=L_\gamma$ and the integrated disagreements $D_{ij}$ from
Sec.~\ref{sec:response-theory}.
The response $\mathcal S_\phi^U$ allows a general fecundity slope.  It is
$\chi_\phi$ times the normalized response $\mathcal S$ of the main text.
The pair equation in Eq.~\eqref{eq:pair-system} retains the source-first
convention and unit global attempt clock.
\labelalias{app:path-dirichlet-data}{eq:JK}
\labelalias{app:path-pair-equations}{eq:pair-system}
Write $\mathbf D$ for the unordered-pair vector. The uniform equation is
$\mathsf M\mathbf D=2\mathbf1$, with the killed-pair matrix of
Sec.~\ref{sec:sign-structure}; its positive Green function proves
$D_{ij}>0$. The source-dependent proof in
Appendix~\ref{app:v10-source-sum-rule} contains the uniform sum rule
\eqref{eq:coalescence-sum-rule}. This contraction is $NZ$ times the expected accumulated quadratic
variation of the neutral fixation martingale $\Phi$ in the physical
clock, linking ancestry to fixation-coordinate fluctuations.
\labelalias{app:path-pair-matrix}{eq:pair-system}
\labelalias{app:path-quadratic-variation}{eq:coalescence-sum-rule}

The contractions in Eq.~\eqref{eq:JK} give
\begin{equation}
 \mathcal S_\phi^U=\frac{2\chi_\phi}{NZ}(bJ-cK),\qquad K>0.
 \label{app:path-response}
\end{equation}
The factor two compares complementary singleton preparations. Their
first-order drift agrees because
$f(\mathbf1-n)=(b-c)\mathbf1-f(n)$ and $L\mathbf1=0$.
These formulas specialize the reproductive-value and coalescent methods
of Refs.~\cite{Maciejewski2014,AllenMcAvoy2019,AllenEtAl2017,AllenEtAl2021,McAvoyAllen2021}.
For $\chi_\phi>0$, the threshold convention of Sec.~\ref{sec:pair-response} gives $J>0$ and $R=K/J$ for a promoter; $J\le0$ gives negative response for $b,c>0$. At zero slope the first-order response vanishes. An empty promoter set has threshold infimum $+\infty$.

\subsection{The universal strict lower bound}
\label{app:path-universal-barrier}

Put $h_0=\Phi$. The drift identity~\eqref{eq:supermartingale} and the edge inequality~\eqref{eq:edge-obstruction} are the input to Proposition~\ref{prop:I-barrier}. Optional stopping turns their pointwise sign into a fixation comparison; integration over neutral histories gives the quantitative threshold gap.
\labelalias{app:path-neutral-coordinate}{eq:reproductive}
\labelalias{app:path-exact-drift}{eq:supermartingale}
\labelalias{app:path-discordant-bound}{eq:edge-obstruction}
\begin{proof}[Proof of Proposition~\ref{prop:I-barrier}]
Orient a discordant edge as $n_i=1,n_j=0$ and set $\Xi_{ij}=1-(Pn)_i+(Pn)_j$. The payoff difference is
$f_i-f_j=-(c-b)-b\Xi_{ij}$. Looplessness and the row sums of $P$ give
\begin{equation}
 \begin{aligned}
 \Xi_{ij}-(P_{ij}+P_{ji})
 &=\sum_{k\ne i,j}\bigl[P_{ik}(1-n_k)+P_{jk}n_k\bigr] \ge0.
 \end{aligned}
 \label{app:path-edge-gap}
\end{equation}
Since $P_{ij}+P_{ji}>0$ on a support edge, Eq.~\eqref{eq:edge-obstruction} is strict throughout $0\le b\le c$, including both endpoints. Connectedness ensures a discordant edge in every mixed state.
Consequently, for positive selection and strictly increasing $\phi$,
$h_0$ is strictly superharmonic before consensus. Dynkin's identity at
the finite-mean absorption time $\tau$ gives
\begin{equation}
 h_\delta(n)=h_0(n)+
 \E_n\!\left[\int_0^\tau\mathcal A h_0(n(t))\,dt\right]<h_0(n)
 \label{app:path-stopping-identity}
\end{equation}
for each mixed initial state. Here $h_\delta$ is the probability of
all-cooperator absorption. Applying the inequality to a singleton and its complement and averaging with $\mu$ proves Eq.~\eqref{eq:finite-no-go}. In particular,
$\rho_C^U<1/N<\rho_D^U$ for $0\leq b\leq c$.

For the quantitative weak bound, fix the common introduction distribution $\mu$ and suppress its superscript on the contractions below. The integrand of $K-J$ is $\sum\gamma_{ij}\Xi_{ij}$ over discordant edges. Multiplying Eq.~\eqref{app:path-edge-gap} by $\gamma_{ij}$, summing over these edges and integrating over neutral singleton histories gives $K-J\ge X_{\rm e}>0$,
with $X_{\rm e}$ defined in Eq.~\eqref{eq:edge-gap-bound}. Thus the
response at $b=c$ is strictly negative under
the positive-slope hypothesis $\chi_\phi>0$ stated above. On a promoter,
$J>0$ ensures $X_{\rm e}<K$, and
\begin{equation}
 R=\frac KJ\ge\frac K{K-X_{\rm e}}>1.\qedhere
 \label{app:path-universal-weak-bound}
\end{equation}
\end{proof}

\subsection{Preparation and the coalescence sum rule}
\label{app:v7-preparation-source}\label{app:v10-source-sum-rule}
Let $\mu$ be a probability vector used for the introduction site of
either invading trait, independently of selection. At neutrality the
initial pair disagreement is $\mu_i+\mu_j$. Integrating its backward
equation gives $\mathsf M D^\mu=g^\mu$ with
$g^\mu_{ij}=N(\mu_i+\mu_j)$, as in
Eq.~\eqref{eq:v7-pair-preparation}. The factor $N$ retains the global
attempt clock. The contractions in Eq.~\eqref{eq:JK} define
$K^\mu,J^\mu$ and give $\mathcal S^\mu=2(bJ^\mu-cK^\mu)/(NZ)$.
\labelalias{eq:v7-general-preparation-pairs}{eq:v7-pair-preparation}
Proposition~\ref{prop:I-barrier} therefore applies to every such preparation.

For $\mu_i^{\mathrm b}=\vartheta_i/N$, dividing a row of
Eq.~\eqref{eq:v7-general-preparation-pairs} by $\vartheta_i+\vartheta_j$ gives the
first-jump recursion. Thus $D_{ij}^{\mathrm b}$ is the expected
number of ancestral jumps up to and including coalescence, and
$D_{ij}^{\mathrm b}\ge1$. The neutral fixation normalization is
\begin{equation}
 \begin{aligned}
 \rho_0^{\mathrm b}
 &=\sum_i\mu_i^{\mathrm b}\pi_i
 =\frac{2}{NZ}\sum_{i<j}\gamma_{ij},\\
 \frac{\chi_c^{\mathrm b}}{\rho_0^{\mathrm b}}
 &=\frac{\sum_{i<j}\gamma_{ij}D_{ij}^{\mathrm b}}
        {\sum_{i<j}\gamma_{ij}}\ge1.
 \end{aligned}
 \label{eq:v7-birth-normalization}
\end{equation}
This relates the absolute response amplitude to an edge-weighted jump
count.

The same matrix fixes a source-dependent normalization. Put $m_{ij}=\pi_i\pi_j$ for $i<j$. For this calculation use ancestral time $\widehat t=t/N$, in which a lineage at $i$ jumps to $k$ at rate $P_{ki}$. The source $g^\mu_{ij}=N(\mu_i+\mu_j)$ retains the physical-time integrated disagreement $D^\mu$. The measure $m$ obeys detailed balance
on every transition between uncoalesced pairs. After these transitions
cancel, the remaining flux is the rate at which the pair coalesces,
\begin{equation}
 (m^\T\mathsf M)_{ij}
 =\pi_i\pi_j(P_{ij}+P_{ji})
 =\frac{\gamma_{ij}}{Z}(\pi_i+\pi_j).
 \label{eq:v10-pair-killing-measure}
\end{equation}
Contracting $\mathsf M D^\mu=g^\mu$, with $g^\mu_{ij}=N(\mu_i+\mu_j)$, therefore gives Eq.~\eqref{eq:v10-source-sum}.
\labelalias{eq:v10-source-sum-rule}{eq:v10-source-sum}
This proof retains the introduction distribution throughout. Uniform
introduction gives $Z(1-\sum_i\pi_i^2)$ on the right. For birth-associated
introduction, $\mu_i=\vartheta_i/N$ with
$\vartheta_i=\sum_jP_{ji}$, it becomes
$Z\sum_i\vartheta_i\pi_i(1-\pi_i)$. On an isothermal graph
$\pi_i=1/N$, so $K^\mu=Z(N-1)/2$ for every common singleton
introduction distribution. For every common source the identity is also the
integrated quadratic-variation identity for the neutral martingale
$\Phi$, because $\E_0^\mu[\Phi_\tau^2-\Phi_0^2]=\sum_i\mu_i\pi_i(1-\pi_i)$; the pair calculation resolves this variation edge by edge.

\subsection{Directed copying and stationary cut balance}
\phantomsection
\label{app:I-directed-barrier}
For this extension, replace reciprocal contacts by any row-stochastic, loopless, strongly connected copying matrix $P$, retaining averaged payoff~\eqref{eq:payoff}. The positive normalized reproductive-value vector now solves~\cite{Maciejewski2014}
\begin{equation}
 P\pi=\operatorname{diag}(\vartheta)\pi,\qquad
 \sum_i\pi_i=1,\qquad \vartheta_i=\sum_jP_{ji}.
 \label{eq:I-directed-stationarity}
\end{equation}
Thus $\Phi(n)=\sum_i\pi_i n_i$ is again the neutral fixation coordinate; the reciprocal formula $\pi_i\propto s_i^{-1}$ is not assumed here.
For this directed rule, the same strict finite-selection comparison holds
for $0\le b\le c$ and every common selection-independent singleton source, with positive,
strictly increasing fecundity. It also holds for accumulated payoffs
$s_i f_i$ with fixed positive encounter factors $s_i$. For averaged
payoffs and $\phi'(0)>0$, every weak promoter has threshold above one,
and none exists on two or three sites. On a bidirected tree, neutral
ancestry is reversible.
\begin{proof}
For the set $S$ of cooperative sites, summing~\eqref{eq:I-directed-stationarity} over $i\in S$ cancels the internal terms and leaves
\begin{equation}
 \sum_{i\in S,j\notin S}\pi_jP_{ij}
 =\sum_{i\notin S,j\in S}\pi_jP_{ij}.
 \label{eq:I-directed-cut}
\end{equation}
At $b\le c$ every cooperator has payoff at most zero and every defector has payoff at least zero. On a copying arc from a cooperator $i$ to a defector $j$, the stronger inequality $f_i\le-c+b(1-P_{ij})<0$ holds. Such an arc exists in every mixed state. Since fecundity is positive and strictly increasing,
\begin{equation}
 \mathcal A\Phi=
 \frac{\sum_{i\in S,j\notin S}\pi_jP_{ij}F_i
 -\sum_{i\notin S,j\in S}\pi_jP_{ij}F_i}
 {\sum_k F_k}<0.
 \label{eq:I-directed-drift}
\end{equation}
The stopping argument~\eqref{app:path-stopping-identity} proves~\eqref{eq:finite-no-go} for every common singleton source. Differentiating this drift at neutrality gives a strictly negative first-order contraction at $b=c$ when $\phi'(0)>0$, so the weak threshold is strictly above one. Accumulated payoffs $s_i f_i$ retain the same finite-selection obstruction for any assigned positive encounter factors $s_i$, since a cooperative payoff is at most $s_i(b-c)$ and is strictly negative on an outgoing discordant arc. A directed $P$ alone does not fix these factors; a weighted realization fixes them as outgoing strengths. 

For averaged payoffs, the benefit part of the neutral drift derivative is
$g_b(n)=N^{-1}\sum_{ij}\pi_jP_{ij}(Pn)_i(n_i-n_j)$ under the normalized fecundity convention. At a singleton,
\begin{equation}
 g_b(e_k)=-\frac{\pi_k}{N}\sum_{i\ne k}P_{ik}^2<0,
 \qquad g_b(\mathbf1-n)=g_b(n).
 \label{eq:I-directed-small}
\end{equation}
The second identity uses the neutral stationarity contraction. Through three sites every mixed state is a singleton or its complement, while the coefficient multiplying $-c$ is the positive outgoing cut flux. Integration up to neutral absorption therefore excludes every weak promoter under any common singleton source. Cut balance extends these obstruction arguments without assuming edgewise detailed balance. On a bidirected tree, however, the ancestral detailed-balance relations $\pi_jP_{ij}=\pi_iP_{ji}$ can be imposed recursively from a root. The unique path to each vertex makes the positive weights consistent, and normalization gives a reversible ancestral measure.
\end{proof}

\subsection{Selection-dependent introduction with accumulated payoffs}
\label{app:I-accumulated-birth-source}
At $b>c$, the introduction rule can change the sign of the response under accumulated payoffs. Return to reciprocal contacts $w_{ij}=w_{ji}$, with $P_{ij}=w_{ij}/s_i$, and use the accumulated payoff $s_i[-cn_i+b(Pn)_i]$. Replacement still uses $P$. We distinguish an imposed, common source $\mu_j^{\mathrm b}=\vartheta_j/N$ from mutations actually arising at births. All-defector fecundities equal one, so newborn cooperators retain that source. In an all-cooperator population the fecundities are $\phi(\delta(b-c)s_i)$. The location distribution $\mu_{D,j}(\delta)$ of a newborn defector and its derivative are therefore
\begin{align}
 \mu_{D,j}(\delta)&=
 \frac{\sum_i\phi(\delta(b-c)s_i)P_{ij}}
      {\sum_i\phi(\delta(b-c)s_i)},\notag\\
 \left.\partial_\delta\mu_{D,j}\right|_0
 &=\frac{b-c}{N}(s_j-\vartheta_j\overline s),
 \quad \overline s=\frac1N\sum_i s_i.
 \label{eq:I-accumulated-source-derivative}
\end{align}
We used $\sum_i s_iP_{ij}=s_j$ and $\phi'(0)=1$. A neutral singleton fixes with probability $\pi_j=1/(Zs_j)$. Contracting the source derivative with $\pi_j$ adds a defector fixation contribution, which enters the paired response with a minus sign. If $\mathcal S_{\rm fixed}$ is the full accumulated-payoff response with the neutral birth source held fixed for both traits, then
\begin{equation}
 \begin{aligned}
 \mathcal S_{\rm newborn}
 &=\mathcal S_{\rm fixed}-(b-c)\Xi_{\rm prep},\\
 \Xi_{\rm prep}&=\frac1Z-\overline s\,\rho_0^{\mathrm b},
 \qquad \rho_0^{\mathrm b}=\frac1N\sum_j\vartheta_j\pi_j.
 \end{aligned}
 \label{eq:I-accumulated-source-response}
\end{equation}
This follows by differentiating the introduction average directly. No complementation identity is used, because $f_i(\one-n)=(b-c)s_i-f_i(n)$ now contains the site-dependent offset $(b-c)s_i$, which removes the argument behind the averaged-payoff relation $\partial_\delta\rho_C=\mathcal S/2$. The source correction vanishes at equal strengths and at $b=c$; it has no universal sign on heterogeneous graphs. For a triangle with one pendant vertex and all four edge weights equal to one, $Z=7/3$ and $\rho_0^{\mathrm b}=11/56$, giving $\Xi_{\rm prep}=1/28>0$.

For a reproducible finite example, take the four-cycle with consecutive weights $(1,40,64000,40)$. Let $h_1=\left.\partial_\delta h_C\right|_0$, with zero values at both consensus configurations. Differentiating the fixation problem gives the mixed-state Dirichlet equation $(-G_{0,\rm mix})h_1=G_1\Phi$, with $G_0g=\sum_{ij}P_{ij}[g(n^{j\leftarrow i})-g(n)]$ and $G_1\Phi=\sum_{ij}P_{ij}f_i\pi_j(n_i-n_j)$. The omitted total-fecundity factor is common to all transitions and leaves the harmonic problem unchanged. Sum $h_1(e_i)+h_1(\one-e_i)$ against $\vartheta_i/N$, then apply Eq.~\eqref{eq:I-accumulated-source-response}. Exact rational elimination gives $\Xi_{\rm prep}\simeq-189.878818$ and the paired coefficients $\chi_{b,\rm newborn}\simeq169.710845$, $\chi_{c,\rm newborn}\simeq238.429136$. Their ratio is $1.404913968\ldots$. The fixed-source benefit coefficient is negative. This example distinguishes two introduction experiments within the accumulated-payoff model. The common-source barrier applies to the imposed source used for both traits; the newborn experiment includes the selection-dependent correction in Eq.~\eqref{eq:I-accumulated-source-response}. In the principal averaged-payoff model, homogeneous fecundities are site-independent and the birth source remains $\vartheta_j/N$ at every selection strength.

\section{Stochastic operators and exact observable reductions}\label{app:v10-operators}
The configuration probability remains normalized by its sum. The spin, bosonic and exterior representations below act on different auxiliary spaces; each represents the same finite copying process in its stated regime.
\subsection{Single occupancy and the operator description}
\label{sec:physics}
The one-particle constraint connects the stochastic operator in
Eq.~\eqref{eq:spin-generator} to a two-state spin at each site.
We spell out that equivalence and the role of the physical clock,
which is needed for the propagator reduction in
Appendix~\ref{app:v3-operator-response}.

Introduce bosons $a_i^\dagger,d_i^\dagger$ creating cooperative and defective occupants, with annihilation operators $a_i,d_i$ and number operators $\hat n_i^C=a_i^\dagger a_i$, $\hat n_i^D=d_i^\dagger d_i$. Restrict $N_i=\hat n_i^C+\hat n_i^D=1$, the spin-$1/2$ Schwinger-boson sector~\cite{ArovasAuerbach1988,Sachdev2025}. Applying the Doi--Peliti construction~\cite{Doi1976SecondQuantization,Peliti1985PathIntegral} to the present copying process gives
\begin{align}
 B_{ij}^{\rm copy}&=\hat n_i^C(a_j^\dagger-d_j^\dagger)d_j
 +\hat n_i^D(d_j^\dagger-a_j^\dagger)a_j,\notag\\
 \widehat{\mathcal L}&=\sum_{i\ne j}B_{ij}^{\rm copy}\widehat r_{ij},
 \qquad\mathcal H_{\rm st}=-\widehat{\mathcal L}.
 \label{eq:v10-bosonic-copy}
\end{align}
The diagonal rate acts before conversion, and each $N_i$ is conserved. On this sector the equivalent Pauli form of Eq.~\eqref{eq:spin-generator} is
\labelalias{eq:doi}{eq:v10-bosonic-copy}
\begin{equation}
 B_{ij}^{\rm copy}=(\sigma_j^x-I)
                  \frac{1-\sigma_i^z\sigma_j^z}{2},
 \label{eq:app-copy-pauli}
\end{equation}
where $\sigma_i^z|n_i\rangle=(2n_i-1)|n_i\rangle$.
The discordance projector acts before the target flip, and the diagonal
rate operator remains on the right.
This representation is used in stochastic many-body
theory~\cite{AlcarazEtAl1994,Schutz2001,DelRazoLammaMerbis2026}.

A physical initial state satisfies $(N_i-1)|p(0)\rangle=0$,
where $N_i=a_i^\dagger a_i+d_i^\dagger d_i$. Every copying term
commutes with every $N_i$, so the constraint is conserved exactly.
Equivalently,
$\Pi_i=\int_0^{2\pi}e^{\mathrm i\theta(N_i-1)}\,d\theta/(2\pi)$ projects
onto it, since the integral is one on number eigenvalue one and zero on every other integer eigenvalue. Projected bosonic coherent states
give the spin-$1/2$ coherent-state integral. Fixing the mean coherent-state
occupation to one does not enforce single occupancy.

The diagonal rates are defined on the physical sector, where their
denominator is positive, and may be extended by zero outside it.
Their coherent-state matrix elements must respect that restriction,
operator order and a consistent discrete-time slicing
\cite{ShibataTakagi1999}. This is the same finite
probability dynamics, without a density or Gaussian approximation.

Two operations affecting its clock remain distinct. Multiplying all
outgoing rates at a configuration by a positive factor leaves the
harmonic equation unchanged but alters times and spectra; it is a
state-dependent time change. Replacing the self-energy $\Sigma(z)$ of Eq.~\eqref{eq:main-self-energy} by $\Sigma(0)$
removes fast-excursion durations. The physical Poisson source
in Appendix~\ref{app:p6-protected} retains those durations.
Neither operation is a similarity transformation of the original
physical-time generator.

The two absorbing point masses and their convex mixtures are stationary, with the reached mixture fixed by preparation. Transitions into either consensus state have no reverse, so the configuration chain has no detailed-balance measure of full support. Resolvent and path-integral calculations must retain these absorbing boundary conditions.

\subsection{The neutral exterior representation}
\label{app:v3-operator-response}
The commuting binary observables in Eq.~\eqref{eq:v4-boolean} have an auxiliary exterior representation of their neutral evolution. This construction gives the full propagator on paths and cycles and the pair inverse used in the response calculation.
\paragraph*{Real neutral spectrum on reciprocal supports.}
\label{app:I-neutral-real-spectrum}
\begin{proof}[Proof of Proposition~\ref{prop:I-neutral}, spectrum]
Order the spin monomials $\sigma_S$ by their degree $k=|S|$.
A copy into a site in $S$ moves its index when the source lies outside
$S$, and removes two indices when both sites belong to $S$.
The coefficient matrix $\mathsf G_{{\rm obs},0}$ is therefore block
triangular. Its degree-$k$ diagonal block is denoted
$\mathsf G^{[k]}$. For $j\in S$, $i\notin S$ and
$T=S\setminus\{j\}\cup\{i\}$, its nonzero entries are
\begin{equation}
 \begin{aligned}
 (\mathsf G^{[k]})_{T,S}&=\frac{P_{ij}}N,\\
 (\mathsf G^{[k]})_{S,S}&=-\frac1N\sum_{j\in S}\vartheta_j.
 \end{aligned}
 \label{eq:I-neutral-degree-block}
\end{equation}
Thus $\mathsf G^{[k]}$ is the forward generator of $k$ unordered
ancestral walkers killed at their first collision. In particular,
$\mathsf G^{[0]}=0$, $\mathsf G^{[1]}=\mathsf G$, and
$\mathsf G^{[2]}=-\mathsf M^{\T}/N$.
Let $m_k(S)=\prod_{j\in S}\pi_j$ and
$\Pi^{[k]}=\operatorname{diag}(m_k(S))$.
For these surviving moves, ancestral balance gives
$m_k(S)P_{ij}=m_k(T)P_{ji}$. Hence
$(\Pi^{[k]})^{-1/2}\mathsf G^{[k]}(\Pi^{[k]})^{1/2}$ is real
symmetric. Its eigenvalues are nonpositive, and strictly negative for
$k\ge2$, because collision occurs almost surely on a finite connected
support. The block-triangular factorization consequently makes the
entire neutral spectrum real.

The unit global-attempt normalization also bounds its extent. For an edge $i--j$, hold all other occupations fixed. Its contribution to the negative symmetrized block is zero when neither endpoint is occupied, is $(P_{ij}+P_{ji})/N$ when both are occupied, and on the two configurations with one occupied endpoint is
\begin{equation}
 \frac1N\begin{pmatrix}
 P_{ji}&-\sqrt{P_{ij}P_{ji}}\\
 -\sqrt{P_{ij}P_{ji}}&P_{ij}
 \end{pmatrix}.
 \label{eq:I-neutral-edge-block}
\end{equation}
This matrix is positive semidefinite with eigenvalues $0$ and $(P_{ij}+P_{ji})/N$.
Summing over edges and using
$\sum_{i<j}(P_{ij}+P_{ji})/N=1$ gives
$0\le-(\Pi^{[k]})^{-1/2}\mathsf G^{[k]}(\Pi^{[k]})^{1/2}\le I$.
Thus every neutral eigenvalue lies in $[-1,0]$.
The degree-$N$ block is the scalar $-1$, so that endpoint belongs to the
spectrum. The degree-zero block and the single stationary one-lineage
mode account for the two zero eigenvalues. The Markov semigroup is
bounded, so neither can support a Jordan block of size greater than one.
\end{proof}

Coincident nonzero eigenvalues in different degree blocks can still be coupled into Jordan blocks. The path and cycle similarity below cancels these collision couplings explicitly. On branching supports the ordered quadratic construction encounters the spectator-sign obstruction described at the end of Appendix~\ref{app:v3-neutral-cas}, so diagonalizability needs a separate argument.

\subsubsection{The neutral path and cycle propagators}
\label{app:v3-neutral-cas}

Retain the physical clock $r_{ij}=P_{ij}/N$ and define the one-spin
coefficient matrix by $\mathsf G_{ij}=r_{ij}$ for $i\ne j$ and
$\mathsf G_{jj}=-\sum_i r_{ij}$. Thus
$\mathcal A_0\sigma_j=\sum_i\mathsf G_{ij}\sigma_i$. This coefficient matrix
$\mathsf G$ is the forward generator of a single ancestral lineage in the
physical clock; its columns sum to zero.
For $S=\{i_1<\cdots<i_k\}$ identify
$\sigma_S$ linearly with $e_S=e_{i_1}\wedge\cdots\wedge e_{i_k}$
in the exterior coefficient space $\Lambda\mathbb C^N$.
Introduce auxiliary fermions $\psi_i,\psi_i^\dagger$ with
$\{\psi_i,\psi_j^\dagger\}=\delta_{ij}$ and
$\{\psi_i,\psi_j\}=\{\psi_i^\dagger,\psi_j^\dagger\}=0$.
The physical observable algebra remains commutative; this identification
concerns its coefficient space.
We now construct the transformations in Proposition~\ref{prop:I-neutral}, retaining the two trait-complement sectors on a cycle.

\begin{proof}[Proof of Proposition~\ref{prop:I-neutral}, propagation]
The common algebraic step is a finite similarity transformation.
For antisymmetric $\mathsf K_{\rm sim}$ write $T_{\mathsf K_{\rm sim}}=\sum_{i<j}(\mathsf K_{\rm sim})_{ij}\psi_j\psi_i$ and
$\mathscr G=d\Gamma(\mathsf G)=\sum \mathsf G_{ij}\psi_i^\dagger\psi_j$.
Here $d\Gamma(\mathsf G)$ applies $\mathsf G$ to each occupied index and sums the terms.
The anticommutation relations give
$[\mathscr G,T_{\mathsf K_{\rm sim}}]=-T_{\mathsf G^{\mathsf T}\mathsf K_{\rm sim}+\mathsf K_{\rm sim}\mathsf G}$, and pair-annihilation bilinears
commute. The next commutator therefore vanishes. Consequently
\begin{equation}
 \begin{gathered}
 \mathsf G^{\mathsf T}\mathsf K_{\rm sim}+\mathsf K_{\rm sim}\mathsf G=-\mathsf C_{\rm coll}\\
 \Longrightarrow\quad \mathscr G+T_{\mathsf C_{\rm coll}}=e^{-T_{\mathsf K_{\rm sim}}}\mathscr G e^{T_{\mathsf K_{\rm sim}}}.
 \end{gathered}
 \label{eq:v10-neutral-similarity}
\end{equation}
This reduces the propagator once the geometry fixes the antisymmetric collision
matrix $\mathsf C_{\rm coll}$, with $(\mathsf C_{\rm coll})_{ji}=-(\mathsf C_{\rm coll})_{ij}$ and $(\mathsf C_{\rm coll})_{ii}=0$, and a kernel $\mathsf K_{\rm sim}$ solving this
Sylvester equation, a linear matrix equation for the unknown kernel.

On an open path, nearest-neighbor moves into empty indices have no
ordering sign, while collisions remove consecutive indices with
coefficient $r_{i,i+1}+r_{i+1,i}$. Thus $\mathsf C_{\rm coll}$ has these upper entries
on adjacent pairs and zero elsewhere. The constant upper kernel
$(\mathsf K_{\rm sim})_{ij}=1$ solves the equation because all remote terms cancel.
Equation~\eqref{eq:v10-neutral-similarity}, with $T=T_{\mathsf K_{\rm sim}}$, gives the
exact spin coefficient matrix $\mathsf G_{{\rm obs},0}=e^{-T}\mathscr G e^T$; in particular
$e^{-T}e_i\wedge e_j$ represents $\sigma_i\sigma_j-1$.
\labelalias{eq:v3-path-intertwiner}{eq:v10-neutral-similarity}

A cycle requires its boundary twist. Let $\mathsf G_o=\mathsf G$ be periodic, and
let $\mathsf G_e$ reverse the two wrap-edge off-diagonal entries, leaving the
diagonal unchanged. Even observables use $\mathsf G_e$, and odd observables use
$\mathsf G_o$. This is complement parity of observables, not parity of the
number of cooperators. The same even-antiperiodic/odd-periodic moment
sectors occur in the Glauber chain~\cite{MayerSollich2004}; the
Pfaffian structure of spatially inhomogeneous nearest-neighbor coalescing walks is derived in
Ref.~\cite[Secs.~1.1--1.2]{GarrodEtAl2018}.

Let $\mathsf C_{{\rm coll},p}$ have upper entries $r_{i,i+1}+r_{i+1,i}$ on ordinary edges,
with wrap entries $(\mathsf C_{{\rm coll},e})_{1N}=r_{1N}+r_{N1}$ and
$(\mathsf C_{{\rm coll},o})_{1N}=-(r_{1N}+r_{N1})$. In parity sector $p$,
\begin{equation}
 \begin{gathered}
 \mathsf G_{{\rm obs},0}=\mathscr G_p+T_{\mathsf C_{{\rm coll},p}}=W_p\mathscr G_pW_p^{-1},\\
 \mathscr G_p=d\Gamma(\mathsf G_p),\qquad W_p=e^{-T_{\mathsf K_{{\rm sim},p}}},\\
 \mathsf G_p^{\mathsf T}\mathsf K_{{\rm sim},p}+\mathsf K_{{\rm sim},p}\mathsf G_p=-\mathsf C_{{\rm coll},p}.
 \end{gathered}
 \label{eq:v3-cycle-intertwiner}
\end{equation}
A wrap hop crosses $k-1$ indices and a wrap collision has sign $(-1)^k$,
which fixes the sector signs. Equation~\eqref{eq:v10-neutral-similarity}
then applies with $(\mathsf G,\mathsf C_{\rm coll},\mathsf K_{\rm sim})=(\mathsf G_p,\mathsf C_{{\rm coll},p},\mathsf K_{{\rm sim},p})$.

The even kernel is $(\mathsf K_{{\rm sim},e})_{ij}=1$ for $i<j$. For a positive reversible
cycle, choose a reversible measure $\widetilde\pi_i>0$ with
$c_{ij}=r_{ij}\widetilde\pi_j=r_{ji}\widetilde\pi_i$ and set
$\mathcal R_i=1/c_{i,i+1}$. The odd kernel is
\begin{equation}
 (\mathsf K_{{\rm sim},o})_{ij}=1-2\frac{\sum_{e=i}^{j-1}\mathcal R_e}
                         {\sum_{e=1}^N\mathcal R_e},\qquad i<j.
 \label{eq:v3-cycle-electrical}
\end{equation}
For the present rule, choose the unnormalized reversible measure
$\widetilde\pi_i=N/s_i=NZ\pi_i$. Then
$c_{ij}=w_{ij}/(s_is_j)=\gamma_{ij}$ and
$\mathcal R_i=\gamma_{i,i+1}^{-1}$, with the same conductance normalization
as the network of ancestry in the main text. To check the formula, set
$u_i=-2\sum_{e<i}\mathcal R_e/\sum_e\mathcal R_e$ and subtract
$u\boldsymbol1^{\mathsf T}-\boldsymbol1u^{\mathsf T}$ from the constant
upper-triangular kernel. Its residual is removed by
$\mathsf G_o^{\mathsf T}u=v$, with $v_1=-2r_{N1}$, $v_N=2r_{1N}$ and $v_i=0$ for $1<i<N$.
Multiplication by $\widetilde\pi_i$ gives Kirchhoff's equation with constant
current on the ordered arc, proving Eq.~\eqref{eq:v3-cycle-electrical}.

Positive undirected weights make $\mathsf G_o$ similar to a symmetric negative
Laplacian with one zero eigenvalue. The antiperiodic $\mathsf G_e$ is similar
to a strictly negative signed Laplacian with flux $\pi$. Hence their
skew Sylvester sums are nonzero, since the sole possible zero-plus-zero
pair wedges a vector with itself. Exponentiating the similarity therefore gives Eq.~\eqref{eq:v3-neutral-exterior}, with the even and odd degree sets specified there. On a path the same $\mathsf G$ and $W$ serve both sectors. All neutral eigenvalues
are the appropriate subset sums, all degeneracies are semisimple
(all Jordan blocks have size one), and there are exactly two zero modes. This proves neutral diagonalizability
for positive undirected paths and cycles without assuming full-support
detailed balance of the absorbing configuration chain.
\end{proof}

For real rates, take $V=\mathbb R^N$ and the split quadratic form
$q_{\rm split}(v+\alpha)=\alpha(v)$ on $V\oplus V^*$.
The creation and annihilation operators generate its Clifford algebra;
their quadratic products act as spin generators on the exterior space.
After subtracting $\operatorname{tr}\mathsf G_p/2$, $\mathscr G_p$ and the annihilation
terms are Clifford bivectors, the quadratic elements generating this action.
The quadratic evolution lies in a spin representation of
$\operatorname{Spin}(N,N)$, or $\operatorname{Spin}(2N,\mathbb C)$
after complexification. The scalar
trace factor multiplies its spinor evolution. Since $T_{\mathsf K_{{\rm sim},p}}$ lowers degree
by two, it is nilpotent, and $W_p=e^{-T_{\mathsf K_{{\rm sim},p}}}$ has every eigenvalue equal
to one. This unipotent operator lies outside the compact real spin group, and its matrix in the auxiliary basis has negative entries, as the representative $\sigma_i\sigma_j-1$ of $e^{-T}e_i\wedge e_j$ on a path shows.

The construction has a definite geometric limit. A hop across an edge
with both an interior and an exterior spectator vertex has opposite
fermionic signs in two configurations of the same parity. One quadratic
hopping coefficient cannot reproduce both positive copying rates.
Consequently this ordered, parity-only ansatz requires every edge to be
consecutive or the wrap edge, and connected admissible graphs are paths
or cycles. Within this ordered exterior basis, branching hops require coefficients that depend on spectator positions, beyond the quadratic ansatz above. The bounded-degree observable filtration holds on every graph; on connected reciprocal supports the block argument in Appendix~\ref{app:I-neutral-real-spectrum} also keeps the spectrum real.

\subsubsection{The neutral pair inverse in the physical clock}
\label{app:v3-sylvester-time}\label{app:I-relaxation-gap}

For the first selection correction, the physical-clock Poisson source
is $g_{\rm sel}=\mathcal A_1\Phi$, the drift forcing of
Eq.~\eqref{eq:absorbing-susceptibility}. More generally, write a quadratic
source and its solution as
$g_{\rm sel}=\sum_{i<j}(J_{\rm src})_{ij}(\sigma_i\sigma_j-1)$ and
$h_1=\sum_{i<j}X_{ij}(\sigma_i\sigma_j-1)$, with
$\mathcal A_0h_1=-g_{\rm sel}$. Extending the coefficients antisymmetrically gives
\begin{equation}
 \mathsf G_eX+X\mathsf G_e^{\mathsf T}=-J_{\rm src},
 \label{eq:v3-pair-sylvester}
\end{equation}
where $\mathsf G_e=\mathsf G$ on a path. The transpose placement differs from the
kernel equation because $X$ is a two-vector coefficient matrix.
Writing $\nu_a^{(e)}\ge0$ for the eigenvalues of $-\mathsf G_e$, the transformed antisymmetric coefficients satisfy
$\widehat X_{ab}=(\widehat J_{\rm src})_{ab}/(\nu_a^{(e)}+\nu_b^{(e)})$ for $a\ne b$, with $\widehat X_{aa}=0$. On a path the single zero mode therefore causes no zero denominator in this sector, although the Sylvester operator is singular on symmetric matrices. The pair problem is thereby reduced to an $N\times N$
diagonalization, using $O(N^3)$ dense operations and $O(N^2)$ storage.
\paragraph*{Relaxation, initial states and conditioning.}
The relaxation statement of Proposition~\ref{prop:I-neutral} compares survival of a mixed configuration with survival of its ancestral pairs at fixed positive weights. Proposition~4.1 and Corollary~4.2 of Ref.~\cite{BenAriEtAl2021} relate voter survival to pair coalescence, including initial distributions satisfying a pair-disagreement condition. The ordered-edge formulation of Ref.~\cite[Sec.~1.5, Proposition~2]{AllardEtAl2024} includes the present weighted copying rule.

\begin{proof}[Proof of Proposition~\ref{prop:I-neutral}, relaxation]
For a uniform singleton, duality gives
$\Pr_U[n_i(t)\ne n_j(t)]=(2/N)\Pr_{ij}[\tau_{\rm meet}>t]$,
where $\tau_{\rm meet}$ is the ancestral meeting time in global attempts.
The unordered two-lineage graph is connected.
To see this, retain a spanning tree and remove a leaf $v$ with neighbor
$u$. If a token occupies $v$ while $u$ is empty, move it to $u$.
If both $v$ and $u$ are occupied and $N\ge3$, move the token at $u$
to another neighbor before moving the token at $v$. Every pair can
therefore reach the smaller tree without collision. Induction ends at
$N=2$, where there is one pair state. All these moves have positive
rates in both directions. The killed pair semigroup is consequently
irreducible.

For a fixed mixed configuration $n$, set
$d^n_{ij}(0)=(n_i-n_j)^2$. Neutral duality gives
\begin{equation}
 d^n(t)=e^{-t\mathsf M/N}d^n(0),\quad
 d^n_{ij}(t)=\Pr_n[n_i(t)\ne n_j(t)].
 \label{eq:I-disagreement-tail}
\end{equation}
The initial vector is nonnegative and nonzero. The least eigenvalue
of the reversible killed pair matrix is simple, with strictly positive
left and right eigenvectors, so each component has a positive leading
amplitude with decay rate $\lambda_{\min}(\mathsf M)/N$.
Since nonconsensus is the union of pair disagreements,
\begin{equation}
 \max_{i<j}d^n_{ij}(t)\le\Pr_n(\tau>t)
 \le\sum_{i<j}d^n_{ij}(t),
 \label{eq:I-every-source-tail}
\end{equation}
\labelalias{eq:I-survival-comparison}{eq:I-every-source-tail}
which fixes that rate for every mixed configuration and every mixture
of such configurations. For the finite backward transient semigroup, the largest survival probability is its induced sup norm. Its exponential rate is therefore the spectral decay rate of the transient block. The configuration matrix is block triangular between the two absorbing states and the mixed states, so its nonzero eigenvalues are those of this transient block. This proves Eq.~\eqref{eq:I-relaxation-pair} even when the configuration-space transient block is reducible. The same upper bound excludes a polynomial
factor multiplying the leading survival exponential.

At neutrality $h_C=\Phi$. At a fixed positive weighting put
$\pi_{\min}=\min_i\pi_i>0$; then
$\pi_{\min}\le\Phi(n)\le1-\pi_{\min}$ on every mixed state.
The Markov property gives
\begin{equation}
 \Pr_n(\tau>t,\,C\text{ fixes})
 =\E_n[\mathbf1_{\{\tau>t\}}\Phi(n(t))].
 \label{eq:I-conditioned-tail}
\end{equation}
This joint probability lies between $\pi_{\min}$ and
$1-\pi_{\min}$ times $\Pr_n(\tau>t)$. Division by the positive
fixation probability preserves the tail exponent. Replacing $\Phi$
by $1-\Phi$ proves the extinction-conditioned statement.
\end{proof}
The tail comparison holds at each fixed positive weighting. Its lower factor $\pi_{\min}=\min_i\pi_i$ and its slow-mode amplitudes can vanish along a contact hierarchy. The
reductions in Appendix~\ref{app:v10-domains} retain the entrance projection
and physical-time source to determine the limiting masses and mean times
in their respective hierarchy regimes.

On a path, the exterior pair spectrum consists of sums of two distinct
eigenvalues of $-\mathsf G$, including its single zero eigenvalue.
Its least rate is therefore the first positive one-lineage rate.
On a cycle the pair sector uses $\mathsf G_e$, giving the sum of its
two smallest rates. For equal cycle weights the Fourier-mode rates
are $[1-\cos((2k+1)\pi/N)]/N$, $k=0,\ldots,N-1$.
The equal minima at $k=0,N-1$ give the expression in the main text,
for every finite $N\ge3$. Changing to attempts per site multiplies
all these rates by $N$. 

For the six-site hierarchy, reflection splits the one-lineage problem into two $3\times3$ sectors. Let $\mathsf H_{\rm rev}$ be the reflection-odd block of $-6\mathsf G$ and put $F_{\rm rev}(z)=\det(zI-\mathsf H_{\rm rev})$. At $u=0$ its simple zero obeys
\begin{equation}
 \partial_zF_{\rm rev}(0)=\frac{1+3v+v^2}{(1+v)^2},\quad
 \partial_uF_{\rm rev}(0)=-\frac2{1+v}.
 \label{eq:I-six-site-spectral-tail}
\end{equation}
The other nonzero rates stay bounded away from zero. Implicit differentiation and the global attempt factor $1/6$ give $\lambda_{\rm rel}=u(1+v)/[3(1+3v+v^2)]+O(u^2)=u[1+O(u+v)]/3$. On the five-cycle, the class reduction in Appendix~\ref{app:c5-reduction} leaves two opposed classes with exit rate $2u[1+O(u+y)]/5$, while its other modes remain fast for $\beta$ in a compact subinterval of $(0,1)$. This proves the two domain-tail limits quoted in the main text.

\section{Formation, collective competition and source control}\label{app:v10-domains}
The limiting competitors must be reached from a singleton. The following reductions retain that entrance stage, the positive physical clock and the order of limits. Appendices~\ref{app:p6-protected}--\ref{app:path-p5-sharpness} treat separated domains on paths, Appendices~\ref{app:c5}--\ref{app:fringe-floors} the contested fringe and its voltage reduction, Appendices~\ref{app:v10-leaky-pins}--\ref{app:v10-tied} further five-site hierarchies, and Appendix~\ref{app:path-general-times} the path time distributions. Their error estimates are at fixed population size and on the stated compact fecundity domains.
\subsection{Exact reduction of the protected six-site path}
\label{app:p6-protected}
This subsection justifies the effective dynamics and time observables in Secs.~\ref{sec:effective-memory} and \ref{sec:finite-selection}. We retain the exact excursion memory first,
then control its reduction and the initial establishment transient.
The vertices are labeled consecutively along the path, with edge weights
$(Q,q,1,q,Q)$. Set $u=1/q$ and $v=q/Q$. The pins are vertices 3 and 4;
the left and right triples are $\{1,2,3\}$ and $\{4,5,6\}$. The
nonzero entries of the source-to-recipient kernel are
\begin{align}
 P_{12}=P_{65}&=1,&
 P_{21}=P_{56}&=\frac1{1+v},\notag\\
 P_{23}=P_{54}&=\frac v{1+v},&
 P_{32}=P_{45}&=\frac1{1+u},\notag\\
 P_{34}=P_{43}&=\frac u{1+u}.
 \label{eq:app-p6-kernel}
\end{align}
Use the backward generator $\mathsf A$ in Eq.~\eqref{eq:backward},
including its total-fecundity denominator and silent attempts.
Payoff and selection parameters range over a compact set where
$\phi$ is continuous and bounded away from zero. Locally Lipschitz
$\phi$ gives the $O(u+v)$ errors; continuity gives convergence alone.
Strict monotonicity is needed only to determine the selected sign.
\labelalias{eq:app-p6-full-generator}{eq:backward}
\subsubsection{Exact elimination retains memory at finite weights}
Let $\mathsf Q=\mathsf A^{\mathsf T}$ be the forward generator and
partition its probability vector into aligned states $S$ and
fast states $F$. Eliminating the fast probability vector
from the time-dependent equation gives the exact identity
\begin{align}
 \dot p_{S}(t)={}&\mathsf Q_{SS}p_{S}(t)
 +\mathsf Q_{SF}e^{t\mathsf Q_{FF}}p_{F}(0)
 \notag\\
 &+\int_0^t\mathsf Q_{SF}e^{(t-r)\mathsf Q_{FF}}
       \mathsf Q_{FS}p_{S}(r)\,dr.
 \label{eq:app-p6-memory}
\end{align}
The second term carries the preparation, and the third carries the
memory of excursions through omitted configurations. Its Laplace
transform is Eq.~\eqref{eq:main-self-energy}.
\labelalias{eq:app-p6-resolvent}{eq:main-self-energy}
Use the backward harmonic Schur operator $L_{\rm harm}$ in Eq.~\eqref{eq:schur}. The zero-frequency forward trace generator is $L_{\rm harm}^{\mathsf T}$. Its trace clock removes excursion durations, whereas the memory kernel retains physical waiting times. The controlled slow-time limit below makes this memory local.
\labelalias{eq:app-p6-harmonic-schur}{eq:schur}

\subsubsection{Harmonic elimination and the physical clock}
The proof of Proposition~\ref{prop:I-fixed-selection} has three reusable ingredients. The killed
fast inverse stays bounded, every exit row of an aligned state has an
exact factor $u$, and the divided mixed-state operator is invertible.
Harmonic entrance then controls fixation, while the eliminated Poisson
source controls physical time. The same argument controls the time distributions and their moments in Sec.~\ref{sec:finite-selection}. Appendix~\ref{app:path-long} verifies the alignment and exit-rate ingredients on longer paths and Appendix~\ref{app:v7-support-extension} on supports with added contacts, after which the same block argument applies.

\begin{proof}[Proof of Proposition~\ref{prop:I-fixed-selection}]
For the four aligned states and their 60-state complement, positivity of the fast
pin-to-follower rates ensures absorption from $F$ into
$S$ in bounded fast time, uniformly on every compact
positive-fecundity domain. Thus $\mathsf A_{FF}^{-1}$ is
uniformly bounded near $(u,v)=(0,0)$, and the entrance matrix is
\begin{equation}
 \mathsf E_{F}=-\mathsf A_{FF}^{-1}
 \mathsf A_{FS}
 =\mathsf E_{F}(0,0)+O(u+v).
\end{equation}
At $(u,v)=(0,0)$ each row of $\mathsf E_F$ is the indicator of the aligned state determined by the pin traits of that configuration. At finite $v$ a follower can overwrite a pin during establishment, and the $O(u+v)$ correction to $\mathsf E_F$ accounts for these events.

There is an exact structural factor in every $S$ row of
$\mathsf A$. At $u=0$ the pin-to-pin contact disappears and each
aligned state is absorbing for every $v$. All transitions out of an
aligned state thus contain the factor $u$. Consequently
\begin{equation}
 L_{\rm harm}=u\widehat L_{\rm harm},\qquad
 \widehat L_{\rm harm}=L_{\rm dom}+O(u+v).
 \label{eq:app-p6-uniform-factor}
\end{equation}
Factoring the exact slow rate before expansion gives an error $O(u+v)$,
uniformly in the ratio $v/u$. Continuous fecundities give convergence
with $o(1)$ in place of the Lipschitz error.

To compute $L_{\rm dom}$, put
$F_{\rm dom}=\phi(\delta(b-c))$, $d_0=\phi(0)=1$ and $h_{\rm dom}=F_{\rm dom}/(F_{\rm dom}+d_0)$. In a mixed
aligned state the leading total fecundity is $3(F_{\rm dom}+d_0)$. The two pin payoffs differ from their homogeneous-domain values by $O(u)$, independently of $v$. Local Lipschitz continuity gives bridge-transfer rates $u h_{\rm dom}/3+O(u^2)$ and $u(1-h_{\rm dom})/3+O(u^2)$ in the two directions. The mean alignment time with the pins fixed is uniformly bounded on the stated compact domain, while competing pin changes have total rate $O(u+v)$. The followers therefore align with the new pin trait with probability $1-O(u+v)$, uniformly in $v/u$. These rates give $L_{\rm dom}$ in Eq.~\eqref{eq:domain-generator} and the effective-rate error $O(u(u+v))$. A uniform
singleton enters each mixed class with probability $1/6+O(u+v)$;
their cooperative exit probability $h_{\rm dom}$ proves
Eq.~\eqref{eq:p6-selected}.
\labelalias{eq:app-p6-selected-generator}{eq:domain-generator}
\labelalias{eq:app-p6-selected-law}{eq:p6-selected}

The harmonic Schur complement determines hitting probabilities, but
eliminating configurations also changes the Poisson source for a time
observable. Let $\mathcal M=\{10,01\}$ and let $m$ be the mean
absorption time. From $\mathsf A m=-\mathbf1$, elimination gives
\begin{align}
 -(L_{\rm harm})_{\mathcal M\mathcal M}m_{\mathcal M}
 &=\mathbf1_{\mathcal M}
 -\mathsf A_{\mathcal MF}\mathsf A_{FF}^{-1}
       \mathbf1_{F},\\
 m_{F}&=-\mathsf A_{FF}^{-1}\mathbf1_{F}
 +\mathsf E_{F\mathcal M}m_{\mathcal M}.
\end{align}
Every $\mathcal M$ row contains the exact factor $u$, making the
first source $\mathbf1+O(u)$. Since
$(L_{\rm harm})_{\mathcal M\mathcal M}=-(u/3)[I+O(u+v)]$, it follows that
$m_{\mathcal M}=(3/u)[\mathbf1+O(u+v)]$. The mixed-sector entrance
mass is $1/3+O(u+v)$, giving $T_{\rm abs}=u^{-1}[1+O(u+v)]$.

For $m_C(n)=\E_n[\tau\,\mathbf1_{\{C\text{ fixes}\}}]$,
the exact backward equation is $\mathsf A m_C=-h_C$, with $h_C$ the
C-fixation committor. Its eliminated mixed source is
$h_{\rm dom}\mathbf1+O(u+v)$. Averaging the resulting first moment over the
same initial preparation and dividing by
$\rho_C^U=h_{\rm dom}/3+O(u+v)$ gives $T_{\rm fix}=(3/u)[1+O(u+v)]$.
Conditioning is uniform because $h_{\rm dom}$ is bounded away from zero on the
stated compact domain. The same block inverse bounds the mean from
every configuration by $Cq$, for a uniform constant $C$. Markov's
inequality and the strong Markov property then give
\begin{equation}
 \sup_n\Pr_n(\tau>2Ckq)\le 2^{-k},\qquad k=1,2,\ldots .
 \label{eq:app-p6-uniform-tail}
\end{equation}
The aggregate success-conditioned tail obeys the same exponential
bound up to division by $\rho_C^U$, which stays bounded away from zero.
These bounds make every moment of $\tau/q$ uniformly integrable.
\end{proof}

Here $\rho_D^U$ uses a D singleton in a C background. These formulas
hold at fixed admissible selection, and imply a positive limiting
advantage for strictly increasing $\phi$, fixed $\delta>0$, and
$b>c$. When $b-c=O(u+v)$, the $O(u+v)$ error can be as large as the leading response, and the finite-weight sign depends on corrections to both the reduced rates and the entrance probabilities. For uniform introduction at weak selection, Appendix~\ref{app:I-long-path-correction} gives $R=1+6/q+O(q^{-2})$ when $Q$ is bounded below by a positive constant times $q^3$.

The selection derivative and time distributions in Sec.~\ref{sec:finite-selection} follow by differentiating the regularized problem and retaining its frequency dependence.
\begin{proof}[Derivation of the selection and time-distribution limits]
For fixed $b,c$ and a $C^1$ fecundity map, the fast inverse, the entrance matrix $\mathsf E_F$ and the divided operator $\widehat L_{\rm harm}$ extend continuously to $(u,v)=(0,0)$ together with their $\delta$ derivatives. The exact factor $u$ in each aligned-state exit row is independent of selection. The mixed-state restriction of $\widehat L_{\rm harm}$ has a uniformly bounded inverse, so differentiating its harmonic equation and then the entrance average gives $C^1$ convergence in $\delta$ on a compact interval about neutrality. In particular,
\begin{equation}
 \lim_{u,v\to0}\left.\partial_\delta\Delta^U\right|_{\delta=0}
 =\frac{\chi_\phi}{6}(b-c).
\end{equation}
A $C^2$ map makes the differentiated entries Lipschitz in $(u,v)$ and gives an $O(u+v)$ differentiated error; continuity of $\phi'$ gives convergence without that rate.

For the time distributions, fix a Laplace parameter $s_{\rm L}>0$ and set $z=us_{\rm L}$ in Eq.~\eqref{eq:main-self-energy}. The resolvent identity and $\mathsf Q_{FS}=O(u)$ give $\Sigma(us_{\rm L})-\Sigma(0)=O(u^2)$, while the transformed entrance term differs from its zero-frequency value by $O(u)$. Under uniform singleton introduction its limit is $(2/3,1/6,1/6,0)^{\T}$ in the order $(00,10,01,11)$. The retained-state Laplace transform in the slow clock is $u\widetilde p_S(us_{\rm L})$, which therefore converges to $(s_{\rm L}I-L_{\rm dom}^{\T})^{-1}$ applied to this entrance vector. Equivalently,
\begin{equation}
 \begin{aligned}
 \E[e^{-s_{\rm L}\tau/q}]&\longrightarrow
       \frac23+\frac{1}{3(1+3s_{\rm L})},\\
 \E[e^{-s_{\rm L}\tau/q}\mathbf1_{\{C\text{ fixes}\}}]
       &\longrightarrow\frac{h_{\rm dom}}{3(1+3s_{\rm L})}.
 \end{aligned}
 \label{eq:I-domain-laplace-limits}
\end{equation}
These transforms give Eq.~\eqref{eq:p6-tails}; division by $\rho_C^U\to h_{\rm dom}/3$ gives the successful-fixation limit. Subtracting its subdistribution from the unconditional distribution gives the extinction survival limit in Eq.~\eqref{eq:v7-extinction-time}. Its mean estimate follows from $T_{\rm ext}^C=(T_{\rm abs}-\rho_C^U T_{\rm fix})/(1-\rho_C^U)$ and Eqs.~\eqref{eq:p6-selected}--\eqref{eq:p6-times}, with error $O(q(u+v))=O(1+q^2/Q)$. The uniform exponential bound~\eqref{eq:app-p6-uniform-tail} permits passage to every fixed moment.

The limiting unconditional first and second moments of $\tau/q$
are 1 and 6, and the conditional moments are 3 and 18. Their variances
are therefore 5 and 9.
\end{proof}

\subsubsection{A closed propagator and the initial layer}
In the column-probability convention, set
$\mathcal H=-L_{\rm dom}^{\mathsf T}$. Its minimal polynomial divides
$\lambda(\lambda-1/3)$, and therefore
\begin{equation}
 \mathcal H^2=\tfrac13\mathcal H,\qquad
 e^{-\widetilde t\mathcal H}=I+3(e^{-\widetilde t/3}-1)\mathcal H.
 \label{eq:app-p6-propagator}
\end{equation}
The repeated eigenvalues $0,0,1/3,1/3$ do not produce Jordan blocks.
For uniform C-singleton preparation, the post-establishment probability
vector is $(2/3,1/6,1/6,0)^{\mathsf T}$. With
$E=e^{-\widetilde t/3}$, $U=z_1z_2z_3$, and $V=z_4z_5z_6$, its limiting
six-variable generating function is
 \begin{align}
 G_\infty(z;\widetilde t)&=\frac23+\frac{1-h_{\rm dom}}{3}(1-E) \label{eq:app-p6-pgf}
\\
 &\quad+\frac E6(U+V)+\frac{h_{\rm dom}}3(1-E)UV,\quad \widetilde t>0. \notag
 \end{align}
 The limit at $\widetilde t\downarrow0$ gives the distribution after fast
establishment. The microscopic singleton polynomial $\sum_i z_i/6$
belongs to the fast initial layer, which collapses to $\widetilde t=0$
on the slow clock.

\subsection{Protected domains on every fixed path with at least six sites}
\label{app:path-general-construction}\label{app:path-long}

Fix $N\geq6$, choose $m_L,m_R\geq2$ with $N=m_L+m_R+2$, and assign
the consecutive path weights
\begin{equation}
 \left(\underbrace{q^3,\ldots,q^3}_{m_L-1},q,1,q,
 \underbrace{q^3,\ldots,q^3}_{m_R-1}\right),\qquad q\to\infty.
 \label{app:path-general-weights}
\end{equation}
The two sites adjacent to the unit-weight edge are protected sites.
Each has a follower path of $m_L$ or $m_R$ sites, and its domain size is
$n_L=m_L+1$ or $n_R=m_R+1$. A protected site broadcasts to its followers
at order-one rate. Its adjacent follower overwrites it only at order
$q^{-2}$, whereas cross-contact between protected sites occurs at order
$q^{-1}$. The condition $m_L,m_R\geq2$ ensures a $q^3$ follower-core
edge that suppresses the reverse rate. An isolated two-site group has
no such edge, which accounts for the dimer establishment contest on $P_5$.
With the main-text fecundity normalization, uniform introduction has the
limiting fixation probabilities and weak response in
Eq.~\eqref{eq:long-path-law}, and the threshold tends to one,
with $N$ and the domain sizes fixed.

\begin{proof}
For fixed $m_L,m_R$, a finite sequence of copies propagates each pin
value through its follower path with probability bounded below on the
same compact positive-fecundity domain. The controlled elimination of
Appendix~\ref{app:p6-protected} therefore applies with $u=q^{-1}$ and
$v=q^{-2}$. Its constants may depend on the domain sizes, and its
fecundity regularity assumptions govern the differentiated errors.
Each pin has entrance mass $1/N$; the changed total fecundity in a
mixed state determines the new holding rates.
\labelalias{app:path-schur-reduction}{eq:schur}

Define
\begin{equation}
 \begin{gathered}
 F_{\rm dom}=\phi\bigl(\delta(b-c)\bigr),\qquad d_0=\phi(0),\\
 h_{\rm dom}=\frac{F_{\rm dom}}{F_{\rm dom}+d_0}.
 \end{gathered}
 \label{app:path-general-fecundities}
\end{equation}
In a homogeneous cooperative domain the limiting payoff is $b-c$;
in a defective domain it is zero. When the left domain is cooperative,
the two absorption rates on slow time $\widetilde t=t/q$ are
\begin{equation}
 \widehat q_{LC}=\frac{F_{\rm dom}}{n_LF_{\rm dom}+n_Rd_0},\quad
 \widehat q_{LD}=\frac{d_0}{n_LF_{\rm dom}+n_Rd_0}.
 \label{app:path-general-left-rates}
\end{equation}
For the opposite mixed state, the denominator is $n_RF_{\rm dom}+n_Ld_0$.
The cooperative exit probability is $h_{\rm dom}$ in either case. The common
denominator cancels in that probability but fixes the holding time.

The entrance masses and the common exit probability prove
Eq.~\eqref{eq:long-path-law}, with the general slope multiplying its
response by $\chi_\phi$. Thus the general-slope response defined in Appendix~\ref{app:path-pair-representation} obeys $\mathcal S_\phi^U\to\chi_\phi(b-c)/N$.
The bounded fast inverse justifies this limit and its selection
derivative under the regularity assumptions above. Since the limiting
benefit coefficient is nonzero, promoters eventually exist and
$R\to1$. The strict barrier, Proposition~\ref{prop:I-barrier}, then completes the uniform-introduction case $N\ge6$ of Proposition~\ref{thm:path-classification}, with $N$ fixed before the limit.
\end{proof}

This identifies $P_6$ as the smallest path approaching one. Among five-site supports, the classification of Sec.~\ref{sec:five-site-supports} excludes threshold one on seventeen; Conjecture~\ref{conj:v10-dense} would exclude the remaining four under the two named introduction rules. Appendix~\ref{app:path-general-times} gives the unequal-domain residence-time distributions and moments. Appendix~\ref{app:v7-support-extension} controls added contacts and the introduction estimates for more general follower cores.
\labelalias{app:path-general-terminal-law}{eq:long-path-law}
\labelalias{app:path-general-response}{eq:long-path-law}
\labelalias{app:path-final-classification}{eq:path-classification}

\paragraph*{First-passage correction for a follower core.}
\phantomsection
\label{app:I-long-path-correction}
The path calculation is a specialization of a weighted-core identity.
Take two fixed connected reciprocal follower cores $\mathcal C_L,\mathcal C_R$,
each with at least two vertices, and join pin $p_\ell$ to a single attachment
vertex $a_\ell$ of its own core by weight $q$. Multiply each fixed positive
internal shape weight $w_{ij}^{\rm core}$ by $Q$, join the pins by weight
one, and include no other contacts. For path cores the attachment vertices
are their adjacent endpoints. Put $u=1/q$ and $v=q/Q$. First set $v=0$
at fixed $u>0$ and define
$p_{\rm br}=u/(1+u)$ and $p_{\rm out}=1/(1+u)$.
In this zero-leakage limit, both introduction rules have the same threshold.
Its first correction is the ancestral mean-hitting-time expression
\eqref{eq:I-core-shape-correction}; for endpoint-attached path cores the
coefficient of $q^{-1}$ is $N/2+3$. Equation~\eqref{eq:I-long-path-leak-errors}
controls the finite-leakage errors.
\begin{proof}
An ancestor at a pin switches to the other pin at rate $p_{\rm br}$ in
ancestral time $t/N$. The uniform pair equation keeps its physical-time
source two, so the pin pair obeys $D_{p_Lp_R}=1/p_{\rm br}$.

Let $T_L$ be the ancestral hitting time from the follower adjacent to
$p_L$ to that pin, and set
$\mathcal F_L(z)=\E[e^{-zT_L}]$; define $\mathcal F_R$ in the same way.
These are Laplace transforms of hitting-time distributions, distinct
from the fecundity function $\phi$. Until that hit the core lineage
and the two-state pin lineage evolve independently. The source
accumulated before $T_L$ cancels between the own-pin and opposite-pin
pair problems. The difference of their probabilities of leaving two
opposed pin lineages is $\mathcal F_L(2p_{\rm br})$. Hence
\begin{equation}
 D_{a_Lp_R}-D_{a_Lp_L}
 =\frac{\mathcal F_L(2p_{\rm br})}{p_{\rm br}}.
 \label{eq:I-core-pair-difference}
\end{equation}
The right core obeys the corresponding identity. Only the bridge conductance remains at $v=0$, with value
$p_{\rm br}^2$. The contractions in Eq.~\eqref{eq:JK} consequently give
 \begin{align}
 K_0&=p_{\rm br}, \notag \\
 J_0&=\frac{p_{\rm br}}2
 \left\{p_{\rm out}\left[\mathcal F_L(2p_{\rm br})
                   +\mathcal F_R(2p_{\rm br})\right]-2p_{\rm br}\right\}, \notag \\
 R_0(u)&=\frac{2}{p_{\rm out}[\mathcal F_L(2p_{\rm br})
                   +\mathcal F_R(2p_{\rm br})]-2p_{\rm br}}. \label{eq:I-core-passage-threshold}
 \end{align}
The denominator is positive for sufficiently small $u$.

Let $d_i^{\rm core}=\sum_jw_{ij}^{\rm core}>0$ be the strength within a
fixed core before multiplication by $Q$. At $v=0$, ancestry jumps inside
that core from $i$ to $j$ at rate $w_{ij}^{\rm core}/d_j^{\rm core}$.
Its reversible weights are $\varrho_i=1/d_i^{\rm core}$, and its attachment
vertex exits to the pin at rate $p_{\rm out}$. Multiplying the mean-time
Poisson equation by $\varrho^{\T}$ cancels every internal edge flux, leaving
\begin{equation}
 \E[T_\ell]=\frac{d_{a_\ell}^{\rm core}}{p_{\rm out}}
       \sum_{i\in\mathcal C_\ell}\frac1{d_i^{\rm core}},
 \qquad \ell=L,R.
 \label{eq:I-core-mean-return}
\end{equation}
All moments are finite for fixed core shapes. Expanding the two transforms
in Eq.~\eqref{eq:I-core-passage-threshold} therefore gives
\begin{equation}
 R_0(u)=1+\left[2+\sum_{\ell=L,R}d_{a_\ell}^{\rm core}
                 \sum_{i\in\mathcal C_\ell}\frac1{d_i^{\rm core}}\right]u
          +O(u^2).
 \label{eq:I-core-shape-correction}
\end{equation}
For an unweighted path of $m\ge2$ vertices attached at an endpoint,
$d_a^{\rm core}=1$ and the inverse-degree sum is $(m+2)/2$.
Equation~\eqref{eq:I-core-shape-correction} then has coefficient $N/2+3$.
For two internally regular cores it has coefficient $N$. Attachment
geometry matters for irregular cores. In a four-vertex star the factor
$d_a^{\rm core}\sum_i1/d_i^{\rm core}$ is $10$ at the center and $10/3$
at a leaf, despite the unchanged core size.
At $v=0$ a follower introduction cannot change either pin and cannot
fix. Each pin has birth-associated mass $p_{\rm br}/N$, compared with
$1/N$ under uniform introduction. Both response coefficients are thus
multiplied by $p_{\rm br}$, and the two thresholds coincide exactly.

Finite leakage is controlled by the same elimination used for the
leading domain limit. At $u=v=0$, every fast state reaches a pin-aligned
state through a finite copying sequence on each fixed core, so the fast
inverse is analytic near the origin. Every exit from an aligned state
has a factor $u$. After division by this factor, the two mixed-domain
committor equations have a nonsingular limit. At neutrality, their solutions and first selection derivatives are
analytic in $(u,v)$. Thus each uniform susceptibility differs from its $v=0$
value by $O(v)$, while each birth-associated susceptibility is its
$v=0$ value $p_{\rm br}\chi^U(u,0)$ plus $O(v)$.
The uniform benefit coefficient tends to $1/N$. The birth coefficient
is bounded below by a positive multiple of $u$ when $v/u$ is sufficiently
small. Specializing to endpoint-attached path cores and dividing the two response
coefficients gives
\begin{align}
 R^U&=1+\frac{N/2+3}{q}
                +O_N\!\left(q^{-2}+\frac qQ\right), \label{eq:I-long-path-leak-errors}  \\
 R^{\mathrm b}&=1+\frac{N/2+3}{q}
                +O_N\!\left(q^{-2}+\frac{q^2}Q\right),
 \quad \frac{q^2}Q\longrightarrow0. \notag
\end{align}
For the fixed weighted cores, the same estimates hold with $N/2+3$
replaced by the coefficient in Eq.~\eqref{eq:I-core-shape-correction};
the error constants depend on the core shapes and attachment vertices.
Sufficient conditions for retaining the leakage-free first coefficient are
$Q/q^2\to\infty$ for uniform and $Q/q^3\to\infty$ for birth-associated
introduction. A remainder $O_N(q^{-2})$ follows, respectively, when
$Q$ is bounded below by a positive constant times $q^3$ or $q^4$.
\end{proof}

The next coefficient involves the second moments of $T_L$ and $T_R$, which grow with core size, while the error constants also depend on core shape. The estimates therefore keep $N$ fixed. A joint large-$N$, large-$q$ limit needs growth bounds for these moments and the fast alignment time relative to the contact scales. For a specified core family, the second moments solve the killed mean-time Poisson equation behind Eq.~\eqref{eq:I-core-mean-return}, with source $2\,\E_i[T_\ell]$ in place of one.

\subsection{Five-site sharpness and the establishment penalty}
\label{app:path-p5-sharpness}

We derive the establishment and exit probabilities used in
Sec.~\ref{sec:five-site-mechanism}, keeping the selective dimer contest
through the reduction. Use the coordinates
\begin{equation}
 u=q^{-1},\qquad \eta=q/A,\qquad v=q/Q.
 \label{app:path-p5-small-coordinates}
\end{equation}
First let $u,v\to0$ with $\eta>0$ fixed. The left pair $\{1,2\}$ aligns
by mutual copying, whereas protected site 3 aligns its follower pair
$\{4,5\}$. The fast closed configurations are
$(\ell,\ell,r,r,r)$, with $\ell,r\in\{0,1\}$.
A mutant in the left dimer must first win the contest in
Eq.~\eqref{eq:dimer-establishment}. Write
$F_-=\phi(-\delta c)$, $F_b=\phi(\delta b)$,
$F_+=\phi(\delta(b-c))=F_{\rm dom}$ and $F_0=\phi(0)=d_0$.
The two possible copying events within a discordant dimer have
relative rates $F_-$ and $F_b$, giving $p_C=F_-/(F_-+F_b)$.
Their common global fecundity denominator cancels from this
probability but remains in residence times.

From a cooperative left dimer beside a defective triple, successful
cooperative takeover of site 3 has leading physical rate
$\eta uF_+/(2F_++3F_0)$. The reverse copy enters the dimer at
rate $uF_0/(2F_++3F_0)$ and establishes with probability $p_D$.
The ratio of these successful rates gives $h_L$ in
Eq.~\eqref{eq:dimer-slow-h}. With the right triple cooperative,
cooperative establishment in the dimer occurs at rate
$uF_+p_C/(2F_0+3F_+)$, whereas defective takeover of site 3
has rate $\eta uF_0/(2F_0+3F_+)$, giving $h_R$.
An unsuccessful copy returns rapidly to the same mixed state.
Consequently the establishment factors enter both the initial
singleton projection and repeated cross-domain attempts.
Uniform introduction then has the fixation probabilities in
Eq.~\eqref{eq:dimer-rho}; omitting either occurrence of the
establishment factors changes the response.

\paragraph*{The selected establishment criterion.}
\phantomsection
\label{app:I-p5-selected}
The same committors decide the finite-selection sign throughout the
fixed-$\eta$ hierarchy. Define the domain fecundity ratio
$x_{\rm dm}=F_{\rm dom}/d_0>0$.
Under uniform introduction, promotion at finite benefit is possible
precisely in the linear and exponential fecundity ranges of
Eq.~\eqref{eq:I-p5-selected-cutoffs}. Within either range there is a unique
benefit threshold, with promotion above it.
\begin{proof}
Substitution into Eq.~\eqref{eq:dimer-rho} gives
\begin{equation}
 \begin{gathered}
 \Delta_\infty=\frac{\eta\Psi_{\rm dm}}
 {5(\eta x_{\rm dm}+p_D)(x_{\rm dm}p_C+\eta)},\\
 \begin{aligned}
 \Psi_{\rm dm}={}&p_C(1+2p_C)x_{\rm dm}^2\\
       &+2\eta(p_C-p_D)x_{\rm dm}-p_D(1+2p_D).
 \end{aligned}
 \end{gathered}
 \label{eq:I-p5-selected-factor}
\end{equation}
The denominator is positive. For linear fecundity,
$x_{\rm dm}=1+\kappa(r-1)\ge1-\kappa>0$ and
$p_C=(1-\kappa)/(1+x_{\rm dm})$. Multiplying $\Psi_{\rm dm}$ by
$(1+x_{\rm dm})^2$ leaves its sign unchanged and gives
 \begin{align}
 \mathscr L_{\rm dm}(x)={}&(1-\kappa-2\eta)x^3
       +\kappa(2\kappa-5-4\eta)x^2  \label{eq:I-p5-linear-polynomial}\\
 &+[-1-5\kappa+2\eta(1-2\kappa)]x-\kappa(1+2\kappa), \notag
 \end{align}
where $x=x_{\rm dm}$. On the physical domain,
$\partial_\eta\mathscr L_{\rm dm}=-2x(x+1)(x+2\kappa-1)<0$.
At $\eta=(1-\kappa)/2$ it factors as
$\kappa(x+1)[(4\kappa-7)x-(2\kappa+1)]<0$.
This proves the nonpromoting side of Eq.~\eqref{eq:I-p5-selected-cutoffs}.
Below that cutoff the cubic coefficients have signs $+,-,-,-$,
so there is one positive root $x_*$ by Descartes' sign rule.
Since $\mathscr L_{\rm dm}(1)=-4\kappa(3+2\eta)<0$, the threshold is
$r_{\rm crit}^{\rm lin}=1+(x_*-1)/\kappa>1$, and promotion occurs
exactly above it.

For exponential fecundity put $\upsilon_{\rm dm}=e^{-2\kappa}$ and
$z_{\rm dm}=e^{\kappa(r+1)}$. Then $x_{\rm dm}=\upsilon_{\rm dm}z_{\rm dm}$,
$p_C=1/(1+z_{\rm dm})$, and
$\Psi_{\rm dm}=z_{\rm dm}\mathscr E_{\rm dm}(z_{\rm dm})/(1+z_{\rm dm})^2$,
where
\begin{equation}
 \mathscr E_{\rm dm}(z)=(\upsilon_{\rm dm}^2-2\eta \upsilon_{\rm dm})z^2
       +3(\upsilon_{\rm dm}^2-1)z+2\eta \upsilon_{\rm dm}-1.
 \label{eq:I-p5-exponential-polynomial}
\end{equation}
On $z\ge e^\kappa>1$ this polynomial decreases with $\eta$. At
$\eta=\upsilon_{\rm dm}/2$ it equals $(\upsilon_{\rm dm}^2-1)(3z+1)<0$.
Below the cutoff its coefficients have signs $+,-,-$, so its sole
positive root $z_*$ determines
$r_{\rm crit}^{\exp}=\log(z_*)/\kappa-1$.
The value $\mathscr E_{\rm dm}(1/\upsilon_{\rm dm})
=(\upsilon_{\rm dm}^2-1)(3+2\eta)/\upsilon_{\rm dm}<0$ makes this threshold greater
than one. These factorizations prove both cutoffs and uniqueness
without solving a second selected configuration problem.
\end{proof}

\paragraph*{Weak response and the sharp hierarchy.}
Assume $\chi_\phi>0$ for the weak-response comparison below. At neutrality, $p_C=p_D=1/2$,
$p_C'=-\chi_\phi(b+c)/4$ and
$(F_+/F_0)'=\chi_\phi(b-c)$. Differentiation gives
$\mathcal S_\phi^U=\chi_\phi\mathcal S$, with the fixed-$\eta$ limit in
Eq.~\eqref{eq:dimer-response}. Its benefit coefficient is positive
for $0<\eta<1/2$, where the limiting threshold is
$(7+2\eta)/(1-2\eta)$. At $\eta=1/2$ the leading benefit response
vanishes while the cost response remains negative; for fixed
$\eta>1/2$ both oppose cooperation in this sector.
In the further hierarchy $1\ll q\ll A\ll Q$, the benefit coefficient
is eventually positive and the threshold tends to seven. Together with
the all-weight bound of Proposition~\ref{prop:I-p5-envelope}, this proves the sharp
uniform-introduction infimum.

\begin{proof}
Sending $\eta\to0$ needs a remainder estimate relative to the leading term, because the leading amplitude $2\chi_\phi\eta/[5(1+2\eta)^2]$ of Eq.~\eqref{eq:dimer-response} also tends to zero. The nonzero exact
kernel entries in these variables are
\begin{equation}
 \begin{gathered}
 P_{12}=P_{54}=1,\\
 P_{21}=\frac1{1+\eta u},\qquad P_{23}=\frac{\eta u}{1+\eta u},\\
 P_{32}=\frac{u}{1+u},\qquad P_{34}=\frac1{1+u},\\
 P_{43}=\frac{v}{1+v},\qquad P_{45}=\frac1{1+v}.
 \end{gathered}
 \label{app:path-p5-regular-kernel}
\end{equation}
At $u=0$ the four internally homogeneous domain configurations are
closed, for every nearby $\eta,v$. The other 28 states have an invertible
fast restriction at the origin. The harmonic Schur complement therefore
has an exact factor $u$. After division by $u$, the two mixed-class
escape rates remain positive at $\eta=v=0$, because the right domain
can still replace the left dimer. The fast inverse, the divided mixed-class
inverse, the harmonic initial projection and their selection derivatives
are consequently regular near the origin, uniformly on a compact
strictly positive-fecundity selection domain.

At $\eta=v=0$, site 3 cannot be overwritten, for any $u$; its initial
state determines fixation. Thus $\mathcal S_\phi^U(u,0,0)=0$. Taylor expansion of the
regular reduced response, using the corresponding $\mathcal S_\phi^U=\chi_\phi\mathcal S$ limit in Eq.~\eqref{eq:dimer-response}, gives
\begin{equation}
 \mathcal S_\phi^U=\frac{2\chi_\phi\eta}{5}(b-7c)
       +O(\eta u+\eta^2+v).
 \label{app:path-p5-uniform-response}
\end{equation}
The remainder holds coefficientwise in the first-order $b,c$ response
on bounded parameter sets. In particular, it controls the ratio even
when the overall response is small.
The hierarchy
\begin{equation}
 1\ll q\ll A\ll Q
 \label{app:path-p5-sharp-hierarchy}
\end{equation}
has $u,\eta,v/\eta\to0$ and gives
\begin{equation}
 \mathcal S_\phi^U\sim\frac{2\chi_\phi q}{5A}(b-7c),\quad
 R=7+O\!\left(\frac1q+\frac qA+\frac AQ\right).
 \label{app:path-p5-sharp-limit}
\end{equation}
The equivalence for $\mathcal S_\phi^U$ assumes $b\ne7c$. For example, $A=q^2$ and
$Q=q^3$ give a benefit coefficient
$2\chi_\phi/(5q)+O(q^{-2})>0$ eventually and $R\to7$. Combined with
the strict inequality of Proposition~\ref{prop:I-p5-envelope}, Eq.~\eqref{eq:v10-p5-strict-grade-bound}, this proves
\begin{equation}
 \inf_{\substack{A,q,Q>0\\J>0}} R=7.
 \label{app:path-p5-infimum}
\end{equation}
The infimum has no finite positive-weight minimizer.
For every fixed $b/c>7$, some finite weighting therefore favors
cooperation at sufficiently weak positive selection.
\end{proof}
The infimum seven is the sharp weak-selection threshold bound over all positive weights. The finite-selection cutoffs of Eq.~\eqref{eq:I-p5-selected-cutoffs} are exact in the fixed-$\eta$ hierarchy. An all-weight extension requires global control of the selected committors, whose sign is not determined by the first-order pair response alone. Starting from a singleton on a path, nearest-neighbor copying keeps the introduced type in a single interval until absorption, so the selected committor can be computed on the interval states.

At neutrality, the two mixed states escape at leading rate
$(1+2\eta)/(10q)$. Uniform singleton preparation enters them with total
mass $2/5$, and neutral cooperative fixation has probability $1/5$.
The corresponding leading times are
\begin{equation}
 T_{\rm abs}\sim\frac{4q}{1+2\eta},\qquad
 T_{\rm fix}\sim\frac{10q}{1+2\eta}.
 \label{app:path-p5-times}
\end{equation}
Here $T_{\rm fix}$ conditions on cooperative fixation from a uniform
cooperator singleton. Along the hierarchy in Eq.~\eqref{app:path-p5-sharp-hierarchy}, the
coefficients tend to four and ten. The threshold improves while its
weak response vanishes and its resolution time grows.

\paragraph*{General introduction masses.}
For the five-site hierarchy at fixed $\eta>0$, the two derivatives in Eq.~\eqref{eq:p5-site-cancellation} and complementary pairing give
\begin{equation}
\begin{aligned}
 \mathcal S^\mu\longrightarrow{}&\frac{2\eta}{(1+2\eta)^2}\bigl[(m_p-m_{\rm dm}\eta)b\\
 &\qquad-\{3m_p+m_{\rm dm}(\eta+2)\}c\bigr].
\end{aligned}
 \label{eq:v10-p5-source-proof}
\end{equation}
Here $m_{\rm dm}$ is the total introduction mass on the two dimer sites. Dividing the two coefficients on $m_p-m_{\rm dm}\eta>0$ proves Eq.~\eqref{eq:v10-p5-source}. Introductions on the two followers have zero leading response at fixed hierarchy ratios. A nonpositive leading benefit with a positive cost coefficient is adverse at every fixed finite $b/c$; higher orders are needed when both coefficients vanish. In particular, vanishing pin mass with nonzero dimer mass gives negative leading benefit. Appendix~\ref{app:v7-preparation} excludes birth-associated promotion at every positive weighting through Eq.~\eqref{eq:v7-birth-p5-sign}.

\subsection{Contested domains on cycles}
\label{app:c5}
This subsection derives the five-cycle expansion~\eqref{eq:c5-expansion} and its physical times. The general cycle formula~\eqref{eq:cN-contested} follows from the voltage calculation in Appendix~\ref{app:wall-cycles}. The retained objects are recurrent classes
of the fast process. Their stationary weights and the probabilities of
entering them determine both collective competition and physical time.

\subsubsection{Kernel and fast classes}
\label{app:c5-reduction}
For the weights~\eqref{eq:c5-weights} put $\bar\beta=1-\beta$. The nonzero
kernel entries are
\begin{equation}
 \begin{gathered}
 P_{15}=P_{23}=\frac1{1+u},\qquad P_{12}=P_{21}=\frac u{1+u},\\
 P_{34}=\frac{\beta}{\beta+uy},\qquad P_{32}=\frac{uy}{\beta+uy},\\
 P_{54}=\frac{\bar\beta}{\bar\beta+uy},\qquad
 P_{51}=\frac{uy}{\bar\beta+uy},\\
 P_{43}=\beta,\qquad P_{45}=\bar\beta.
 \end{gathered}
 \label{eq:c5-kernel}
\end{equation}
They are analytic in $(u,y)$ near the origin at fixed $\beta\in(0,1)$.
At $u=0$ the pins cannot be overwritten. The closed classes are then the
two consensus states and, for each assignment of opposite traits to the
pins, the set of single-wall chain configurations; the other $22$ of the
$32$ configurations are transient. Indeed, on the arc
$1$--$5$--$4$--$3$--$2$ a copy from $i$ into $j$ removes the discordance of
the edge $\{i,j\}$ and changes that of the other edge at $j$, so the number
of discordant edges changes by $0$ or $-2$. With pins of opposite traits it
is odd, and every history reaches one wall.

Write $F_i=\phi(\delta f_i)$ and let $k$ be the number of cooperative chain
sites when pin $1$ is cooperative. The unnormalized copying generator,
with rates $F_iP_{ij}$, moves the wall by
\begin{equation}
 \begin{gathered}
 0\to1:\ F_1,\qquad 1\to2:\ F_5,\qquad 2\to3:\ \beta F_4,\\
 1\to0:\ \bar\beta F_4,\qquad 2\to1:\ F_3,\qquad 3\to2:\ F_2,
 \end{gathered}
 \label{eq:c5-wall-rates}
\end{equation}
with the fecundities evaluated in the current configuration and the fast
kernel at $u=0$. At neutrality detailed balance gives
Eq.~\eqref{eq:c5-wall}, the series voltage division of
Appendix~\ref{app:wall-cycles}, which also gives the chain weights behind
Eq.~\eqref{eq:cN-contested}; exact elimination of the pair system for
$C_N$, $N=4,\ldots,8$, reproduces the limits $3$, $3/2$, $5/4$, $7/6$
and $9/8$ under both introduction rules.

\paragraph*{Reduction onto the fast classes.}
Fix $0<\beta_{\min}<1/2$ and
$\beta\in[\beta_{\min},1-\beta_{\min}]$. Payoff and selection parameters
remain in a compact set on which all fecundities are strictly positive.
Let $\mathsf Q$ be the configuration-basis matrix of the physical forward generator $\mathcal L$, the transpose of the backward generator in Eq.~\eqref{eq:backward}. Denote its $u=0$ limit by $\mathsf Q_{\rm fr}$; this limit is
independent of $y$ and retains the chosen selection intensity.
If $\widetilde\varpi$ is the invariant probability of the
unnormalized wall generator in Eq.~\eqref{eq:c5-wall-rates}, its
physical invariant probability is
\begin{equation}
 \varpi_k^{\rm phys}
 =\frac{F_\Sigma(k)\widetilde\varpi_k}
        {\sum_lF_\Sigma(l)\widetilde\varpi_l},
 \qquad F_\Sigma(k)=\sum_iF_i(k).
 \label{eq:c5-physical-class-measure}
\end{equation}
This identity holds for any finite fast class. For its backward
matrices, $\mathsf A_{\rm fr}=\operatorname{diag}(F_\Sigma)^{-1}
\widetilde{\mathsf A}_{\rm fr}$, so
$\varpi^{\rm phys}\mathsf A_{\rm fr}$ is proportional to
$\widetilde\varpi\widetilde{\mathsf A}_{\rm fr}=0$.
Averaging an exit rate consequently cancels its local $F_\Sigma(k)$ denominator,
but retains the common normalization in
Eq.~\eqref{eq:c5-physical-class-measure}. That normalization cancels
from an exit-probability ratio and remains in the residence time.

Let the four columns of $\mathsf E_{\rm cl}$ be the stationary
probabilities of the fast classes, embedded in the 32 configurations,
and let $(\mathsf C_{\rm cl})_{an}$ be the probability of reaching
class $a$ from configuration $n$ under $\mathsf Q_{\rm fr}$. Then
\begin{equation}
 \mathsf C_{\rm cl}\mathsf E_{\rm cl}=I_4,\qquad
 \mathsf Q_{\rm fr}\mathsf E_{\rm cl}=0,\qquad
 \mathsf C_{\rm cl}\mathsf Q_{\rm fr}=0.
 \label{eq:c5-class-projection}
\end{equation}
The complementary space to
$\mathsf E_{\rm cl}\mathsf C_{\rm cl}$ contains 28 fast directions.
The 22 transient configurations have bounded mean hitting times of
the fast classes. Each four-state single-wall class is an irreducible finite
birth--death chain with its three nonstationary modes bounded away
from zero on the stated compact parameter sets. The inverse on the
full complementary space is therefore uniformly bounded. This
includes the six relaxation directions within the wall classes,
in addition to the 22 transient directions.

In this class--fast decomposition, write the blocks of $\mathsf Q$
as $\mathsf Q_{cc},\mathsf Q_{cf},\mathsf Q_{fc},\mathsf Q_{ff}$.
The source--inverse--readout identity of
Sec.~\ref{sec:effective-memory} applies in this basis. At zero
frequency its class operator is
\begin{equation}
 \mathsf Q_{cc}
 -\mathsf Q_{cf}\mathsf Q_{ff}^{-1}\mathsf Q_{fc}
 =u\,\widehat{\mathsf Q}_{\rm cl}(u,y).
 \label{eq:c5-divided-class-operator}
\end{equation}
Every transition changing a pin assignment contains $u$, and the
bounded fast inverse makes the divided operator regular at the
origin. After the two consensus classes are removed, its transient
$2\times2$ block has strictly positive escape rates at $u=y=0$
and hence a bounded inverse. The full four-class operator remains
singular because consensus is absorbing. Keeping the entrance term
and the frequency dependence of the Schur reduction, or equivalently
the corrected Poisson source in Appendix~\ref{app:p6-protected},
then controls both committors and physical passage times. For the
linear and exponential fecundity maps these divided inverses are
analytic near the origin and have bounded derivatives on smaller
compact parameter sets. In particular, $uT_{\rm abs}$ and
$uT_{\rm fix}$ for uniform introduction are regular there; their
coefficients are determined below. These finite-selection and time estimates keep $\beta$ in a compact subset of $(0,1)$. At either endpoint one wall position becomes transient while the other three form a recurrent class. Extending the bounds requires a reduction controlled by both leakage ratios $\lambda_{\mathrm c1}=y/\beta$ and $\lambda_{\mathrm c2}=y/(1-\beta)$, including the class entrance and physical-time terms. Appendix~\ref{app:I-c5-boundary} proves the corresponding uniform control for weak response when both leakage ratios vanish.

\subsubsection{Slow competition, normalization and exit probabilities}
\label{app:c5-slow-competition}
Let $\widetilde\varpi_\delta$ be the stationary distribution of the
unnormalized generator~\eqref{eq:c5-wall-rates} on a contested class,
with components $\widetilde\varpi_k(\delta)$ for wall position $k$.
The global fecundity normalization is a state-dependent time change and
cancels in exit probabilities (the source-sensitive voltage reduction in Appendix~\ref{app:circuit} treats the
general case). The following calculation extends Proposition~\ref{prop:v10-contest} by retaining the
exact dependence on the leak parameter $y$. To leading order in $u$,
with pin $1$ cooperative,
\begin{equation}
 h_{12}=\frac{\langle F_1+(y/\beta)\,n_3F_3\rangle}
 {\langle F_1+F_2+(y/\bar\beta)(1-n_5)F_5+(y/\beta)\,n_3F_3\rangle},
 \label{eq:c5-exit}
\end{equation}
where $\langle\cdot\rangle$ averages over $\widetilde\varpi_\delta$. The first
terms are bridge copies; the terms proportional to $y$ are copies from a
disagreeing chain end into its pin, $P_{51}\simeq uy/\bar\beta$ and
$P_{32}\simeq uy/\beta$. The class with pin $2$ cooperative gives $h_{21}$
by reflection. Uniform introduction enters each contested class with mass
$1/5$ at a pin, while chain introductions are erased by copies from the pins, which are effectively
frozen on the fast time scale.
Hence $\Delta=\tfrac25(h_{12}+h_{21}-1)+O(u)$. The pin payoffs are
$f_1=-cn_1+bn_5$ and $f_2=-cn_2+bn_3$ in the fast kernel.
At $y=0$ their first-order response uses the neutral wall weights.
For nonzero $y$, the leak exits depend on wall position already at
neutrality, so the derivative of $\widetilde\varpi_\delta$ also enters.
Detailed balance in Eq.~\eqref{eq:c5-wall-rates} gives
\begin{equation}
 \left.\partial_\delta\log
 \frac{\widetilde\varpi_3(\delta)}{\widetilde\varpi_0(\delta)}\right|_0=-3c-2b.
\end{equation}
Since $\widetilde\varpi_3(\delta)/\beta=\widetilde\varpi_0(\delta)/\bar\beta=1/3$
at neutrality, the change of wall weights contributes
$-y(c+2b/3)$ to the difference of the two exit weights.
Direct payoff differentiation contributes
$2b/3-c+y(b-c)/3$. Combining these terms with the neutral total
exit weight $2(1+y/3)$ in Eq.~\eqref{eq:c5-exit} gives
\begin{equation}
 \begin{gathered}
 \mathcal S=\frac{b\,(2-y)-c\,(3+4y)}{5\,(3+y)}+O(u),\\
 R=\frac{3+4y}{2-y}+O(u),
 \end{gathered}
 \label{eq:c5-reduced-R}
\end{equation}
independently of $\beta$. Copies from a disagreeing chain end into its pin
lower the benefit coefficient and raise the cost coefficient, and promoters
require $y<2$ in this limit. The full ten-pair rational function of
Appendix~\ref{app:c5-pair} confirms the closed form, since setting $u=0$
in it gives exactly $(3+4y)/(2-y)$ for $y\ne2$ and $0<\beta<1$. At neutrality the total exit rate of a
contested class is $\tfrac25u(1+y/3)$ per attempt; with entrance mass
$2/5$ this gives $T_{\rm abs}$ and $T_{\rm fix}$ in Sec.~\ref{sec:c5-contested}. A copy
across an edge keeps the cooperative set of a singleton introduction an arc
of the cycle, so the process runs on the twenty proper arcs. Exact
solution of this arc chain gives
\begin{equation}
 \begin{gathered}
 T_{\rm abs}=\frac{3}{(3+y)\,u}+c_{\rm abs}+O(u),\\
 T_{\rm fix}=\frac{15}{2(3+y)\,u}+c_{\rm fix}+O(u),
 \end{gathered}
 \label{eq:c5-times-exact}
\end{equation}
with $uT_{\rm abs}$ and $uT_{\rm fix}$ analytic at $u=y=0$ for each $\beta\in(0,1)$.
The rational functions $c_{\rm abs}(y,\beta)$ and $c_{\rm fix}(y,\beta)$ are the
$u$ derivatives of these rescaled times at $u=0$; at $y=0$, with $p=\beta\bar\beta$,
$c_{\rm abs}=(10p^2+149p+593)/[3(5p+32)]$ and $c_{\rm fix}=\tfrac56(p+10)$. With unequal
pin contacts and $x=w_{12}/w_{23}$, $x'=w_{12}/w_{51}$, the same chain gives
$T_{\rm abs}=2/(x+x')+O(1)$ and $T_{\rm fix}=5/(x+x')+O(1)$ at $y=0$. A separately
written $32$-configuration solver reproduces these rational functions at
ten exact points. The class projection and divided inverse above control the remainders
and their selection derivatives in the stated compact domains.
Uniform introduction has a nonzero limiting mass in each contested
class and a positive limiting fixation probability, so conditioning
on fixation preserves this control.

Birth-associated introduction needs a relative estimate because its
response vanishes. Each pin receives introduction mass
$u/5+O(u(u+y))$. A chain singleton can change a pin only through a
chain-to-pin copy of rate $O(uy)$ before its fast extinction;
the bounded fast lifetime therefore makes its fixation response
$O(uy)$, including its derivative at neutrality. At $y=0$ this
channel vanishes identically, while each pin mass is exactly
$u/[5(1+u)]$. The equality below concerns the continued $y=0$
generator; finite positive core weights have $y>0$ and obey the
relative estimate that follows it,
\begin{equation}
 \begin{gathered}
 \Delta^{\rm b}(u,0)=\frac{u}{1+u}\,\Delta^U(u,0),\\
 \mathcal S^{\rm b}=u\left[\frac{2b-3c}{15}
       +O\bigl((u+y)(|b|+|c|)\bigr)\right].
 \end{gathered}
 \label{eq:c5-relative-birth-response}
\end{equation}
The susceptibility estimate is coefficientwise in $b,c$.
After division by $u$, its benefit coefficient tends to $2/15>0$,
which proves $R^{\rm b}\to3/2$. This argument retains the fast
entrance contribution when $y>0$; the full source-sensitive
decomposition appears in Appendix~\ref{app:circuit}.

\subsubsection{The exact ten-pair function}
\label{app:c5-pair}
The pair system~\eqref{eq:pair-system} has coefficients rational in
$(u,y,\beta)$. Its pair solution contains the diverging pin residence
time. After cancellation of the common singular factor in $K/J$,
the exact threshold is regular at $u=y=0$ for fixed
$\beta\in(0,1)$. Exact Taylor division gives
\begin{equation}
 \begin{aligned}
 R={}&\frac32+\bigl(6+\beta\bar\beta\bigr)u+\frac{11}{4}\,y
 +\frac{61+6\beta\bar\beta}{4}\,u^2\\
 &+c_{uy}(\beta)\,uy+\frac{11}{8}\,y^2+O\bigl((u+y)^3\bigr),
 \end{aligned}
 \label{eq:c5-taylor}
\end{equation}
where $c_{uy}(\beta)=\left.\partial_u\partial_yR(u,y,\beta)\right|_{u=y=0}$
is determined by the regular pair ratio above, with $c_{uy}(1/2)=2269/38$.
The $y$ and $y^2$ coefficients agree with the
expansion of Eq.~\eqref{eq:c5-reduced-R}, which they test by an
independent route. The birth-associated function has the same limit
and $u$-coefficient. Its $y$-coefficient is
\begin{equation}
 \begin{gathered}
 c^{\mathrm b}=\frac{\Pi_6(\xi_{\rm cyc})}{2\xi_{\rm cyc}\,(1+\xi_{\rm cyc})^2(32\xi_{\rm cyc}^2+69\xi_{\rm cyc}+32)},\\
 \begin{aligned}
 \Pi_6(\xi_{\rm cyc})&=186(1+\xi_{\rm cyc}^6)+2895(\xi_{\rm cyc}+\xi_{\rm cyc}^5)\\
 &\phantom{={}}+10453(\xi_{\rm cyc}^2+\xi_{\rm cyc}^4)+15548\xi_{\rm cyc}^3,
 \end{aligned}
 \end{gathered}
 \label{eq:c5-birth-coefficient}
\end{equation}
with $\xi_{\rm cyc}=Q_2/Q_1=\bar\beta/\beta$; it equals $761/19$ at $\beta=1/2$.
The birth-associated limit at $u=0$ is also rational in closed form. At
$\beta=1/2$,
\begin{equation}
 R^{\mathrm b}\big|_{u=0}=\frac{714y^2+1192y+57}{38-220y-257y^2},
 \label{eq:c5-birth-u0}
\end{equation}
so birth-associated promoters require $257y^2+220y<38$. Its exact positive endpoint is $y_{\rm cut}=(\sqrt{21866}-110)/257\simeq0.1474$, so the range is $y<y_{\rm cut}$; uniform introduction admits $y<2$. Its expansion
$3/2+(761/19)y+(376621/1444)y^2+O(y^3)$ agrees with an independent
corner expansion of the same function.

\paragraph*{Unequal core scales and introduction amplitudes.}
\phantomsection
\label{app:I-c5-boundary}
The weak-response limit also permits the core ratio to approach either
endpoint, when both core contacts dominate $q^2$. To see the
required scale directly, define
\begin{equation}
 \lambda_{\mathrm c1}=\frac{q^2}{Q_1},\qquad
 \lambda_{\mathrm c2}=\frac{q^2}{Q_2},\qquad
 \beta=\frac{\lambda_{\mathrm c2}}{\lambda_{\mathrm c1}+\lambda_{\mathrm c2}}.
\end{equation}
As $q\to\infty$ and both leakage ratios vanish, both weak thresholds
tend to $3/2$, uniformly as $\beta$ approaches either endpoint.
If $\lambda_{\mathrm c1}>0$ stays fixed while $\lambda_{\mathrm c2}\to0$,
the limits are those of Eq.~\eqref{eq:I-c5-boundary-threshold}.
\begin{proof}
In the ten-pair system, divide the row for $(1,2)$ by $u=1/q$ and denote
the resulting matrix by $\widehat{\mathsf M}_{\rm cyc}$. Its first row is
\begin{equation}
 \begin{aligned}
 &\left[\frac2{1+u}+\frac{\lambda_{\mathrm c1}}{1+\lambda_{\mathrm c1}u}
 +\frac{\lambda_{\mathrm c2}}{1+\lambda_{\mathrm c2}u}\right]D_{12}\\
 &\hspace{1em}-\frac{\lambda_{\mathrm c1}D_{13}}{1+\lambda_{\mathrm c1}u}
 -\frac{\lambda_{\mathrm c2}D_{25}}{1+\lambda_{\mathrm c2}u}.
 \end{aligned}
\end{equation}
For birth-associated introduction the source of this row is the expression
in brackets. For uniform introduction use $uD^U$ as the unknown vector;
the first source entry is then $2$ and the remaining entries are $2u$.
These equations remain regular when the pin waiting time diverges.

Temporarily treating $\beta$ as independent of the two leakage parameters,
the limit $u=\lambda_{\mathrm c1}=\lambda_{\mathrm c2}=0$ has determinant
\begin{equation}
 \det\widehat{\mathsf M}_{\rm cyc}
 =18\left[32+5\beta(1-\beta)\right]\ge576,
 \quad 0\le\beta\le1.
\end{equation}
The factors are the pin-pair entry $2$, two single-pin blocks of determinant
$3$, and a fringe-pair block of determinant $32+5\beta(1-\beta)$.
The physical contacts obey
$\lambda_{\mathrm c1}\beta=\lambda_{\mathrm c2}(1-\beta)$ throughout.
For either limiting solution the pin-pair entry is $1$.
The difference $d_i^{\rm pin}=D_{1i}-D_{2i}$, using $uD^U$ for the uniform case,
has boundary values $d_1^{\rm pin}=-1$, $d_2^{\rm pin}=1$ and satisfies
\begin{equation}
 \begin{aligned}
 (1+\beta)d_3^{\rm pin}&=1+\beta d_4^{\rm pin},\\
 2d_4^{\rm pin}&=d_3^{\rm pin}+d_5^{\rm pin},\\
 (2-\beta)d_5^{\rm pin}&=-1+(1-\beta)d_4^{\rm pin}.
 \end{aligned}
\end{equation}
Thus $d_3^{\rm pin}-d_5^{\rm pin}=4/3$. For the conductances, put $g_{\rm cyc}=\lambda_{\mathrm c1}\beta=\lambda_{\mathrm c2}(1-\beta)$ on the physical contacts. The rescaled readout is smooth too,
 \begin{gather}
 \frac{\gamma_{12}}{u^2}=\frac1{(1+u)^2},  \label{eq:I-c5-regular-readout}\\
 \frac{\gamma_{23}}{u^2}=\frac{\lambda_{\mathrm c1}}{(1+u)(1+\lambda_{\mathrm c1}u)},\quad
 \frac{\gamma_{51}}{u^2}=\frac{\lambda_{\mathrm c2}}{(1+u)(1+\lambda_{\mathrm c2}u)}, \notag \\
 \frac{\gamma_{34}}{u^2}=\frac{g_{\rm cyc}}{1+\lambda_{\mathrm c1}u},\quad
 \frac{\gamma_{45}}{u^2}=\frac{g_{\rm cyc}}{1+\lambda_{\mathrm c2}u}, \notag \\
 \frac Zu=\frac2{1+u}
 +u\left[\frac{\lambda_{\mathrm c1}}{1+\lambda_{\mathrm c1}u}
 +\frac{\lambda_{\mathrm c2}}{1+\lambda_{\mathrm c2}u}+g_{\rm cyc}\right]. \notag 
 \end{gather}
 Taking $g_{\rm cyc}=\lambda_{\mathrm c1}\beta$ extends these expressions smoothly to independent $\beta$; conclusions are restricted back to the physical contacts. The limiting cost and benefit contractions are respectively
$1$ and $(d_3^{\rm pin}-d_5^{\rm pin})/2=2/3$ after their natural powers of $u$ are removed.
The determinant bound makes the inverse and its derivatives uniformly
bounded over the full core-ratio interval. Consequently,
\begin{align}
 \mathcal S^U
 &=\frac{2b-3c}{15}
 +O\!\left((u+\lambda_{\mathrm c1}+\lambda_{\mathrm c2})(|b|+|c|)\right),\\
 \mathcal S^{\rm b}
 &=u\left[\frac{2b-3c}{15}
 +O\!\left((u+\lambda_{\mathrm c1}+\lambda_{\mathrm c2})(|b|+|c|)\right)\right].
\end{align}
These estimates hold as $u,\lambda_{\mathrm c1},\lambda_{\mathrm c2}\to0$, uniformly in $\beta\in[0,1]$, and control the weak response. The finite-selection and time estimates rest on the wall-class reduction of Appendix~\ref{app:c5-reduction}, which keeps $\beta$ in a compact subset of $(0,1)$ for the reason given at its end.

An asymmetric boundary shows why control relative to the birth amplitude
is necessary. Fix $\lambda_{\mathrm c1}>0$ and let $u,\lambda_{\mathrm c2}\to0$. The regular pair
matrix has limiting determinant $576$. Its limiting birth solution obeys
$D_{13}=D_{12}+1$ and
$(2+\lambda_{\mathrm c1})D_{12}-\lambda_{\mathrm c1}D_{13}=2+\lambda_{\mathrm c1}$, hence
$D_{12}=1+\lambda_{\mathrm c1}$. For completeness, with $x=1+\lambda_{\mathrm c1}$, the other entries needed by the benefit readout are
\begin{equation}
 \begin{gathered}
 D_{15}=(x+7)/3,\quad D_{23}=1,\quad
 D_{24}=(x+8)/3,\\
 D_{25}=(2x+7)/3,\qquad D_{34}=(4x+35)/12.
 \end{gathered}
\end{equation}
The bridge and the $23$ contact survive in $\gamma/u^2$. Their contraction in Eq.~\eqref{eq:JK} gives
\begin{equation}
 \frac{K^{\rm b}}{u^2}\longrightarrow1+2\lambda_{\mathrm c1},
 \quad
 \frac{J^{\rm b}}{u^2}\longrightarrow\frac{16+\lambda_{\mathrm c1}}{24},
 \quad \frac Zu\longrightarrow2.
\end{equation}
Uniform introduction still has $K^U/u\to1$ and $J^U/u\to2/3$.
Therefore
\begin{align}
 \mathcal S^U&=\frac{2b-3c}{15}+O(u+\lambda_{\mathrm c2}),\\
 \frac{\mathcal S^{\rm b}}u
 &=\frac{16+\lambda_{\mathrm c1}}{120}b-\frac{1+2\lambda_{\mathrm c1}}{5}c
 +O(u+\lambda_{\mathrm c2}),\\
 R^U&\longrightarrow\frac32,
 \qquad R^{\rm b}\longrightarrow
 \frac{24(2\lambda_{\mathrm c1}+1)}{\lambda_{\mathrm c1}+16}.
\end{align}
Here $b,c$ are fixed, and the error bounds are uniform for $\lambda_{\mathrm c1}$ in
compact subsets of $(0,\infty)$. This joint limit allows every
$Q_2/q^2\to\infty$, including $Q_2=q^3$. The response of introductions at the weaker core contact survives at the same order as the suppressed birth-associated pin response. In particular, $Q_1/q\to\infty$ alone
does not preserve the birth threshold $3/2$.
\end{proof}

\subsection{Finite selection on the five-cycle}
\label{app:c5-selection}
This subsection derives the selected five-cycle thresholds of Sec.~\ref{sec:contested-selection}, with the hierarchy limit taken first at fixed $(\kappa,r)$.
By the matrix-tree theorem~\cite{Chaiken1982}, the numerator and denominator of each fixation
probability on the arc chain are $u^2$ times a positive forest sum plus
higher orders, so the hierarchy limit exists, is independent of the path
$u,y\to0$, and is continuous in the fecundities. For linear fecundity,
$\beta=1/2$ and $r=b/c$, the exact evaluation of $h_{12}$ and $h_{21}$ at
$u=y=0$ gives
 \begin{gather}
\Delta_\infty=\frac{2\kappa\,[2+\kappa(r-1)]\,\mathcal B_{\rm lin}(r,\kappa)}{5\,\mathscr D_{\rm lin}(r,\kappa)},  \label{eq:c5-selected-linear}\\
 \mathcal B_{\rm lin}=(2r-3)\bigl[2+2\kappa(r-1)\bigr]-\kappa^2(5r^2-3r+2), \notag
 \end{gather}
with
\begin{equation}
 \begin{aligned}
 \mathscr D_{\rm lin}={}&24+48\kappa(r-1)+\kappa^2(30r^2-78r+38)\\
 &+\kappa^3(6r^3-36r^2+44r-14)\\
 &+\kappa^4(-5r^3+10r^2-9r+2).
 \end{aligned}
\end{equation}
The denominator is positive throughout the admissible domain. With
$t_\kappa=1-\kappa$, its exact decomposition is
\begin{align}
\mathscr D_{\rm lin}={}&2+6t_\kappa+8t_\kappa^2+6t_\kappa^3+2t_\kappa^4\notag\\
&+\kappa r(5+17t_\kappa+17t_\kappa^2+9t_\kappa^3)\notag\\
&+\kappa^2r^2(4+16t_\kappa+10t_\kappa^2)
 +\kappa^3r^3(1+5t_\kappa).
\end{align}
Thus $\mathscr D_{\rm lin}\ge2$ for $0\le\kappa<1$ and $r\ge0$.
The pin fecundity difference in Eq.~\eqref{eq:c5-exit} is
$\langle F_1-F_2\rangle=\delta(2b-3c)/3-\delta^3b(b+c)(b+2c)/9+O(\delta^4)$.
Selection therefore leaves the sum of the two pins' agreement probabilities unchanged
at first order and lowers it by $\delta^2(b+c)(b+2c)/9$ at second order.
The zero of $\mathcal B_{\rm lin}$ is the root of
\begin{equation}
 (4\kappa-5\kappa^2)r^2+(4-10\kappa+3\kappa^2)r-(6-6\kappa+2\kappa^2)=0,
\end{equation}
which gives Eq.~\eqref{eq:c5-rcrit}, continued as
$+\tfrac{455}{64}\kappa^4-\tfrac{2975}{256}\kappa^5+O(\kappa^6)$, the exact
value $r_{\rm crit}(1/2)=7/3$, and divergence as $\kappa\uparrow4/5$, where
the $r^2$ coefficient vanishes. For $\kappa\ge4/5$, $\mathcal B_{\rm lin}<0$ for every
$r>1$. The exact selected arc
chain reproduces Eq.~\eqref{eq:c5-selected-linear} identically in
$(\kappa,r)$, both with the full payoffs at $y=0$ and with the kernel
payoffs along $y=u$.

For exponential fecundity put $X=e^{\kappa r/2}$ and $Y=e^{\kappa/2}$. Then
$\Delta_\infty=-\tfrac25\,\mathcal N/\widehat{\mathcal N}$, with
\begin{equation}
 \begin{aligned}
 \mathcal N={}&X^3Y^8-X^3Y^6-2X^4Y^4+2X^2Y^6\\
 &-2X^2Y^2+2Y^4+XY^2-X ,
 \end{aligned}
\end{equation}
and $\widehat{\mathcal N}$ the same sum with every sign positive. Since
$\mathcal N(1/X,1/Y)=-\mathcal N(X,Y)/(X^4Y^8)$, the threshold is even in
$\kappa$, $r_{\rm crit}=3/2+57\kappa^2/64+127\kappa^4/4096+O(\kappa^6)$.
Write $\mathcal N=X^3Y^4g(X)$ with
\begin{equation}
 \begin{gathered}
 g(X)=Y^4-Y^2-2X+\varphi(X),\\
 \varphi(X)=\frac{2(Y^4-1)}{XY^2}+\frac{2}{X^3}+\frac{Y^2-1}{X^2Y^4}.
 \end{gathered}
\end{equation}
For $\kappa>0$ each term of $\varphi$ is positive and decreasing in $X$, so
$g$ decreases strictly in $r$, with $g(1)>0$ for $Y>1$ and
$g(X)\to-\infty$ as $X\to\infty$. Every $\kappa>0$ therefore has exactly one
threshold, and $\Delta_\infty>0$ if and only if $r>r_{\rm crit}$. The root
lies between $X_0=Y^2(Y^2-1)/2$ and $X_0+\varphi(X_0)/2$.
Define the positive relative bracket width
$\bar\varepsilon(\kappa)=\varphi(X_0)/(2X_0)$. Then
\begin{equation}
 0<r_{\rm crit}-4+\frac{2\ln2}{\kappa}-\frac2\kappa\ln\bigl(1-e^{-\kappa}\bigr)
 <\frac2\kappa\ln(1+\bar\varepsilon).
\end{equation}
Here $\bar\varepsilon=4e^{-3\kappa}[1+O(e^{-\kappa})]$ as $\kappa\to\infty$. As $r\to\infty$
at fixed $\kappa$, $\Delta_\infty\to2/5$. The class reduction of Appendix~\ref{app:c5-reduction} bounds the fast inverse on compact positive-fecundity domains. As $\kappa r$ grows, exponential fecundity changes both relaxation and exit rates. The contact separation needed to approach the fixed-$(\kappa,r)$ threshold can therefore grow; a joint limit requires control of these rates and the elimination remainder along the chosen scaling. 

\subsection{Why the exponential selected threshold increases}
\label{app:v10-exponential-monotonicity}
Take equal core proportions on the five-cycle, uniform singleton
introduction, exponential fecundity, $r=b/c$ and $\kappa=\delta c>0$.
The hierarchy limit is taken first at each fixed $(r,\kappa)$.
Appendix~\ref{app:c5-selection} proves uniqueness of the threshold
at each fixed $\kappa$.
That threshold increases strictly with $\kappa$, from $3/2$ as
$\kappa\downarrow0$ toward four as $\kappa\to\infty$.
\begin{proof}
Using its variables $X,Y$ and sign polynomial
$\mathcal N$, division by $2X^2Y^4>0$ and grouping opposite
exponentials gives
\begin{align}
F(\kappa,r)={}&-\mathcal N/(2X^2Y^4)\notag\\
={}&2\sinh(\kappa r)-2\sinh\kappa
-\sinh[\kappa(r/2+2)]\notag\\
&+\sinh[\kappa(r/2+1)].
\label{eq:v10-exponential-sinh}
\end{align}
To control the change with selection, fix $3/2<r<4$ and put
$a=r/2+2$, $b_0=r/2+1$. The odd Taylor coefficients of $F$ are
\begin{equation}
c_m(r)=2(r^m-1)-(a^m-b_0^m),\quad m=1,3,5,\ldots .
\end{equation}
For real $m>0$ define
\begin{equation}
Q(m)=\frac{a^m-b_0^m}{r^m-1}
=\frac{\int_{b_0}^a x^{m-1}\,dx}{\int_1^r x^{m-1}\,dx}.
\end{equation}
The derivative of $\log Q$ is the difference between the means of
$\log x$ under the positive density $x^{m-1}$ restricted to the two
intervals. Both endpoints of $[b_0,a]$ exceed those of $[1,r]$.
Their conditional distributions are therefore strictly ordered. On an overlap, the difference of their distribution functions is affine in the unnormalized cumulative mass and has the required sign at both endpoints. Hence $Q'(m)>0$.
Since $Q(1)=1/(r-1)<2$ and $Q(m)\to\infty$ as $m\to\infty$,
the coefficients $c_m$ have a positive initial part and a negative tail,
with at most one intervening zero.

Let $m_0$ be the largest odd index with positive coefficient.
Then $F(\kappa,r)/\kappa^{m_0}$ decreases strictly for $\kappa>0$.
Every lower positive power decreases after division, and every higher
negative power also decreases. Termwise differentiation is valid on
compact positive $\kappa$ intervals by absolute convergence.
For small $\kappa$, $F=(2r-3)\kappa+O(\kappa^3)>0$; for large
$\kappa$, its negative exponential with exponent $a>r$ dominates.
Every $r\in(3/2,4)$ is therefore crossed exactly once.
At $r=3/2$, the linear term vanishes and all higher odd coefficients
are negative, whereas
\begin{equation}
F(\kappa,4)=\sinh(4\kappa)+\sinh(3\kappa)-2\sinh\kappa>0.
\end{equation}
The unique threshold lies strictly in $(3/2,4)$.
It is continuous by the nonzero $X$ derivative, and the single crossing
at every intermediate $r$ makes it injective, hence strictly monotone.
Its weak-selection limit is $3/2$ and its large-$\kappa$ limit is four,
so it is strictly increasing.
\end{proof}

\subsection{Elimination with a source and a measured response}
\label{app:circuit}
The voltage reduction concerns a sequence of positive weightings on a
fixed support with protected sites $p_1,p_2$ joined by a bridge.
This subsection states its preparation assumptions explicitly and bounds the cost and benefit coefficients separately, so that the estimates remain valid at $b\mathfrak a=c$, where the leading terms of $\mathcal S^\mu$ cancel (Eq.~\eqref{eq:v10-circuit-amplitudes}), and when the leading benefit itself vanishes.

The killed ancestral-pair generator is $-\mathsf M$ in ancestral time
$\widehat t=t/N$, with single-lineage rates $P_{ki}$. The source
$g^\mu_{ij}=N(\mu_i+\mu_j)$ keeps $D^\mu$ in physical global-attempt units.
All stopping times in this subsection use $\widehat t$. Write $X_{\widehat t}$ for the unordered ancestral pair in this clock. Let $S=\{p_1,p_2\}$ be the pin-pair state, let $\widehat\tau_S$ be its hitting time, and put $\widehat\tau_{\rm meet}=\tau_{\rm meet}/N$ for coalescence in this clock. The strong Markov property gives the exact identity
\begin{gather}
D_{ij}^\mu=A_{ij}^\mu+\varphi_{ij}D_S^\mu, \label{eq:v10-circuit-decomposition}
 \\
A_{ij}^\mu=\mathbb E_{ij}\int_0^{\widehat\tau_S\wedge\widehat\tau_{\rm meet}}g^\mu(X_{\widehat t})\,d\widehat t,\quad
\varphi_{ij}=\Pr_{ij}(\widehat\tau_S<\widehat\tau_{\rm meet}). \notag
\end{gather}
Fix $\mu$ and suppress its superscript on $D,K,J$ and the stopped-source
terms in the following calculation. Then $A_S=0$, $\varphi_S=1$, and consequently
\begin{equation}
K=D_S\Gamma+K_A,\qquad J=D_SJ_{\rm hit}+J_A,
\end{equation}
where $\Gamma=\sum_{i<j}\gamma_{ij}\varphi_{ij}$,
$J_{\rm hit}=-\sum_{i<j}\mathsf B_{ij}\varphi_{ij}$, and the remaining two
terms are the same contractions against $A$.

Let $\tau_F$ be the largest mean single-lineage hitting time of the two
pins from a non-pin site and put
\begin{equation}
t_p=\vartheta_{p_1}+\vartheta_{p_2},\quad
\varepsilon_T=t_p\tau_F,\quad
\varepsilon_\gamma=\gamma_{\rm br}^{-1}\sum_{e\ne\{p_1,p_2\}}\gamma_e.
\end{equation}
Define $\tau_2=\max_{\{i,j\}\ne S}\mathbb E_{ij}(\widehat\tau_S\wedge\widehat\tau_{\rm meet})$.
If $\varepsilon_T<1$,
\begin{equation}
\tau_2\le\frac{2\tau_F}{1-\varepsilon_T}.
\label{eq:v10-circuit-fringe-time}
\end{equation}
Indeed, a lineage at a pin moves before a fringe lineage reaches either
pin with probability at most $\varepsilon_T$. Starting with two fringe lineages,
their first pin hit costs at most $\tau_F$; the remaining fringe hit costs
at most another $\tau_F$, unless a pin moves first. Thus
$\tau_2\le2\tau_F+\varepsilon_T\tau_2$.

Let $h_{\rm hit}(i)$ be the probability that a neutral lineage hits $p_1$ before
$p_2$, and define the conditional assortment
\begin{equation}
\mathfrak a=(Ph_{\rm hit})_{p_1}-(Ph_{\rm hit})_{p_2}.
\end{equation}
The following coefficient bounds make the reduction quantitative. With
$G=\max g_{ij}^\mu$,
\begin{equation}
\Lambda_\mu=G\tau_2/D_S,\qquad
\varepsilon_\mu=N|\mu_{p_1}-\mu_{p_2}|\tau_F/D_S,
\end{equation}
the exact inequalities are
\begin{align}
0&\le\frac{K}{\gamma_{\rm br}D_S}-1\le \varepsilon_\gamma(1+\Lambda_\mu),\\
\left|\frac{J}{\gamma_{\rm br}D_S}-\mathfrak a\right|
&\le \varepsilon_T+2\varepsilon_\gamma+\varepsilon_\mu+2(\varepsilon_T+\varepsilon_\gamma)\Lambda_\mu.
\label{eq:v10-circuit-errors}
\end{align}
\begin{proof}
Write $L_\gamma=\gamma_{\rm br}\mathbf d_{\rm br}\mathbf d_{\rm br}^{\mathsf T}+L'$ with
$\mathbf d_{\rm br}=e_{p_1}-e_{p_2}$. Then
$\gamma_{\rm br}\le\Gamma\le\gamma_{\rm br}(1+\varepsilon_\gamma)$.
The nonbridge part of $\mathsf B$ has absolute off-diagonal coefficient
sum at most $2\varepsilon_\gamma\gamma_{\rm br}$. For a fringe site $i$, the probability
that a pinned ancestral pair reaches $S$ is $1-h_{\rm hit}(i)$ or $h_{\rm hit}(i)$, according
to which pin it starts at, up to the event that the pin moves first.
The latter has probability at most $\vartheta_{p_j}\tau_F$.
Using $\sum_i(P_{p_1i}-P_{p_2i})=0$ gives
$|J_{\rm hit}-\gamma_{\rm br}\mathfrak a|
\le\gamma_{\rm br}(\varepsilon_T+2\varepsilon_\gamma)$.

For the source remainder, couple the walks started at $(p_1,i)$ and
$(p_2,i)$ through the same fringe lineage. Until its first pin hit or
the first of the two independent pin clocks, their reward difference is
$N(\mu_{p_1}-\mu_{p_2})$.
The exceptional pin-clock event has probability at most $\varepsilon_T$; after it,
each remaining stopped reward is at most $G\tau_2$.
Therefore
\begin{equation}
|A_{p_1i}-A_{p_2i}|
\le N|\mu_{p_1}-\mu_{p_2}|\tau_F+2G\varepsilon_T\tau_2.
\end{equation}
The row-difference norm $\sum_i|P_{p_1i}-P_{p_2i}|$ is at most two,
so the bridge contribution and nonbridge bound together give
\begin{equation}
|J_A|/\gamma_{\rm br}\le
N|\mu_{p_1}-\mu_{p_2}|\tau_F+2(\varepsilon_T+\varepsilon_\gamma)G\tau_2.
\end{equation}
Also $0\le K_A\le \varepsilon_\gamma\gamma_{\rm br}G\tau_2$, since $A_S=0$.
Combining these inequalities proves Eq.~\eqref{eq:v10-circuit-errors}.
\end{proof}

Thus $R^\mu\to1/\mathfrak a_\infty$ if
$\varepsilon_T,\varepsilon_\gamma\to0$, $\mathfrak a\to\mathfrak a_\infty>0$,
$\Lambda_\mu=O(1)$ and $\varepsilon_\mu\to0$.
The conditions $\Lambda_\mu=O(1)$ and $\varepsilon_\mu\to0$ restrict the introduction rule; for the two named rules they are verified as follows. For uniform introduction,
$G=2$, $D_S\ge2/t_p$, so
$\Lambda_U\le2\varepsilon_T/(1-\varepsilon_T)$ and $\varepsilon_U=0$.
For birth-associated introduction $g_{ij}^{\mathrm b}=\vartheta_i+\vartheta_j$,
$D_S^{\mathrm b}\ge1$, $G_{\mathrm b}\le2(N-1)$ and $\varepsilon_{\mathrm b}\le \varepsilon_T$.
Bounded $\tau_F$ then gives $\Lambda_{\mathrm b}\le2G_{\mathrm b}\tau_F/(1-\varepsilon_T)=O(1)$.
The pin-pair source under birth-associated introduction is $t_p$, whose accumulation during the first residence of mean $1/t_p$ already gives $D_S^{\mathrm b}\ge1$. The bound on $\Lambda_{\mathrm b}$ also retains source accumulated during fringe excursions.

The same decomposition controls the response when leakage remains
comparable to the bridge. Under uniform introduction, if $\varepsilon_T\to0$,
$\varepsilon_\gamma=O(1)$ and $\liminf J_{\rm hit}/\Gamma>0$, then
\begin{equation}
 \frac{R^U}{\Gamma/J_{\rm hit}}\longrightarrow1.
 \label{eq:v10-bounded-leak-circuit}
\end{equation}
Indeed, the source remainders obey
$K_A/(D_S\Gamma)=o(1)$ and $J_A/(D_S\Gamma)=o(1)$.
The separate estimates of $K$ and $J$ hold without the lower bound on $J_{\rm hit}/\Gamma$ and so remain valid when the leading benefit vanishes; that bound enters the ratio $R^U=K/J$, where it keeps the $o(D_S\Gamma)$ remainder of $J$ small relative to $D_SJ_{\rm hit}$.

Return now to $\varepsilon_\gamma\to0$. The amplitudes distinguish the two
preparations even when their thresholds agree. The direct coalescence rate out of the pin pair is
$\lambda_{\rm br}=\gamma_{\rm br}(s_{p_1}+s_{p_2})$.
If $q_e$ is the chance that its first jump is a fringe excursion,
\begin{equation}
\frac{q_e}{1-q_e}=
\frac{t_p-\lambda_{\rm br}}{\lambda_{\rm br}}\le \varepsilon_\gamma.
\end{equation}
Renewal at that first jump and its return bounds the accumulated source by
\begin{equation}
\begin{gathered}
0\le t_pD_S^U/2-1\le \varepsilon_\gamma(1+t_p\tau_2),\\
0\le D_S^{\mathrm b}-1\le \varepsilon_\gamma(1+G_{\mathrm b}\tau_2).
\end{gathered}
\end{equation}
Under the corresponding bounded-source hypotheses,
$D_S^U=2t_p^{-1}[1+o(1)]$ and $D_S^{\mathrm b}=1+o(1)$.
Retain the exact amplitude $\mathfrak A_\mu=2\gamma_{\rm br}D_S^\mu/(NZ)$
defined in the main text. The renewal estimates give
\begin{equation}
\begin{gathered}
\mathfrak A_U=\frac4{NZ(s_{p_1}+s_{p_2})}[1+o(1)],\\
\mathfrak A_{\mathrm b}=\frac{2\gamma_{\rm br}}{NZ}[1+o(1)].
\end{gathered}
\end{equation}
For either named introduction rule $\mu$, the physical coefficient statements are
\begin{equation}
\chi_c^\mu=\mathfrak A_\mu[1+o(1)],\qquad
\chi_b^\mu=\mathfrak A_\mu[\mathfrak a+o(1)],
\end{equation}
or equivalently
\begin{equation}
\mathcal S^\mu=\mathfrak A_\mu[b\mathfrak a-c+o(b+c)],\quad
\mathfrak A_{\mathrm b}/\mathfrak A_U=\tfrac12t_p[1+o(1)].
\label{eq:v10-circuit-amplitudes}
\end{equation}
The additive form remains meaningful at $b\mathfrak a=c$.
If in addition $Zs_{p_1}s_{p_2}/(s_{p_1}+s_{p_2})\to1$, then with
$\xi_{\rm pin}=s_{p_2}/s_{p_1}$,
$\mathfrak A_U=4\xi_{\rm pin}/[N(1+\xi_{\rm pin})^2][1+o(1)]$.

For the pin and partner variables $X,Y$ defined in Sec.~\ref{sec:c5-contested}, neutral ancestry gives
$\mathbb E[Y\mid X=1]=(Ph_{\rm hit})_{p_1}$ and
$\mathbb E[Y\mid X=0]=(Ph_{\rm hit})_{p_2}$.
Thus $\operatorname{Cov}(X,Y)/\operatorname{Var}(X)=\mathfrak a$.
This is the regression coefficient of the partner's type on the pin's type, conditional on the contested state, as in Sec.~\ref{sec:c5-contested}.

At finite selection apply the time-change identity
\eqref{eq:c5-physical-class-measure} to the fast fringe. If
$\widetilde\varpi_\delta$ is its stationary distribution before division
by total fecundity, physical exit rates average against
$F_\Sigma\widetilde\varpi_\delta/\langle F_\Sigma\rangle$.
The local denominator cancels, leaving averages of fecundities with
$\widetilde\varpi_\delta$, together with bridge factors and a common
normalizer.
With the opposed pins fixed at $n_{p_1}=1$, $n_{p_2}=0$, define
\begin{equation}
\mathfrak a_\delta=
\bigl\langle(Pn)_{p_1}-(Pn)_{p_2}\bigr\rangle_{\widetilde\varpi_\delta}.
\end{equation}
This average uses the selected fast distribution before division by total
fecundity. At neutrality it reduces to the ancestral voltage difference
$\mathfrak a$. For affine fecundity the difference is proportional to
$\delta(b\mathfrak a_\delta-c)$.
For nonlinear fecundity one must retain
$\langle\phi(\delta f_{p_1})\rangle$ and
$\langle\phi(\delta f_{p_2})\rangle$; evaluating $\phi$ at average payoff
changes the answer. The selected cycle formulas use these averages with
the hierarchy limit taken first at fixed admissible selection parameters.
\subsection{Voltage drops, wall motion and cycle limits}
\label{app:wall-cycles}
Let $p_1=c_0,c_1,\ldots,c_m,c_{m+1}=p_2$ be a chain with opposite
frozen types at its ends, positive neutral rates and chain-only copying.
Its unique stationary distribution has a single wall, with bond
probabilities proportional to the ancestral resistances,
as in Eq.~\eqref{eq:v10-wall-voltage}.
\begin{proof}
Copying preserves the number of walls or lowers it by two. Opposite
fixed ends make that number odd, and a finite sequence of copies from
one end reaches a one-wall configuration with positive probability.
The one-wall class is closed and irreducible, so it supports the unique
stationary distribution. Its wall lies on bond $(c_k,c_{k+1})$ with
probability $h_{\rm hit}(c_k)-h_{\rm hit}(c_{k+1})$, since $h_{\rm hit}(c_k)$ is the probability
that the ancestry at $c_k$ reaches $p_1$. Kirchhoff's constant current
along the chain makes these voltage drops proportional to the
resistances. Thus
\begin{equation}
 \varpi_k=\frac{\mathcal R_k}{\sum_l\mathcal R_l},\qquad
 \mathcal R_k=\frac{s_{c_k}s_{c_{k+1}}}{w_{c_kc_{k+1}}}
            =\gamma_{c_kc_{k+1}}^{-1}.
 \label{eq:v10-wall-voltage}
\end{equation}
These are also the detailed-balance weights of the wall walk. The
voltage argument connects that stationary description directly to
ancestry.
\end{proof}

\labelalias{prop:wall}{app:wall-cycles}

Applying the wall distribution to the cycle hierarchy proves the threshold
limits in Eq.~\eqref{eq:cN-contested} under both introduction rules.
The pins have
$s_p=1+q$ and $\vartheta_p=(1+q)^{-1}+q/(q+Q)\to0$.
The fringe lineage walk has end absorption rates $q/(1+q)\to1$ and
interior rates of order one, so $\tau_F$ is bounded and $\varepsilon_T\to0$.
With $\gamma_{\rm br}=(1+q)^{-2}$, every nonbridge edge has
$\gamma/\gamma_{\rm br}\asymp q^2/Q$, hence $\varepsilon_\gamma\to0$.
Bounded $\tau_F$ also verifies the source condition for both named
introduction rules. The voltage reduction therefore applies. Its
limiting series consists of two end resistances $Q$, two adjacent resistances $2Q$ and $m-3$ remaining resistances $4Q$ for $m\ge3$; for $m=2$, all three resistances equal $Q$. Substitution in
Eq.~\eqref{eq:v10-wall-voltage} determines $\mathfrak a$ and proves
Eq.~\eqref{eq:cN-contested} under both rules.

\labelalias{prop:cycles}{eq:cN-contested}

\subsection{Fringe bounds and their equality geometry}
\label{app:fringe-floors}
Under the hypotheses of Proposition~\ref{prop:v10-contest}, assume also
that every fringe lineage reaches a pin in the limiting walk.
The limiting assortment is at most $1/3$ with two fringe sites and
$2/3$ with three. Both bounds are attained; the three-site equality
cases have the five-cycle limiting geometry up to relabeling.
\begin{proof}
To apply the extremal comparison to Proposition~\ref{prop:v10-contest}, take a subsequence along which $P\to\bar P$ and the limiting ancestral walk reaches a pin from every fringe site. A fringe site's total ancestral jump rate is at most $N-1$, so $\tau_F\ge1/(N-1)$. The hypothesis $\varepsilon_T\to0$ therefore forces both pin overwrite rates to vanish. Hence $\bar P_{vp}=0$, and the pins send only to fringe sites, which send only within the fringe. The limiting killed-fringe matrix is invertible by the assumed absorption. Continuity of its inverse gives $h\to\bar h$ and $\mathfrak a_\infty=\mathfrak a(\bar P)$. This identifies the limiting response with the voltage functional optimized below.
Allowing all such limiting matrices, without imposing reversibility, enlarges the admissible set. An upper bound proved on this larger set also bounds every reciprocal realization.
A one-site fringe is excluded by looplessness and the row-sum condition.
Over this relaxed set,
\begin{equation}
 \mathfrak a\le\frac13\quad\text{for two fringe sites},\quad
 \mathfrak a\le\frac23\quad\text{for three}.
 \label{eq:v10-fringe-floors}
\end{equation}
Both bounds are attained. With three fringe sites, equality holds
exactly when the pins send to distinct sites $x,y$, these send only
to the third site $z$, and $z$ sends to $x$, $y$ or both. This is the
five-cycle limiting geometry, up to relabeling.

Parametrize the pin rows by unnormalized vectors $a,b\ge0$ and each
fringe row by an unnormalized vector on the other fringe sites. Let
$S_a=\sum a$, $S_b=\sum b$ and let $S_v$ be the sum of the
unnormalized fringe row at $v$. Write the normalized harmonic system as
$Mh=\alpha$, with $\alpha=a/S_a$. Explicitly,
$M_{vv}=\sum_i\bar P_{iv}$ sums incoming copying from both pins and the
fringe, and $M_{vz}=-\bar P_{zv}$ for distinct fringe sites $v,z$.
The source is $\alpha_v=\bar P_{p_1v}$, since ancestry follows each
received copy back to its parent. Fix the uncancelled common clearing
\begin{equation}
\mathscr C_{\rm row}=S_aS_b\prod_vS_v,\quad
M'=\mathscr C_{\rm row}M,\quad \rho'=\mathscr C_{\rm row}\alpha.
\end{equation}
The cleared matrix is a nonsingular M-matrix precisely on the
absorbing domain, so $\det M'>0$ there. Cramer's rule writes
$\mathfrak a=r_{\rm num}/[(\sum a)(\sum b)\det M']$, with
$r_{\rm num}=(S_ba-S_ab)^{\T}\operatorname{adj}(M')\rho'$, where $\operatorname{adj}$ denotes the transposed cofactor matrix.
The polynomial
$\mathfrak a_{\max}(\sum a)(\sum b)\det M'-r_{\rm num}$ has nonnegative
coefficients for $\mathfrak a_{\max}=1/3$ with two fringe sites (14 terms, degree 10)
and for $\mathfrak a_{\max}=2/3$ with three (14\,076 terms, degree 17), by exact
coefficient arithmetic. This proves both inequalities; the four- and
five-cycle limits attain them.

The same polynomial determines equality. A nonnegative-coefficient
polynomial vanishes on the closed orthant exactly when no monomial has
all its variables nonzero. The 14\,076 monomials have 1\,224 distinct
supports and 14 maximal allowed coordinate sets. Two have a pin with a zero outgoing row and three have a fringe site with a zero outgoing row, so their row normalization is undefined. Three others have a fringe site that receives no copies and is therefore an absorbing ancestral state outside the pins. The remaining six are the stated configurations up to relabeling.
\end{proof}

It follows that $\mathfrak a=1$ requires at least four fringe sites;
the two separated triples of $P_6$ attain it. Perfect assortment needs
separately aligned followers, while a shared two- or three-site fringe
forces mixed ancestry. These bounds use the hypotheses of Proposition~\ref{prop:v10-contest} and an absorbing limiting fringe. The dense conjecture also covers finite weights and hierarchies outside this reduction; Appendix~\ref{app:I-dense-bulk} identifies a further proved region.
\labelalias{prop:six}{app:fringe-floors}
\labelalias{prop:floors}{eq:v10-fringe-floors}
\labelalias{prop:unique}{app:fringe-floors}

\subsection{Leaky two-pin limits and arbitrary singleton preparation}
\label{app:v10-leaky-pins}
A controlled leak can remain in the slow competition. To specify this
limit, take three fringe sites, two pins and the contacts
\begin{equation}
 \begin{gathered}
 w_{p_1p_2}=1,\qquad
 w_{p_1x}=q\eta_{1x},\qquad
 w_{p_2x}=q\xi_0\eta_{2x},\\
 w_{xy}=q^2\xi_0G_{xy},\qquad q\longrightarrow\infty.
 \end{gathered}
 \label{eq:v10-leaky-shape}
\end{equation}
Here $\xi_0=\lim\xi_{\rm pin}>0$ is the limiting pin-strength ratio, $G$ is symmetric with zero diagonal and positive
entries off the diagonal, and the two positive pin rows satisfy
$\sum_x\eta_{ix}=1$. These shape parameters remain fixed in the
$q$ limit. Writing $g_x=\sum_yG_{xy}$, the limiting fringe row is
$\bar P_{xy}=G_{xy}/g_x$. Reversibility fixes the normalized
nonbridge conductances as
\begin{equation}
 \widehat\gamma_{p_ix}=\frac{\eta_{ix}}{g_x},\qquad
 \widehat\gamma_{xy}=\frac{G_{xy}}{g_xg_y},\qquad
 \widehat\gamma_{p_1p_2}=1.
\end{equation}
The hats mean division by the bridge conductance before taking the
limit.
For this fixed-shape hierarchy, every common selection-independent
singleton introduction satisfies $\liminf_{q\to\infty}R^\mu\ge3/2$
on the promoting domain, including contact-dependent mixtures of sources.
\begin{proof}
Increasing the overall scale of $G$ by $t_{\rm fr}>0$ leaves
$\bar P$ fixed and divides each nonbridge conductance by $t_{\rm fr}$.
Thus the uniform-introduction limit has the form
\begin{equation}
 R_\infty^U(t_{\rm fr})=
 \frac{1+\Gamma_1/t_{\rm fr}}
 {\mathfrak a+J_1/t_{\rm fr}},
 \label{eq:v10-leaky-mobius}
\end{equation}
where $\Gamma_1,J_1$ are the nonbridge contractions against the frozen
pair-hitting probabilities at $t_{\rm fr}=1$.
The denominator must be positive. The frozen fringe bound gives
$2-3\mathfrak a\ge0$. Solving the three single-lineage harmonic
equations and the three fringe-pair equations gives
$2\Gamma_1-3J_1$ as a rational function with positive numerator and
denominator coefficients. Consequently $R_\infty^U\ge3/2$.

The preparation can also be concentrated on the fringe, where its fast
source cannot be dropped. Let $\mathsf G_{\rm fr}(y,x)$ be the expected
ancestral time at $x$ before a lineage starting at $y$ reaches a pin,
and put $\lambda_{yi}=\eta_{iy}/g_y$. In the following fringe-source
calculation, $D^{(x)},A^{(x)},K^{(x)},J^{(x)}$ denote the limits after
removing the common reward factor $N$ and using $\widehat\gamma$ in
the contractions. The functions $h_y=h_{\rm hit}(y)$ and $\varphi$ are the frozen-kernel single-lineage and pair-hitting probabilities. Thus the finite-weight
fringe-channel contractions divided by $N\gamma_{\rm br}$ converge to these
$K^{(x)},J^{(x)}$. Renewal of the pin pair gives
\begin{equation}
 D_S^{(x)}=
 \frac{\sum_y(\lambda_{y1}+\xi_0\lambda_{y2})\mathsf G_{\rm fr}(y,x)}
 {1+\xi_0+\sum_y\lambda_{y1}(1-h_y)
                  +\xi_0\sum_y\lambda_{y2}h_y}.
 \label{eq:v10-fringe-source-renewal}
\end{equation}
The stopped source has
$A_{p_iy}^{(x)}=\mathsf G_{\rm fr}(y,x)$; its three fringe-pair entries solve
the same killed pair equations with reward
$\mathbf1\{x\text{ belongs to the pair}\}$ and these boundary values.
Contracting $D^{(x)}=A^{(x)}+\varphi D_S^{(x)}$ then gives
$2K^{(x)}-3J^{(x)}>0$. In the rational parametrization by the three
entries of $G$, two unnormalized pin rows and $\xi_0$, the numerator
and denominator have positive coefficients. The pin channels have a diverging source scale. For a singleton at
either pin, $g^\mu_S=N$ gives $D_S\ge N/t_p$, while the stopped fast
source remains bounded at fixed shape. Its relative contribution
therefore vanishes. Indeed, $\max_{i<j}g_{ij}^\mu=N$, $\Lambda_\mu\le t_p\tau_2\to0$ and $\varepsilon_\mu\le\varepsilon_T\to0$. The source-remainder estimates in Appendix~\ref{app:circuit} then give $K_A/(D_S\Gamma),J_A/(D_S\Gamma)\to0$, as in the argument for Eq.~\eqref{eq:v10-bounded-leak-circuit}. The same positive leak
gap proves the corresponding pin-channel inequality. Summing the five channel inequalities with the nonnegative
introduction masses proves, on the promoting domain of
Eq.~\eqref{eq:v10-leaky-shape},
\begin{equation}
 \liminf_{q\to\infty}R^\mu\ge\frac32.
 \label{eq:v10-leaky-all-source-floor}
\end{equation}
This includes both named introduction rules and contact-dependent
mixtures. The exact coefficient expansions have respectively
$152\,736/163\,248$ terms for the uniform leak gap and
$3\,652\,260/3\,943\,758$ for a fringe channel, after reduction over the
rational-function field. For these expansions write $\eta_{1x}=a_x/\sum a$ and
$\eta_{2x}=b_x/\sum b$, in variables
$(G_{01},G_{02},G_{12},a_0,a_1,a_2,b_0,b_1,b_2,\xi_0)$.
Numerators and denominators are reduced over $\mathbb Q$ and normalized
with positive denominator coefficients. Their signs prove the inequalities throughout the positive shape domain.
\end{proof}

The non-strict floor persists under limits of the positive shape parameters taken after $q\to\infty$. If the shape varies with $q$, the relative transition scales can change. Extending Eq.~\eqref{eq:v10-leaky-all-source-floor} then requires controlling the shape dependence of the eliminated inverse and the response remainder along that joint limit.

\subsection{Three protected sites with a two-site fringe}
\label{app:v10-three-pins}
The same elimination controls a different five-site boundary. Let
$p=0,1,2$ label the pins and $x,y$ the fringe, with
\begin{equation}
 \begin{gathered}
 w_{xy}=1,\qquad w_{ij}=\varepsilon_3^2B_{ij},\\
 w_{px}=\varepsilon_3 X_p,\qquad w_{py}=\varepsilon_3 Y_p.
 \end{gathered}
 \label{eq:v10-three-pin-shape}
\end{equation}
The nine shape coefficients are fixed and positive while
$\varepsilon_3\downarrow0$. The fringe sites copy one another at
leading order, whereas the three pins change through both bridges and
leaks on the $\varepsilon_3$ scale.
For either named introduction rule, each promoting fixed positive shape
has limiting threshold strictly above $5/2$. The infimum is $5/2$,
approached by a further ordered limit of the shape parameters.

\begin{proof}
Let $s$ index the three pin pairs and $f$ the seven pairs containing a
fringe site, and expand $\mathsf M=\mathsf M_0+
\varepsilon_3\mathsf M_1+O(\varepsilon_3^2)$. Define
\begin{equation}
 \mathsf E_3=-(\mathsf M_0^{ff})^{-1}\mathsf M_0^{fs},\qquad
 \mathsf L_3=\mathsf M_1^{ss}+\mathsf M_1^{sf}\mathsf E_3.
\end{equation}
Both killed inverses exist for positive shape coefficients. Uniform
introduction has $D=\varepsilon_3^{-1}D_{-1}+O(1)$, with
\begin{equation}
 \mathsf L_3(D_{-1})_s=2\one,
 \qquad(D_{-1})_f=\mathsf E_3(D_{-1})_s.
\end{equation}
For birth-associated introduction $D=D_0+O(\varepsilon_3)$ and
\begin{equation}
 \begin{aligned}
 \mathsf L_3(D_0)_s&=(g_1)_s-
 \mathsf M_1^{sf}(\mathsf M_0^{ff})^{-1}(g_0)_f,\\
 (D_0)_f&=(\mathsf M_0^{ff})^{-1}(g_0)_f+\mathsf E_3(D_0)_s.
 \end{aligned}
\end{equation}
Here $g=g_0+\varepsilon_3g_1+O(\varepsilon_3^2)$ is the corresponding
pair source. Thus the fast source remains in both the slow equation and
the response readout.

Contracting these solutions with the leading cost and benefit rows
gives $K_\infty^\mu,J_\infty^\mu$, after removing their common
$\varepsilon_3^{-1}$ factor for uniform introduction.
For either named introduction rule, the rational function
$K_\infty^\mu-\tfrac52J_\infty^\mu$ has a positive-coefficient numerator
and denominator. Their term counts in reduced form are
$74\,478/90\,827$ for uniform introduction and
$282\,072/159\,801$ for birth-associated introduction. The variables are the six $X_p,Y_p$ and the three $B_{01},B_{02},B_{12}$,
with reduction over $\mathbb Q$ and a positive denominator normalization.
Hence
\begin{equation}
 R_\infty^\mu>\frac52\qquad(J_\infty^\mu>0)
 \label{eq:v10-three-pin-floor}
\end{equation}
throughout the positive shape domain.

The infimum is $5/2$ under both rules. One explicit ordered limit first
takes $\varepsilon_3\to0$, then multiplies all $B_{ij}$ by a common
factor tending to infinity. In the resulting bridge competition let
the leading pin-strength coefficients $X_p+Y_p$ have ratios
$(1,L^{-5},L^3)$, and let the bridge weights have ratios
$(L^{19},1,L^{20})$ in the order $(01,02,12)$. Take the pin rows on
$(x,y)$ proportional to $(1,L^{19})$, $(L^{20},1)$ and
$(L^{29},L^5)$, respectively, normalizing each row to sum to $X_p+Y_p$.
The leading terms of
the reduced rational functions give $R_\infty^\mu\to5/2$ as $L\to\infty$.
\end{proof}

The strict floor and its infimum refer to shapes held fixed during $\varepsilon_3\to0$, with the infimum taken afterwards. If $X_p,Y_p,B_{ij}$ vary with $\varepsilon_3$, the three slow pin pairs and seven fast pairs can develop further time scales. A joint reduction must track the shape dependence of $(\mathsf M_0^{ff})^{-1}$ and $\mathsf L_3^{-1}$ and control discarded terms relative to the retained slow rates.

\subsection{A readout can eliminate an undetermined slow coefficient}
\label{app:v10-tied}
\label{app:v10-readout-cancellation}
The singular pair inverse may contain information that the response does not
need. This occurs in a five-site hierarchy with sites
$p,f,e,d_1,d_2$ and contacts
\begin{gather}
w_{fe}=1,\quad w_{pf}=\varepsilon_h L,\quad
w_{pe}=\varepsilon_h W,\quad w_{d_1d_2}=\varepsilon_h A, \notag \\
w_{pd_i}=\varepsilon_h^2X_i,\quad
w_{fd_i}=\varepsilon_h Y_i,\quad w_{ed_i}=\varepsilon_h Z_i . \label{eq:v10-tied-weights}
\end{gather}
Here $i=1,2$. All nine coefficients are positive and held fixed as $\varepsilon_h\downarrow0$; any later limit of these coefficients is taken after this one. Define
\begin{gather}
s_p^{(0)}=L+W,\quad a_i=A+Y_i+Z_i,\quad s_d=a_1+a_2, \notag\\
\ell=L/s_p^{(0)},\quad w=W/s_p^{(0)},\notag \\
\alpha_i=A/a_i,\quad y_i=Y_i/a_i,\quad z_i=Z_i/a_i,\\
r_d=\alpha_1+\alpha_2, \notag \\
t_f=1+\ell+y_1+y_2,\quad t_e=1+w+z_1+z_2 . \notag 
\end{gather}
Here $s_p^{(0)}$ is the leading coefficient of the pin strength, $s_p=\varepsilon_hs_p^{(0)}+O(\varepsilon_h^2)$. The limiting backward walk has two recurrent classes, the pin and the dimer.
Its probabilities $h_f,h_e$ of reaching the pin from the fringe solve
\begin{equation}
\begin{pmatrix}t_f&-1\\-1&t_e\end{pmatrix}
\binom{h_f}{h_e}=\binom\ell w .
\label{eq:v10-tied-first-solve}
\end{equation}
The determinant $t_ft_e-1$ is positive. The stationary probabilities within
the ancestral dimer are $(a_2/s_d,a_1/s_d)$.

Let $h^{\rm sep}_{ij}$ be the probability that two limiting ancestral walks reach
different recurrent classes before meeting. Write these probabilities as the vector $\mathbf h_{\rm sep}$. Its entries are
\begin{equation}
\begin{gathered}
h^{\rm sep}_{pf}=1-h_f,\quad h^{\rm sep}_{pe}=1-h_e,\quad h^{\rm sep}_{pd_i}=1,\\
h^{\rm sep}_{fd_i}=h_f,\quad h^{\rm sep}_{ed_i}=h_e,\quad h^{\rm sep}_{d_1d_2}=0,\\
\begin{aligned}
h^{\rm sep}_{fe}=q_f
&=\frac{\ell(1-h_e)+(y_1+y_2)h_e}{t_f+t_e}\\
&\quad+\frac{w(1-h_f)+(z_1+z_2)h_f}{t_f+t_e}.
\end{aligned}
\end{gathered}
\end{equation}
These are the first-step equations for the fast killed pair process.
In the copying process with opposed pin and dimer types,
$h_f,h_e$ are the fringe means and
$m_f=(h_f+h_e-q_f)/2$ is their joint occupation probability.
The leading cost and benefit contractions do not depend on the component
of the regular pair solution along the slow nullvector. This cancellation
reduces their calculation to two-dimensional systems and the readouts
\eqref{eq:v10-tied-readouts}.

\begin{proof}
Writing $\mathsf M=\mathsf M_0+\varepsilon_h\mathsf M_1+O(\varepsilon_h^2)$, the right kernel of
$\mathsf M_0$ is spanned by $\mathbf h_{\rm sep}$. Its normalized left kernel is supported on
$(p,d_1),(p,d_2)$ with masses $(a_2/s_d,a_1/s_d)$.
The left--right pairing is one. Put
\begin{gather}
p_0=L(1-h_f)+W(1-h_e)+X_1/a_1+X_2/a_2, \notag \\
\lambda_i=p_0+X_i/s_p^{(0)}+Y_i h_f+Z_i h_e,\\
\lambda_*=(a_2\lambda_1+a_1\lambda_2)/s_d>0. \notag
\end{gather}
Here $(\mathsf M_1\mathbf h_{\rm sep})_{pd_i}=\lambda_i$. For a source
$\mu(\varepsilon_h)=\mu^0+O(\varepsilon_h)$, the normalized left nullvector
averages the source over the two pin--dimer pair states as
\begin{equation}
\bar g_\mu=5\left(\mu_p^0+
\frac{a_2\mu_{d_1}^0+a_1\mu_{d_2}^0}{s_d}\right).
\end{equation}
The simple slow pole then gives
\begin{equation}
\begin{gathered}
D^\mu=\varepsilon_h^{-1}C_\mu\mathbf h_{\rm sep}+D_0^\mu
       +\zeta_\mu(\varepsilon_h)\mathbf h_{\rm sep}+O(\varepsilon_h),\\
C_\mu=\bar g_\mu/\lambda_* ,\qquad \zeta_\mu=O(1).
\end{gathered}
\label{eq:v10-tied-laurent}
\end{equation}
The bounded scalar can be absorbed into a fixed $D_0^\mu$, up to
$O(\varepsilon_h)$, for either named analytic introduction rule. For a
general source the following readouts cancel it exactly. The only singular conductance is
$\gamma_{d_1d_2}=\varepsilon_h^{-1}\gamma_*+O(1)$, where
$\gamma_*=A/(a_1a_2)$. Its cost readout annihilates $\mathbf h_{\rm sep}$ because
$h^{\rm sep}_{d_1d_2}=0$. Its benefit readout is
\begin{equation}
\begin{aligned}
\frac{\gamma_*}{2}\bigl[&-r_dD_{d_1d_2}
 +(y_1-y_2)(D_{fd_2}-D_{fd_1})\\
 &+(z_1-z_2)(D_{ed_2}-D_{ed_1})\bigr],
\end{aligned}
\end{equation}
which also annihilates $\mathbf h_{\rm sep}$. The apparent double pole cancels. Both leading
contractions depend on $D_0^\mu$ modulo its undetermined multiple of $\mathbf h_{\rm sep}$.

This observation reduces the remaining calculation to four scalar differences.
Define the remaining source contractions
\begin{equation}
\begin{gathered}
g_d=5(\mu_{d_1}^0+\mu_{d_2}^0),\\
g_\Delta=5(\mu_{d_2}^0-\mu_{d_1}^0).
\end{gathered}
\end{equation}
For uniform introduction $(\bar g_\mu,g_d,g_\Delta)=(2,2,0)$; for birth-associated
introduction this triple is $(2A/s_d,r_d,\alpha_1-\alpha_2)$.
Suppress the fixed preparation superscript on $D_0^\mu$ and set $d_{\rm pair}=(D_0)_{d_1d_2}$ and
$d_j=(D_0)_{jd_2}-(D_0)_{jd_1}$, for $j=p,f,e$. Then
\begin{align}
d_{\rm pair}={}&\frac{g_d}{r_d}+\frac{C_\mu}{r_d}\bigl[(Y_1+Y_2)h_f\notag\\
&+(Z_1+Z_2)h_e+(X_1+X_2)/s_p^{(0)}\bigr],\\
d_p&=\frac{g_\Delta+C_\mu(\lambda_1-\lambda_2)}{r_d}.
\end{align}
With $\Delta x=(X_2-X_1)/s_p^{(0)}$, $\Delta Y=Y_2-Y_1$ and
$\Delta Z=Z_2-Z_1$, the second two-dimensional solve is
\begin{equation}
\begin{pmatrix}t_f+r_d&-1\\-1&t_e+r_d\end{pmatrix}
\binom{d_f}{d_e}=\binom{b_f}{b_e},
\label{eq:v10-tied-second-solve}
\end{equation}
where
\begin{align}
b_f={}&\ell d_p+(y_1-y_2)d_{\rm pair}+g_\Delta \\
&-C_\mu[\Delta x(2h_f-1)+\Delta Yh_f+\Delta Z(h_f-q_f)], \notag\\
b_e={}&w d_p+(z_1-z_2)d_{\rm pair}+g_\Delta \\
&-C_\mu[\Delta x(2h_e-1)+\Delta Y(h_e-q_f)+\Delta Zh_e]. \notag
\end{align}
Subtracting the two dimer-index equations cancels the common first-order
fringe-rate corrections and gives these equations directly. Define
\begin{align}
k_{\rm sep}={}&\ell(1-h_f)+w(1-h_e)+q_f\\
&+(y_1+y_2)h_f+(z_1+z_2)h_e+\sum_i\frac{X_i}{s_p^{(0)}a_i},\notag \\
j_{\rm sep}={}&-\ell^2(h_e-m_f)-w^2(h_f-m_f)-q_f\notag\\
&+\sum_i[(y_i+z_i-2y_iz_i)m_f-y_i^2h_f-z_i^2h_e]\notag\\
&+\sum_i\frac{X_i}{s_p^{(0)}a_i}[(\ell-y_i)h_f+(w-z_i)h_e].
\end{align}
The result is the rational pair
\begin{align}
\bar K^\mu&=\gamma_*d_{\rm pair}+C_\mu k_{\rm sep},\\
\bar J^\mu&=\frac{\gamma_*}{2}
[-r_dd_{\rm pair}+(y_1-y_2)d_f+(z_1-z_2)d_e]\notag\\
&\quad+C_\mu j_{\rm sep},
\label{eq:v10-tied-readouts}
\end{align}
with $K^\mu=\varepsilon_h^{-1}\bar K^\mu+O(1)$ and the analogous expansion
for $J^\mu$. All denominators are bounded away from zero on compact positive
parameter sets.
\end{proof}
The expansion is uniform on compact positive shape-parameter sets. Joint limits toward a shape boundary require estimates for the resulting growth of the inverse coefficients and the remainder. The infima below take $\varepsilon_h\downarrow0$ first, at each fixed positive shape. For a source supported on $\{f,e\}$, $\bar g_\mu=g_d=g_\Delta=0$, so $\bar K^\mu=\bar J^\mu=0$. The first nonzero terms of both contractions must be found before taking their ratio. If $\bar J^\mu=0<\bar K^\mu$, promotion requires a positive subleading benefit term and its threshold diverges; the first nonzero term determines the rate of divergence.

Within this fixed-ratio family, the promoting limiting thresholds are
strictly above $3/2$ under uniform introduction and $5/2$ under
birth-associated introduction. These are the respective infima,
approached by further ordered limits of positive shape parameters.

\begin{proof}
Set $L=1$ by $\varepsilon_h'=L\varepsilon_h$ and replace
$(A,W,Y_i,Z_i,X_i)$ by $(A/L,W/L,Y_i/L,Z_i/L,X_i/L^2)$.
This changes the hierarchy coordinate and is not a common rescaling of
physical contacts. For uniform introduction,
\begin{equation}
\begin{gathered}
\widehat K=(\lambda_*/2)\bar K^U=K_0+X_1K_1+X_2K_2,\\
\widehat J=(\lambda_*/2)\bar J^U=J_0+X_1J_1+X_2J_2.
\end{gathered}
\end{equation}
The six rational coefficients depend on $(A,W,Y_1,Y_2,Z_1,Z_2)$ and are
obtained directly from Eqs.~\eqref{eq:v10-tied-first-solve}--%
\eqref{eq:v10-tied-readouts}. Write each pair with its least common
positive polynomial denominator as $K_i=U_i/D_i$, $J_i=V_i/D_i$.
Every coefficient of $D_i,U_i$ is positive. The nonzero positive and
zero coefficients of $2U_i-3V_i$, counted on the union of supports, are
$(19411,0)$, $(10008,2)$, $(10008,2)$ for $i=0,1,2$.
Thus $2\widehat K-3\widehat J>0$ at all positive parameters.
The family
\begin{equation}
\begin{gathered}
L=A=1,\quad W=Y_1=Y_2=Z_1=X_2=q^{-2},\\
Z_2=q,\quad X_1=q^4
\end{gathered}
\end{equation}
has $K_0\sim4/3$, $J_0\sim-1/3$,
$X_1K_1\sim q^4$, $X_1J_1\sim2q^4/3$ and negligible $X_2$ terms.
Taking $\varepsilon_h\downarrow0$ first and then $q\to\infty$ proves
the sharp unattained uniform floor $3/2$.

For birth-associated introduction, reducing the rational ratio from the
same equations gives $U/V$ with $87290$ positive coefficients in $U$,
and $87278$ positive and 12 zero coefficients in $2U-5V$.
Thus $R^{\rm b}>5/2$ for every fixed positive parameter choice with $V>0$, and its infimum is $5/2$.
Sharpness follows from
\begin{equation}
\begin{gathered}
L=1,\quad A=X_1=X_2=Y_1=Z_2=W=q^{-2},\\
Y_2=q^{-1},\quad Z_1=q.
\end{gathered}
\end{equation}
The unique leading monomial is $Y_2^7Z_1^9$, with coefficients 75 and
30 in $U,V$; all other powers along this ray are smaller by at least one.
Thus $U=75q^2+O(q)$, $V=30q^2+O(q)$.
\end{proof}

For any prescribed closeness to either floor, first choose a positive
shape along the corresponding ray, then take $\varepsilon_h$ sufficiently
small. The ordered limits thus construct positive finite weightings
arbitrarily close to the two infima.

A boundary face illustrates how earlier one-pin values fit together.
Set $W=Y_1=Y_2=Z_1=X_2=0$, retain $L=1$, first let $X_1\to\infty$,
and then put $A=\xi q$, $Z_2=q$ with $q\to\infty$ at fixed $\xi>0$.
The uniform threshold becomes
\begin{equation}
\begin{gathered}
R_{\rm face}(\xi)=
\frac{2(\xi+1)(\xi+3)(11\xi+8)(7\xi^2+6\xi+1)}{\mathscr D_{\rm face}(\xi)},\\
\mathscr D_{\rm face}(\xi)=22\xi^5+71\xi^4+264\xi^3+362\xi^2+189\xi+32.
\end{gathered}
\end{equation}
It has values $3/2$ at zero, $1064/235$ at one and limiting value seven
at infinity. Subtracting $3/2$ leaves the positive numerator
$\xi(242\xi^4+1507\xi^3+2320\xi^2+1354\xi+269)$ over twice
$\mathscr D_{\rm face}(\xi)$. Its derivative is $R_{\rm face}^\prime(\xi)=-2P_8(\xi)/\mathscr D_{\rm face}(\xi)^2$, where
\begin{equation}
\begin{aligned}
P_8={}&3993\xi^8-6424\xi^7-101644\xi^6\\
&-264520\xi^5-329569\xi^4-231956\xi^3\\
&-95312\xi^2-21664\xi-2152.
\end{aligned}
\end{equation}
Descartes' rule and opposite endpoint signs give one positive root, hence
one maximum. The value seven is approached from above as $\xi\to\infty$.
At finite $\xi$, the same face includes thresholds below seven, such as
$R_{\rm face}(1)=1064/235$, so seven is not a lower bound across all
one-pin hierarchies.

\subsection{Fixation-time distribution and control of its moments}
\label{app:path-general-times}

We resume the fixed-$N$ path family in Eq.~\eqref{app:path-general-weights}, with $q\to\infty$, domain sizes $n_L=m_L+1$ and $n_R=m_R+1$, and the domain fecundities $F_{\rm dom},d_0$ and cooperative exit probability $h_{\rm dom}$ of Eq.~\eqref{app:path-general-fecundities}. The fecundities remain in a fixed compact positive domain, and $\widetilde t=t/q$ is the slow time.

The residence means in the two mixed states on slow time are
 \begin{gather}
 d_L=\frac{n_LF_{\rm dom}+n_Rd_0}{F_{\rm dom}+d_0},\qquad
 d_R=\frac{n_RF_{\rm dom}+n_Ld_0}{F_{\rm dom}+d_0},\notag \\
 d_L+d_R=N.  \label{app:path-general-residence}
 \end{gather}
In either state the limiting exponential waiting time is independent
of the absorbing destination. With uniform cooperator-singleton
preparation and $\widetilde t>0$, 
\begin{align}
 \Pr\nolimits_U(\tau/q>\widetilde t)&\longrightarrow
 \frac{e^{-\widetilde t/d_L}+e^{-\widetilde t/d_R}}{N},\notag\\
 \Pr\nolimits_U(\tau/q>\widetilde t\mid C\text{ fixes})&\longrightarrow
 \frac12\bigl(e^{-\widetilde t/d_L}+e^{-\widetilde t/d_R}\bigr).
 \label{app:path-general-time-mixtures}
\end{align}
The unconditional limit has an atom $1-2/N$ at zero, representing the
collapsed fast extinction stage. It is absent in the conditional distribution
because fast-extinct cooperator singletons do not contribute to
cooperative fixation.

For locally Lipschitz $\phi$, the same reduction controls the
time-Poisson equations with quantitative errors at fixed $N$ on the
stated compact positive-fecundity parameter domain. After the two
consensus configurations are removed, the mixed-state Schur complement
is $q^{-1}$ times a matrix that differs from an invertible limit by
$O(q^{-1})$. The inverse Schur complement therefore has a leading term
proportional to $q$ and a bounded remainder. The leading coefficient is the harmonic extension
of the two residence means, combined with the initial fast projection.
Applying the block inverse to the unit source and to the committor-weighted
source gives, respectively,
\begin{equation}
 T_{\rm abs}=q+O(1),\qquad T_{\rm fix}=\frac N2q+O(1).
 \label{app:path-general-mean-times}
\end{equation}
This resolvent argument controls the bounded remainder. Convergence
of the rescaled distribution, together with uniform integrability,
would imply only an $o(q)$  remainder in the means.

The second-moment Poisson equations have sources twice the corresponding
first-moment functions. The same block inverse gives a leading $q^2$
term with an $O(q)$ remainder. For continuous positive fecundity,
the rescaled moments still converge, without these error rates.
The limiting second moments give
\begin{equation}
 \begin{aligned}
 \frac{\operatorname{Var}_U\tau}{q^2}&\longrightarrow
 N-1+\frac{(d_L-d_R)^2}{N},\\
 \frac{\operatorname{Var}(\tau\mid C\text{ fixes})}{q^2}
 &\longrightarrow\frac{N^2}{4}+\frac{(d_L-d_R)^2}{2},\\
 d_L-d_R&=\frac{(n_L-n_R)(F_{\rm dom}-d_0)}{F_{\rm dom}+d_0}.
 \end{aligned}
 \label{app:path-general-variances}
\end{equation}
The success-conditioned limit is a single exponential precisely when $(n_L-n_R)(F_{\rm dom}-d_0)=0$. The unconditional entrance atom remains. Otherwise the two residence rates contribute separately. For locally Lipschitz fecundity, the scaled variances in Eq.~\eqref{app:path-general-variances} have $O(q^{-1})$ errors at fixed $N$ on the stated compact parameter domains.

The same argument determines extinction-conditioned moments.
Since fixation has limiting probability $2h_{\rm dom}/N$, for every fixed
integer $k\ge1$ the three moment limits are
 \begin{align}
 q^{-k}\E_U[\tau^k]&\longrightarrow
            \frac{k!}{N}(d_L^k+d_R^k),\notag \\
 q^{-k}\E_U[\tau^k\mid C\text{ fixes}]&\longrightarrow
            \frac{k!}{2}(d_L^k+d_R^k),  \label{eq:v7-path-all-moments}\\
 q^{-k}\E_U[\tau^k\mid C\text{ extinct}]&\longrightarrow
       \frac{k!(1-h_{\rm dom})}{N-2h_{\rm dom}}(d_L^k+d_R^k). \notag
 \end{align}
Uniform exponential tail bounds follow by iterating a fixed positive
probability of absorption on successive time intervals of length
$O(q)$. They justify passage from the limiting mixtures to each
fixed moment. The third formula follows by subtracting the
successful subdistribution. For $N=6$ with equal domains,
$d_L=d_R=3$ and its first moment is
Eq.~\eqref{eq:v7-extinction-time}.

The moment and tail bounds keep $N$ fixed. Their constants depend on the domain sizes through the fast inverse and the probability of absorption during a slow interval. A joint large-$N$, large-$q$ limit needs this size dependence controlled relative to $q$, using follower-alignment and first-passage estimates as in Appendix~\ref{app:path-long}. For strictly increasing fecundity at fixed positive selection, the limiting advantage from Eq.~\eqref{app:path-general-terminal-law} has the sign of $b-c$. For locally Lipschitz $\phi$, its finite-$q$ error is bounded by a constant times $q^{-1}$ on a fixed compact positive-fecundity parameter domain, with $N$ fixed. The finite-weight sign follows when the limiting advantage exceeds this bound in modulus. Inside that margin, a sharper response estimate or the exact finite-chain committor is needed.
Setting a weak edge to zero before the long-time limit creates additional
closed classes and changes the fixation problem.

\section{Support inequalities and the five-site algebra}\label{app:path-classification}
This appendix turns the response into global support bounds. Most bounds are sign statements about the coefficients of a specified numerator. For these, the polynomial normalization is fixed first; cancelling a common factor afterwards can change coefficient signs (Appendix~\ref{app:v10-support-identities}).
\subsection{Absence of promoters on paths with at most four sites}
\label{app:path-small-no-go}\label{app:path-small}

The singleton/complement argument in Sec.~\ref{sec:small-tree-obstructions}
excludes every support with at most three sites, for every introduction
distribution. On the four-site path, $J^\mu<0$ for every positive weighting
and every common selection-independent singleton distribution $\mu$.
Thus no path with at most four sites admits a weak promoter.
\labelalias{app:path-singleton-no-go}{eq:small-benefit}

\begin{proof}
For the four-site path, normalize the consecutive weights to $(a_{\rm w},b_{\rm w},1)$,
where $a_{\rm w},b_{\rm w}>0$. Let $J^{(\ell)}$ denote the benefit coefficient when the
initial singleton is at site $\ell$, numbering the sites $0,1,2,3$ along
the path. Its pair source is
$4(\mathbf1_{\{i=\ell\}}+\mathbf1_{\{j=\ell\}})$ in place of the
constant two in~Eq.~\eqref{app:path-pair-equations}. Define the positive common denominator
\begin{equation}
\begin{aligned}
 \mathscr D_4={}&(a_{\rm w}+b_{\rm w})(b_{\rm w}+1)\\
 &\times(a_{\rm w}+b_{\rm w}+1)(a_{\rm w}b_{\rm w}+2a_{\rm w}+b_{\rm w}).
\end{aligned}
\end{equation}
Substitution of the path
kernel into that system and contraction with $(-\mathsf B_{ij})$ gives
\begin{equation}
 J^{(0)}=-\frac{4p_0(a_{\rm w},b_{\rm w})}
 {\mathscr D_4},
 \label{app:path-p4-endpoint}
\end{equation}
where
\begin{equation}
 \begin{split}
 p_0={}&a_{\rm w}^2b_{\rm w}^2+3a_{\rm w}^2b_{\rm w}+a_{\rm w}^2+2a_{\rm w}b_{\rm w}^3\\
      &+7a_{\rm w}b_{\rm w}^2+4a_{\rm w}b_{\rm w}+a_{\rm w}+b_{\rm w}^4+3b_{\rm w}^3+b_{\rm w}^2.
 \end{split}
 \label{app:path-p4-endpoint-polynomial}
\end{equation}
For the adjacent internal site,
\begin{equation}
 J^{(1)}=-\frac{4p_1(a_{\rm w},b_{\rm w})}
 {(a_{\rm w}+b_{\rm w})\mathscr D_4c_1(a_{\rm w},b_{\rm w})},
 \label{app:path-p4-internal}
\end{equation}
with
\begin{equation}
 \begin{split}
 c_1={}&15a_{\rm w}^2b_{\rm w}^2+34a_{\rm w}^2b_{\rm w}+16a_{\rm w}^2+31a_{\rm w}b_{\rm w}^3\\
      &+72a_{\rm w}b_{\rm w}^2+34a_{\rm w}b_{\rm w}+12b_{\rm w}^4+31b_{\rm w}^3+15b_{\rm w}^2,
 \end{split}
 \label{app:path-p4-denominator-polynomial}
\end{equation}
and
\begin{align}
 p_1={}&a_{\rm w}^5(16b_{\rm w}^4+82b_{\rm w}^3+135b_{\rm w}^2+82b_{\rm w}+16)\notag\\
 &+a_{\rm w}^4\bigl(76b_{\rm w}^5+415b_{\rm w}^4+755b_{\rm w}^3\notag\\
 &\qquad+551b_{\rm w}^2+164b_{\rm w}+16\bigr)\notag\\
 &+a_{\rm w}^3\bigl(128b_{\rm w}^6+752b_{\rm w}^5+1503b_{\rm w}^4\notag\\
 &\qquad+1253b_{\rm w}^3+461b_{\rm w}^2+66b_{\rm w}\bigr)\notag\\
 &+a_{\rm w}^2\bigl(88b_{\rm w}^7+582b_{\rm w}^6+1304b_{\rm w}^5\notag\\
 &\qquad+1227b_{\rm w}^4+525b_{\rm w}^3+93b_{\rm w}^2\bigr)\notag\\
 &+a_{\rm w}\bigl(20b_{\rm w}^8+179b_{\rm w}^7+483b_{\rm w}^6\notag\\
 &\qquad+519b_{\rm w}^5+250b_{\rm w}^4+51b_{\rm w}^3\bigr)\notag\\
 &+16b_{\rm w}^8+63b_{\rm w}^7+79b_{\rm w}^6+41b_{\rm w}^5+9b_{\rm w}^4.
 \label{app:path-p4-internal-polynomial}
\end{align}
Every coefficient is positive. Reversing the path and rescaling all
weights gives
\begin{equation}
 J^{(3-\ell)}(a_{\rm w},b_{\rm w},1)=a_{\rm w}^{-1}J^{(\ell)}(a_{\rm w}^{-1},b_{\rm w}/a_{\rm w},1),
 \quad \ell=0,1.
 \label{app:path-p4-reversal}
\end{equation}
The factor $a_{\rm w}^{-1}$ occurs because $\gamma$ and $\mathsf B$ scale inversely with a
common weight multiplication, whereas the pair solution depends only
on $P$. Hence all four $J^{(\ell)}$ are negative. Linearity of the pair
equations in the introduction distribution proves $J^\mu<0$ for every introduction distribution,
including uniform introduction. This completes the small-path no-go.
\end{proof}

\subsection{Five-site response and its positive forest expansion}
\label{app:path-p5-global}\label{v4:alg:p5-forest}\label{app:path-five}

The proof separates a positive cost contraction from a signed benefit
polynomial. Rooted forests control the former; outer-degree decomposition
controls the latter and determines an exact envelope.

Every positively weighted five-site path, with vertices numbered along it, can be scaled to consecutive
weights $(A,1,q,Q)$ with $A,q,Q>0$. No ordering of the weight magnitudes is assumed. Its kernel has nonzero entries
\begin{equation}
 \begin{gathered}
 P_{12}=P_{54}=1,\\
 P_{21}=\frac{A}{A+1},\qquad P_{23}=\frac{1}{A+1},\\
 P_{32}=\frac{1}{q+1},\qquad P_{34}=\frac{q}{q+1},\\
 P_{43}=\frac{q}{Q+q},\qquad P_{45}=\frac{Q}{Q+q}.
 \end{gathered}
 \label{app:path-p5-kernel}
\end{equation}
\paragraph*{The pair determinant.}

The forest expansion proves coefficient positivity by counting outgoing
rows. Let $\mathsf M$ be
the ten-state killed ancestral-pair matrix in the clock with rates $P$.
The pair source is $\mathsf M\mathbf D=2\,\mathbf1$. Put
\begin{equation}
 \begin{gathered}
 T=(A+1)(q+1)(Q+q),\quad
 \mathscr D=\frac{T^8\det \mathsf M}{q},\\
 U=\mathscr D K,\quad V=\mathscr D J.
 \end{gathered}
 \label{v4:alg:p5-normalization}\labelalias{v4:eq:p5-normalization}{v4:alg:p5-normalization}
\end{equation}
This normalization agrees with the reduced five-site fraction. The neutral spectral reduction in Eq.~\eqref{eq:v10-p5-spectrum} evaluates the determinant from the one-lineage characteristic polynomial. The following identity extends that calculation to longer paths.

The product of pair sums has a useful symmetric-function form.
For $n\ge2$ positive rates and the staircase partition
$\varrho^{(n)}=(n-1,n-2,\ldots,1,0)$, the Schur polynomial is
\begin{equation}
 \begin{aligned}
 s_{\varrho^{(n)}}(\lambda_1^{\rm anc},\ldots,\lambda_n^{\rm anc})
 &=\prod_{a<b}(\lambda_a^{\rm anc}+\lambda_b^{\rm anc})\\
 &=\det[\mathsf e_{n-2i+j}]_{i,j=1}^{n-1},
 \end{aligned}
 \label{app:path-staircase-schur}
\end{equation}
where $\mathsf e_0=1$ and $\mathsf e_k=0$ outside $0\le k\le n$.
The first equality follows by dividing the Vandermonde determinant
in the squared rates by that in the rates; the second is the dual
Jacobi--Trudi identity~\cite[Chap.~I, Sec.~3]{Macdonald1995}.
Polynomial continuity includes repeated rates. Since the pair spectrum on a path consists of the sums of two distinct one-lineage rates, including the zero rate (Appendix~\ref{app:v3-sylvester-time}), the pair determinant
on a path is $\mathsf e_n s_{\varrho^{(n)}}$ with $n=N-1$,
and Eq.~\eqref{eq:v10-p5-spectrum} is its $n=4$ case.
Its positivity follows from the product, while the determinant
computes its value directly from the one-lineage characteristic
coefficients without solving for individual rates.

\paragraph*{Cost grades from forests.}
\labelalias{v4:alg:p5-cost-positive}{v4:alg:p5-forest}
With the normalization in Eq.~\eqref{v4:alg:p5-normalization}, $U$ is a polynomial
with nonnegative coefficients in $(A,q,Q)$. Its total degree in the two
outer variables $(A,Q)$ is exactly fifteen, and its separate degrees in
$A$ and $Q$ are at most eight.

\begin{proof}
The cost contraction and the pair source give
\begin{equation}
 qU=2T^8\sum_{i=1}^4\gamma_{i,i+1}
          [\operatorname{adj}(\mathsf M)\mathbf1]_{i,i+1}.
 \label{v4:alg:p5-cost-forest}
\end{equation}
For a killed row generator with cemetery $0$, the cofactor
$\operatorname{adj}(\mathsf M)_{ij}$ sums directed forests rooted at $0$ and
$j$ in which $i$ reaches $j$. Each edge contributes its transition
rate~\cite{Chaiken1982,AGNamGunawardena2025}.
A transition with denominator $s_k$, $k=2,3,4$, can leave only a pair
containing a neighbor of source vertex $k$. Each source has two neighbors,
and exactly $\binom52-\binom32=7$ pair rows contain at least one of them.
A forest selects at most one outgoing edge from each row. Its adjugate
monomial therefore has denominator exponent at most seven in each of
$A+1,q+1,Q+q$. Coalescence edges ending at the cemetery obey the same
count, and multiple channels to the cemetery are expanded before counting.

The four simplified conductances are
\begin{equation}
 \begin{aligned}
 \gamma_{12}&=\frac1{A+1},&
 \gamma_{23}&=\frac1{(A+1)(q+1)},\\
 \gamma_{34}&=\frac q{(q+1)(Q+q)},&
 \gamma_{45}&=\frac1{Q+q}.
 \end{aligned}
 \label{v4:alg:p5-conductances}
\end{equation}
They add at most one denominator power. Consequently $T^8$ clears every
forest term individually, leaving products of positive monomials and
nonnegative powers of the three positive linear strength polynomials.

The division by $q$ is also termwise. At $q=0$, the physical graph splits
into components of sizes three and two. The six cross-component pair
states form a closed class with no route to the cemetery. For a cost row
$12,23,$ or $45$, a contributing two-root forest requires at least one
transition across the deleted edge. If its transient root were in the
closed cross-component class, the cost row could not reach it; if the
root were outside that class, the class would contain no root and would
force a directed cycle. Thus every such forest contains a numerator
factor $q$. The remaining cost row $34$ already contains that factor in
$\gamma_{34}$. Every cleared term is divisible by $q$, proving polynomiality
and nonnegative coefficients without cancellation.

Under $(A,Q)\mapsto t(A,Q)$, $t\to\infty$, the limiting ancestral chain
has unique absorbing site three, with arrows
$1\leftrightarrow2\to3\leftarrow4\leftrightarrow5$. Every lineage reaches
three in finite mean time, so the limiting killed pair matrix is invertible
and $\mathbf D$ has a finite strictly positive limit. Each $\gamma_{i,i+1}$ is
$t^{-1}g_i+O(t^{-2})$, $g_i>0$, giving
$K=t^{-1}k_\infty+O(t^{-2})$, $k_\infty>0$. Since
$T^8\sim t^{16}[AQ(q+1)]^8$ and $\det \mathsf M$ tends to a positive constant,
$U\sim C t^{15}$ with $C>0$. Nonnegative polynomial coefficients then
fix the total outer degree to fifteen. At $A\to\infty$ with $q,Q$ fixed,
the ancestral chain has the unique closed class $\{3,4,5\}$, so
$K$ and the pair inverse remain bounded while $\mathscr D=O(A^8)$.
This proves $\deg_AU\le8$; the $Q$ argument is its reversal.
\end{proof}

\subsection{Outer-degree monotonicity and the exact inner-ratio envelope}
\label{v4:p5}
The cost-positivity argument leaves the benefit grades to be determined.
Their signs will identify losses that disappear when the outer contacts
are strengthened together.

Retain $T,\mathscr D,U,V$ and the ten-pair matrix $\mathsf M$ in the
normalization of Eq.~\eqref{v4:alg:p5-normalization}.
$U,V$ are integer polynomials, each with 736 nonzero monomials, individual
degrees at most eight in $A,Q$, and total outer degree fifteen.
We now determine the signed outer-degree structure of $V$.

Grade by outer contact degree, $U=\sum_{m=0}^{15}U_m$ and
$V=\sum_{m=0}^{15}V_m$. The lower grades obey Eq.~\eqref{eq:v4-p5-grades}.
\labelalias{v4:eq:p5-sign-lemma}{eq:v4-p5-grades}
The top grades satisfy $U_{15}-7V_{15}>0$ for $A,q,Q>0$.
Together these inequalities imply $U-7V>0$, so every finite promoting
weighting has threshold strictly greater than seven.
\begin{proof}[Proof of Proposition~\ref{prop:I-p5-envelope}, signed grades]
For $m\le13$ the benefit assertion is an exact finite coefficient-sign
lemma. It is verified by constructing the kernel in
Eq.~\eqref{app:path-p5-kernel}, solving $\mathsf M\mathbf D=2\,\mathbf1$ over
$\mathbb Q(A,q,Q)$, forming $V=(T^8\det \mathsf M/q)J$, and extracting
$[\ell^m]V(\ell A,q,\ell Q)$. Each resulting coefficient polynomial
has strictly negative nonzero coefficients. The grades $V_0,\ldots,V_{13}$ contain 698 of the 736 coefficients of $V$, all negative; the forms of $V_{14}$, $U_{15}$ and $V_{15}$ given below come from the same exact polynomials. This finite calculation is the computer-assisted step of the proof; the nonnegativity of $U$ follows from the forest argument without it.
The exceptional next grade has the explicit form
\begin{align}
 -V_{14}&=A^6Q^6\bigl\{3(1+q^7)(A-Q)^2 \\
 &\qquad+A^2P_7(q)+Q^2q^7P_7(q^{-1})+AQ C_7(q)\bigr\}, \notag \\
 P_7(q)={}&129+829q+4967q^2+17415q^3 \notag \\
 &+30466q^4+26838q^5+11328q^6+1797q^7, \notag \\
 C_7(q)={}&4671(q+q^6)+26351(q^2+q^5)+56009(q^3+q^4). \notag
\end{align}
Hence two adverse expanded coefficients are absorbed by a square. The
top terms are
\begin{align}
 U_{15}&=2A^7Q^7(q+1)^2(11q^2+28q+11)\notag\\
 &\qquad\times[A A_2(q)+Q B_2(q)], \notag \\
 V_{15}&=A^7Q^7(2q+3)(3q+2)[Q\Pi_5(q)-A F_4(q)],
\end{align}
where
\begin{align}
 \Pi_5(q)&=11q^4-16q^3-124q^2-128q-33, \notag \\
 F_4(q)&=33q^4+128q^3+124q^2+16q-11,\\
 A_2(q)&=20q^2+44q+21,\quad B_2(q)=21q^2+44q+20. \notag
\end{align}

These identities are the sign input for the outer-ray part of Proposition~\ref{prop:I-p5-envelope}, completed at the end of this proof, where the top-grade comparison also gives the strict bound $R>7$.
\labelalias{v4:prop:p5-ray}{prop:I-p5-envelope}

Define $H_4(q)=913q^4+4164q^3+6656q^2+4304q+913$ and
$H_4^\vee=q^4H_4(q^{-1})$. Exact multiplication gives
\begin{equation}
 U_{15}-7V_{15}=A^7Q^7[Aq(2q+3)H_4^\vee+Q(3q+2)H_4]>0.
\label{eq:v10-p5-strict-grade-bound}
\end{equation}
Together with the lower-grade signs this proves $U-7V>0$ throughout
the positive orthant. For a fixed outer ray, $\ell^{-15}U(\ell A,q,\ell Q)=\sum_m\ell^{m-15}U_m(A,q,Q)$, and similarly for $V$. The grade signs give $(\ell^{-15}U)'\le0$ and, since $V_0<0$, $(\ell^{-15}V)'>0$. Where $V>0$, the quotient rule and $U>0$ therefore give Eq.~\eqref{eq:v4-p5-ray}. The scaled benefit increases from $-\infty$ to $V_{15}$, proving the unique-onset criterion.
\end{proof}

\paragraph*{The fixed-inner-ratio envelope.}
\phantomsection
\label{app:v4-p5-envelope}
We complete Proposition~\ref{prop:I-p5-envelope} by minimizing the top-degree quotient at fixed $q\ge1$ and controlling the hierarchy that approaches its boundary value.
\labelalias{v4:prop:p5-envelope}{prop:I-p5-envelope}
\labelalias{v4:eq:p5-envelope}{eq:v4-p5-envelope}
\begin{proof}[Proof of Proposition~\ref{prop:I-p5-envelope}, envelope]
$F_4(q)>0$ for $q\ge1$. The top benefit is positive precisely if
$\Pi_5>0$ and $Q/A>F_4/\Pi_5$. Descartes' rule gives one positive zero of $\Pi_5$,
and $\Pi_5(1)<0$. 
$F_4-3\Pi_5=176q^3+496q^2+400q+88>0$, so every finite promoter has $Q>3A$.
Writing $\theta=A/Q$, the outer-ray limit is
\begin{equation}
 R_\infty(\theta,q)=
 \frac{2(q+1)^2(11q^2+28q+11)[\theta A_2(q)+B_2(q)]}
 {(2q+3)(3q+2)[\Pi_5(q)-\theta F_4(q)]}.
\end{equation}
Its derivative in $\theta$ is positive on the promoter domain, since
$A_2\Pi_5+B_2F_4=(q^2-1)(913q^4+4304q^3+6856q^2+4304q+913)>0$.
By Eq.~\eqref{eq:v4-p5-ray}, every promoting weighting obeys $R(A,1,q,Q)>R_\infty(A/Q,q)>R_\infty(0,q)$. The $\theta\downarrow0$ limit therefore gives the lower bound in Eq.~\eqref{eq:v4-p5-envelope}.
The degree bounds ($\deg_A,\deg_Q\le8$, total outer degree fifteen) make $A^7Q^8=L^{23}$ the unique leading outer monomial
under $(A,Q)=(L,L^2)$; its benefit coefficient is positive for $q>q_c$.
Thus $R\to R_5^{\rm env}(q)$ along the positive family $(A,Q)=(L,L^2)$, and this single limit proves sharpness. Finally,
\begin{equation}
 \begin{aligned}
 R_5^{\rm env}(q)-7&=\frac{H_4(q)}{(2q+3)\Pi_5(q)}>0,
 \\
 R_5^{\rm env}(q)&=7+\frac{83}{2q}+O(q^{-2}),
 \end{aligned}
\end{equation}
and direct differentiation gives
\begin{equation}
 {R_5^{\rm env}}'(q)=\frac{-2(q+1)S_7(q)}{(2q+3)^2\Pi_5(q)^2}<0,
\end{equation}
where $S_7$ has coefficients $10043$, $81565$, $275289$, $494107$,
$501247$, $281169$, $77957$, $7623$ in descending powers.
All are positive.
\end{proof}

\subsection{Why the four-cycle is minimal in sites and edges}
\label{app:v10-minimal-cycle}

The main-text singleton argument excludes all graphs on at most three
sites, while Appendix~\ref{app:path-small} excludes $P_4$.
Under either named introduction rule, the four-cycle attains the minimum
numbers of sites and edges among connected supports admitting a weak
promoter, namely four in each case.
\begin{proof}
The only other connected simple graph with four sites and three edges
is the star. Its obstruction is pointwise. If its central
state is $n_0$ and $m=(Pn)_0$ is the weighted leaf average,
every leaf $i$ has $(Pn)_i=n_0$. Thus every edge contribution
to the benefit Dirichlet form is
\begin{equation}
 \gamma_{0i}(m-n_0)(n_0-n_i)\le0.
 \label{eq:app-star-obstruction}
\end{equation}
In a mixed configuration at least one term is strict. Together with
the negative cost drift, this excludes all positive star weights.
A connected graph with fewer than four edges has at most four sites,
so the promoting cycle is minimal in both counts.
\end{proof}
The response inequalities below complete the four-site classification,
excluding the paw and giving threshold infimum three on $C_4$, the
diamond and $K_4$.

\labelalias{app:path-p5-certificate}{eq:v10-p5-strict-grade-bound}

\begin{table*}[!tbp]
\caption{Retained edges and their variable order for the support calculations. The local vertex labels are $0,\ldots,N-1$, and $ij$ denotes the unordered edge $\{i,j\}$.}
\label{tab:I-support-edges}
\begin{ruledtabular}
\begin{tabular}{llll}
Support & Ordered edges & Support & Ordered edges\\
\hline
$P_4$ & $01,12,23$ & fork & $04,13,23,34$\\
$K_{1,3}$ & $01,02,03$ & cricket & $04,14,23,24,34$\\
paw & $01,12,02,03$ & butterfly & $01,04,14,23,24,34$\\
$C_4$ & $01,12,23,03$ & house & $01,12,23,34,04,02$\\
diamond & $01,12,23,03,02$ & $K_{2,3}$ & $02,03,04,12,13,14$\\
$K_4$ & $01,12,23,03,02,13$ & dart & $01,12,13,14,23,24$\\
$P_5$ & $01,12,23,34$ & kite & $01,13,14,23,24,34$\\
bull & $01,12,02,13,24$ & $K_4$ with a pendant & $04,12,13,14,23,24,34$\\
triangle with a two-edge tail & $01,12,02,23,34$ & $K_{1,1,3}$ & $03,04,13,14,23,24,34$\\
$C_4$ with a pendant & $01,12,23,03,04$ & gem & $01,12,23,34,04,02,03$\\
$C_5$ & $01,12,23,34,04$ & $C_5$ with crossing chords & $01,12,23,34,04,02,13$\\
$P_6$ & $01,12,23,34,45$ &  & 
\end{tabular}
\end{ruledtabular}
\end{table*}

\subsection{Support bounds and exact response identities}
\label{app:c5-supports}\label{app:v10-support-identities}
Proposition~\ref{prop:small-supports} rests on the sign structure of
Sec.~\ref{sec:sign-structure}. For each support with a proved infimum and each introduction rule, one
integer numerator decides the question, namely the numerator of $J$
when there is no weak promoter, and that of $K-r_{\rm lev}J$, with $r_{\rm lev}$ the
claimed infimum, otherwise. Here $K_{n_1,\ldots,n_k}$ denotes the complete
multipartite graph, with groups of the listed sizes and all edges between
different groups. The edge lists in Table~\ref{tab:I-support-edges} fix the less familiar support names.
The eight supports containing $C_5$ are the cycle, the house, the gem,
the wheel $W_4$, the cycle with two crossing chords, $K_5$ less a path of
two edges, $K_5$ less an edge, and $K_5$. The six with longest cycle $C_4$
are $C_4$ with a pendant, the dart, the kite, $K_{2,3}$, $K_4$ with a
pendant and $K_{1,1,3}$. The path, the bull, the triangle with a two-edge
tail and the butterfly have a spanning path but no four- or five-cycle;
the star, the fork and the cricket have no spanning path.

Under both rules the following hold for every positive weighting.
\begin{itemize}
\item $J<0$ on the paw, the star $K_{1,4}$, the fork and the cricket, in
addition to $P_4$ and $K_{1,3}$ (Appendices~\ref{app:path-small} and \ref{app:v10-minimal-cycle}, respectively).
\item $K-3J>0$ on $C_4$, the diamond and $K_4$.
\item $K-7J>0$ and $J^{\mathrm b}<0$ on the path $P_5$, the bull (a triangle with pendant vertices at two different triangle sites), the
triangle with a two-edge tail and the butterfly.
\item $K-2J>0$ and $K^{\mathrm b}-2J^{\mathrm b}>0$ on $C_4$ with a
pendant, the dart, the kite, $K_{2,3}$, $K_4$ with a pendant and
$K_{1,1,3}$.
\item $2K-3J>0$ and $2K^{\mathrm b}-3J^{\mathrm b}>0$ on $C_5$, the
house, the gem and the five-cycle with two crossing chords.
\end{itemize}

\begin{proof}
We prove the listed sign inequalities by expressing each physical
response as an integer polynomial divided by a positive denominator.
Matching upper limits then follow from connected spanning supports.
Write $\mathsf M D=g$ and
$K=k^\T D$, $J=j^\T D$, with
$k_{ij}=\gamma_{ij}$ and $j_{ij}=-\mathsf B_{ij}$. For a proposed floor
$r_{\rm lev}=p_*/q_*$, with positive integers $p_*,q_*$, put
$\ell=q_*k-p_*j$. Clear the denominators of row $\rho$ of
$[\mathsf M\mid g]$ by a positive multiplier $d_\rho$ built from the strengths,
and those of $\ell$ by $d_\ell>0$. The resulting polynomial matrix
$\mathsf M'$ and vectors $g',\ell'$ obey
\begin{equation}
 \begin{gathered}
 \mathcal P=-\det\begin{pmatrix}\mathsf M'&g'\\\ell'^{\T}&0\end{pmatrix},\\
 \mathcal P=d_\ell\det\mathsf M'\,(q_*K-p_*J),\\
 \det\mathsf M'=\left(\prod_\rho d_\rho\right)\det\mathsf M>0.
 \end{gathered}
 \label{eq:v10-support-bordered-sign}
\end{equation}
The last inequality follows because $\mathsf M$ is the nonsingular
killed-pair matrix on a connected retained support. Nonnegative
coefficients of a nonzero $\mathcal P$ prove a strict inequality for
all positive contacts, including tied weights. For a no-promotion
statement choose $\ell=-j$, so $\mathcal P$ is a positive factor
times $-J$.

Table~\ref{tab:I-support-edges} fixes the retained contacts and their variable order in the polynomial calculations. We order ancestral pairs lexicographically.

The source is $g_{ij}=2$ for uniform introduction and $g_{ij}=\vartheta_i+\vartheta_j$ for birth-associated introduction. All rational entries are reduced before their denominator least common multiples are taken. Over $\mathbb Q[w]$, denominator least common multiples are monic in degree reverse lexicographic order, with contacts ordered as in Table~\ref{tab:I-support-edges}. When integer coefficients are required, each cleared row is then multiplied by the least positive integer clearing its coefficient denominators, included in its row or readout multiplier. No row content is divided out. The two interpolation constructions below specify their distinct normalization choices.

\begin{table*}[!tbp]
\caption{Sharp threshold infima and exact coefficient signs for the $32$ completed reduced-response cases. The coprime homogeneous integer pair $U,V$ is normalized jointly, with $U$ positive coefficientwise. At a finite floor $r_{\rm lev}=p_*/q_*$, the tested polynomial is $q_*U-p_*V$; for $R_*=+\infty$ it is $-V$. The column $n_+$ counts its nonzero coefficients, all strictly positive. The homogeneous degree is $d$ and the final column gives its exact individual degrees in the contact order of Table~\ref{tab:I-support-edges}. The last contact can be set to one during elimination; both polynomials are then homogenized to $d=\max(\deg U,\deg V)$ before counting. Joint primitive normalization of $U,V$ does not require the tested sign polynomial to have integer coefficient content one.}
\label{tab:I-reduced-response-cases}
\begin{ruledtabular}
\begin{tabular}{llrrrr}
Support & Introduction & $R_*^\mu$ & $n_+$ & $d$ & Individual degrees\\
\hline
$P_4$ & $U$ & $+\infty$ & 46 & 10 & $(6,10,6)$\\
 & $\mathrm b$ & $+\infty$ & 33 & 8 & $(5,8,5)$\\
$K_{1,3}$ & $U$ & $+\infty$ & 28 & 6 & $(6,6,6)$\\
 & $\mathrm b$ & $+\infty$ & 25 & 6 & $(5,5,5)$\\
paw & $U$ & $+\infty$ & 649 & 18 & $(12,12,12,7)$\\
 & $\mathrm b$ & $+\infty$ & 772 & 19 & $(12,12,12,8)$\\
$C_4$ & $U$ & $3$ & 1000 & 23 & $(12,12,12,12)$\\
 & $\mathrm b$ & $3$ & 988 & 23 & $(12,12,12,12)$\\
diamond & $U$ & $3$ & 9932 & 25 & $(13,13,13,13,13)$\\
 & $\mathrm b$ & $3$ & 13839 & 27 & $(14,14,14,14,14)$\\
$K_4$ & $U$ & $3$ & 108958 & 27 & $(14,14,14,14,14,14)$\\
 & $\mathrm b$ & $3$ & 205020 & 31 & $(16,16,16,16,16,16)$\\
$P_5$ & $U$ & $7$ & 734 & 21 & $(8,14,14,8)$\\
 & $\mathrm b$ & $+\infty$ & 726 & 21 & $(8,14,14,8)$\\
bull & $U$ & $7$ & 13708 & 27 & $(17,17,17,10,10)$\\
 & $\mathrm b$ & $+\infty$ & 17892 & 29 & $(18,18,18,11,11)$\\
triangle with a two-edge tail & $U$ & $7$ & 16004 & 32 & $(16,17,17,16,8)$\\
 & $\mathrm b$ & $+\infty$ & 18122 & 33 & $(16,18,18,17,8)$\\
$C_4$ with a pendant & $U$ & $2$ & 18742 & 33 & $(17,16,16,17,10)$\\
 & $\mathrm b$ & $2$ & 21458 & 34 & $(18,16,16,18,11)$\\
$C_5$ & $U$ & $3/2$ & 25720 & 39 & $(16,16,16,16,16)$\\
 & $\mathrm b$ & $3/2$ & 25590 & 39 & $(16,16,16,16,16)$\\
house & $U$ & $3/2$ & 427193 & 43 & $(18,18,18,16,18,19)$\\
 & $\mathrm b$ & $3/2$ & 548437 & 45 & $(19,19,19,16,19,21)$\\
$K_{2,3}$ & $U$ & $2$ & 440316 & 43 & $(18,18,18,18,18,18)$\\
 & $\mathrm b$ & $2$ & 570986 & 45 & $(19,19,19,19,19,19)$\\
dart & $U$ & $2$ & 238894 & 36 & $(11,18,18,18,18,18)$\\
 & $\mathrm b$ & $2$ & 317333 & 38 & $(12,19,19,19,19,19)$\\
kite & $U$ & $2$ & 312275 & 37 & $(10,19,19,18,18,19)$\\
 & $\mathrm b$ & $2$ & 455733 & 40 & $(11,21,21,19,19,21)$\\
$P_6$ & $U$ & $1$ & 15466 & 36 & $(10,18,18,18,10)$\\
 & $\mathrm b$ & $1$ & 15344 & 36 & $(10,18,18,18,10)$\\
\end{tabular}
\end{ruledtabular}
\end{table*}

A second representation computes the rational response in lowest
terms, $R=U/V$, with $U$ chosen positive on the positive orthant and
$V$ retaining the sign of $J$. Table~\ref{tab:I-reduced-response-cases}
organizes the $32$ reduced fractions on $16$ supports, of which $23$
cases are promoting and nine have $J<0$ everywhere. In every case, $U$ has strictly positive nonzero coefficients. Since $K>0$, the exact identity $UJ=VK$ gives
\begin{equation}
 q_*K-p_*J=\frac KU(q_*U-p_*V).
 \label{eq:I-reduced-sign-transfer}
\end{equation}
This identity also holds at $J=0$. The sign polynomial is therefore $q_*U-p_*V$ for a finite floor and $-V$ for a no-promotion case. Table~\ref{tab:I-reduced-response-cases} identifies its degree and every nonzero coefficient sign. To fix its normalization, form the separately cleared $K$ and $J$ borders, restore their readout multipliers as in Eq.~\eqref{eq:v10-support-reduced-identity}, and cancel the common polynomial factor from numerator and denominator. Restore the last contact by homogenizing both polynomials to the common degree $d=\max(\deg U,\deg V)$ of the dehomogenized pair. Normalize the combined integer coefficient content to one and choose the common sign by $U>0$ at unit contacts. This determines the coprime reduced pair independently of the positive row factors used during elimination.

The count comprises six four-site supports, nine five-site supports and
$P_6$. The bordered-determinant construction also treats supports outside
this reduced-fraction collection, as detailed below and in
Table~\ref{tab:I-determinant-bounds}.

The star $K_{1,4}$ follows directly from the pointwise argument in
Eq.~\eqref{eq:app-star-obstruction}, which applies to any number of
leaves and either introduction rule. The fork, cricket and butterfly
use division-free bordered calculations. Expanding each determinant
row by row and retaining the minors indexed by column subsets
evaluates the polynomial in Eq.~\eqref{eq:v10-support-bordered-sign}
using exact polynomial additions and multiplications. All contacts remain symbolic for these three supports. The common pair-row multiplier is the denominator least common multiple of that row of $\mathsf M$ together with $\vartheta_i+\vartheta_j$, used under both introduction rules. The readout is cleared separately, and the sign counts refer to unreduced bordered polynomials. The $J$ numerators for
the fork and cricket have strictly negative coefficients under both
sources. For the fork the term counts are $725$ and $696$; for the
cricket they are $13993$ and $13228$, in uniform and birth-associated
order. On the butterfly, the uniform numerator of $K-7J$ has $498206$
positive terms, and the birth-associated $J$ numerator has $490880$
negative terms. The cleared determinants are positive on the contact
orthant. Together with the star sign, these comparisons settle four of the eight proved supports outside Table~\ref{tab:I-reduced-response-cases}; the independent bordered construction below settles $K_4$ with a pendant, $K_{1,1,3}$, the gem and the five-cycle with two crossing chords. 

The largest exact calculations use interpolation in finite fields to
avoid expanding determinants with rational-function entries. The
logical steps can be stated independently of their implementation.
For an integer polynomial matrix $\mathsf C_{\rm bord}=(C_{ij})$, let
$L_{ij}=\|C_{ij}\|_1$, where $\|\cdot\|_1$ is the sum of the absolute
values of its coefficients. The determinant expansion gives
\begin{equation}
 \|\det\mathsf C_{\rm bord}\|_1
 \le\operatorname{perm}L
 \le\prod_i\sum_jL_{ij}.
 \label{eq:v10-determinant-coefficient-bound}
\end{equation}
The permanent is the unsigned permutation sum. Both expressions are
integer upper bounds, calculable before the polynomial determinant.
For each contact $w_e$, a valid individual-degree bound is
\begin{equation}
 \deg_{w_e}\det\mathsf C_{\rm bord}\le d_e,\quad
 d_e=\max_{\sigma}\sum_i\deg_{w_e}C_{i,\sigma(i)};
 \label{eq:v10-determinant-degree-bound}
\end{equation}
here zero entries exclude a permutation product. The simpler sums of
row maxima or column maxima are also valid bounds. Evaluating on
$d_e+1$ distinct field values for every interpolated variable, on the
complete Cartesian product, and inverting the Vandermonde matrices
therefore recovers every coefficient modulo the field prime. Variables
left symbolic need no interpolation. Homogeneity permits one contact
to be set to one without merging coefficients, since its exponent is
fixed by the other exponents and the total degree.

Let $B$ bound every integer coefficient and let $Q_{\rm CRT}$ be the
product of the interpolation primes. Once $Q_{\rm CRT}>2B$, the signed
Chinese-remainder representative is the unique possible integer
coefficient. This reconstruction is an exact polynomial identity argument over finite fields and integers. Any distinct nodes on the complete tensor grids recover the same modular coefficients, and the coefficient bound fixes their integer representatives. Off-grid and floating-point evaluations test the implementation separately; the inequality follows from the reconstructed coefficient signs.

For the house, $K_{2,3}$, dart and kite, under both introduction rules,
the reduced fractions came from two $11\times11$ bordered
matrices with common first ten rows. Their last rows contain the
separately cleared $K$ and $J$ readouts. If their denominator-clearing
factors are $d_K,d_J$, then
\begin{equation}
 R=\frac{d_J\det\mathsf C_K}{d_K\det\mathsf C_J}.
 \label{eq:v10-support-reduced-identity}
\end{equation}
Here $\mathsf C_K,\mathsf C_J$ denote the two bordered matrices. The last listed contact is set to one before construction. Each augmented pair row and each readout is cleared separately over $\mathbb Q[w]$ and then integerized as above. For $X=K,J$ define
\begin{equation}
 \begin{aligned}
 h_e^X&=\min\left\{\sum_i\max_j\deg_{w_e}(\mathsf C_X)_{ij},\right.\\
 &\hspace{4em}\left.\sum_j\max_i\deg_{w_e}(\mathsf C_X)_{ij}\right\},
 \end{aligned}
 \label{eq:v10-support-reduced-degrees}
\end{equation}
omitting zero entries from each maximum. Set $h_e=\max(h_e^K,h_e^J)$. The two variables with smallest $h_e$ are interpolated, with ties broken by listed order; the other three remain symbolic. A deterministic choice is the complete tensor grid with nodes $1,\ldots,h_e+1$ on each interpolated axis, evaluated on the cleared polynomial matrices. The reported calculation used distinct pseudorandom nodes on complete grids of the same sizes, which recover the same polynomials. Multiply the recovered determinants by $d_J,d_K$ as in Eq.~\eqref{eq:v10-support-reduced-identity} before exact greatest-common-divisor division. Homogenize both reduced polynomials to their larger total degree in the deleted contact, clear any remaining rational coefficient denominators jointly, and divide their common integer content. Choose the common sign so that $U$ has positive leading coefficient.

The largest permanent bound over the sixteen integer determinants, before multiplication by the readout clearers, is
$44889117071728277716992<2^{76}$. The two proved primes
$4611686018427387847$ and $4611686018427387817$ have product
\begin{gather}
 Q_{\rm CRT}=21267647932558653302378126310941659999,\notag \\
 Q_{\rm CRT}>2^{123}>2B.  \label{eq:v10-support-crt-size}
\end{gather}
The full interpolation grids and the exact coefficient bound thus fix
these eight physical response identities deterministically. Direct
fraction-free elimination treats the other completed reduced cases,
including $K_4$ under birth-associated introduction. Coefficient positivity is checked on the reduced numerator itself, because dividing out a common polynomial factor can create negative coefficients, as in $(t^3+1)/(t+1)=t^2-t+1$.

The independent bordered construction covers $K_4$ with a pendant,
$K_{1,1,3}$, the house, the gem and the five-cycle with two crossing
chords, under both rules. Here the full homogeneous matrices are constructed over $\mathbb Z[w]$ before any contact is fixed. Coprime integer-polynomial denominators have positive leading coefficient; their least common multiples, including integer content, clear each augmented pair row and the combined readout $q_*k-p_*j$. Let $D_{\rm tot}$ be the sum of the homogeneous row degrees of the bordered matrix. Use the assignment bounds in Eq.~\eqref{eq:v10-determinant-degree-bound}, replacing $d_e$ by $\min(d_e,D_{\rm tot})$, and fix the first contact of maximal $d_e$ to one. Interpolate every other contact on $1,\ldots,d_e+1$, with none left symbolic. The deleted exponent is $D_{\rm tot}$ minus the sum of the retained exponents. For the supports in the stated order, the fixed edges are $14,03,02,02,01$ and the total degrees are $44,50,48,51,52$, under both rules. This construction uses four primes
$2147483647$, $2147483629$, $2147483587$, $2147483579$.
Their product exceeds $2^{123}$, while an integer coefficient bound
$B=2^{107}$ suffices for every case, as the row-norm bounds in
Table~\ref{tab:I-determinant-bounds} show. Exact reconstruction gives
nonnegative coefficients with at least one positive coefficient for
$K-2J$ on the first two supports and for $2K-3J$ on the last three.
The inequality then follows from Eq.~\eqref{eq:v10-support-bordered-sign}; the reconstruction determines $\mathcal P$ exactly, and the signs of its coefficients give the sign of $q_*K-p_*J$ on the positive orthant.

\begin{table*}[t]
\caption{Finite reconstruction bounds for the independent bordered
calculation in Appendix~\ref{app:v10-support-identities}.
The degree vector lists valid individual-degree bounds in the
edge order of Table~\ref{tab:I-support-edges} after removal of the fixed edge; $D_{\rm tot}$ is the
homogeneous degree before that removal. With
$L_{ij}=\|(\mathsf C_{\rm bord})_{ij}\|_1$, the integer row-norm
product satisfies $\prod_i\sum_jL_{ij}<2^{b_{\rm row}}$ and bounds
every coefficient. These bounds apply under both introduction rules.
For $K_{1,1,3}$ the last degree bound is 20 under birth-associated
introduction and 21 under uniform introduction. The final columns count
positive coefficients; all other coefficients vanish.}
\label{tab:I-determinant-bounds}
\begin{ruledtabular}
\begin{tabular}{lcccrrr}
Support&Fixed edge&$D_{\rm tot}$&Retained degree bounds&$b_{\rm row}$&Uniform&Birth-associated\\
\hline
$K_4$ with a pendant&14&44&$(12,22,22,22,23,23)$&102&7\,193\,269&7\,090\,999\\
$K_{1,1,3}$&03&50&$(21,21,21,21,21,21)$&101&7\,580\,916&7\,456\,656\\
House&02&48&$(20,20,20,18,20)$&90&706\,616&702\,536\\
Gem&02&51&$(21,20,22,20,21,23)$&105&9\,960\,967&9\,877\,544\\
$C_5$ with crossing chords&01&52&$(22,22,20,20,22,22)$&107&13\,270\,341&13\,217\,491
\end{tabular}
\end{ruledtabular}
\end{table*}

The matching upper limits follow from three connected spanning supports. The house, gem and cycle with crossing chords contain a spanning $C_5$, whose hierarchy in Eq.~\eqref{eq:c5-weights} approaches $3/2$ under both introduction rules. Every connected five-site support containing $C_4$ has a spanning $C_4$ with a pendant, obtained by retaining one edge from the fifth vertex to the cycle. The explicit hierarchy below approaches two on this subgraph under both rules. The bull, the triangle with a two-edge tail and the butterfly contain a spanning $P_5$, whose uniform hierarchy approaches seven. Choose a sufficiently advanced finite hierarchy point, so that the retained support is connected and $J>0$. Appendix~\ref{app:v7-support-extension} therefore permits all remaining edges to be given sufficiently small positive weights while preserving $J>0$ and changing $R$ by an arbitrarily small amount. Choosing the hierarchy point before the added contacts proves the matching infimum even when the response amplitude vanishes along the hierarchy. The $C_4$ hierarchy gives the four-site upper value three in the same way. Combined with the strict lower bounds, these constructions prove the named support infima and their nonattainment at finite positive weights. For the four remaining dense supports, $W_4$, $K_5$ less a two-edge
path, $K_5$ less an edge and $K_5$, the same spanning-cycle construction
gives the upper value $3/2$. Their matching lower bound remains Conjecture~\ref{conj:v10-dense}. Appendix~\ref{app:I-dense-bulk} identifies its proved finite-contrast region.

The value two has a direct mechanism. On $C_4$ with a pendant, the
hierarchy $w_{01}=q^8$, $w_{13}=q^2$, $w_{14}=q^6$, $w_{23}=1$,
$w_{24}=q^3$ protects the pins $2$ and $3$. Site $4$ follows pin $2$,
whereas site $1$ and the pendant $0$ form a strongly coupled pair contested
by pin $3$ and by site $4$. Pin $3$ thus agrees with its neighbor with
probability $1/2$, and pin $2$ with probability one. Pin $3$ overwrites pin
$2$ about $q$ times more often than the reverse, so either contested state
almost always ends with pin $3$'s trait, and selection enters through the
ratio of the two pin fecundities. Summing both contested states gives
$\mathcal S=2(b-2c)/(5q)+O(q^{-2})$, hence $R\to2$ with a vanishing
amplitude, as on the five-site path. The exact ratio along this hierarchy
is $R=2+5u^2+17u^3+O(u^4)$ with $u=1/q$, and
$q\chi_b=2/5-4u/5+O(u^2)$, $q\chi_c=4/5-8u/5+O(u^2)$. Under birth-associated
introduction $R^{\mathrm b}=2+34u/7+O(u^2)$, with amplitudes of order
$q^{-3}$.
\end{proof}

\subsection{An exact dense comparison with two exchangeable groups}
\label{app:v10-dense-block}
The following family tests whether strong internal contact alone can promote on a dense support, and proves Eq.~\eqref{eq:v10-dense-block}. Partition five sites into $\mathcal G_2$ and $\mathcal G_3$ with sizes two and
three, with internal contacts $w_{\mathcal G_2},w_{\mathcal G_3}>0$ and common
cross contact $w_\times>0$. Normalize the cross contact to one and set
$a=w_{\mathcal G_2}/w_\times$, $d=w_{\mathcal G_3}/w_\times$.
All contractions in this subsection refer to this normalized graph.
The contractions $J,K$ on the original graph are the corresponding
normalized contractions divided by $w_\times$; $Z$ has the same scaling,
so the physical susceptibilities are unchanged. Put
\begin{gather}
s_{\mathcal G_2}=a+3,\quad s_{\mathcal G_3}=2d+2,\\
x=s_{\mathcal G_2}^{-1},\quad y=s_{\mathcal G_3}^{-1},\quad\alpha=ax,\quad\beta=dy,\quad\gamma_\times=xy. \notag
\end{gather}
For every $a,d\ge0$ and every common selection-independent singleton
source, $J^\mu<0$. If $a,d\to\infty$ at fixed positive ratio, the physical
benefit susceptibility tends to $-1/5$ under both named introduction rules.
\begin{proof}
Permutation symmetry makes the ancestral pair problem three dimensional,
with three pair orbits, within $\mathcal G_2$, within $\mathcal G_3$, and between groups. Their integrated disagreements are $D_{\rm pair},D_{\rm triple},D_\times$, and their pair sources are labeled in the same way. Dividing each same-group row by two gives
\begin{equation}
\begin{gathered}
\mathsf M_*\begin{pmatrix}D_{\rm pair}\\D_{\rm triple}\\D_\times\end{pmatrix}
=\begin{pmatrix}g_{\rm pair}/2\\g_{\rm triple}/2\\g_\times\end{pmatrix},\\
\mathsf M_*=\begin{pmatrix}
\alpha+3y&0&-3y\\0&\beta+2x&-2x\\-x&-2y&3y+2x
\end{pmatrix}.
\end{gathered}
\label{eq:v10-block-pair}
\end{equation}
The two independent preparation channels, uniform within $\mathcal G_2$ or
$\mathcal G_3$, define $J_{\mathcal G_2},J_{\mathcal G_3}$ and
$K_{\mathcal G_2},K_{\mathcal G_3}$. Their right sides are $(5/2,0,5/2)^{\mathsf T}$ and
$(0,5/3,5/3)^{\mathsf T}$. Their contractions follow from
\begin{align}
K/\gamma_\times={}&(\alpha/y)D_{\rm pair}+(3\beta/x)D_{\rm triple} +6D_\times,\\
J/\gamma_\times={}&-
\bigl[(\alpha^2/y+3\alpha-3y)D_{\rm pair}\notag\\
&+3(\beta^2/x+2\beta-2x)D_{\rm triple}\notag\\
&+6(3x+2y-1)D_\times\bigr].
\end{align}
Eliminating the first two rows of Eq.~\eqref{eq:v10-block-pair}, followed
by the last row, gives
\begin{equation}
\frac{J_{\mathcal G_2}}{\gamma_\times}
=-\frac{5(d+1)\Pi_{\rm pair}}{(a+3)\mathscr D_{\rm blk}},\quad
\frac{J_{\mathcal G_3}}{\gamma_\times}
=-\frac{5\Pi_{\rm triple}}{2(d+1)\mathscr D_{\rm blk}},
\label{eq:v10-block-sign}
\end{equation}
where
\begin{align}
\mathscr D_{\rm blk}={}&a^3(6d^2+15d)\notag\\
&+a^2(8d^3+66d^2+159d+20)\notag\\
&+a(56d^3+282d^2+621d+152)\notag\\
&+126d^2+549d+180,\\
\Pi_{\rm pair}={}&a^4(6d^2+21d)\notag\\
&+a^3(8d^3+84d^2+282d+56)\notag\\
&+a^2(56d^3+432d^2+1380d+464)
\notag\\
&+a(492d^2+2334d+1032)\notag\\
&+234d^2+1359d+720,\\
\Pi_{\rm triple}={}&a^4(6d^3+15d^2)\notag\\
&+a^3(8d^4+92d^3+242d^2+78d+28)\notag\\
&+a^2(104d^4+668d^3+1608d^2+818d+260)\notag\\
&+a(208d^4+1444d^3+3750d^2+2218d+676)\notag\\
&+702d^3+2673d^2+1782d+540.
\end{align}
All three polynomials are strictly positive for $a,d\ge0$. Each starting
site within a group has the same response, so source linearity proves
Eq.~\eqref{eq:v10-dense-block} for every singleton introduction distribution.
The conclusion includes zero within-group contacts by continuity, while the
cross contact must remain positive. This excludes weak-selection promotion for this exchangeable
family with the stated copying and payoff kernels.

The adjoint formulation explains why pointwise signs of occupation rewards
are unnecessary. Let $\mathbf j_{\rm blk}$ be the three-vector of coefficients of $J/\gamma_\times$, with entries $j_{{\rm blk},i}$. Put
$p_{\rm blk}=\alpha+3y$, $q_{\rm blk}=\beta+2x$ and $t_{\rm blk}=3y+2x$. Then
\begin{equation}
\begin{gathered}
z_3=\frac{j_{{\rm blk},3}+3yj_{{\rm blk},1}/p_{\rm blk}+2xj_{{\rm blk},2}/q_{\rm blk}}{t_{\rm blk}-3xy/p_{\rm blk}-4xy/q_{\rm blk}},\\
z_1=\frac{j_{{\rm blk},1}+xz_3}{p_{\rm blk}},\quad z_2=\frac{j_{{\rm blk},2}+2yz_3}{q_{\rm blk}}.
\end{gathered}
\end{equation}
The denominator is the positive Schur complement of $\mathsf M_*$, and
$\mathsf M_*^{\mathsf T}z=\mathbf j_{\rm blk}$. The two source contractions are
$J_{\mathcal G_2}/\gamma_\times=5(z_1+z_3)/2$ and
$J_{\mathcal G_3}/\gamma_\times=5(z_2+z_3)/3$. Equation~\eqref{eq:v10-block-sign}
determines the signs of precisely these two combinations.

For $a=\Lambda A_0$, $d=\Lambda D_0$ with $A_0,D_0>0$ fixed,
\begin{equation}
J_{\mathcal G_2}\sim-\frac5{2a},\qquad J_{\mathcal G_3}\sim-\frac5{4d},\qquad
Z\sim\frac2a+\frac3{2d}.
\end{equation}
Both uniform and birth-associated introductions have $m_{\mathcal G_2}\to2/5$.
Consequently their physical benefit susceptibility $2J/(5Z)$ tends to
$-1/5$. The corresponding cost terms are
\begin{equation}
K_{\mathcal G_2}\sim\frac{5(9a+2d)}{a(3a+4d)},\qquad
K_{\mathcal G_3}\sim\frac{5(3a+14d)}{2d(3a+4d)},
\end{equation}
and hence
\begin{equation}
\chi_c\longrightarrow
\frac{2(9A_0^2+78A_0D_0+8D_0^2)}{5(3A_0+4D_0)^2}>0.
\end{equation}
Strong internal contacts suppress the raw contractions while leaving a negative physical benefit susceptibility. An introduced singleton remains disfavored at weak selection.
\end{proof}

\subsection{Leading hierarchy types on five sites}
\label{app:I-leading-hierarchies}
To organize the possible high-contrast limits, we classify the closed ancestral classes of the fastest copying graph. Assume, as after Conjecture~\ref{conj:v10-dense}, a unique largest incident exponent at each site.
The leading ancestral walk then has either one, two or three cold
singletons, or one cold singleton together with a closed pair.
\begin{proof}
Every vertex has one outgoing leading copying arrow. A directed cycle of length at least three would require a strict cyclic chain of inequalities between its consecutive reciprocal edge exponents. Hence every component is a mutual pair with directed trees feeding into it. Each component contains at least two vertices, so five vertices have component sizes $5$ or $3+2$.

The ancestral walk reverses the copying arrows. A leaf of a copying tree has no incoming copying arrow and is a closed cold singleton for ancestry. A mutual pair is a closed lineage class only when its entire copying component has size two. A five-site component has one, two or three leaves; a $3+2$ partition has one cold singleton and a closed pair. These are precisely the four possible types.
\end{proof}
A given support can realize fewer of these types. The classification fixes the closed ancestral classes of the fastest dynamics. Their slower transitions and source-dependent entrance probabilities determine the response. Appendices~\ref{app:v10-leaky-pins}, \ref{app:v10-three-pins} and \ref{app:v10-tied} compute it for fixed positive shapes. Further separation of the connecting rates can change the relevant slow classes; staged elimination is one route to controlling the response in such nested hierarchies.

\subsection{Birth-associated exclusion on the five-site path}
\label{app:v7-preparation}\label{app:v7-birth-polynomials}
The classification depends on the source as well as the copying operator.
For birth-associated introduction, $J^{\mathrm b}<0$ on the five-site path
at every positive weighting, so this path admits no weak promoter.

\begin{proof}
The bordered determinant of Eq.~\eqref{eq:bordered} has the following equivalent adjugate form for birth-associated introduction.
Order the unordered pairs lexicographically, put
$g^{\mathrm b}_{ij}=\vartheta_i+\vartheta_j$, $k_{ij}=\gamma_{ij}$ and
$j_{ij}=-\mathsf B_{ij}$, and form
\begin{equation}
 \frac{U}{V}
 =\frac{k^{\mathsf T}\operatorname{adj}(\mathsf M)g^{\mathrm b}}
        {j^{\mathsf T}\operatorname{adj}(\mathsf M)g^{\mathrm b}}.
 \label{eq:v7-birth-reconstruction}
\end{equation}
Here $\operatorname{adj}(\mathsf M)$ is the adjugate matrix. Clear rational
denominators, remove the common polynomial factor, and choose the
common sign so that $U$ is positive at unit weights. This defines
coprime integer polynomials up to a common positive multiplier.
The identity $UJ^{\mathrm b}=VK^{\mathrm b}$ is an exact consequence
of the reconstruction. For the five-site path, $U$ has positive
coefficients, so $V$ has the sign of $J^{\mathrm b}$ throughout
the positive domain.

For the five-site path with consecutive weights $(A,1,q,Q)$,
the $10\times10$ reconstruction, with $U$ and $V$ of total degree 21, individual degrees at most 8 in $A$ and $Q$ and 14 in $q$, has 728 nonzero coefficients in
$U$, all positive, and 726 in $V$, all negative. Consequently
\begin{equation}
 J^{\mathrm b}(A,1,q,Q)<0\qquad(A,q,Q>0).
 \label{eq:v7-birth-p5-sign}
\end{equation}
These are full polynomial coefficient signs over the three-variable
domain.
\end{proof}
On the hierarchy $(q^2,1,q,q^3)$, the exact susceptibilities
give the more explicit interpretation
\begin{equation}
 \chi_b^{\mathrm b}=-\frac{14}{5q^2}+O(q^{-3}),\qquad
 \chi_c^{\mathrm b}=\frac8{5q}+O(q^{-2}).
 \label{eq:v7-birth-p5-asymptotic}
\end{equation}
The rapidly overwritten dimer receives frequent introductions, while
the protected pin receives few. The negative benefit term in this
hierarchy is consistent with the global sign in
Eq.~\eqref{eq:v7-birth-p5-sign}.

\subsection{Source-independent spectral comparison and finite contact contrast}
\label{app:I-dense-bulk}
We first prove Proposition~\ref{prop:I-source-spectral} and a bound separating stationary and relaxing ancestral modes. We then apply the bound to all ten contacts of $K_5$, obtaining finite-contrast regions of Conjecture~\ref{conj:v10-dense}.
\begin{proof}[Proof of Proposition~\ref{prop:I-source-spectral} and its mode bound]
For this comparison define the diagonal matrix $\mathsf R=\operatorname{diag}(s_i^{-1})$ and the following ancestral matrices,
\begin{equation}
 \begin{aligned}
 \mathsf A_{\rm anc}&=\mathsf R^{-1}L_\gamma
 =\operatorname{diag}(\vartheta)-P^{\T}\\
 &=-N\mathsf G^{\T},\\
 \mathsf A_{\rm s}&=\mathsf R^{-1/2}L_\gamma\mathsf R^{-1/2},\\
 \mathsf P_{\rm s}&=\mathsf R^{1/2}P\mathsf R^{-1/2}.
 \end{aligned}
 \label{eq:I-bulk-operators}
\end{equation}
Thus $-\mathsf A_{\rm anc}/N$ is the backward one-lineage generator, while $\mathsf A_{\rm s}$ is its positive symmetric representative with the clock factor removed. The unit vector $\upsilon_i=\sqrt{\pi_i}$ obeys $\mathsf A_{\rm s}\upsilon=0$ and $\mathsf P_{\rm s}\upsilon=\upsilon$. Split $\mathsf P_{\rm s}=\mathsf P_{\rm sym}+\mathsf P_{\rm skew}$ into its symmetric and skew-symmetric parts. If $0=\lambda_0<\lambda_1\le\cdots\le\lambda_{N-1}$ are the eigenvalues of $\mathsf A_{\rm s}$, define a symmetric filtered matrix in an orthonormal eigenbasis by
\begin{equation}
 (\mathsf T_{\rm filt})_{ab}
 =\frac{\lambda_a-\lambda_b}{\lambda_a+\lambda_b}
 (\mathsf P_{\rm skew})_{ab},\qquad
 (\mathsf T_{\rm filt})_{00}=0.
 \label{eq:I-bulk-filter}
\end{equation}
Indices $a,b$ label ancestral eigenmodes; entries with site indices $i,j$ below refer to the original site basis. The fraction is used away from $(a,b)=(0,0)$. Its modulus is at most one, so a small relaxation rate does not produce a divergent inverse-gap estimate.

Write $\mathfrak g_i=N\mu_i$, a site source whose sum on a pair is $g^\mu_{ij}$ in Eq.~\eqref{eq:v7-pair-preparation}, and let $D^\mu$ be the symmetric disagreement matrix with zero diagonal. Put $\nu_i=\sum_jP_{ji}D^\mu_{ij}$. The full pair equation is
\begin{equation}
 \mathsf A_{\rm anc}D^\mu+D^\mu\mathsf A_{\rm anc}^{\T}
 =\mathfrak g\mathbf1^{\T}+\mathbf1\mathfrak g^{\T}
       -2\operatorname{diag}(\mathfrak g+\nu).
 \label{eq:I-bulk-pair-matrix}
\end{equation}
For any symmetric $X$ satisfying $X\mathbf1=0$ and $\mathsf A_{\rm anc}^{\T}X+X\mathsf A_{\rm anc}=C$, tracing this identity against $X$ gives
$-\tfrac12\operatorname{Tr}(CD^\mu)=\sum_i(\mathfrak g_i+\nu_i)X_{ii}$.
For the cost and benefit matrices $C=L_\gamma,\mathsf B$, respectively, the two solutions are
\begin{equation}
 \begin{aligned}
 \mathsf R^{-1/2}X_K\mathsf R^{-1/2}
   &=\tfrac12(I-\upsilon\upsilon^{\T}),\\
 \mathsf R^{-1/2}X_J\mathsf R^{-1/2}
   &=\tfrac12(\mathsf P_{\rm sym}+\mathsf T_{\rm filt}
                 -\upsilon\upsilon^{\T}).
 \end{aligned}
 \label{eq:I-bulk-readout-solutions}
\end{equation}
Indeed, Eq.~\eqref{eq:I-bulk-filter} solves
$\mathsf A_{\rm s}\mathsf T_{\rm filt}+\mathsf T_{\rm filt}\mathsf A_{\rm s}
 =[\mathsf A_{\rm s},\mathsf P_{\rm skew}]$. The map $X\mapsto\mathsf A_{\rm s}X+X\mathsf A_{\rm s}$ is invertible away from the $(0,0)$ mode, so $\mathsf T_{\rm filt}$ is the unique solution with $\upsilon^{\T}\mathsf T_{\rm filt}\upsilon=0$ and does not depend on the eigenbasis chosen within a degenerate eigenspace.
Also $\mathsf T_{\rm filt}\upsilon=\mathsf P_{\rm skew}\upsilon$, which ensures $X_J\mathbf1=0$. Since $D^\mu\ge0$ is an integrated disagreement probability, the weights $\omega_i=s_i^{-1}(\mathfrak g_i+\nu_i)$ are nonnegative. Using $P_{ii}=0$ gives
\begin{equation}
 \begin{aligned}
 K^\mu&=\tfrac12\sum_i\omega_i(1-\pi_i),\\
 J^\mu&=\tfrac12\sum_i\omega_i
                 [(\mathsf T_{\rm filt})_{ii}-\pi_i].
 \end{aligned}
 \label{eq:I-bulk-local-response}
\end{equation}
Since $\mu$ is a probability distribution, at least one $\omega_i$ is positive; $0<\pi_i<1$ then gives $K^\mu>0$. The introduction rule enters solely through these nonnegative weights.

For equal strengths, $P$ is symmetric, $\pi_i=1/N$ and $\mathsf P_{\rm skew}=0$. Equation~\eqref{eq:I-bulk-local-response} then gives the exact source-independent identity
\begin{equation}
 \frac{J^\mu}{K^\mu}=-\frac1{N-1}.
 \label{eq:I-isothermal-all-source}
\end{equation}
This recovers the isothermal obstruction~\cite{AllenLippnerNowak2019} for every common singleton introduction distribution. On an unequal-strength graph the filtered skew part is precisely the contribution that can overcome this adverse response.

Set $\overline\pi=P^{\T}\pi$. Terms in Eq.~\eqref{eq:I-bulk-filter} with precisely one zero-mode index contribute
$2\upsilon_i(\mathsf P_{\rm skew}\upsilon)_i=\pi_i-\overline\pi_i$ to its diagonal. If $\mathsf T_\perp$ contains the positive-mode terms, then
$(\mathsf T_{\rm filt})_{ii}-\pi_i=(\mathsf T_\perp)_{ii}-\overline\pi_i$.
Thus the stationary contribution is always adverse to the benefit response. To bound the other contribution, put
\begin{equation}
 \eta_{\rm sp}=\frac{\lambda_{N-1}-\lambda_1}
                    {\lambda_{N-1}+\lambda_1},\qquad
 \varsigma_{\rm sk}=\|\mathsf P_{\rm skew}\|_F,
 \label{eq:I-bulk-mode-bounds}
\end{equation}
where $\|X\|_F^2=\sum_{ij}|X_{ij}|^2$ is the Frobenius norm. In the positive-mode basis, the coordinates $u_a$ of the site vector satisfy $\sum_{a>0}u_a^2=1-\pi_i$. The filtered matrix has zero diagonal in this basis. Cauchy--Schwarz and $\sum_{a>0}u_a^4\ge(1-\pi_i)^2/(N-1)$ therefore give
\begin{equation}
 \begin{aligned}
 |(\mathsf T_\perp)_{ii}|
 &\le\eta_{\rm sp}\varsigma_{\rm sk}
       \sqrt{(1-\pi_i)^2-\sum_{a>0}u_a^4}\\
 &\le\sqrt{\frac{N-2}{N-1}}\eta_{\rm sp}\varsigma_{\rm sk}(1-\pi_i).
 \end{aligned}
 \label{eq:I-bulk-filter-bound}
\end{equation}
Here projection onto the positive modes cannot increase the Frobenius norm. Since $K^\mu>0$, Eq.~\eqref{eq:I-bulk-local-response} implies the general comparison
\begin{equation}
 \frac{J^\mu}{K^\mu}
 \le\sqrt{\frac{N-2}{N-1}}\eta_{\rm sp}\varsigma_{\rm sk}
       -\min_i\frac{\overline\pi_i}{1-\pi_i}.
 \label{eq:I-bulk-general-bound}
\end{equation}
It separates strength imbalance, which enters the skew part of the copying matrix, from ancestral relaxation and the adverse stationary mode.
\end{proof}

\paragraph*{A rational sufficient region on $K_5$.}
On $K_5$ with all ten contacts positive, averaged payoffs and the common kernel of Proposition~\ref{prop:I-source-spectral}, let $H$ be the ratio of largest to smallest contact weight. For every common selection-independent singleton source, $H\le7/4$ implies $J^\mu<0$, and $H\le11/3$ implies $2K^\mu-3J^\mu>0$, including both upper endpoints.
\begin{proof}
Normalize the smallest contact to one and let $H$ be the largest, so $1\le w_{ij}\le H$ for all ten pairs. The quadratic form
\begin{equation}
 z^{\T}\mathsf A_{\rm s}z
 =\sum_{i<j}w_{ij}
       \left(\frac{z_i}{\sqrt{s_j}}-\frac{z_j}{\sqrt{s_i}}\right)^2
 \label{eq:I-bulk-form}
\end{equation}
lies between the unit-contact sum $Zz^{\T}(I-\upsilon\upsilon^{\T})z$ and $H$ times that sum. This gives the matrix inequality
\begin{equation}
 Z(I-\upsilon\upsilon^{\T})\le\mathsf A_{\rm s}
 \le HZ(I-\upsilon\upsilon^{\T}).
 \label{eq:I-bulk-spectrum-comparison}
\end{equation} Consequently $\eta_{\rm sp}\le(H-1)/(H+1)$. Also $P_{ji}\ge1/(4H)$, so $\overline\pi_i\ge(1-\pi_i)/(4H)$.
For the remaining norm, direct summation gives
\begin{equation}
 \varsigma_{\rm sk}^2
 =\frac12\sum_{i<j}P_{ij}P_{ji}
       \left(\frac{s_i}{s_j}-\frac{s_j}{s_i}\right)^2.
 \label{eq:I-bulk-skew-norm}
\end{equation}
For one edge put $x=w_{ij}$, $a_0=s_i-x$ and $b_0=s_j-x$. Its contribution is
\begin{equation}
 h(x,a_0,b_0)=
 \frac{x^2(a_0-b_0)^2(2x+a_0+b_0)^2}
 {2(x+a_0)^3(x+b_0)^3}.
 \label{eq:I-bulk-edge-norm}
\end{equation}
At fixed $x$, with $p=x+a_0\le q_0=x+b_0$, this decreases with $p$ and increases with $q_0$. The logarithmic derivatives are
$-(3q_0^2+p^2)/[p(q_0^2-p^2)]$ and
$(q_0^2+3p^2)/[q_0(q_0^2-p^2)]$; equality of the strengths follows by continuity. Since $3\le a_0,b_0\le3H$, the contribution is bounded by $h(x,3,3H)$.

Its $x$ derivative has the sign of
 \begin{align}
 p_H(x)&=-4x^3-2A_Hx^2+(8B_H-A_H^2)x+2A_HB_H,\notag \\
 A_H&=3(H+1),\qquad B_H=9H. \label{eq:I-bulk-edge-derivative}
 \end{align}
The second derivative $p_H''=-24x-4A_H$ is negative. Hence its minimum on $[1,H]$ occurs at an endpoint. Those values are $45H^2+102H-19$ and $H(-19H^2+102H+45)$, both positive for $1\le H\le11/3$. The edge norm thus increases up to $x=H$ throughout this interval. Summing ten contributions gives
\begin{equation}
 \varsigma_{\rm sk}^2\le
 \frac{45(H-1)^2(5H+3)^2}{64H(H+3)^3}.
 \label{eq:I-bulk-skew-bound}
\end{equation}
Combining these bounds with Eq.~\eqref{eq:I-bulk-general-bound} proves
 \begin{align}
 \frac{J^\mu}{K^\mu}&\le\mathcal C_{\rm ctr}(H)-\frac1{4H},\qquad 1\le H\le\frac{11}3,  \label{eq:I-dense-contrast-bound}\\
 \mathcal C_{\rm ctr}(H)^2&=
 \frac{135(H-1)^4(5H+3)^2}{256H(H+1)^2(H+3)^3},
 \qquad\mathcal C_{\rm ctr}(H)\ge0. \notag 
 \end{align}
For either proposed endpoint $H_0$, every weighting of contrast at most $H_0$ satisfies $1\le w_{ij}\le H_0$ after normalization. It therefore suffices to insert $H_0$ in Eq.~\eqref{eq:I-dense-contrast-bound}, using the rational comparisons
\begin{equation}
 \begin{aligned}
 \mathcal C_{\rm ctr}(7/4)^2&=\frac{24155415}{1487250688}<\frac1{49},\\
 \mathcal C_{\rm ctr}(11/3)^2&=\frac{6912}{13475}<\left(\frac{97}{132}\right)^2.
 \end{aligned}
 \label{eq:I-dense-rational-endpoints}
\end{equation}
Since $1/(4H)=1/7$ at the first endpoint and $97/132=2/3+3/44$ at the second, $J^\mu<0$ for $H\le7/4$ and $2K^\mu-3J^\mu>0$ for $H\le11/3$. At $H=1$ the bound becomes the exact equal-strength identity $J^\mu/K^\mu=-1/4$. Strict rational slack makes both stated ranges closed at their upper endpoints.
\end{proof}
The rational comparisons in Eq.~\eqref{eq:I-dense-rational-endpoints} give sufficient contrast ranges, without optimizing their endpoints. Completing a missing contact by weight $\varepsilon\to0$ sends $H$ to infinity, beyond these estimates. The remaining conjecture requires control of $J^\mu/K^\mu$ at arbitrarily high contrast for the two named introduction rules; the fixed-shape reductions of Appendices~\ref{app:v10-leaky-pins}--\ref{app:v10-tied} treat parts of this boundary.

\paragraph*{Finite tests of the dense conjecture.}
\label{app:I-dense-finite-tests}
The sampled contact sets in Sec.~\ref{sec:five-site-supports} were
constructed as follows. In this paragraph vertices are $0,\ldots,4$
and the ten pairs are ordered lexicographically. One continuous
Python \texttt{random.Random(503819)} stream generated all three
rational families; integer draws were independent and uniform on the
stated inclusive ranges. The first 100 sets have $w_{ij}=10^{k_{ij}}$
with $k_{ij}\in[-8,8]$. For each of the next 60, first draw
$m\in[2,6]$ and put $\varepsilon_s=10^{-m}$. Then draw nine factors
$\zeta_\ell=(n_\ell/d_\ell)10^{a_\ell-b_\ell}$, in increasing $\ell$,
with successive draws $n_\ell,d_\ell\in[1,9]$ and
$a_\ell,b_\ell\in[0,4]$. Assign $w_{12}=1$,
$(w_{01},w_{02},w_{34})=\varepsilon_s(\zeta_1,\zeta_2,\zeta_3)$,
$(w_{03},w_{04})=\varepsilon_s^2(\zeta_4,\zeta_5)$ and
$(w_{13},w_{14},w_{23},w_{24})=\varepsilon_s(\zeta_6,\zeta_7,\zeta_8,\zeta_9)$.
For each of the last 40, draw $q_1,q_2$ in that order as $n10^a$,
with $n\in[1,9]$, $a\in[1,5]$, followed by $Q_1,Q_2$ with
$n\in[1,9]$, $a\in[5,10]$. Draw ten background weights
$10^{-m}$, $m\in[6,12]$, in pair order, then overwrite the cycle
$01,12,23,34,04$ by $1,q_1,Q_1,Q_2,q_2$.
The five overwritten draws are retained in the random stream.
The second and third families follow the one-pin and five-cycle contact patterns, but their finite random factors can make nominally separate hierarchy scales comparable. Individual draws can therefore lie outside the asymptotic regimes of those reductions. Exact rational pair solves give strictly positive
$2K^\mu-3J^\mu$ for all 200 sets under both named sources.

The floating-point search used NumPy's \texttt{default\_rng(8102026)}
and 3000 initial logweight vectors $x_{ij}=k_{ij}\log2$,
with independent uniform integer exponents $k_{ij}\in[-10,10]$,
evaluated after normalization
$w_{ij}=\exp(x_{ij}-\max x)$. It minimized the scale-invariant gap
$2-3J^\mu/K^\mu$ for the uniform and birth-associated sources and
the smallest such gap among singleton sites. Every evaluation updated
all three best vectors. Starting from the current best
vector for each target, in sitewise, uniform and birth-associated
order, coordinate refinement used logweight steps
$1,0.4,0.15,0.05,0.015$, at most six sweeps per step, a fresh random
coordinate order per sweep, and negative then positive proposals.
Each proposal was centered by subtracting its mean; proposals with
logweight range above 22 were skipped. Strict improvements were
accepted, and a sweep without improvement ended that step.
A 45-second cutoff was checked after each step and target and was
not reached. The run made 3993 evaluations, including adaptive
trials. No negative gap was found; the smallest named-source gaps
were approximately $0.01075$ and $0.02614$.
The two distinct stored minimizing weight vectors also have positive
gaps under both named sources when their decimal entries are treated
as exact rationals. Independent rational solves checked all 202
vectors, their pair-equation residuals and the source sum rule,
Eq.~\eqref{eq:v10-source-sum}. A proof of the unrestricted conjecture must
also cover the contact weights between these samples and arbitrary
separations of their relative scales.

\section{Payoff, encounter and support extensions}
\label{app:v7-robustness}

The response contraction and the domain reduction test different aspects
of the model. The former determines how a general payoff matrix or a
changed encounter kernel modifies weak selection; the latter identifies
which enlarged supports retain protected competitors. This appendix
derives the extensions in Sec.~\ref{sec:v7-robustness}, together with the birth-associated six-site expansions of Sec.~\ref{sec:v10-preparation-comparison}, and states which
finite-selection conclusions follow from each change.

\subsection{Paired response for a general two-strategy game}
\label{app:v7-general-game}
For a two-strategy game symmetric under exchanging player roles, the row player's matrix is
$\left(\begin{smallmatrix}a_{CC}&a_{CD}\\a_{DC}&a_{DD}\end{smallmatrix}\right)$.
With the same averaged kernel $P$, its payoff at site $i$ is
 \begin{align}
 f_i={}&a_{DD}+(a_{CD}-a_{DD})n_i
 +(a_{DC}-a_{DD})(Pn)_i \notag\\
 &+\Gamma_{\rm game} n_i(Pn)_i, \notag \\
 \Gamma_{\rm game}={}&a_{CC}-a_{CD}-a_{DC}+a_{DD}.  \label{eq:v7-game-payoff}
 \end{align}
The constant term cancels from the first drift correction of $\Phi$.
The last term is quadratic, but the paired comparison reduces its
response to the same pair observables as the donation game.
To see the cancellation, put $x=n_i$, $y=n_j$,
$m=(Pn)_i$ and $z=(Pn)_j$, and write $f(x,m)$ for
Eq.~\eqref{eq:v7-game-payoff}. Direct expansion gives
\begin{equation}
 \begin{aligned}
 &(x-y)[f(x,m)-f(y,z)]\\
 &\quad +(y-x)[f(1-x,1-m)-f(1-y,1-z)]\\
 &=(a_{CC}+a_{CD}-a_{DC}-a_{DD})(x-y)^2\\
 &\quad +(a_{CC}-a_{CD}+a_{DC}-a_{DD})(x-y)(m-z).
 \end{aligned}
 \label{eq:v7-game-cancellation}
\end{equation}
Integrating this identity over the neutral process and using any
common selection-independent singleton preparation proves
Eq.~\eqref{eq:v7-general-game}, with that preparation's $K,J$.
Since $K-J>0$, the weak criterion takes the
structure-coefficient form~\cite{TarnitaEtAl2009},
\begin{equation}
 \sigma_{\rm str}(a_{CC}-a_{DD})>a_{DC}-a_{CD},\qquad
 \sigma_{\rm str}=\frac{K+J}{K-J}.
 \label{eq:v7-structure-coefficient}
\end{equation}
For donation promoters,
$\sigma_{\rm str}=(R+1)/(R-1)$. Its monotonic dependence on $R$
converts the sharp uniform support infima into the suprema stated
in Sec.~\ref{sec:v7-robustness}. When $a_{CC}>a_{DD}$, these
determine whether a contact weighting can overcome a prescribed
off-diagonal disadvantage. In particular, a six-site path can
overcome every fixed finite value of $a_{DC}-a_{CD}$.

At fixed selection, keep all payoffs fixed and the fecundities
positive and bounded away from zero. In the two-domain path limit,
the cooperative and defective aligned domains have fecundities
$\phi(\delta a_{CC})$ and $\phi(\delta a_{DD})$. The bounded fast
inverse preserves the two pin entrance probabilities $1/N$.
Consequently
\begin{equation}
 \Delta_\infty=\frac2N
 \frac{\phi(\delta a_{CC})-\phi(\delta a_{DD})}
      {\phi(\delta a_{CC})+\phi(\delta a_{DD})}.
 \label{eq:v7-general-game-domain}
\end{equation}
For strictly increasing fecundity and positive selection, its sign
is the diagonal payoff difference. The off-diagonal entries control finite-hierarchy
establishment corrections. Equation~\eqref{eq:v7-game-cancellation} removes $\Gamma_{\rm game}$ from the paired first-order response. Finite-selection calculations retain the mixed-configuration term $\Gamma_{\rm game}n_i(Pn)_i$. The wall and dimer establishment calculations must therefore use the general-game payoffs of Eq.~\eqref{eq:v7-game-payoff}; the aligned-domain result~\eqref{eq:v7-general-game-domain} already applies to this game.

\subsection{Distinct interaction and replacement kernels}
\label{app:v7-encounter-kernel}
Let $Q_{\rm int}$ be a nonnegative row-stochastic interaction
kernel, while reciprocal replacement continues through $P$.
The payoff is $-cn_i+b(Q_{\rm int}n)_i$. Neutral ancestry and $K$
are unchanged, and the benefit contraction becomes
\begin{equation}
 J_{\rm int}=-\sum_{i<j}
 \left[\frac{Q_{\rm int}^{\mathsf T}L_\gamma+
                   L_\gamma Q_{\rm int}}2\right]_{ij}D_{ij}.
 \label{eq:v7-separate-benefit}
\end{equation}
When $J_{\rm int}>0$, set $R_{\rm int}=K/J_{\rm int}$ for the weak threshold. If $J_{\rm int}\le0$, no positive benefit--cost ratio favors cooperation at first order.
On a replacement interface directed from C to D, the payoff
difference is at most $b-c$. Thus the reproductive-value drift
is adverse for $b<c$ and nonpositive at $b=c$. If interactions
are loopless, a cooperative singleton has payoff $-c$ whereas
every defective site has nonnegative payoff. Every mixed state
can reach a cooperative singleton before absorption, so the
integrated drift is strictly adverse also at $b=c$. The strict
weak barrier $R_{\rm int}>1$ therefore survives on the positive-benefit domain.
Allowing self-interaction preserves the non-strict barrier; for
example, $Q_{\rm int}=I$ gives exact neutrality at $b=c$.

For a quantitative finite-weight bound, put $\mathsf E_{\rm int}=Q_{\rm int}-P$ and
\begin{equation}
 \omega_{\rm int}(\mathsf E_{\rm int})=\max_{(i,j):\,w_{ij}>0}\frac12\sum_k|(\mathsf E_{\rm int})_{ik}-(\mathsf E_{\rm int})_{jk}|.
 \label{eq:v7-encounter-oscillation}
\end{equation}
Every row of $\mathsf E_{\rm int}$ sums to zero. Hence, for binary $n$,
$|(\mathsf E_{\rm int} n)_i-(\mathsf E_{\rm int} n)_j|\le\omega_{\rm int}(\mathsf E_{\rm int})$ on every replacement edge, and
\begin{equation}
 |(\mathsf E_{\rm int} n)^{\mathsf T}L_\gamma n|
 \le\omega_{\rm int}(\mathsf E_{\rm int})n^{\mathsf T}L_\gamma n.
 \label{eq:v7-encounter-form-bound}
\end{equation}
Integration gives
\begin{equation}
 \left|\frac{J_{\rm int}}K-\frac JK\right|
 \le\omega_{\rm int}(\mathsf E_{\rm int})
 \le2\max_i\operatorname{TV}((Q_{\rm int})_i,P_i),
 \label{eq:v7-encounter-stability}
\end{equation}
where $\operatorname{TV}(p,q)=\tfrac12\sum_k|p_k-q_k|$ is
total-variation distance between two probability rows. Thus a
finite promotion margin at $r>R$ survives whenever twice the
maximum row distance is less than $1/R-1/r$. The version in
Eq.~\eqref{eq:v7-encounter-stability} remains meaningful if the
changed benefit coefficient is nonpositive.

For $Q_{\rm int}^{\alpha}=(1-\alpha)P+\alpha Q^{\rm unif}$,
$0\le\alpha\le1$, let $Q^{\rm unif}_{ij}=(1-\delta_{ij})/(N-1)$ as in
Sec.~\ref{sec:v7-robustness}. The identity
$(Q^{\rm unif}n)_i-(Q^{\rm unif}n)_j=-(n_i-n_j)/(N-1)$ gives
\begin{equation}
 J_\alpha=(1-\alpha)J-\frac{\alpha K}{N-1}.
 \label{eq:v7-encounter-affine}
\end{equation}
This proves Eq.~\eqref{eq:v7-encounter-map-main}. The map is
strictly increasing in $R$ on the positive-benefit domain, so
each sharp support infimum $R_*$ transforms to
\begin{equation}
 R_{*,\alpha}=
 \frac{R_*}{1-\alpha-\alpha R_*/(N-1)},
 \quad 0\le\alpha<\frac{N-1}{N-1+R_*}.
 \label{eq:v7-encounter-infimum}
\end{equation}
At and above the endpoint, $J_\alpha\le0$ for every weighting.
There is also a selected-process identity for exponential fecundity.
Writing $M(n)=\sum_i n_i$, the payoff is
\begin{equation}
 f_i=-c_{\rm eff}n_i+b_{\rm eff}(Pn)_i
        +\frac{\alpha b}{N-1}M(n),
 \label{eq:v7-encounter-effective-payoff}
\end{equation}
with $b_{\rm eff}=(1-\alpha)b$ and
$c_{\rm eff}=c+\alpha b/(N-1)$. The final term multiplies all
exponential fecundities by the same factor, which cancels from
normalized source selection. The full generator and its physical
clock are thus unchanged after replacing $b,c$ by the effective
parameters. For linear fecundity, factor $1+\delta\alpha bM(n)/(N-1)$ from every fecundity. The same normalized rates then have the effective payoffs above and the configuration-dependent selection strength
\begin{equation}
 \delta_{\rm eff}(n)=
 \frac{\delta}{1+\delta\alpha bM(n)/(N-1)}.
 \label{eq:I-linear-encounter-intensity}
\end{equation}
The common factor cancels from the global attempt normalization, so this identity retains physical time. It is a state-dependent rate identity; a single constant intensity does not reproduce the linear-fecundity process in general. Its derivative at neutrality recovers the exact weak-response map. On protected-domain supports the aligned payoff difference
is $g_{\rm int}=b[1-N\alpha/(N-1)]-c$, and exponential fecundity gives
$\Delta_\infty=(2/N)\tanh(\delta g_{\rm int}/2)$.

\subsection{Adding contacts and retaining protected domains}
\label{app:v7-support-extension}
For connected supports $G_0,G_1$ on the same vertex set with
$E(G_0)\subset E(G_1)$, the threshold infima obey
$R_*(G_1)\le R_*(G_0)$, with the same inequality for $R_*^{\mathrm b}$.
Consequently, every support containing a spanning path with $N\ge6$ has threshold infimum one under both introduction rules, using the path limits in Eqs.~\eqref{eq:I-long-path-leak-errors} and \eqref{eq:v7-broader-birth-response} and the strict barrier. We then control the protected-domain
mechanism on the enlarged support.
\begin{proof}
Retain a fixed positive weighting on
$G_0$ and give each new edge weight $\varepsilon_{\rm edge}>0$. At $\varepsilon_{\rm edge}=0$ the
retained support is connected, so the killed pair matrix is
invertible. The contractions are analytic there, and every
positive-benefit threshold on $G_0$ is approached from positive
weights on $G_1$. Thus $R_*(G_1)\le R_*(G_0)$.
Together with the universal lower bound, this proves
$R_*(G)=1$ whenever $G$ contains a spanning path with $N\ge6$.
The same continuity argument applies to birth-associated
preparation, whose source is analytic in the replacement kernel.
\end{proof}

To control the contrast while retaining all extra edges, suppose
the vertices partition as
\begin{equation}
 V=\{p_1\}\sqcup F_L\sqcup\{p_2\}\sqcup F_R.
 \label{eq:v7-two-core-partition}
\end{equation}
The pins $p_1,p_2$ are adjacent. Each follower-induced subgraph is
connected, contains at least two vertices and contacts its own pin.
Assign all follower-core edges weight $Q$. Choose nonempty sets of
pin-to-own-follower edges with fixed positive proportions of a total
weight $q$ at each pin. Give every other edge weight one, including
$p_1p_2$. Put $u=1/q$ and $v=q/Q$.

At $u=v=0$, the pins are fixed. Connected follower cores align with
them through finite copying sequences whose probabilities are bounded
below on compact positive-fecundity domains. Every exit from an aligned
configuration has a factor $u$. These facts verify the fast-inverse and
slow-factor assumptions of Appendix~\ref{app:p6-protected}; the remaining
check concerns the extra contacts and their leading transitions.

Direct pin-to-pin copies change the receiving domain. A pin copy
to an opposite follower occurs at order $u$, but leaves both pins
unchanged; fast relaxation returns to the same aligned state.
Copies from a follower to the opposite domain have rate $O(uv)$.
Thus extra edges contribute fast return excursions or higher-order
domain transitions, and the leading two-domain competition remains
the same. With $F_{\rm dom}=\phi(\delta(b-c))$, $d_0=\phi(0)$ and
$h_{\rm dom}=F_{\rm dom}/(F_{\rm dom}+d_0)$, either mixed domain has cooperative exit probability
$h_{\rm dom}$. Uniform preparation enters each such domain state with
probability $1/N+O(u+v)$, so
\begin{equation}
 \rho_C=\frac{2h_{\rm dom}}{N}+O(u+v),\quad
 \rho_D=\frac{2(1-h_{\rm dom})}N+O(u+v).
 \label{eq:v7-broader-domain-probabilities}
\end{equation}
Near neutrality, the bounded inverse controls the differentiated
remainders and gives
$\mathcal S=(b-c)/N+O((u+v)(|b|+|c|))$.
Taking $Q=q^2$ makes $v=u$, so the benefit coefficient is positive eventually and $R=1+O_G(q^{-1})$ on this fixed support.

For birth-associated introduction, the two pin masses are
$u/N+O(u^2+v)$. At $v=0$, a follower cannot overwrite a pin,
so its fixation response vanishes; regularity makes it $O(v)$.
Consequently
\begin{equation}
 \mathcal S^{\mathrm b}
 =\frac{u(b-c)}N+O((u^2+v)(|b|+|c|)).
 \label{eq:v7-broader-birth-response}
\end{equation}
Taking $v=u^2$, or $Q=q^3$, gives a positive benefit coefficient eventually and $R^{\mathrm b}=1+O_G(q^{-1})$. The source-dependent amplitude has been removed before taking this ratio. The
no-go results for paths through four sites hold for every singleton
preparation. Together with Eq.~\eqref{eq:v7-birth-p5-sign}, this
completes the path classification in
Table~\ref{tab:I-introduction-thresholds}. Increasing contact contrast separates copying time scales at fixed $N$ and fixed support. The constants in $O_G(q^{-1})$ depend on the follower cores and their attachment geometry. Growing cores require quantitative alignment and first-passage estimates to determine how large $q$ must be relative to their size (Appendix~\ref{app:path-long}).
\subsubsection{Birth-associated response across source and leakage scales}
\label{app:I-birth-leak}
On $(Q,q,1,q,Q)$ put $u=1/q$ and $v=q/Q$. A fixed-ratio expansion cannot control $v/u\to\infty$. Subtracting the common pin-pair residence gives a system that remains regular throughout $u,v\to0$.
The joint expansion~\eqref{eq:I-birth-joint-corner} holds without fixing
$v/u$. Within this corner, the birth-associated threshold tends to one
exactly when $v/u\to0$; a fixed ratio in $(0,1/2)$ gives
Eq.~\eqref{eq:I-birth-leak-threshold}, while the critical scaling
$v=u/2$ is nonpromoting for sufficiently small $u$.
\begin{proof}
Define $\alpha_{\rm c}=1/(1+v)$, $\varepsilon_{\rm c}=v/(1+v)$,
$p_{\rm out}=1/(1+u)$ and $p_{\rm br}=u/(1+u)$.
For cross-domain pairs set
$\widehat D_{ij}=D_{ij}^{\mathrm b}-D_{34}^{\mathrm b}$ and
$y_{\rm p}=p_{\rm br}D_{34}^{\mathrm b}$.
Reflection reduces the exact pair equations to the nine variables
\begin{equation}
 \bigl(D_{12}^{\mathrm b},D_{13}^{\mathrm b},D_{23}^{\mathrm b},
 \widehat D_{14},\widehat D_{15},\widehat D_{16},
 \widehat D_{24},\widehat D_{25},y_{\rm p}\bigr).
 \label{eq:I-birth-regular-chart}
\end{equation}
The transformed equations contain no divisions by $u$ or $v$. For example, the pin row is exactly
\begin{equation}
 y_{\rm p}=p_{\rm br}+\varepsilon_{\rm c}(1+\widehat D_{24}).
 \label{eq:I-birth-regular-pin}
\end{equation}
At the origin it fixes $y_{\rm p}=0$. Denote the within-domain corner values by $D_{ij}^{[0]}$. Their three rows are
\begin{equation}
 3D_{12}^{[0]}-D_{13}^{[0]}=3,\quad
 D_{13}^{[0]}-D_{23}^{[0]}=1,\quad
 2D_{23}^{[0]}-D_{13}^{[0]}=2.
 \label{eq:I-birth-leak-fast}
\end{equation}
The cross-domain rows next fix $\widehat D_{14}-\widehat D_{24}=1$ and $2\widehat D_{24}-\widehat D_{14}=2$, then the remaining three offsets. The resulting blocks have nonzero determinants $1,3,1,3$, respectively, and the unique corner solution in the order~\eqref{eq:I-birth-regular-chart} is
\begin{equation}
 (7/3,4,3,4,7,8,3,6,0).
 \label{eq:I-birth-corner-values}
\end{equation}
Analytic invertibility therefore controls a full neighborhood of the origin, without restricting $v/u$. In particular,
$y_{\rm p}=u+4v+O((u+v)^2)$.

The exact cost contraction is
$K^{\mathrm b}/u=2v\alpha_{\rm c}(D_{12}^{\mathrm b}+p_{\rm out}D_{23}^{\mathrm b})+p_{\rm out}y_{\rm p}$.
Substituting the corner values in this cost contraction and in the benefit contraction of Eq.~\eqref{eq:JK}, and using the exact normalizer $Z$, gives
\begin{equation}
 \begin{aligned}
 K^{\mathrm b}/u&=u+44v/3+O((u+v)^2),\\
 J^{\mathrm b}/u&=u-2v+O((u+v)^2),\\
 Z/u&=2[v(1+\alpha_{\rm c})+p_{\rm out}].
 \end{aligned}
 \label{eq:I-birth-leak-contractions}
\end{equation}
Since $\chi_b^{\mathrm b}=J^{\mathrm b}/(3Z)$ and $\chi_c^{\mathrm b}=K^{\mathrm b}/(3Z)$, this proves Eq.~\eqref{eq:I-birth-joint-corner}.
For fixed $\lambda_{\rm leak}=v/u$, extend the notation to all pairs by $D_{ij}^{[0]}=\lim_{u\to0}D_{ij}^{\mathrm b}(u,\lambda_{\rm leak}u)$. The cross-domain limits now depend on this fixed ratio. The pin row reads
\begin{equation}
 2(1+\lambda_{\rm leak})D_{34}^{[0]}
 -\lambda_{\rm leak}(D_{24}^{[0]}+D_{35}^{[0]})
 =2(1+\lambda_{\rm leak}),
 \label{eq:I-birth-leak-slow}
\end{equation}
so $D_{34}^{[0]}=1+4\lambda_{\rm leak}$. The joint expansion specializes to
\begin{equation}
 \begin{aligned}
 \chi_b^{\mathrm b}&=\frac u6(1-2\lambda_{\rm leak})+O(u^2),\\
 \chi_c^{\mathrm b}&=\frac u6(1+44\lambda_{\rm leak}/3)+O(u^2),
 \end{aligned}
 \label{eq:I-birth-leak-response}
\end{equation}
which proves Eq.~\eqref{eq:I-birth-leak-threshold} on its positive-benefit domain. If $u/v\to0$, then $\chi_b^{\mathrm b}=-(v/3)[1+o(1)]<0$. Thus every power family $Q=q^a$ with $1<a<2$ is eventually adverse. Within $u,v\to0$, the threshold approaches one exactly when $v/u\to0$. For necessity, any subsequence with $v/u$ bounded away from zero after passage to a further subsequence either has adverse strong leakage or approaches a finite positive ratio. Ratios below $1/2$ give a threshold strictly above one; at $1/2$ the benefit is smaller order than the positive cost; above $1/2$ the benefit is adverse.

At $\lambda_{\rm leak}=1/2$ one more order resolves the vanishing leading benefit. Write $D_{ij}^{\mathrm b}=D_{ij}^{[0]}+uD_{ij}^{[1]}+O(u^2)$. Expanding the pair equations along $v=u/2$ gives
\begin{equation}
 (D_{12}^{[1]},D_{13}^{[1]},D_{23}^{[1]},D_{24}^{[1]},D_{34}^{[1]})
 =(31/9,37/3,26/3,4,2).
 \label{eq:I-birth-leak-critical-pairs}
\end{equation}
The next benefit contraction is
\begin{align}
 J^{\mathrm b}/u^3
 &=D_{24}^{[1]}-\tfrac32D_{23}^{[1]}-\tfrac32D_{12}^{[1]}
       +\tfrac12D_{13}^{[1]}-\tfrac{43}{6}+O(u)\notag\\
 &=-\tfrac{91}{6}+O(u).
 \label{eq:I-birth-leak-critical}
\end{align}
Thus $\chi_b^{\mathrm b}=-91u^2/36+O(u^3)<0$ for sufficiently small $u$. The critical value also excludes promotion in this asymptotic family.
\end{proof}

\subsection{An exact death-first comparison on the four-cycle}
\label{app:I-db-c4}
Here the update rule is dB, with uniform death and fecundity-weighted reproduction among its neighbors, averaged donation payoffs and uniform singleton introduction. The contacts are $(1,q,Q,q)$ as in Eq.~\eqref{eq:I-db-c4-threshold}. The neutral ancestral walk uses $P$ and has stationary weights proportional to $s_i$, which differ from the source-first reproductive values. Put
\begin{equation}
 \begin{gathered}
 r_{\rm d}=\frac q{1+q},\quad s_{\rm d}=\frac q{q+Q},\quad
 k_{\rm d}=r_{\rm d}+s_{\rm d},\\
 \pi^{\dB}=\frac{(s_{\rm d},s_{\rm d},r_{\rm d},r_{\rm d})}{2k_{\rm d}}.
 \end{gathered}
\end{equation}
For $q,Q>0$ with $Q\ne q^2$, the weak threshold is
Eq.~\eqref{eq:I-db-c4-threshold}. At $Q=q^2$ the benefit coefficient
vanishes and the cost coefficient is positive. With $Q=1$, the
threshold tends to one as $q\to0$.
\begin{proof}
Let $t_{ij}^{\dB}$ be the mean number of pair jumps before coalescence, choosing one of the two lineages with probability $1/2$ at each jump and moving it by $P$. It obeys $t_{ii}^{\dB}=0$ and
\begin{equation}
 t_{ij}^{\dB}=1+\frac12\sum_l
 (P_{il}t_{lj}^{\dB}+P_{jl}t_{il}^{\dB}),\qquad i\ne j.
\end{equation}
In global-attempt time each lineage jumps at rate $1/4$, so physical pair times are twice these jump counts. The symmetry leaves four values
$a_{\rm d}=t_{12}^{\dB}$, $d_{\rm d}=t_{34}^{\dB}$,
$z_{\rm d}=t_{13}^{\dB}=t_{24}^{\dB}$ and
$h_{\rm d}=t_{14}^{\dB}=t_{23}^{\dB}$, satisfying
\begin{equation}
 \begin{aligned}
 a_{\rm d}&=1+r_{\rm d}z_{\rm d},\\
 d_{\rm d}&=1+s_{\rm d}z_{\rm d},\\
 h_{\rm d}&=1+\tfrac12(2-k_{\rm d})z_{\rm d},\\
 z_{\rm d}&=1+\tfrac12[(2-k_{\rm d})h_{\rm d}
                       +r_{\rm d}d_{\rm d}+s_{\rm d}a_{\rm d}].
 \end{aligned}
 \label{eq:I-db-four-times}
\end{equation}
Thus $z_{\rm d}=8/[4k_{\rm d}-k_{\rm d}^2-4r_{\rm d}s_{\rm d}]>0$.
Define the stationary $m$-step contractions
$\mathcal T_m^{\dB}=\sum_{ij}\pi_i^{\dB}(P^m)_{ij}t_{ij}^{\dB}$.
The four pair equations give
\begin{equation}
 \begin{aligned}
 \mathcal T_2^{\dB}&=\mathcal T_1^{\dB}-1
 =\frac{2r_{\rm d}s_{\rm d}(2-k_{\rm d})z_{\rm d}}{k_{\rm d}}>0,\\
 \mathcal T_3^{\dB}-\mathcal T_1^{\dB}
 &=(k_{\rm d}-1)^2\mathcal T_2^{\dB}.
 \end{aligned}
 \label{eq:I-db-contractions}
\end{equation}
For the second identity, the three-step minus one-step probability from a site in the $12$ pair to its partner is $r_{\rm d}(2-k_{\rm d})(k_{\rm d}-1)$; the opposite-pair contribution has the reverse change, with the corresponding factor $s_{\rm d}$ for the $34$ pair. The difference of the contractions is consequently proportional to
$a_{\rm d}+d_{\rm d}-2h_{\rm d}=2(k_{\rm d}-1)z_{\rm d}$.
Define $\mathcal S_{\dB}=\partial_\delta(\rho_C^U-\rho_D^U)|_{\delta=0}$ for this update rule. Its ancestral expression~\cite{AllenEtAl2017}, in this pair-jump convention and with $\phi(0)=\phi'(0)=1$, becomes
\begin{equation}
 \mathcal S_{\dB}
 =\frac{\mathcal T_2^{\dB}}4
       \left[b(k_{\rm d}-1)^2-c\right].
 \label{eq:I-db-response}
\end{equation}
Since $\mathcal T_2^{\dB}>0$, the bracket fixes the sign of each coefficient, including the vanishing benefit response at $Q=q^2$. Taking its threshold proves~\eqref{eq:I-db-c4-threshold}. At $Q=1$, the limit $q\to0$ gives $R_{\dB}=[(1+q)/(1-q)]^2\to1$.
\end{proof}
In the limit $Q=1$, $q\to0$, the long-lived units are the dimers $\{1,2\}$ and $\{3,4\}$, maintained by rapid internal copying. No site is protected, since every site dies at the neutral rate $1/4$ on the global-attempt clock, and the persistence is collective.

\labelalias{eq:p5-certificate}{app:path-p5-certificate}
\labelalias{eq:copy-operator}{eq:app-copy-pauli}

\labelalias{eq:access-classification}{tab:I-introduction-thresholds}
\labelalias{eq:v7-preparation-classification}{tab:I-introduction-thresholds}
\labelalias{prop:cn-main}{eq:cN-contested}

\bibliography{references}
\end{document}